\documentclass[10pt]{amsart}
\usepackage[utf8]{inputenc}
\usepackage[a4paper,width=170mm,top=7.5mm,bottom=15mm,left=13mm, right=13mm]{geometry}
\usepackage{etex}
\usepackage[all]{xy} 
\usepackage{cite}
\usepackage{etoolbox} 
\usepackage{lmodern}
\usepackage{array}
\usepackage{xstring}
\usepackage{longtable}
\usepackage[listings]{tcolorbox} 
\tcbuselibrary{breakable}
\tcbuselibrary{skins}
\usepackage{amsmath}
\usepackage{tikz,tikz-cd,stmaryrd}
\usepackage{fancyhdr}
\usepackage{amssymb}
\usepackage{amsthm}
\usepackage{amsmath} 
\usepackage{amsfonts}
\usepackage{amsthm}
\usepackage{pifont}
\usepackage{amsmath}
\usepackage{soul}
\usepackage{epsfig}
\usepackage{latexsym}
\usepackage{amssymb}
\usepackage{color}
\usepackage{amscd}
\usepackage{multirow}
\usepackage{graphicx}
\usepackage{lscape}
\usepackage{varioref}
 \usepackage{tabularx}
\usepackage{float}
\usepackage{xcolor}
\usepackage{bm}
\usepackage{tikz-cd}
 \makeatletter
    \let\stdchapter\section
    \renewcommand*\section{%
    \@ifstar{\starchapter}{\@dblarg\nostarchapter}}
    \newcommand*\starchapter[1]{%
        \stdchapter*{#1}
        \thispagestyle{fancy}
        \markboth{\MakeUppercase{#1}}{}
    }
    \def\nostarchapter[#1]#2{%
        \stdchapter[{#1}]{#2}
        \thispagestyle{fancy}
    }
\makeatother

\newtheorem{theorem}{Theorem}[section]
\newtheorem*{theorem*}{Theorem}
\newtheorem{lemma}[theorem]{Lemma}
\newtheorem{proposition}[theorem]{Proposition}
\theoremstyle{definition}
\newtheorem{definition}[theorem]{Definition}
\theoremstyle{corollary}

\theoremstyle{remark}
\newtheorem{remark}[theorem]{Remark}

\theoremstyle{conclusion}

 \usepackage{acronym}
\usepackage{amsmath}
\usepackage{amssymb}
\usepackage{amsfonts}
\usepackage{esvect}
\usepackage{mathrsfs}
\usepackage{authblk}
\usepackage{geometry}
\usepackage{abstract}
\usepackage{xcolor}

\makeatletter 
\@addtoreset{equation}{section} %\@addtoreset{equation}{subsection} 
\makeatother

\newcommand{\figref}[1]{\figurename~\ref{#1}}
\usepackage[colorlinks=true,urlcolor=blue]{hyperref}

\begin{document}
\thispagestyle{empty}

\begin{center}
	\Large{{\bf Polynomial algebra from the Lie algebra reduction chain $\mathfrak{su}(4) \supset \mathfrak{su}(2)  \times \mathfrak{su}(2)$: \\The supermultiplet model }}
\end{center}
\vskip 0.5cm
\begin{center}
	\textsc{Rutwig Campoamor-Stursberg$^{1,\star}$, Danilo Latini$^{2,*}$, Ian Marquette$^{3,\bullet}$,\\ Junze Zhang$^{4,\dagger}$ and Yao-Zhong Zhang$^{4,\ddagger}$}
\end{center}
\vskip 0.2cm
\begin{center}
	$^1$ Instituto de Matem\'{a}tica Interdisciplinar and Dpto. Geometr\'{i}a y Topolog\'{i}a, UCM, E-28040 Madrid, Spain
\end{center}

\begin{center}
	$^2$ Dipartimento di Matematica “Federigo Enriques”, Università degli Studi di Milano, Via C. Saldini 50, 20133 Milano, Italy \& INFN Sezione di Milano, Via G. Celoria 16, 20133 Milano, Italy
\end{center}

\begin{center}
	$^3$ Department of Mathematical and Physical Sciences, La Trobe University, Bendigo, VIC 3552, Australia
\end{center}

\begin{center}
	$^4$ School of Mathematics and Physics, The University of Queensland, Brisbane, QLD 4072, Australia
\end{center}
\begin{center}
	\footnotesize{$^\star$\textsf{rutwig@ucm.es} \hskip 0.25cm$^*$\textsf{danilo.latini@unimi.it} \hskip 0.25cm $^\bullet$\textsf{i.marquette@uq.edu.au} \hskip 0.25cm 
 $^\dagger$\textsf{junze.zhang@uq.net.au} \hskip 0.25cm 
 $^\ddagger$\textsf{yzz@maths.uq.edu.au}}
\end{center}
\vskip  1cm
%\hrule
%\begin{center}
%	{\bf Abstract}
%\end{center}
\hypersetup{colorlinks=true, linkcolor=black}

\begin{abstract}
\vskip 0.15cm
\noindent  The supermultiplet model, based on the reduction chain $\mathfrak{su}(4) \supset \mathfrak{su}(2) \times \mathfrak{su}(2)$, is revisited through the lens of commutants within universal enveloping algebras of Lie algebras. From this analysis, a collection of twenty polynomials up to degree nine emerges from the commutant associated with the $\mathfrak{su}(2) \times \mathfrak{su}(2)$ subalgebra. This study is conducted in the Poisson (commutative) framework using the Lie-Poisson bracket associated with the dual of the Lie algebra under consideration.  As the main result, we obtain the polynomial Poisson algebra generated by these twenty linearly independent and indecomposable polynomials, with five elements being central. This incorporates polynomial expansions up to degree seventeen in the Lie algebra generators. We further discuss additional algebraic relations among these polynomials, explicitly detailing some of the lower-order ones.  As a byproduct of these results, we also show that the recently introduced ‘grading method' turns out to be essential for deriving the Poisson bracket relations when the degree of the expansions becomes so high that standard approaches are no longer applicable, due to computational limitations.  These findings represent a further step toward the systematic exploration of polynomial algebras relevant to nuclear models.
\end{abstract}
\vskip 0.35cm
\hrule

\tableofcontents

\section{Introduction}
\label{int}

Lie algebras, quadratic algebras and, more generally, higher-order polynomial algebras have been shown to be central objects in many mathematical and physical problems, specifically in topics where techniques of symplectic geometry, (super-)integrable systems and quantum mechanics intersect\cite{MR0365629,MR3493688,MR1806263}. In this context, and specifically referring to the quantum mechanical framework, both generalized Casimir operators and representations are crucial to determine and interpret the spectra of quantum systems \cite{MR2226333,MR2804582,MR2337668, MR1306244,MR2804560}. The structure of polynomial Poisson algebras, which extends classical Poisson algebras by incorporating polynomial functions, plays a significant role in describing the algebraic and geometric properties of various physical models. Their natural connection to Lie algebras, through embedding chains and representations, provides a powerful framework for understanding symmetry, conservation laws, and the dynamics of physical systems.

The mathematical framework provided by embedding chains of Lie algebras into polynomial Poisson algebras is also of particular interest in the study of super-integrable systems and field theories \cite{CampoamorStursberg2021,MR4660510}. In these models, the hierarchical embedding of Lie algebras helps to organize the algebraic structure of conserved quantities and symmetries, leading to a better understanding of the (exact) solvability and stability of the system. Moreover,  such embeddings also contribute to the broader field of mathematical physics by providing insight into the classification of solutions to differential equations and the quantization of classical systems.

 A systematic investigation of polynomial algebras appeared in the study of Lie algebra reduction chains associated with nuclear models in \cite{MR4660510}. In that paper, two relevant chains, namely the chain $\mathfrak{su}(3) \supset \mathfrak{so}(3)$, associated with the Elliott model \cite{elliott1,elliott2}, and the chain $\mathfrak{so}(5) \supset \mathfrak{su}(2) \times \mathfrak{u}(1)$, related to the Seniority model \cite{Helmers61, Racah1965,Hecht65, HS70}, were analyzed from this perspective. Remarkably, the (cubic) polynomial algebras associated with these models were shown to be related to algebraic structures that naturally arise in the study of superintegrable systems \cite{MR1814439}, highlighting their fundamental role across different physical contexts. Developing further along these lines, in this paper, we explore the polynomial Poisson algebra associated with the embedding chain $\mathfrak{su}(4) \supset \mathfrak{su}(2) \times \mathfrak{su}(2)$, related to the supermultiplet model \cite{Wig, MR189602, Hecht, Brunet, Draayer, MR1296410}. By investigating how these  polynomial structures emerge in the classical (commutative) setting, we aim to highlight the deep interplay between this type of algebraic structures and the different physical models. The analysis of this connection not only enriches the theory of polynomial Poisson algebras, but also offers new perspectives on longstanding problems in mathematical physics.

This paper is divided into the following sections: In Section $\ref{sec2}$, we review fundamental aspects of Poisson centralizers and enveloping algebras, notably focusing on polynomial algebras that emerge from linearly independent, indecomposable elements in the commutants.  The discussion extends to detailing the construction of polynomial algebras in dual spaces, the latter being endowed with the Lie-Poisson (or Berezin) brackets. We describe the methods for obtaining polynomial invariants which ultimately define symmetry algebras derived from the analysis of the commutants associated with suitable subalgebras,  also examining the computational efficiency of determining the commutants analytically.  The narrative highlights the challenge of managing polynomial decompositions while maintaining algebraic independence up to a specified degree. The section concludes with a detailed discussion of the role of the grading of generators, in the context of structural brackets,  as introduced in the paper \cite{campoamor2025On}. Emphasis is placed on the importance of implementing such a grading method in the analysis of high-degree expansions, as it significantly facilitates the computations.

In Section $\ref{sec3}$, after introducing the reduction chain of the Lie algebra $\mathfrak{su}(4) \supset \mathfrak{su}(2) \times \mathfrak{su}(2)$ associated with the supermultiplet model, we investigate the construction of homogeneous polynomials derived from the study of the commutant associated with the $\mathfrak{su}(2) \times \mathfrak{su}(2)$ subalgebra. The analysis focuses on the explicit construction of the generators of the polynomial algebra, with a detailed discussion on the methods used to obtain them in order to overcome computational limitations. These limitations arise when finding high-degree polynomial solutions to the system of linear partial differential equations (PDEs) associated with the analysis of the commutant.  Therefore,  key aspects of this inspection include identifying the linearly independent and indecomposable polynomials that are instrumental for closing an algebra.  It turns out that, for this specific problem, it will be sufficient to consider polynomials up to degree nine. Specifically, these steps involve solving the commutant constraints and exploring the Berezin brackets degree by degree. The process culminates in the definition of new polynomials,  defined through Berezin brackets, to ensure the closure of the polynomial algebra. The obtained results are then compared with those in the existing related literature, e.g. \cite{nuclear}, where other methods to obtain the homogeneous polynomials were implemented.

Building upon the generators obtained in Section $\ref{sec3}$, in Section $\ref{sec4}$ we discuss the derivation of the non-trivial Poisson brackets of the polynomial algebra degree by degree. The analysis is divided into detailed subsections, each specific to a particular degree, with particular emphasis placed on the role of gradings in these derivations. Through an organized sequence of propositions and assumptions related to subalgebras and the structural  properties of polynomial algebras,  we meticulously identify the  admissible monomials and elucidate the intricate  relations so obtained. Using the grading  method discussed in Section $\ref{sec3}$, we provide a detailed exposition of the grading process  that is necessary to reduce the number of terms appearing in the expansions of the Poisson brackets, offering comprehensive formulae,  including a comparison of the actual number of terms appearing in the expansions before and after the application of the grading, and  ultimately leading us to the explicit polynomial  relations for different degrees.  We remark once again that this methodical analysis enhances the simplification of calculations through the consistent application of homogeneous gradings, thus leading to the derivation of polynomial  expansions involving the allowable generator sets specifically tailored for designated degree brackets. The proposed grading technique is shown to effectively streamline computational efforts while clarifying potential polynomials appearing in the expansions of the Poisson brackets.  The main result obtained in this section is therefore represented by the complete set of relations defining the polynomial algebra, spanned by the linearly independent and indecomposable polynomials, derived from the analysis of the commutant. Furthermore, existing algebraic relations among these polynomials are explicitly reported up to degree seventeen, again with a consistent comparison to the results already known in the literature.

We conclude in Section $\ref{sec5}$ with a comprehensive summary of our research work, and delineate potential directions for future developments and iterations of this project.

\section{Polynomial algebras from commutants of Lie subalgebras}
\label{sec2}
In this Section \ref{sec2}, we review some fundamental aspects of Poisson centralizers and enveloping algebras (see \cite{CampoamorStursberg2021,moeglin1976factorialite,campoamor2022some,campoamor2024superintegrable,MR1451138}   for additional details). Moreover, we discuss the type of polynomial algebras arising from the linearly independent and indecomposable elements in  the commutants.
\subsection{Commutants in enveloping algebras}
\label{subsec2.1}
 Let $\mathfrak{g}:=\left\{X_1, \dots, X_n \, :\, [X_i,X_j]= \sum_{k=1}^nC_{ij}^k X_k\right\}$ be a $n$ dimensional semisimple or reductive Lie algebra over a field $\mathbb{F}$, and let us denote its universal enveloping algebra by $U(\mathfrak{g})$. Here $[\cdot,\cdot]: \mathfrak{g} \times \mathfrak{g} \rightarrow \mathfrak{g}$ is a commutator and $C_{ij}^k$  are the structure constants of $\mathfrak{g}$. From the Poincar\'e-Birkhoff-Witt (PBW in short) theorem, any element of $U(\mathfrak{g})$ is spanned by the ordered basis elements
\begin{equation}
	\left\{X_1^{a_1}X_2^{a_2}\cdots X_n^{a_n} \, : \, a_1, a_2, \dots, a_n \in \mathbb{N}_{0} : = \mathbb{N} \cup \{0\}\right\} \, .
\end{equation}
If $h$ is a positive integer, let $U_h(\mathfrak{g})$ be the linear space generated by all monomials $X_1^{a_1}\cdots X_n^{a_n}$ such that the inequality $a_1+a_2+\ldots +a_n\leq h$ is satisfied. This allows us to define the degree $\delta$ of an arbitrary element $P\in U(\mathfrak{g})$ as $\delta={\rm inf}\left\{ h : P\in U_h(\mathfrak{g})\right\}$. The subspaces $U_h(\mathfrak{g})$ determine a natural filtration in the enveloping algebra $U(\mathfrak{g})$, i.e., for integers $h,l\geq 0$ the relations  
\begin{equation}
	U_0(\mathfrak{g})=\mathbb{F},\quad U_{h}(\mathfrak{g})U_l(\mathfrak{g})\subset U_{h+l}(\mathfrak{g}),\quad U_{h}(\mathfrak{g})\subset U_{h+k}(\mathfrak{g}),\; k\geq 1 \label{fil1}
\end{equation}
are satisfied, the base field being $\mathbb{F}=\mathbb{R,C}$.  

\medskip
\noindent By virtue of the universal property of the enveloping algebra, $U(\mathfrak{g})$ possesses the same commutator as those in $\mathfrak{g}$. Essentially, the adjoint action of $\mathfrak{g}$ on the enveloping algebra $U(\mathfrak{g})$ is denoted by $\mathrm{ad}:\mathfrak{g} \rightarrow \mathrm{End}\left(U(\mathfrak{g})\right)$ and is defined as 
\begin{equation}\label{adja}
	\begin{array}[c]{rl}
		P\in U(\mathfrak{g}) \mapsto  \mathrm{ad}_{X_i}(P):= \left[X_i,P\right]= X_i P-P X_i\in U(\mathfrak{g})  .%\\
	%	P\left(x_1,\dots ,x_n\right)\in S(\mathfrak{g})\mapsto &\displaystyle  \widehat{X}_i(P)=C_{ij}^{k} x_k \partial_{x_j} P\in S(\mathfrak{g}),\\
	\end{array}
\end{equation} Let us now consider the reduction chain $\mathfrak{g} \supset \mathfrak{g}'$, where $\mathfrak{g}^\prime$ constitutes an $s$-dimensional subalgebra of $\mathfrak{g}$. In order to facilitate the analysis, without introducing unnecessary complexity, we label the generating elements of $\mathfrak{g}^\prime$ as $X_1,\ldots,X_s$. Here, $s :=\dim \mathfrak{g}^\prime$. Utilizing Equation $\eqref{adja}$, we can observe that the kernel of the adjoint representation corresponding to the subalgebra $\mathfrak{g}^\prime$ on the universal enveloping algebra $U(\mathfrak{g})$ provides a framework for formulating the ensuing definition. 

\medskip
\begin{definition}
     Let $\mathfrak{g}' \subset \mathfrak{g}$ be a $s$-dimensional subalgebra of $\mathfrak{g}.$ The commutant $C_{U(\mathfrak{g})}(\mathfrak{g}^{\prime})$ of an $s$-dimensional subalgebra $\mathfrak{g}'$ is given by means of the condition  
\begin{equation}
	C_{U(\mathfrak{g})}(\mathfrak{g}^{\prime})=\left\{ P\in U(\mathfrak{g}): [X,P]=0,\quad \forall X\in\mathfrak{g}^{\prime}\right\} .\label{comm}
\end{equation} For simplicity, in the following, we denote the commutant $C_{U(\mathfrak{g})}(\mathfrak{g}^{\prime})$ by $U(\mathfrak{g})^{\mathfrak{g}^{\prime}}.$
\end{definition}

\begin{remark}
   The center of the enveloping algebra $U(\mathfrak{g})$
\begin{equation}\label{INVS1}
	Z\left(U(\mathfrak{g})\right) := \left\{ P\in U(\mathfrak{g}) :\left[\mathfrak{g},P\right]=\{0\}\right\}= U(\mathfrak{g})^\mathfrak{g}
\end{equation} consists of $\mathfrak{g}$-invariant polynomials.
\end{remark}
\noindent Determining the commutant directly in the enveloping algebra is generally a cumbersome task \cite{CampoamorStursberg2021,campoamor2022some}, for which reason it is computationally more efficient to consider the equivalent analytical formulation. 

\subsection{Polynomial Poisson algebras from polynomials in $S(\mathsf{\mathfrak{g}})\cong \mathsf{Pol}(\mathfrak{g^*})$}
\label{subsec2.2}

In this Section \ref{subsec2.2}, we will provide more details on the construction of polynomial algebras in the dual space. Let $\mathfrak{g}^*$ be the dual space of $\mathfrak{g}$ with an ordered basis $\beta_{\mathfrak{g}^*} := \{e_1,\ldots,e_n\}$ such that $e_i(X_j) =\delta_{ij},$ where $\delta_{ij}$ is the Kronecker delta. In what follows, assume that $\boldsymbol{x}=\{x_1, \ldots, x_n\}$ are the linear coordinates in $\mathfrak{g}^*$.  Recall that one can define a Poisson structure on the algebra of functions on $\mathfrak{g}^*.$  For any $f,g \in C^\infty(\mathfrak{g}^*)$, 
\begin{equation}
	\{f,g\}=\sum_{i,j,k=1}^n C_{ij}^k x_k \partial_{x_i} f \partial_{x_j} g  \, ,
	\label{FPB}
\end{equation} where $\partial_{x_i} := \partial/\partial x_i.$ See, for instance, \cite{MR2906391}.  The algebra of smooth functions inherits the structure of a Poisson algebra that contains a subalgebra isomorphic to $\frak{g}$.  From this, one can define the linear Poisson bracket (a.k.a. Poisson-Lie structure) relations on $\mathfrak{g}^*$ as follows:
\begin{equation}
	\{x_i,x_j\}= \sum_{k=1}^nC_{ij}^k x_k  \, .
	\label{xcoord}
\end{equation} 
Let $S(\mathfrak{g}) \cong \textsf{Pol}(\mathfrak{g}^*):=\mathbb{F}\left[ x_{1},\ldots ,x_{n}\right]$ be the symmetric algebra of $\mathfrak{g}.$  By means of the Lie-Poisson (or Berezin) bracket defined in $\eqref{FPB}$, we deduce that $\{\cdot,\cdot\} :S(\mathfrak{g}) \times S(\mathfrak{g}) \rightarrow S(\mathfrak{g})$ is given by
\begin{equation}
	\{p,q\}= \sum_{i,j,k =1}^n C_{ij}^k x_k \partial_{x_i} p\partial_{x_j} q  \, , \quad \text{ for any } p,q \in S(\mathfrak{g}).
	\label{LPB}
\end{equation}

 We now construct the polynomial algebras from the subalgebras $\mathfrak{g}^\prime$ of $\mathfrak{g}$. As mentioned in Section $\ref{subsec2.1}$, the Lie subalgebra $\mathfrak{g}'$ is spanned by a subset $\boldsymbol{X}' \subset \boldsymbol{X}$ composed by $s<n$ generators, each of which has a corresponding linear coordinate on $\mathfrak{g^*}$, the whole subset being $\boldsymbol{x}' \subset \boldsymbol{x}$. We now look at the ad-action of the subalgebra $\mathfrak{g}^\prime$ on the symmetric algebra, respectively. For the symmetric algebra, the derivation of the coadjoint action $\mathrm{Ad}^*$ from the subgroup to the automorphism of $S(\mathfrak{g})$ will induce a Poisson-Lie bracket. This can be seen as follows: For any $X \in \mathfrak{g}^\prime$ and $p \in S(\mathfrak{g})$,  the derivative along the vector fields $\hat{X}$ is given by
 
 \begin{align}
   X \cdot p (\boldsymbol{x}')=   \hat{X}(p)(\boldsymbol{x}') = \left.\dfrac{d}{dt} \right|_{t = 0} p\left((\mathrm{Ad}_{\exp(tX)}^*(\boldsymbol{x}')\right) := \mathrm{ad}_X^*  (p) (\boldsymbol{x}'), \label{eq:action}
 \end{align} where $\mathrm{Ad}_{\exp(tX)}^*$ is the induced action (coadjoint action) of the Lie group $G'$. By the definition of coadjoint action $\mathrm{ad}^*$ and $\eqref{eq:action}$, we immediately obtain the following  

 \begin{proposition}
 \label{2.1}
     The coadjoint action $\mathrm{ad}^* : \mathfrak{g}' \rightarrow \mathrm{Aut} \left(S(\mathfrak{g})\right)$ induces a Poisson-Lie bracket $\{\cdot,\cdot\}.$  
 \end{proposition}
 \begin{remark}
     For computational reasons, the action in $\eqref{eq:action}$ can be written   as follows: for any $X_u \in \mathfrak{g}^\prime$ with $1 \leq u \leq s,$ \begin{align}
   \mathrm{ad}_{X_u}^*:     P(x_1,\ldots,x_n) \in S(\mathfrak{g}) & \mapsto \{x_u,P  \} = \hat{X}_u(P ) = \sum_{l,k=1}^n C_{uk}^lx_l \dfrac{\partial P}{\partial x_k} \in S(\mathfrak{g}), \label{eq:dual}
  \end{align} where $\hat{X}_u = \sum_{l,k=1}^n C_{uk}^l x_l \dfrac{\partial}{\partial x_k}$ corresponds to the infinitesimal generators of the one-parameter subgroup determined by the generators $X_u$ by means of the coadjoint representation and $p(x_1,\ldots,x_n)$ is a polynomial in $S(\mathfrak{g})$ with the generic form 
  \begin{align}
    p(x_1,\ldots,x_n) = \sum_{a_1 + \ldots + a_n   \leq h}  \Gamma_{i_1,\ldots,i_n}\, x_1^{a_1} \cdots x_n^{a_n} . \label{eq:ci}
\end{align}  
For instance, see \cite{gtp,MR0760556}.
 \end{remark}

  Inspired from Proposition $\ref{2.1}$, we have the following definition.

\begin{definition}
\label{2.3} 
%  Let $\mathfrak{g}$ be a finite-dimensional Lie algebra and $\mathfrak{g}^*$ be its dual. Let $(U(\mathfrak{g}),[\cdot,\cdot])$ and $(S(\mathfrak{g}),\{\cdot,\cdot\})$ be the universal enveloping algebra and associated symmetry algebra of $\mathfrak{g}$, respectively.
The kernel of the coadjoint action of $\mathfrak{g}^\prime$ on the  symmetric algebra $S(\mathfrak{g})$ is given by \begin{equation*}
	C_{S(\frak{g})}\left( \frak{g}^{\prime}\right) =\left\{ p\in S\left( 
	\frak{g}\right) :\left\{x ,p\right\} =0,\; x\in \frak{g}^{\prime *}\right\} .
\end{equation*} For  convenience, we denote the Poisson centralizer  $C_{S(\frak{g})}\left( \frak{g}^{\prime}\right)$ by $S(\mathfrak{g})^{\mathfrak{g}^\prime}.$
 \end{definition}

\begin{remark}
\label{2.6}
   (i)  Taking into account the noncommutative framework (see (\ref{INVS1})), it can be easily seen that the center of the symmetric algebra 
\begin{equation*}
	Z\left(S\left( \frak{g}\right)\right)=\left\{ p\in S\left( \frak{g}\right) :\left\{ x,p\right\} =0,\;x\in
	 \frak{g}^*  \right\}  = S(\mathfrak{g})^\mathfrak{g},
\end{equation*} which contains all the Casimir  elements in $S(\mathfrak{g}).$ From Chevalley's theorem, if $\mathfrak{g}$ is semisimple or reductive, then there exist finitely many homogeneous algebraically independent invariants $p_1,\ldots,p_r$ such that $S(\mathfrak{g})^{\mathfrak{g}} = \mathbb{F}[p_1,\ldots,p_r]$, where $r = \mathrm{rank} (\mathfrak{g}).$

 (ii) The polynomials  commuting with $\mathfrak{g}^{\prime}$ are determined as the solutions of the equations in the following system of PDEs corresponding to the subalgebra generators:
\begin{equation}
	\widehat{X}_u\left( p\right) :=\left\{ x_u,p\right\} = \sum_{ l,k=1}^n C_{uk}^lx_l%
	\partial_{x_k} p=0,\; \text{ } 1\leq u\leq s=\dim \frak{g}^{\prime},
	\label{mlpa}
\end{equation}
where $(x_1, \dots, x_s)$ are coordinates in a dual basis of $\mathfrak{g}'^*$.
 
 (iii) In a finite-dimensional Lie algebra, a finite generating set within $S(\mathfrak{g})^{\mathfrak{g}'}$ may not exist. However, if $\mathfrak{g}$ is semisimple or reductive, it is shown that $S(\mathfrak{g})$ is Noetherian, allowing us to deduce that the centralizer $S(\mathfrak{g})^{\mathfrak{g}'}$ is finitely-generated (see e.g. \cite[Chapter 2]{MR1451138}). Hence, once a maximal set of indecomposable polynomials $\left\{p_{k_1},\dots ,p_{k_m}\right\}$ has been determined, there exists an integer $\zeta \in \mathbb{N}$ such that for all $j \geq 1$, the polynomial $p_{\zeta+k_j}$ is decomposable. A polynomial $p \in S(\mathfrak{g})$ is considered decomposable if there exists another polynomial $p' \in S(\mathfrak{g})$ of lower degree such that $p \equiv 0 \mod p'$, indicating that $p'$ divides $p$. It is also crucial to recognize that elements in the generating set of a centralizer subalgebra do not automatically reflect their algebraic independence.
\end{remark}

Within the framework of the labeling problem, as described in \cite{MR2749089}, there exists a class of functions which are not restricted to polynomials and which satisfy the prescribed system of PDEs given in (\ref{mlpa}). Such functions are identified with labeling operators, which means that they are used to distinguish different states in the case where a representation of the Lie algebra $\mathfrak{g}$ is broken down to a subalgebra $\mathfrak{g}^{\prime}$ (for further details, refer to sources such as \cite{MR170975}). In particular, the system described by (\ref{mlpa}) precisely admits $\mathcal{M}_0=\dim \mathfrak{g}-\dim\mathfrak{g}^{\prime}+\ell_0$ solutions that are functionally independent. Here, $\ell_0$ represents the total count of invariants related to $\mathfrak{g}$, which are characterized solely by their dependence on the variables pertaining to the subalgebra $\mathfrak{g}^{\prime}$. As shown in \cite{MR0411412}, the identity 
\begin{equation}
	\mathcal{M}_0=2n_0+N\left(\mathfrak{g}\right)+N\left(\mathfrak{g}^{\prime}\right)
\end{equation}
holds, where $n_0$ is the number of required (internal) labeling operators. Based on the analysis, we infer that there exist three distinct categories of elements which possess the ability to commute with the subalgebra. These are specifically the Casimir operators associated with $\mathfrak{g}$, the Casimir operators corresponding to $\mathfrak{g}^{\prime}$, and an additional set consisting of $2n$ labeling operators that exhibit dependence on the generators of both $\mathfrak{g}$ and $\mathfrak{g}^{\prime}$. As demonstrated in sources such as \cite{MR0411412}, it is evident that within this collection of $2n$ operators, a maximum of $n$ can mutually commute. This observation provides significant insight, indicating that the set of solutions to equation (\ref{mlpa}) effectively constructs a functional algebra when considered within the context of the Poisson bracket. When one narrows the scope exclusively to polynomial cases, i.e., within the framework of the symmetric algebra, these solutions give rise to a polynomial algebra. Notably, this algebra is typically non-Abelian, signifying complexity beyond mere commutativity. %The symmetrized operators thus generate a polynomial algebra in the enveloping algebra. 

We now recall a map that allows us to compare the $\mathfrak{g}$-module structure of $U(\mathfrak{g})$ and $S(\mathfrak{g}).$ This gives rise to more algebraic structures in $S(\mathfrak{g}).$ Define a $\mathfrak{g}$-invariant isomorphism of vector spaces $\Lambda: S(\mathfrak{g}) \rightarrow U(\mathfrak{g})$ given on the monomials basis by   \begin{equation}
 	\Lambda\left(x_{i_1} \cdots x_{i_n}\right) =\frac{1}{n!} \sum_{\sigma \in \Sigma_n}X_{i_{\sigma(1)}} \cdots X_{i_{\sigma(n)}} 
 	\label{symma}
 \end{equation}
with $\Sigma_n$ being the permutation group in $\{1, 2, \dots, n\}$ that commutes with the adjoint action. The map $\Lambda$ is called the \textit{symmetrization map}  (see, for example, \cite{gtp}). In this sense, the symmetric algebra is embedded in $U(\mathfrak{g})$ as the symmetric tensors, with $U(\mathfrak{g})$ serving as a non-commutative deformation of $S(\mathfrak{g}).$ Note that $U(\mathfrak{g})$ and $S(\mathfrak{g})$ are clearly not isomorphic as algebras as one being commutative, the other generally noncommutative, it could happen that the $\mathfrak{g}$-invariant centralizer $S(\mathfrak{g})^\mathfrak{g}$ and $ U(\mathfrak{g})^\mathfrak{g}$ are isomorphic. See, for example, \cite{MR2816610}. Furthermore, recall that $S(\mathfrak{g})$ is a graded algebra. That is, $S(\mathfrak{g})  = \bigoplus_{k \geq 0} S^k (\mathfrak{g}),$ where \begin{align*}
     S^k(\mathfrak{g}) : = \mathrm{span} \left\{x_1^{a_1} \cdots x_n^{a_n}: a_1 + \ldots + a_n = k, \quad a_j \in \mathbb{N}_0\right\} 
 \end{align*} is a subalgebra of $S(\mathfrak{g})$ consisting of all the degree $k$ polynomials and $S^0(\mathfrak{g}) = \mathbb{F}$. Define the quotient vector subspace of $U(\mathfrak{g})$ by $U^k(\mathfrak{g}) = U_k(\mathfrak{g})/U_{k-1}(\mathfrak{g})$. We then have a associative algebra \begin{align}
     \mathrm{gr} \, U(\mathfrak{g}) = \bigoplus_{k \geq 0} U^k(\mathfrak{g}).
 \end{align} Using PBW theorem, $\Lambda$ induces an algebra isomorphism $\tilde{\Lambda}$ between $S(\mathfrak{g})$ and $\mathrm{gr} \, U(\mathfrak{g})$.  Through the isomorphsim $\tilde{\Lambda}$, we can define $U^l(\mathfrak{g})=\tilde{\Lambda}\left(S^l(\mathfrak{g})\right)$ as an algebra isomorphism that induces decomposition $U_{l}(\mathfrak{g})=\bigoplus_{k=0}^{l} U^{k}(\mathfrak{g})$. This, in particular, implies the relation 
\begin{equation*}
	\left[ P,Q\right]\in U_{l+\ell-1}(\mathfrak{g}),\quad P\in U_{l}(\mathfrak{g}), \quad Q\in U_{\ell}(\mathfrak{g}).
\end{equation*}

%The commutant in the enveloping algebra $U(\mathfrak{g})$ is obtained from the symmetrization map:   $ U(\mathfrak{g})^{\mathfrak{g}^{\prime}}=\Lambda\left(  S(\frak{g})^{\frak{g}^{\prime}}\right).$ 

\subsection{Construction of indecomposable generators of Poisson centralizers}
\label{subsec2.3}

\medskip

 In this Subsection \ref{subsec2.3}, we develop a comprehensive methodology for determining the polynomial solutions of equation $\eqref{mlpa}$. 
 There are essentially two methods for computing a polynomial that sits in the centralizer of a given subalgebra: one involves engaging directly with systems of PDEs via the method of characteristics, while the other involves formulating an Ansatz specific to polynomials. Regarding reductive Lie algebras, the elements in the commutant can be interpreted as polynomials that incorporate variables from the dual space. Therefore, we will determine the possible solution of $\eqref{mlpa}$ by adopting a polynomial Ansatz. However, it is important to note that this polynomial representation is not universally applicable for all types of Lie algebras. For non-semisimple Lie algebras, solutions may involve rational or transcendental functions.

Define the vector space of $\mathfrak{g}'$-invariant  $k$-homogeneous polynomials as 
\begin{align*}
S^k(\mathfrak{g})^{\mathfrak{g}^\prime}  = \left\{ p^{(k)} \in S^k(\mathfrak{g}): \left\{x,p^{(k)}(\boldsymbol{x})\right\} = 0  \quad \forall x \in {\mathfrak{g}'}^*\right\} .
\end{align*}  This provides a procedure that allows us to concentrate, in the subsequent discussion, on homogeneous polynomials of degree $k \geq 1$, as well as on their generic structure
\begin{equation}
	p^{(k)}(\boldsymbol{x})=\sum_{a_1+\dots + a_n = k} \Gamma_{a_1, \dots, a_n} \,x_1^{a_1} \dots x_n^{a_n}  \in S^k\left( \frak{g}\right) \,,
	\label{polynomials}
\end{equation} where $\Gamma_{a_1,\ldots,a_n}$  are constants.  For embedding chains, elements in the Poisson centralizer $S(\mathfrak{g})^{\mathfrak{g}^\prime}$ correspond to polynomials that Poisson commute with the corresponding linear coordinates associated to the generators of the subalgebra $\mathfrak{g}'$ w.r.t. the Lie-Poisson bracket \eqref{LPB}. With the assumed general expansion \eqref{polynomials}, this condition is expressed as 
\begin{equation}
	 \left\{x_u, p^{(k)}(\boldsymbol{x}) \right\}  =\sum_{ l,k=1}^n C_{uk}^lx_l%
	\partial_{x_k} p^{(k)}=0\quad 1\leq u \leq s.
	\label{dbd}
\end{equation}
These linear PDEs can be addressed sequentially by investigating homogeneous solutions of various degrees ($1 \leq k \leq \zeta$).  Here, $\zeta$ represents the highest degree of indecomposable polynomials that solve $\eqref{dbd}$. Note that, as discussed in Remark $\ref{2.6}$ (iii), the existence of $\zeta$ is ensured by the assumption that $\mathfrak{g}$ is reductive or semisimple. Beyond this point, we only find additional linearly dependent polynomials, which can be all expressed in terms of polynomials of the lower-degree solutions. According to its definition, at degree $k=1$, all generators are contained within the centralizer subalgebra $\mathfrak{t}:=\mathfrak{g}^{\mathfrak{g}'}$. Let $m_1 \geq 1$ denote the number of linearly independent degree-one polynomials. Without loss of generality, all degree one $\mathfrak{g}'$-invariant polynomials are  $$\boldsymbol{p}^{(1)}:=\left\{p_1^{(1)}(\boldsymbol{x}), \ldots, p_{m_1}^{(1)}(\boldsymbol{x})\right\}.$$ We then proceed with building degree-two polynomials. By the definition of indecomposability, any quadratic solution to \eqref{dbd} formed by products such as $p_a^{(1)}(\boldsymbol{x})p_b^{(1)}(\boldsymbol{x})$ for $1\leq a,b\leq m_1$ must be discarded. Let $$\boldsymbol{p}^{(2)}:=\left\{p_1^{(2)}(\boldsymbol{x}), \ldots, p_{m_2}^{(2)}(\boldsymbol{x})\right\} \, $$ be the set of all quadratic linearly independent and indecomposable $\mathfrak{g}'$-invariant polynomial solutions.  Applying recursively the method, excluding polynomials that can be obtained from those of lower degrees, a finite set of linearly independent solutions can be found. This set comprises precisely $m_{[\zeta]}:=\text{Card}(\boldsymbol{Q}^{[\zeta]})=m_1+m_2+\ldots +m_\zeta < \infty$ distinctive polynomials that are both indecomposable and linearly independent, each within a predetermined maximum degree of $\zeta$:
\begin{equation}
	\boldsymbol{Q}^{[\zeta]}:=\bigsqcup_{i=1}^\zeta \boldsymbol{p}^{(i)} = \boldsymbol{p}^{(1)} \sqcup \dots  \sqcup \boldsymbol{p}^{(i)}  \sqcup \dots \sqcup \boldsymbol{p}^{(\zeta)}   \, ,
	\label{listN}
\end{equation}
where the  $i$-th set $\boldsymbol{p}^{(i)}:=\left\{p_1^{(i)}(\boldsymbol{x}), \dots,p_{m_i}^{(i)}(\boldsymbol{x})\right\}$ corresponds to $i$ ranging from 1  up to the highest degree $\zeta$. This method guarantees that the collection of polynomials denoted as $\boldsymbol{Q}^{[\zeta]}$  encompass all indecomposable and linearly independent solutions, up to a given degree limit. It should be reiterated that the components of the set referred to in equation (\ref{listN}) do not, in general, exhibit algebraic independence. However, it is important to note that the  dependence relations among these elements involve rational functions. Consequently, these dependencies are situated beyond the domain of the polynomial algebra $\mathsf{Pol}(\mathfrak{g^*})$. For an expanded discussion and further elaboration on this topic, we refer to \cite{campoamor2023algebraic}.

\medskip
Our objective is to construct the finitely generated polynomial (Poisson) algebras obtained from the set $\boldsymbol{Q}^{[\zeta]}$. Let $\textbf{Alg} \left\langle \boldsymbol{Q}^{[\zeta]} \right\rangle $ denote the algebra generated by the set $\boldsymbol{Q}^{[\zeta]}.$ Furthermore, this polynomial algebra  admits  the following filtration structure: \begin{equation*}
\mathcal{Q}_0 := \mathbb{F} \subset \textbf{Alg}\langle \boldsymbol{Q}^{[1]} \rangle:= \mathfrak{t} \subset \dots \subset \textbf{Alg}\langle \boldsymbol{Q}^{[\zeta]} \rangle .
\end{equation*} Note that $\textbf{Alg} \left\langle \boldsymbol{Q}^{[\zeta]} \right \rangle$ is indeed an infinite-dimensional vector space. However, as a finitely generated algebra, we further denote \begin{equation*}
    \dim_{FL}\textbf{Alg}\langle \boldsymbol{Q}^{[\zeta]} \rangle=  m_{[\zeta]}.
\end{equation*} Here, the notation $\dim_{FL}$ represents the count of indecomposable monomials that constitute  a set of generators of $\boldsymbol{Q}^{[\zeta]} $. It is crucial to recognize $\dim_{FL}$ as a theoretical upper limit for the rank of this finitely generated algebra, bearing in mind that the basis elements of $\boldsymbol{Q}^{[\zeta]} $ might not be  functionally independent. It follows that $\dim_{FL} \textbf{Alg}\langle \boldsymbol{Q}^{[\zeta]} \rangle \geq \dim_{KL} \textbf{Alg}\langle \boldsymbol{Q}^{[\zeta]} \rangle$, where $\dim_{KL}$ refers to the Krull dimension associated with the algebra under consideration. See, for instance, \cite{MR1322960}. In what follows, we aim to close the Poisson-Lie bracket in $\textbf{Alg} \left\langle \boldsymbol{Q}^{[\zeta]} \right\rangle $ by 

\begin{equation}
	\left\{p_{i_1}^{(h_{i_1})} , p_{i_2}^{(h_{i_2})}\right\}  =  \sum_{k_1+\ldots + k_r =h_{i_1} + h_{i_2}-1} \Gamma^{s_1,...,s_r}_{i_1,i_2} p^{(k_1)}_{s_1} \cdots p^{(k_r)}_{s_r}  . \, 
	\label{polrel}
\end{equation} Here $h_{i_1},h_{i_2},k_1,\ldots,k_r \leq \zeta$ and $ 1  \leq s_1,\ldots,s_r \leq m_{[\zeta]}$.     Since the generators in $\boldsymbol{Q}^{[\zeta]}$ are not functionally independent, to ensure $\eqref{polrel}$ is indeed a Poisson bracket, further polynomial relations between these generators need to be included. In the framework of constructing finitely generated polynomial Poisson algebras, if one encounters a relation among the generators $\boldsymbol{p}^{(M)}$ (for index values ranging from $2$ to $\zeta$) characterized by
\begin{equation}
	\sum_{i=1}^{m_M} p_i^{(M)}=P^{[M]}(\boldsymbol{p}^{(1)}, \dots, \boldsymbol{p}^{(M-1)})
	\label{pol}
\end{equation}
we proceed by omitting one polynomial from the collection $\boldsymbol{p}^{(M)}$. Subsequent to this omission, all further relations defined by \eqref{polrel} are evaluated using the reduced set. This omitted polynomial, specifically the one corresponding to $i = i^*$ in the above relation \eqref{pol}, is deduced based on the others as follows:
\begin{equation}
	p_{i^*}^{(M)} = P^{[M]}(\boldsymbol{p}^{(1)}, \dots, \boldsymbol{p}^{(M-1)}) - p_1^{(M)} - \dots - p_{i^*-1}^{(M)} - p_{i^*+1}^{(M)} - \dots - p_{m_M}^{(M)} \, .
	\label{eq:polel}
\end{equation}

Let \[ d = \max_{1 \leq k_j \leq \zeta} \sum_{\begin{matrix}
    j \in I \\
    I = \left\{1,\ldots,r: p_{s_j}^{(k_j)} \notin \mathcal{Z}\right\}
\end{matrix}} k_j \] represents the degree of the given polynomial algebra $\textbf{Alg} \left\langle \boldsymbol{Q}^{[\zeta]} \right\rangle$. Here $\mathcal{Z} := \left\{p \in \textbf{Alg}\left\langle \boldsymbol{Q}^{[\zeta]}\right\rangle :  \{p,q\} = 0, \text{ } \forall q \in \textbf{Alg}\left\langle \boldsymbol{Q}^{[\zeta]}\right\rangle \right\}$ is the center of $\textbf{Alg}\left\langle \boldsymbol{Q}^{[\zeta]}\right\rangle.$   For convenience in future discussions, we introduce the notation $\mathcal{Q}_\mathfrak{g}(d) := \left(\textbf{Alg} \left\langle \boldsymbol{Q}^{[\zeta]} \right\rangle,\{\cdot,\cdot\} \right)$.  By the way it is constructed, we find that \begin{align*}
    \mathcal{Q}_\mathfrak{g}(d) &= \mathfrak{t} \oplus \bigoplus_{k \in \Omega} \mathcal{Q}_k,
\end{align*} indicating that it is, indeed, a graded polynomial algebra.  The symbol $\Omega \subset \mathbb{N}$ denotes an ordered index set, while $\mathcal{Q}_k$ refers to the vector space that encompasses $\mathfrak{g}^\prime$-invariant polynomials, each characterized by a degree of $k$.

\medskip
 
Deriving the explicit structure of the polynomial relations indicated in $\eqref{polrel}$ presents significant computational challenges. Let us now recall some  terminologies given in \cite{MR4660510}.   We will use the following notational convention to indicate polynomials of a given degree. Once the representatives for each subset $\boldsymbol{p}^{(k)}$ ($k=1, \dots, \zeta$) have been found, taking into account the condition \eqref{eq:polel} to eliminate unnecessary polynomials, elements of degree one will be indicated with the uppercase letter $A_i$, forming the set $\textbf{A}_1$, while elements of degree two with $B_j$,  forming the set $\textbf{A}_2$,  and so on, following alphabetical order. Eventually, all the representative of each $\boldsymbol{p}^{(k)}$ is reformulated in $\textbf{A}_k.$ Furthermore, central elements will be denoted with lowercase letters, again following alphabetical order to keep track of the degree of homogeneous polynomials.  The concept of decomposing polynomials up to a specific degree inherently suggests that the Poisson bracket $\{\cdot,\cdot\}$ found in $\eqref{polrel}$ produces polynomials of higher degrees, eventually culminating in a result that splits into polynomials that are functionally independent.  Considering a polynomial of degree $\zeta$, the Poisson bracket $\{\cdot,\cdot\}$ applied to these \textit{compact forms} results in certain configurations when evaluated up to degree $\zeta$ \begin{align}
\nonumber
    \{\textbf{A}_1,\textbf{A}_2\} \sim &  \textbf{A}_2 + \textbf{A}_1^2; \\
    \nonumber
    \{\textbf{A}_2,\textbf{A}_2\} \sim & \text{ } \textbf{A}_3 + \textbf{A}_1 \{\textbf{A}_1,\textbf{A}_2\}; \\
      \{\textbf{A}_2,\textbf{A}_3\} \sim & \text{ } \textbf{A}_4 +  \textbf{A}_2^2 + \textbf{A}_1 \{\textbf{A}_2,\textbf{A}_2\}; \label{eq:compact} \\
      \nonumber
       \{\textbf{A}_2,\textbf{A}_4\} \sim & \text{ } \textbf{A}_5 +  \textbf{A}_2 \textbf{A}_3 + \textbf{A}_1 \{\textbf{A}_2,\textbf{A}_3\}\\ 
       \nonumber
       \{\textbf{A}_2,\textbf{A}_5\} \sim & \text{ } \textbf{A}_6 + \textbf{A}_2 \textbf{A}_4 +  \textbf{A}_2^3 +  \textbf{A}_3^2 + \textbf{A}_1 \{\textbf{A}_2,\textbf{A}_4\} \\
       & \vdots  \nonumber
\end{align} For instance, in the relation of degree two above, the relations that we look for can only adopt the following form: \begin{align}
    \{A_k,B_l\} =   \sum_{u=1}^{m_1} \sum_{v=1}^{m_1} \Gamma_{kl}^{uv}A_u A_v + \sum_{w=1}^{m_2} \Gamma_{kl}^w B_w. 
\end{align} Here $\Gamma_{kl}^{uv},\Gamma_{kl}^w$ are some constants.  Inductively, for any $\textbf{A}_k,\textbf{A}_l$ with $1 \leq k,l \leq \zeta$, the Poisson bracket in terms of the compact terms is given as follows \begin{align}
    \{\textbf{A}_k,\textbf{A}_l\} \sim \{\textbf{A}_g,\textbf{A}_h\}\sim &\, \textbf{A}_{k+l-1} + \textbf{A}_2 \textbf{A}_{k+l-3} + \textbf{A}_3 \textbf{A}_{k+l-4} + \ldots + \prod_{ j_1+ \ldots + j_\zeta = k+ l -1} \textbf{A}_{j_1} \cdots \textbf{A}_{j_\zeta} + \ldots + \textbf{A}_1 \{\textbf{A}_k,\textbf{A}_{l-1}\} \nonumber \\
    \sim &\, \sum_{m_1 a_1 + m_2 a_2 + \ldots + m_\zeta a_\zeta = k + l-1 } \textbf{A}_1^{a_1} \textbf{A}_2^{a_2} \cdots \textbf{A}_\zeta^{a_\zeta}. \label{eq:cp} 
\end{align} Here $\textbf{A}_j^{a_j} = \left(p_1^{(j)}\right)^{w_1} \cdots \left(p_{m_j}^{(j)}\right)^{w_{m_j}} $ with $w_1 + \ldots + w_{m_j} = a_j $ and $k + l = g + h \leq \zeta.$ In Equation \eqref{eq:cp}, Poisson brackets with identical degrees yield the same compact expansions. For instance, when $\deg \{\textbf{A}_2, \textbf{A}_2\} = \deg \{\textbf{A}_1, \textbf{A}_3\}$, the expansion of $\{\textbf{A}_2, \textbf{A}_2\}$ aligns with that of $\{\textbf{A}_1, \textbf{A}_3\}$, which manifests as $\textbf{A}_3 + \textbf{A}_1 \{\textbf{A}_1, \textbf{A}_2\}$. Moreover, we further consider that $\textbf{A}_t = \{0\}$ for any $t \leq 0$. It can easily be shown that the generators allowed in compact form $\textbf{A}_1^{a_1} \cdots \textbf{A}_\zeta^{a_\zeta} $ are \begin{align}
    \binom{m_1+a_1 -1}{a_1  } \times \cdots \times \binom{m_\zeta+a_\zeta -1}{a_\zeta  }. \label{eq:counting}
\end{align} Here $m_j$ is the number of elements in $\textbf{A}_j$ for all $1 \leq j \leq \zeta$. 

An in-depth analysis of the embedding chain of Lie algebras reveals that certain coefficients previously discussed could be zero. The lack of a deeper understanding of the inherent structure of polynomial generators leads to each polynomial $p_l^{(k)}$ containing numerous monomials of any degree. This makes it particularly difficult to determine the coefficients $\Gamma^{s_1,...,s_r}_{i_1,i_2}$ appearing in equation \eqref{polrel} and to identify the permissible monomials in formulations like $\textbf{A}_1^{a_1} \textbf{A}_2^{a_2} \cdots \textbf{A}_\zeta^{a_\zeta}$. The subsequent section will delineate a methodology designed to facilitate the identification of a polynomial appearing in the expansions of the Poisson brackets and to systematically construct the associated polynomial algebra, particularly when imposed constraints are present.

\subsection{The grading of the generators of $\mathcal{Q}_\mathfrak{g}(d)$}
\label{2.4}

In this Subsection \ref{2.4}, we recall the terminology and the concept of grading in the context of polynomial algebras as outlined in \cite{campoamor2025On}.  In the following, let $\mathfrak{g} = \bigoplus_{r \in J} \mathfrak{g}_r$. This decomposition is subject to a constraint on the Lie bracket, such that for any elements $r, s, t$ within the finite set $J$, the bracket operation satisfies $[\mathfrak{g}_r,\mathfrak{g}_s] \subset \mathfrak{g}_t$.  From this linear decomposition we deduce that $S(\mathfrak{g}) \cong \bigotimes_{r \in J} S(\mathfrak{g}_r)$. Furthermore, consider the subalgebra $\mathfrak{g}^\prime$, defined as $\mathfrak{g}^\prime : = \bigoplus_{i \in I} \mathfrak{g}_i \subset \mathfrak{g}$, where $I\subset J \subset \mathbb{N}$. As described in Subsection \ref{subsec2.3}, this gives rise to a graded polynomial algebra, denoted $\mathcal{Q}_\mathfrak{g}(d) = \mathfrak{t} \oplus \bigoplus_{k \in \Omega} \mathcal{Q}_k$.  This refinement allows us to simplify the description of generators and relations in the polynomial algebra. 

\begin{definition}
\label{gradingd}
   For any non-zero monomial $p \in \mathcal{Q}_k,$ a $\textit{grading of a monomial}$ is defined by $\mathcal{G} : \mathcal{Q}_k/\{0\}\rightarrow \overbrace{\mathbb{N}_0 \times \ldots \times \mathbb{N}_0}^{\text{$m$-times}}$ given by $p   \mapsto (i_1,\ldots,i_m)$, where $i_j \in \mathbb{N}_0 $ is the number of generators of $\mathfrak{g}_j$ in a monomial. 
\end{definition}

\begin{remark}
\label{re2.4}
  (i) Note that when $p$ is constant, we denote $\mathcal{G}(p)$ as $(0, \ldots, 0)$. However, due to our formulation of polynomial algebras, we exclude $\mathcal{Q}_0 = \mathbb{F}$. Therefore, we shall disregard the case where $\mathcal{G}(p) = (0, \ldots, 0)$.

  (ii) For any non-zero monomials $p ,q \in \mathcal{Q}_k$,  $\mathcal{G}(pq) = \mathcal{G}(p) + \mathcal{G}(q)$.

 (iii) Consider that, if a given monomial $p$ belongs to $\mathcal{Q}_\mathfrak{g}(d)$ and is defined as decomposable, then by definition, we can express $p$ in the form $p = \prod_{j=1}^w p_j$, where each individual component $p_j$ is an indecomposable monomial also found within $\mathcal{Q}_\mathfrak{g}(d)$. Moreover, there exists the possibility of encountering an indecomposable monomial $p' \in \mathcal{Q}_\mathfrak{g}(d)$ such that the grading $\mathcal{G}(p)$ is equal to the sum of the gradings of its constituent parts, $\sum_{j=1}^w \mathcal{G}(p_j)$, which itself is equivalent to $\mathcal{G}(p')$. Consequently, to eliminate any potential ambiguity, we denote the grading of any decomposable polynomial as a discrete sum. Specifically, this is represented by $\mathcal{G}(p) = \sum_{j=1}^w \mathcal{G}(p_j)  $.
 
(iv) Suppose that $p$ is expressed as $p = \sum_{j=1}^w \gamma_j p_j$, where it constitutes an indecomposable polynomial $p_j$. Here, $\gamma_j \in \mathbb{F}$ are the coefficients for each index $j$, and each $p_j$ is also an indecomposable component. Subsequently, we define the operation $\tilde{+}$ to allow $\mathcal{G}(p)$ to be represented as $\mathcal{G}(p_1) \tilde{+} \ldots \tilde{+} \mathcal{G}(p_w)$, ensuring that associativity and commutativity are maintained in this framework. Specifically, in cases where elements $p_i$ and $p_j$ satisfy the condition $\mathcal{G}(p_i) = \mathcal{G}(p_j)$, it leads to a simplification where $\mathcal{G}(p)$ can be rewritten as $\mathcal{G}(p_1) \tilde{+} \ldots \tilde{+} \mathcal{G}(p_i) \tilde{+} \ldots \tilde{+} \mathcal{G}(p_{j-1}) \tilde{+} \mathcal{G}(p_{j+1})  \tilde{+} \ldots \tilde{+} \mathcal{G}(p_w)$. This implication highlights that, when equivalent elements occur, only one of them needs to be included in the final representation. From these observations, we define $\mathcal{G}(p)$ as \textit{homogeneous grading} under the condition that $\mathcal{G}(p) = \mathcal{G}(p_1) = \ldots = \mathcal{G}(p_w)$. Note that any grading obtained from a monomial is invariably homogeneous. This also allows us to simplify some of the relations in the indecomposable generators. See, for instance, $\eqref{suitgene}$ below. %This provides constraints on the polynomials allowed in the commutator or Poisson bracket for the closure.
\end{remark}

In the following, we illustrate how the grading method helps us to find the allowed polynomials within the structure of the Poisson brackets. Consider generators $p, q \in \boldsymbol{Q}^{[\zeta]}$. For our purposes, assume, without loss of generality, that these elements are expressed as: 
\begin{align*}
p := & \, p^{(i_1 + \ldots + i_m)} = x_1^{a_1} \cdots x_{s_1}^{a_{s_1}} \cdots \underbrace{x_{{s_{t-1}}+1}^{a_{s_{t-1}+1}} \cdots x_{s_t}^{a_{s_t}}}_{\text{denoted as $Y_j$ below}} \cdots x_{s_{m-1}+1}^{a_{s_{m-1}+1}} \cdots x_{s_m}^{a_{s_m}}, \\
q := & \, q^{(j_1 + \ldots + j_m)} = x_1^{b_1} \cdots x_{s_1}^{b_{s_1}} \cdots \underbrace{x_{{s_{t-1}}+1}^{b_{s_{t-1}+1}} \cdots x_{s_t}^{b_{s_t}}}_{\text{denoted as $Z_j$ below}} \cdots x_{s_{m-1}+1}^{b_{s_{m-1}+1}} \cdots x_{s_m}^{b_{s_m}}.
\end{align*} Here $a_{s_{t-1} + 1} + \ldots + a_{s_t} = i_t$ and $b_{s_{t-1}+1} + \ldots +b_{ s_t} =  j_t. $
By definition, $\mathcal{G}(p) = (i_1, \ldots, i_m)$ and $\mathcal{G}(q) = (j_1, \ldots, j_m)$. In this context, elements $x_{s_{j-1}+1}, \ldots, x_{s_j}$ serve uniquely as coordinates in $\mathfrak{g}_j^*$. Employing the Leibniz rule, it becomes possible to show, using induction, that for any given products $ p = \prod_{j=1}^m Y_j $ and $ q = \prod_{k=1}^m Z_k$, the following identity is satisfied:
\begin{align} 
     \left\{\prod_{j=1}^m Y_j,\prod_{k=1}^m  Z_k\right\}   = \sum_{1 \leq i_s \leq m} \left\{Y_{i_s},\prod_{k=1}^m  Z_k\right\} \prod_{j \neq i_s} Y_j. \label{eq:muli} %+\sum_{i_s,j_r} \{A_{i_s},B_{j_r}\}   \prod_{j \neq i_s} A_j \prod_{k \neq j_r} B_k
 \end{align} 
  From the preceding expression, it is straightforward to verify that the condition $\left\{Y_{i_s},\prod_{k=1}^m Z_k\right\} = 0$ is satisfied for each index $i_s$ within the set $I.$ Consequently, when attempting to determine the grading that corresponds to the bracket, our analysis can be refined by concentrating on the expression comprising the summation of the term $\left\{Y_{i_s},\prod_{k=1}^m Z_k\right\} \prod_{j \neq i_s} Y_j $, where $i_s$ is excluded from the set $I.$ Incorporating this with the established commutator relations intrinsic to the structure of the Lie algebra, we may deduce:
\begin{align}
    \mathcal{G} \left(\{p,q\}\right) = \left(i_1 + j_1 + g_{11}, \ldots, i_m + j_m + g_{mt}\right) \tilde{+} \ldots \tilde{+} \left(i_1 + j_1 + g_{1 \xi}, \ldots, i_m + j_m + g_{m \xi}\right), \label{eq:gradg}
\end{align} where $g_{11}, \ldots, g_{m \xi}$ belongs to the set $\{-2, -1, 0, 1, 2\}$ and $1 \leq t \leq \xi.$ All permissible terms, as indicated by the grading associated with $\{p,q\}$ in $\eqref{eq:gradg}$, are encompassed within the set
\begin{align*}
    \mathcal{A} := \{p_1,\ldots,p_r\} \sqcup \{p_{k_1}p_{l_1},\ldots,p_{k_r}p_{l_r}\} \sqcup  \ldots \sqcup \left\{p_{t_1}\cdots p_{t_s}, \ldots, p_{\ell_1}\cdots p_{\ell_s}\right\}.
\end{align*} 
For any element $p$ belonging to $\mathcal{A},$ we can assert that the grading $\mathcal{G}(f)$ either equals the grading $\mathcal{G}(\{p,q\})$ or  equals some particular homogeneous gradings within $\mathcal{G}(\{p,q\}).$

\section{The supermultiplet model}
\label{sec3}

 The Lie group $SU(4)$ was first considered in the context of nuclear physics in 1937, when Wigner established that the charge independence of nuclear forces was related to an observed approximate fourfold degeneracy in the energy levels \cite{Wig}. This motivated to introduce a nucleon that distinguished isospin quantum numbers in combination with ordinary spin, leading to the spin-isospin multiplet model. In the group-theoretical framework, this model corresponds to the state labeling problem for the reduction chain
$\mathfrak{su}(4) \supset \mathfrak{su}(2) \times \mathfrak{su}(2)$, which has been analyzed by various authors using different techniques (see \cite{MR189602,Hecht,Brunet,Draayer,nuclear} and references therein). In this Section \ref{sec3}, we focus on the construction of the polynomial algebra associated to the Wigner supermultiplet model, which in particular contains the labeling operators. For convenience, this section is divided into two parts. In the first, specifically subsection $\ref{subsec3.1}$, we elaborate on the generators along with the commutator  relations relevant to the Lie algebra $\mathfrak{su}(4)$ and its subalgebra $\mathfrak{su}(2) \times \mathfrak{su}(2)$. Following this, in subsection $\ref{subsec3.2},$ we determine all indecomposable polynomial entities that contribute fundamentally to the generation of a polynomial algebra through the use of the Poisson bracket $\{\cdot,\cdot\}$.
\vskip 0.5cm

\subsection{The reduction chain $\mathfrak{su}(4) \supset \mathfrak{su}(2) \times \mathfrak{su}(2)$}
\label{subsec3.1}
The fifteen-dimensional Lie algebra $\mathfrak{su}(4)$ is given in terms of the basis: 
\begin{equation}
\{S_1, S_2, S_3, T_1, T_2, T_3, Q_{11}, Q_{12}, Q_{13}, Q_{21}, Q_{22}, Q_{23}, Q_{31}, Q_{32}, Q_{33}\} 
\label{gen}
\end{equation}
and the $\mathfrak{su}(2) \times \mathfrak{su}(2)$ subalgebra is spanned by the six generators $\{S_1, S_2, S_3, T_1, T_2, T_3\}$. The commutation relations, for $1 \leq i,j,k,\alpha,\beta \leq 3$, are the following:
\begin{align}
[S_i, S_j]&={\rm i} \epsilon_{ijk}S_k \qquad [T_\alpha, T_\beta]={\rm i} \epsilon_{\alpha \beta \gamma}T_\gamma \qquad [S_i, T_\alpha]=0 \label{a} \\
[S_i, Q_{j\alpha}]&={\rm i} \epsilon_{ijk} Q_{k \alpha} \quad \,\, [T_\alpha, Q_{i\beta}]={\rm i} \epsilon_{\alpha \beta \gamma} Q_{i \gamma} \quad [Q_{i \alpha}, Q_{j \beta}]=\frac{{\rm i}}{4}(\delta_{\alpha \beta} \epsilon_{ijk}S_k+\delta_{ij}\epsilon_{\alpha \beta \gamma}T_\gamma) \, , \label{b}
\end{align}
where summation over repeated indices is understood and $\rm i = \sqrt{-1}$. Here $\delta_{ij}$ is the Kronecker delta and $\epsilon_{ijk}$ is the Levi-Civita symbol. Under the basis $\eqref{gen},$ the Lie algebra $\mathfrak{su}(4)$ admits the (vector space) decomposition $ \mathfrak{su}(4) = \mathfrak{g}_1 \oplus \mathfrak{g}_2 \oplus \mathfrak{g}_3$ with 
\begin{align*}
   \mathfrak{g}_1 = \mathrm{span} \{S_1,S_2,S_3\} , \text{ } \mathfrak{g}_2 = \mathrm{span} \{T_1,T_2,T_3\} \text{ and } \mathfrak{g}_3 = \mathrm{span} \{Q_{11},\ldots,Q_{33}\}.
\end{align*} In this case, we consider the following linear coordinates: 
\begin{align}
\boldsymbol{x}:&=\{x_1, x_2, x_3, x_4, x_5, x_6,x_7,x_8,x_9,x_{10},x_{11},x_{12}, x_{13}, x_{14}, x_{15}\} \nonumber \\ &=\{s_1,s_2,s_3,t_1,t_2,t_3,q_{11},q_{12},q_{13},q_{21},q_{22},q_{23},q_{31},q_{32},q_{33}\} \, ,
\label{lincoord}
\end{align}
 and restrict to consider the commutant related to the subset of elements: 
 \begin{equation}
 \boldsymbol{x}':= \{x_1, x_2, x_3, x_4, x_5, x_6\} = \{s_1, s_2, s_3, t_1, t_2, t_3\} 
 \label{subalgebcoordin}
 \end{equation}
with respect to the Berezin bracket  \eqref{FPB}. In the commutative setting, the brackets of coordinates \eqref{xcoord} explicitly read:
\begin{equation}
\begin{split}
\{s_i, s_j\}&={\rm i} \epsilon_{ijk}s_k \qquad \{t_\alpha, t_\beta\}={\rm i} \epsilon_{\alpha \beta \gamma}t_\gamma \qquad \{s_i, t_\alpha\}=0 \\
\{s_i,q_{j\alpha}\}&={\rm i} \epsilon_{ijk} q_{k \alpha} \quad \,\, \{t_\alpha, q_{i\beta}\}={\rm i} \epsilon_{\alpha \beta \gamma} q_{i \gamma} \quad \{q_{i \alpha}, q_{j \beta}\}=\frac{{\rm i}}{4}(\delta_{\alpha \beta} \epsilon_{ijk}s_k+\delta_{ij}\epsilon_{\alpha \beta \gamma}t_\gamma) \, .
\label{classicalrels}
\end{split}
\end{equation}
 These relations will serve as the building blocks for deriving the polynomial algebra relations.
\vskip 0.5cm

\subsection{Construction of the generators of the commutant}
\label{subsec3.2}

We now construct a set of linearly independent and indecomposable polynomials arising from the analysis of the commutant of the subalgebra $\mathfrak{su}(2) \times \mathfrak{su}(2)$ in the enveloping algebra of $\mathfrak{su}(4)$. This is done in the Poisson (commutative) setting. Let us remark that elements of degree one will be denoted by uppercase letters $A_i$, elements of degree two by $B_j$, and so on, following the alphabetical order. Similarly, central elements will be represented by lowercase letters, such as $a_k$, $b_l$, $c_m$, and so on, again adhering to the alphabetical order to indicate the degree of the corresponding homogeneous polynomials. This convention facilitates tracking both the degree of the generators and the central elements while ensuring a clear distinction between them. For the compact forms, we will instead indicate with ${\bf A} \equiv  {\bf A}_1$, ${\bf B} \equiv  {\bf A}_2$,   ${\bf C} \equiv  {\bf A}_3$, and so on, the corresponding sets containing homogeneous polynomials of degree $1, 2, 3$, and so on, respectively (see \eqref{eq:compact} for a direct comparison with the notation used in subsection \ref{subsec2.3}). We proceed with our systematic procedure as given in subsection $\ref{subsec2.3}$ to find the classical (unsymmetrized) elements defining $\mathcal{Q}_{\mathfrak{su}(4) }(d)$ corresponding to the coordinates $\boldsymbol{x}'$  reported in \eqref{subalgebcoordin}.
Taking into account the set of elements $\boldsymbol{x}$ as reported in \eqref{lincoord}, the general homogeneous polynomial at degree $k$ is given by 
\begin{equation}
p^{(k)}(\boldsymbol{x})=\sum_{a_1+\ldots + a_{15}= {k}} \Gamma_{a_1, \dots, a_{15}} \,x_1^{a_1} \cdots x_{15}^{a_{15}}  
\label{eq:exp1}
\end{equation}
with the commutant constraint:
\begin{equation}
\begin{split}
\{s_i, p^{(k)}(\boldsymbol{x})\}&=0 \, , \quad i=1,2,3 \\
 \{t_\alpha, p^{(k)}(\boldsymbol{x})\}&=0  \, ,\quad \alpha=1,2,3.
\label{pdes}
\end{split}
\end{equation}
If $ k \geq 1$ is the degree of the considered  polynomial, and $ n$ is the number of generators of the Lie algebra, then the expansion \eqref{eq:exp1} involves a total number of terms $N$ given by
\begin{equation}
N=\binom{ n+k-1}{k}=\frac{( n+k-1)!}{k!( n-1)!} \, .
\end{equation}
For $ n=15$ we obtain 
\begin{equation}
N=\binom{14+k}{k}=\frac{(14+k)!}{k!14!} \, .
\end{equation}
Thus, a first-degree ($k=1$) expansion consists of $N=15!/(1!14!)=15$ terms, while a second-degree ($k=2$) expansion includes $N=16!/(2!14!)=120$ terms. For third-degree ($k=3$), there are $N=17!/(3!14!)=680$ terms, and so forth.  It is easy to verify that there are no degree-one elements that Poisson commute simultaneously with $\{s_i, t_\alpha\}$. Moving to the degree $k=2$ expansion, the equations \eqref{pdes} yield three independent and indecomposable solutions, given by 
\begin{equation}
\begin{split}
p_1^{(2)}&:=s_1^2 + s_2^2 + s_3^2, \\
p_2^{(2)}&:=t_{1}^2 + t_{2}^2 + t_{3}^2 ,\\
p_3^{(2)}&:=q_{11}^2+q_{12}^2+q_{13}^2+q_{21}^2+q_{22}^2+q_{23}^2+q_{31}^2+q_{32}^2+q_{33}^2 .
\end{split}
\end{equation}
It is easily seen that the (unsymmetrized) quadratic Casimir of $\mathfrak{su}(4)$ corresponds to the combination:
\begin{equation}
c^{[2]}=p_1^{(2)}+p_2^{(2)}+4 p_3^{(2)} \, ,
\end{equation}
where $p_1^{(2)}$ and $p_2^{(2)}$ are the quadratic Casimir elements associated to the subalgebra $\mathfrak{su}(2) \times \mathfrak{su}(2)$. The third-order expansion leads, once solved the commutant constraint for the undetermined coefficients, to the following two cubic polynomials as solutions:
\begin{equation}
\begin{split}
p_1^{(3)}&:=(q_{11} s_1  + q_{21} s_2 + q_{31} s_3) t_1 + (q_{12} s_1  + q_{22} s_2  + 
q_{32} s_3) t_2 + (q_{13} s_1  + q_{23} s_2  + q_{33} s_3) t_3 ,\\
p_2^{(3)}&:=(q_{12} q_{23}  - q_{13} q_{22}) q_{31} + (q_{13} q_{21}  - q_{11} q_{23}) q_{32} + (q_{11} q_{22}  - 
q_{12} q_{21}) q_{33}  \, . 
\end{split}
\end{equation}
It is straightforward to verify that the (unsymmetrized) cubic Casimir of $\mathfrak{su}(4)$ corresponds to:
\begin{equation}
c^{[3]}=p_1^{(3)}-4 p_2^{(3)} \, .
\end{equation} 
Increasing the degree, for fourth-degree polynomials, ten solutions to the system of PDEs are identified. Among these, three are excluded as they result from the squares of $p_1^{(2)}$, $p_2^{(2)}$, and the product $p_1^{(2)} p_2^{(2)}$, respectively. The remaining indecomposable polynomial solutions are:

\begin{align}
p_1^{(4)}&= (q_{11}^2 + q_{12}^2  + q_{13}^2) s_1^2 + 2 (q_{11} q_{21}  + 
q_{12} q_{22}  +  q_{13} q_{23})s_1 s_2 + (q_{21}^2 + q_{22}^2  + 
q_{23}^2) s_2^2 \nonumber \\
&+ 
2 ((q_{11} q_{31}  + q_{12} q_{32}  + q_{13} q_{33})s_1 s_3  + (q_{21} q_{31}  + q_{22} q_{32}  + 
q_{23} q_{33})s_2 s_3)  + (q_{31}^2 + q_{32}^2 + q_{33}^2) s_3^2 \, , \nonumber \\
p_2^{(4)}&= (q_{21}^2  + q_{22}^2  + q_{23}^2  + q_{31}^2  + q_{32}^2  + 
q_{33}^2) s_1^2 - 2 (q_{11} q_{21}  + q_{12} q_{22}  + q_{13} q_{23} )s_1 s_2 \nonumber  \\
& + 
(q_{11}^2 + q_{12}^2 + q_{13}^2  + q_{31}^2  + q_{32}^2 + 
q_{33}^2 )s_2^2 - 2 ((q_{11} q_{31}  + q_{12} q_{32}  + q_{13} q_{33})s_1 \nonumber  \\
&+ (q_{21} q_{31} + q_{22} q_{32}  + 
q_{23} q_{33}) s_2) s_3 + (q_{11}^2 + q_{12}^2 + q_{13}^2 + q_{21}^2 + q_{22}^2 + 
q_{23}^2) s_3^2 \, ,\nonumber   \\
p_3^{(4)}&= (q_{12} q_{23}  - q_{13} q_{22} ) s_3 t_1 + (q_{13} q_{21} - q_{11} q_{23}) s_3 t_2 + (q_{11} q_{22} - 
q_{12} q_{21}) s_3 t_3\nonumber  \\
&+ ((q_{13} q_{32} - q_{12} q_{33}) t_1 + (q_{11} q_{33} - 
q_{13} q_{31}) t_2 + (q_{12} q_{31} - q_{11} q_{32}) t_3) s_2 \nonumber \\
&+ ((q_{22} q_{33} - 
q_{23} q_{32}) t_1 + (q_{23} q_{31} - q_{21} q_{33}) t_2 + (q_{21} q_{32} - 
q_{22} q_{31} ) t_3) s_1 \, , \nonumber  \\
p_4^{(4)}&=(q_{11}^2 + q_{21}^2 + q_{31}^2) t_1^2 + 
2 (q_{11} q_{12} + q_{21} q_{22} + q_{31} q_{32}) t_1 t_2 + (q_{12}^2 + q_{22}^2 + 
q_{32}^2) t_2^2\nonumber  \\
&+ 2 ((q_{11} q_{13} + q_{21} q_{23} + q_{31} q_{33}) t_1 + (q_{12} q_{13} + q_{22} q_{23} + 
q_{32} q_{33}) t_2) t_3 + (q_{13}^2 + q_{23}^2 + q_{33}^2) t_3^2 \, , \nonumber  \\
p_5^{(4)}&=(q_{12}^2 + q_{13}^2 + q_{22}^2 + q_{23}^2 + q_{32}^2 + q_{33}^2) t_1^2 - 
2 (q_{11} q_{12} + q_{21} q_{22} + q_{31} q_{32}) t_1 t_2 \nonumber \\
&+ (q_{11}^2 + q_{13}^2 + q_{21}^2 + 
q_{23}^2 + q_{31}^2 + q_{33}^2) t_2^2 - 
2 ((q_{11} q_{13} + q_{21} q_{23} + q_{31} q_{33}) t_1 \nonumber \\
& + (q_{12} q_{13} + q_{22} q_{23} + 
q_{32} q_{33}) t_2) t_3 + (q_{11}^2 + q_{12}^2 + q_{21}^2 + q_{22}^2 + q_{31}^2 + 
q_{32}^2) t_3 \, ,  \nonumber \\
p_6^{(4)}&=q_{11}^4 + 2 q_{11}^2 q_{12}^2 + q_{12}^4 + 2 q_{11}^2 q_{13}^2 + 
2 q_{12}^2 q_{13}^2 + q_{13}^4 + 2 q_{11}^2 q_{21}^2 + q_{21}^4 + 4 q_{11} q_{12} q_{21} q_{22}\nonumber \\
& + 
2 q_{12}^2 q_{22}^2 + 2 q_{21}^2 q_{22}^2 + q_{22}^4 + 4 q_{11} q_{13} q_{21} q_{23} + 
4 q_{12} q_{13} q_{22} q_{23} + 2 q_{13}^2 q_{23}^2 + 2 q_{21}^2 q_{23}^2 + 
2 q_{22}^2 q_{23}^2\nonumber  \\
&+ q_{23}^4 + 2 q_{11}^2 q_{31}^2 + 2 q_{21}^2 q_{31}^2 + q_{31}^4 + 4 q_{11} q_{12} q_{31} q_{32} + 4 q_{21} q_{22} q_{31} q_{32} + 2 q_{12}^2 q_{32}^2 \nonumber \\
&+ 
2 q_{22}^2 q_{32}^2 + 2 q_{31}^2 q_{32}^2 + q_{32}^4 + 4 q_{11} q_{13} q_{31} q_{33} + 4 q_{21} q_{23} q_{31} q_{33} + 4 q_{12} q_{13} q_{32} q_{33} + 4 q_{22} q_{23} q_{32} q_{33}\nonumber  \\
&+ 2 q_{13}^2 q_{33}^2 + 2 q_{23}^2 q_{33}^2 + 2 q_{31}^2 q_{33}^2 + 2 q_{32}^2 q_{33}^2 + q_{33}^4 \, ,\nonumber  \\
p_7^{(4)} &=q_{12}^2 q_{21}^2 + q_{13}^2 q_{21}^2 - 2 q_{11} q_{12} q_{21} q_{22} + q_{11}^2 q_{22}^2 + 
q_{13}^2 q_{22}^2 - 2 q_{11} q_{13} q_{21} q_{23} - 2 q_{12} q_{13} q_{22} q_{23} + q_{11}^2 q_{23}^2 \nonumber \\
&+ 
q_{12}^2 q_{23}^2 + q_{12}^2 q_{31}^2 + q_{13}^2 q_{31}^2 + q_{22}^2 q_{31}^2 + 
q_{23}^2 q_{31}^2 - 2 q_{11} q_{12} q_{31} q_{32} - 2 q_{21} q_{22} q_{31} q_{32} + q_{11}^2 q_{32}^2 + 
q_{13}^2 q_{32}^2 \nonumber \\
&+ q_{21}^2 q_{32}^2 + q_{23}^2 q_{32}^2 - 2 q_{11} q_{13} q_{31} q_{33} - 2 q_{21} q_{23} q_{31} q_{33} - 2 q_{12} q_{13} q_{32} q_{33} - 2 q_{22} q_{23} q_{32} q_{33} + q_{11}^2 q_{33}^2 \nonumber \\
&+ q_{12}^2 q_{33}^2 + q_{21}^2 q_{33}^2 + q_{22}^2 q_{33}^2 \, ,
\end{align}
to which the following three relations have to be added:
\begin{align}
p_1^{(4)}+p_2^{(4)}=p_1^{(2)}p_3^{(2)} \qquad p_4^{(4)}+p_5^{(4)}=p_2^{(2)}p_3^{(2)}  \qquad  p_6^{(4)}+2 p_7^{(4)}=\bigl(p_3^{(2)}\bigl)^2 \, ,
\end{align}
meaning that we can drop the unessential polynomials $p_2^{(4)}, p_5^{(4)}$ and $p_7^{(4)}$. We are thus left with four degree-four polynomials $p_1^{(4)}, p_3^{(4)}, p_4^{(4)}, p_6^{(4)}$. At this level, the (unsymmetrized) quartic Casimir of $\mathfrak{su}(4)$ is given by:
\begin{equation}
c^{[4]}=p_1^{(4)}-2 p_3^{(4)}+p_4^{(4)}-2p_6^{(4)}-p_1^{(2)}p_3^{(2)}-p_2^{(2)}p_3^{(2)}-\frac{1}{8}\bigl(\bigl(p_1^{(2)}\bigl)^2+\bigl(p_2^{(2)}\bigl)^2\bigl) \, .
\end{equation}
Thus, up to degree four, we are left with nine polynomials and we choose the following elements:
{\small \begin{equation}
\{b_1, b_2, b_3, c_1, C_2, d_1, D_2, D_3, D_4\}:=\{p_1^{(2)}, p_2^{(2)},p_3^{(2)},p_1^{(3)}-4 p_2^{(3)}, p_1^{(3)}, p_1^{(4)}-2 p_3^{(4)}+p_4^{(4)}-2p_6^{(4)}, p_3^{(4)}, p_4^{(4)}, p_6^{(4)} \} \, ,
\label{poldeg4}
\end{equation}}

\noindent where we have indicated with lowercase letters the central elements and with uppercase letters the generators. In this notation, the three (unsymmetrized) Casimir elements of $\mathfrak{su}(4)$ corresponds to the following combinations:
\begin{align}
c^{[2]}=b_1+b_2+4 b_3 , \qquad c^{[3]}=c_1 , \qquad c^{[4]}=d_1-(b_1 +b_2) b_3-\frac{1}{8}(b_1^2+b_2^2).
\end{align}
Among the nine polynomials derived, five serve as central elements while the remaining four are intended as generators to close a polynomial algebra. To achieve closure of the algebra, an additional degree-five element must be incorporated, specifically derived from the commutant constraint \eqref{pdes} with the expansion \eqref{eq:exp1} for $k=5$. Solving the equations, seven degree-five polynomials, namely $p_1^{(5)}, \dots, p_7^{(5)}$, are produced. Only one of them, i.e. $p_4^{(5)}$, unable to be expressed as a combination of the lower-degree ones. Explicitly, it reads:
{\small \begin{align}
p_4^{(5)}&=q_{11}^3 s_1 t_1 + q_{11} q_{12}^2 s_1 t_1 + q_{11} q_{13}^2 s_1 t_1 + q_{11} q_{21}^2 s_1 t_1 + 
q_{12} q_{21} q_{22} s_1 t_1 + q_{13} q_{21} q_{23} s_1 t_1 + q_{11} q_{31}^2 s_1 t_1\nonumber \\
& + 
q_{12} q_{31} q_{32} s_1 t_1 + q_{13} q_{31} q_{33} s_1 t_1 + q_{11}^2 q_{21} s_2 t_1 + 
q_{21}^3 s_2 t_1 + q_{11} q_{12} q_{22} s_2 t_1 + q_{21} q_{22}^2 s_2 t_1 + 
q_{11} q_{13} q_{23} s_2 t_1 \nonumber \\
&+ q_{21} q_{23}^2 s_2 t_1 + q_{21} q_{31}^2 s_2 t_1 + 
q_{22} q_{31} q_{32} s_2 t_1 + q_{23} q_{31} q_{33} s_2 t_1 + q_{11}^2 q_{31} s_3 t_1 + 
q_{21}^2 q_{31} s_3 t_1 + q_{31}^3 s_3 t_1\nonumber \\
& + q_{11} q_{12} q_{32} s_3 t_1 + 
q_{21} q_{22} q_{32} s_3 t_1 + q_{31} q_{32}^2 s_3 t_1 + q_{11} q_{13} q_{33} s_3 t_1 + 
q_{21} q_{23} q_{33} s_3 t_1 + q_{31} q_{33}^2 s_3 t_1 + q_{11}^2 q_{12} s_1 t_2\nonumber  \\
&+ q_{12}^3 s_1 t_2 +
q_{12} q_{13}^2 s_1 t_2 + q_{11} q_{21} q_{22} s_1 t_2 + q_{12} q_{22}^2 s_1 t_2 + 
q_{13} q_{22} q_{23} s_1 t_2 + q_{11} q_{31} q_{32} s_1 t_2 + q_{12} q_{32}^2 s_1 t_2\nonumber  \\
&+ 
q_{13} q_{32} q_{33} s_1 t_2 + q_{11} q_{12} q_{21} s_2 t_2 + q_{12}^2 q_{22} s_2 t_2 + 
q_{21}^2 q_{22} s_2 t_2 + q_{22}^3 s_2 t_2 + q_{12} q_{13} q_{23} s_2 t_2 + q_{22} q_{23}^2 s_2 t_2 \nonumber \\
& +
q_{21} q_{31} q_{32} s_2 t_2 + q_{22} q_{32}^2 s_2 t_2 + q_{23} q_{32} q_{33} s_2 t_2 + 
q_{11} q_{12} q_{31} s_3 t_2 + q_{21} q_{22} q_{31} s_3 t_2 + q_{12}^2 q_{32} s_3 t_2 + 
q_{22}^2 q_{32} s_3 t_2 \nonumber \\
&+ q_{31}^2 q_{32} s_3 t_2 + q_{32}^3 s_3 t_2 + q_{12} q_{13} q_{33} s_3 t_2 +
q_{22} q_{23} q_{33} s_3 t_2 + q_{32} q_{33}^2 s_3 t_2 + q_{11}^2 q_{13} s_1 t_3 + 
q_{12}^2 q_{13} s_1 t_3 + q_{13}^3 s_1 t_3 \nonumber \\
&+ q_{11} q_{21} q_{23} s_1 t_3 + 
q_{12} q_{22} q_{23} s_1 t_3 + q_{13} q_{23}^2 s_1 t_3 + q_{11} q_{31} q_{33} s_1 t_3 + 
q_{12} q_{32} q_{33} s_1 t_3 + q_{13} q_{33}^2 s_1 t_3 + q_{11} q_{13} q_{21} s_2 t_3 \nonumber  \\
&+ 
q_{12} q_{13} q_{22} s_2 t_3 + q_{13}^2 q_{23} s_2 t_3 + q_{21}^2 q_{23} s_2 t_3 + 
q_{22}^2 q_{23} s_2 t_3 + q_{23}^3 s_2 t_3 + q_{21} q_{31} q_{33} s_2 t_3 + 
q_{22} q_{32} q_{33} s_2 t_3 \nonumber \\
&+ q_{23} q_{33}^2 s_2 t_3 + q_{11} q_{13} q_{31} s_3 t_3 + 
q_{21} q_{23} q_{31} s_3 t_3 + q_{12} q_{13} q_{32} s_3 t_3 + q_{22} q_{23} q_{32} s_3 t_3 + 
q_{13}^2 q_{33} s_3 t_3 + q_{23}^2 q_{33} s_3 t_3 \nonumber  \\
&+ q_{31}^2 q_{33} s_3 t_3 + 
q_{32}^2 q_{33} s_3 t_3 + q_{33}^3 s_3 t_3 \, .
\label{poldeg5}
\end{align}}

\noindent We define this polynomial to be $E_1:=p_4^{(5)}$.  The degree-six expansion, once solved the commutant constraint \eqref{pdes}, leads to twenty-nine polynomials $p_1^{(6)}, \ldots, p_{29}^{(6)}$. Among them, only  four turn out to be non expressible in terms of the lower-order ones, i.e. $p_{11}^{(6)}, p_{12}^{(6)}, p_{17}^{(6)}, p_{23}^{(6)}$. We note these indecomposable solutions by $F_j$ as follows  
\begin{equation}
F_1:=p_{11}^{(6)}, \quad F_2:=p_{12}^{(6)}, \quad F_3:=p_{17}^{(6)}, \quad F_4:=p_{23}^{(6)} \, .
\end{equation} 
Their explicit expressions are reported in the Appendix \ref{appendixA}.  Concerning degree-seven expansions, the commutant constraint lead to twenty-five solutions $p_1^{(7)}, \dots, p_{25}^{(7)}$, among which just two of them are not expressible as combinations of the lower-order polynomials, namely $p_9^{(7)}$ and $p_{16}^{(7)}$. Thus, besides the polynomials \eqref{poldeg4}, we need to consider the following additional polynomials of degree six and seven respectively:
\begin{equation}
\{F_1, F_2, F_3, F_4, G_1, G_2\}:=\{p_{11}^{(6)}, p_{12}^{(6)}, p_{17}^{(6)}, p_{23}^{(6)}, p_9^{(7)}, p_{16}^{(7)} \}\, = \boldsymbol{p}^{(6)} \sqcup \boldsymbol{p}^{(7)} .
\end{equation}
We thus have the following set composed by sixteen polynomials up to degree seven coming from a systematic analysis of the commutant related to the subalgebra $\mathfrak{su}(2) \times \mathfrak{su}(2)$:
\begin{equation}
\{b_1, b_2, b_3, c_1, C_2, d_1, D_2, D_3, D_4, E_1, F_1, F_2, F_3, F_4, G_1, G_2 \} \, .
\end{equation}
It is also recommended to compute the expansion for $k=8$ and beyond. A brief look at the Berezin brackets involving polynomials of total degree eight reveals that they cannot all be expressed as combinations of lower-degree polynomials. This suggests the presence of linearly independent degree-eight elements. Nevertheless, the expansion for $k=8$ includes $N=22!/(8!14!)=319770$ terms, which makes computation impractical within a reasonable timeframe. This conclusion is supported by analyzing the computation time $t(k)$ as a function of the expansion degree $k$. Up to degree seven, the recorded computation times (as derived from our computer with our specific code implementation) are as follows:
{\footnotesize \begin{equation}
\{deg \,1, deg \,2, deg \,3, deg\, 4, deg \,5, deg \,6, deg \,7 \}=\{3.70139 \cdot 10^{-2} \,s, 3.37363 \cdot 10^{-1} \,s, 1.59624 \, s, 2.41622\cdot 10^1 \, s, 5.95406 \cdot 10^2 \, s, 1.13846 \cdot 10^4 \, s, 1.94800\cdot 10^5 \, s\} \, .\nonumber
\end{equation}}
%If we plot $t(n)$ vs $n$, we get what it looks like an exponential behavior as follows:
%  \begin{figure}[http]
%	\centering
%	\includegraphics[scale=0.65]{1.pdf}
%	\caption{Plot of the computation time $t(n)$ as a function of the degree of the expansion $n$.}
%	\label{fig1}
%\end{figure}

\noindent  In order to estimate roughly the time needed to compute the $k=8$ expansion, we can fit these points with a parametric function $f_{\alpha,\beta}(x)=\alpha \exp (\beta x)$ (we are essentially doing a two-parameter exponential fit). By doing so with the help of Wolfram Mathematica$^{\text{\textregistered}}$, in particular the \textsf{FindFit} function, we find the values $\alpha^*= 4.51331 \cdot 10^{-4} \,s$, $\beta^*=2.84043$ (see the left plot reported in \figref{fig1} below). 
\noindent Thus, the estimated time to solve the commutant constraint \eqref{pdes} for degree $k=8$  expansions is around $f_{\alpha^*, \beta^*}(8)=3.33549 \cdot 10^6 \, s$ (see the right plot reported in \figref{fig1}) , which is roughly equivalent to 39 days. It is important to mention that completing the degree-seven expansions took approximately 54 hours, namely more than two days. 
\begin{figure}[http]
	\centering
	\includegraphics[scale=0.45]{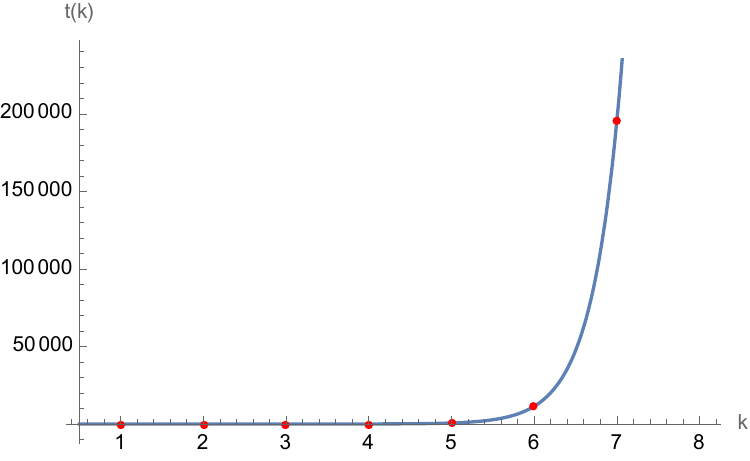} \quad 	\includegraphics[scale=0.45]{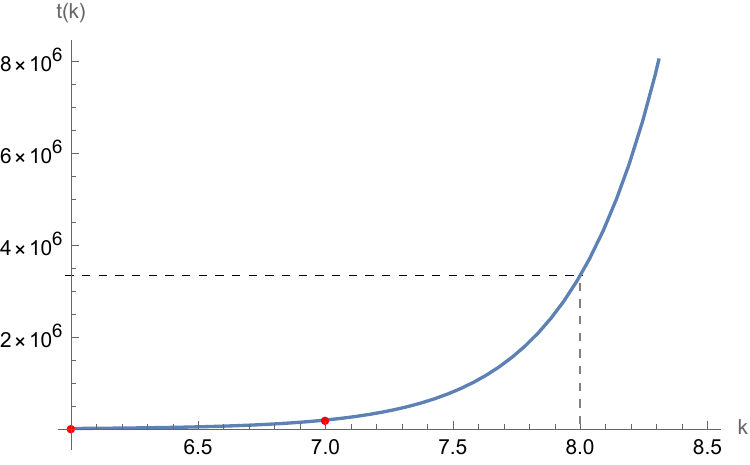}
	\caption{Plot of the function $f_{\alpha^*, \beta^*}(x)$.}
	\label{fig1}
\end{figure}

\noindent Therefore, the best strategy entails creating new indecomposable polynomials of degree eight by employing Berezin brackets from lower-degree polynomials. From a computational standpoint, constructing these systematically up to degree seven is achievable within a reasonable amount of time. Nevertheless, it is conceivable that more efficient programming techniques,  combined with the use of more powerful computers, could generate these polynomials within a more acceptable computation time.

To establish closure in the polynomial algebra, it is essential to include the new generators that appear within the Berezin brackets of higher degree. To identify higher degree indecomposable generators, we can compute the expansion pertaining to the non-zero Poisson brackets. Specifically, non-zero Berezin brackets start from the sixth degree. For instance:
\begin{equation}
\{C_2, D_2\} \, ,\qquad \{C_2, D_3\} \, ,\qquad  \{C_2, D_4\} \, .
\end{equation}
When considering expansions, we formally denote the sets with elements of degree two, three, four, and six as $\mathbf{B}$, $\mathbf{C}$, $\mathbf{D}$, and $\mathbf{F}$ respectively. In the compact form, we seek degree-six relations of the form:
\begin{equation}
\{\mathbf{C}, \mathbf{D}\} \sim \mathbf{F}+\mathbf{B}\mathbf{D}+\mathbf{C}\mathbf{C}+\mathbf{B}\mathbf{B}\mathbf{B} \, .
\label{deg6}
\end{equation}
In particular, the above equation formally represents the following expansion, such as for $C_2$ and $D_2$:
\begin{equation}
\begin{split}
\{C_2, D_2\}&=\alpha_1  F_1+\alpha_2  F_2+\alpha_3  F_3+\alpha_4  F_4+(\beta_1 b_1+\beta_2 b_2+\beta_3 b_3)d_1+(\beta_4 b_1+\beta_5 b_2+\beta_6 b_3)D_2\\
&+(\beta_7 b_1+\beta_8 b_2+\beta_9 b_3)D_3 +(\beta_{10} b_1+\beta_{11} b_2+\beta_{12} b_3)D_4+(\gamma_1 c_1^2+\gamma_2 C_2^2+\gamma_3 c_1 C_2)\\
&+(\delta_1 b_1^3+\delta_2 b_2^3+\delta_3 b_3^3+\delta_4 b_1^2 b_2+ \delta_5 b_1 b_2^2+\delta_6 b_1^2 b_3+\delta_7 b_1 b_3^2+\delta_8 b_2^2 b_3+\delta_9 b_2 b_3^2+\delta_{10} b_1 b_2 b_3) \, ,
\end{split}
\end{equation}
which needs to be solved for the unknown coefficients $\{\alpha_1, \dots, \alpha_{4}, \beta_1, \dots, \beta_{12}, \gamma_1, \dots, \gamma_{3}, \delta_1, \dots, \delta_{10}\}$. This same expansion applies to $\{C_2, D_3\}$ and $\{C_2, D_4\}$, that is, for Berezin brackets that result in degree-six polynomials. We will consistently use this formal notation to represent the expansions considered throughout the paper.
 At this level, by solving the algebraic equations obtained for the undetermined coefficients, we get the following:
\begin{equation}
\{C_2, D_2\}=-{\rm i} (F_1+F_3) \, , \qquad \{C_2, D_3\}= -2 {\rm i} F_3 \, , \qquad \{C_2, D_4\}=0 \, .
\end{equation}
At degree-seven, for Berezin brackets involving $\{\mathbf{C},\mathbf{E}\}$ and $\{\mathbf{D},\mathbf{D}\}$, we look for expansions of the type:
\begin{equation}
\{\mathbf{C},\mathbf{E}\} \sim \{\mathbf{D},\mathbf{D}\}  \sim \mathbf{G}+\mathbf{B}\mathbf{E}+\mathbf{C}\mathbf{D}+\mathbf{B}\mathbf{B}\mathbf{C} \, .
\label{deg7}
\end{equation}
The only non-zero Berezin brackets are the following:
\begin{equation}
\{D_2, D_3\}=\{D_3, D_4\}=2{\rm i}G_2 \, , \qquad \{D_2, D_4\}=\{E_1, C_2\}={\rm i}(G_1+G_2) \, .
\end{equation}
These polynomials of degree seven Poisson commute with all the generators of the subalgebra and, instead of presenting their explicit expressions (which are cumbersome), we simply notice that they can be expressed as:
 \begin{equation}
G_2=-\frac{{\rm i}}{2} \{D_2, D_3\}=-\frac{{\rm i}}{2}\{D_3, D_4\} \, , \qquad G_1 =-{\rm i}\{D_2, D_4\}-G_2=-{\rm i}\{E_1, C_2\}-G_2\, .
 \end{equation}
At this stage, we are prepared to calculate the Berezin brackets of degree eight. If these brackets can be written as combinations of lower-order polynomials, then the Berezin brackets $\{\mathbf{C}, \mathbf{F}\}$ and $\{\mathbf{D}, \mathbf{E}\}$ will employ the subsequent formal expansion:
\begin{equation}
\{\mathbf{C}, \mathbf{F}\} \sim \{\mathbf{D}, \mathbf{E}\}\sim  \mathbf{B}\mathbf{F}+\mathbf{C}\mathbf{E}+\mathbf{D}\mathbf{D}+\mathbf{B}\mathbf{C}\mathbf{C}+\mathbf{B}\mathbf{B}\mathbf{D}+\mathbf{B}\mathbf{B}\mathbf{B}\mathbf{B} \,  .
\label{deg8}
\end{equation}
The latter, once solved the algebraic equations obtained for the undetermined coefficients, lead to the following relations:
\begin{align}
\{C_2, F_{1}\}&= {\rm i}\bigl(-b_1 b_2 b_3^2-b_3 C_2^2+\frac{1}{4}b_1 \bigl(c_1-C_2\bigl)C_2+2 b_2 b_3(d_1+2D_2)+(d_1+4 D_2)D_2 \nonumber \\
&\hskip 0.5cm+((b_1-2 b_2)b_3-D_2)D_3+\bigl(\bigl(-\frac{1}{2}b_1+4 b_3\bigl)b_2+2D_2\bigl)D_4+C_2 E_1-b_2 F_2-\frac{1}{4}b_1 F_4\bigl) \, ,\nonumber  \\
\{C_2, F_{2}\}&= {\rm i} H_1  \, ,\nonumber \\
\{C_2, F_{3}\}&=   {\rm i}\bigl(-b_1 b_2 b_3^2-b_3 C_2^2+\frac{1}{4}b_2 \bigl(c_1-C_2\bigl)C_2+ b_2 b_3(d_1+2D_2)+2D_2^2+((2b_1- b_2)b_3+D_2)D_3 \nonumber \\
&\hskip 0.5cm+\bigl(-\frac{1}{2}b_1+2 b_3\bigl)b_2D_4+C_2 E_1-\frac{1}{2}b_2 F_2-b_1 F_4\bigl)\, ,\nonumber \\
\{C_2, F_{4}\}&={\rm i} H_2   \, , \nonumber  \\
\{D_2, E_{1}\}&=\frac{{\rm i}}{4}(b_2 F_1+b_1 F_2-H_1-H_2)  \, , \nonumber \\
\{D_3, E_{1}\}&=\frac{{\rm i}}{2}(H_2+b_2 F_1+4 b_3 F_3) \, , \nonumber \\
\{D_4, E_{1}\}&=\frac{{\rm i}}{2}(H_1+H_2+2 b_3(F_1+F_3)) \, ,
\end{align}
meaning that we need to introduce two new linearly independent and indecomposable degree-eight polynomials $H_1, H_2$ defined through the Berezin brackets, i.e.:
\begin{equation}
H_1:=-{\rm i}\{C_2, F_2\} \, , \qquad H_2:=-{\rm i}\{C_2, F_4\} \, .
\end{equation}
Notice that in some of the relations appear the new indecomposable polynomials on the right hand side. In fact, once a new indecomposable polynomial is found, it has to be added to the considered expansion. At this degree, there is also an algebraic relation involving the generators, i.e.:
\begin{align}
\frac{1}{2}b_1 b_2 b_3^2&+b_3 C_2^2-b_2 b_3(d_1+2D_2)-D_2^2-((b_1-b_2)b_3-d_1-2D_2+D_3)D_3 \nonumber \\
&+\bigl(\bigl(\frac{1}{2}b_1-2b_3\bigl)b_2+2D_3\bigl)D_4-2 C_2 E_1+\frac{1}{2}b_2 F_2+\frac{1}{2}b_1 F_4=0 \, .
\end{align}
 We remark that, although $H_1, H_2$ are indecomposable,  we cannot be sure at this level that they are the only polynomials of degree eight in $\boldsymbol{p}^{(8)}$. To understand this point, it is sufficient to look back at the degree-six generators and associated Berezin brackets of the same degree. In fact, in the Berezin brackets above, just two of the four indecomposable polynomials appear, i.e. $F_1$, $F_3$. That said, if we manage to close a polynomial algebra by adding these two polynomials to the previous generators, then the eighteen elements:
\begin{equation}
\{b_1, b_2, b_3, c_1, C_2, d_1, D_2, D_3, D_4, E_1, F_1, F_2, F_3, F_4, G_1, G_2, H_1, H_2\} 
\label{generators}
\end{equation}
would represent the minimal set of linearly independent and indecomposable generators that are necessary to reach closure. Consequently, we proceed by calculating the ninth-degree expansions. These expansions encompass $\{\mathbf{C},\mathbf{G} \}$ and $\{\mathbf{D},\mathbf{F} \}$ as demonstrated in the subsequent formal expansion:
\begin{equation}
\{\mathbf{C}, \mathbf{G}\} \sim \{\mathbf{D}, \mathbf{F}\}\sim  \mathbf{B} \mathbf{G}+ \mathbf{C} \mathbf{F}+ \mathbf{D} \mathbf{E}+ \mathbf{C} \mathbf{C} \mathbf{C}+ \mathbf{B} \mathbf{C} \mathbf{D}+\mathbf{B} \mathbf{B} \mathbf{E}+ \mathbf{B} \mathbf{B} \mathbf{B} \mathbf{C}\,  .
\label{deg9}
\end{equation}
After implementing the expansions and solving the corresponding algebraic equations for the undetermined coefficients, we get the following relations:

\begin{align}
\{D_2, F_1\}&={\rm i}\bigl(\bigl(\frac{1}{16}b_1^2 b_2+\frac{1}{4}b_1 b_2 b_3\bigl)c_1+\bigl(-\frac{1}{16}b_1^2 b_2-\frac{1}{4}b_1 b_2 b_3-\frac{1}{2}b_1 b_3^2\bigl)C_2+\frac{1}{2}c_1 C_2^2-\frac{3}{4}C_2^3  \nonumber \\
&\hskip 0.5cm+\bigl(-\frac{1}{2}b_2 c_1+\bigl(\frac{3}{4}b_2+b_3\bigl)C_2\bigl)d_1+\bigl(\bigl(-\frac{b_1}{4}-b_2\bigl)c_1+\bigl(\frac{1}{2}b_1+\frac{3}{2}b_2+3 b_3\bigl)C_2\bigl)D_2 \nonumber \\
&\hskip 0.5cm+\bigl(\bigl(-\frac{1}{4}b_1+\frac{1}{4}b_2\bigl)c_1+\bigl(\frac{1}{2}b_1-\frac{3}{4}b_2-b_3\bigl)C_2\bigl)D_3+\bigl(-b_2 c_1+\bigl(\frac{3}{2}b_2+2 b_3\bigl)C_2\bigl)D_4 \nonumber \\
&\hskip 0.5cm-\frac{1}{4}b_1 b_2 E_1-D_2 E_1-\frac{1}{2}C_2 F_2\bigl) \, , \nonumber \\
\{D_2, F_2\}&={\rm i}((-c_1+2 C_2)F_1+4 b_3 G_1-b_1 G_2) \, , \nonumber \\
\{D_2, F_3\}&={\rm i}\bigl(\bigl(\frac{1}{16}b_1 b_2^2+\frac{1}{4}b_1 b_2 b_3\bigl)c_1+\bigl(-\frac{1}{16}b_1 b_2^2-\frac{1}{4}b_1 b_2 b_3-\frac{1}{2}b_2 b_3^2\bigl)C_2+\frac{1}{2}c_1 C_2^2-\frac{3}{4}C_2^3 \nonumber \\
&\hskip 0.5cm+\bigl(-\frac{1}{4}b_2 c_1+\frac{1}{2}b_2C_2\bigl)d_1+\bigl(-\frac{3}{4}b_2c_1+\bigl(\frac{3}{2}b_2+b_3\bigl)C_2\bigl)D_2\nonumber \\
&\hskip 0.5cm+\bigl(\bigl(-\frac{1}{2}b_1+\frac{1}{4}b_2\bigl)c_1+\bigl(\frac{3}{4}b_1-\frac{1}{2}b_2+b_3\bigl)C_2\bigl)D_3+\bigl(-\frac{1}{2}b_2 c_1+b_2 C_2\bigl)D_4 \nonumber \\
&\hskip 0.5cm-\frac{1}{4}b_1 b_2 E_1-D_2 E_1-\frac{1}{2}C_2 F_2\bigl) \, ,\nonumber \\
\{D_2, F_4\}&={\rm i}((-c_1+2 C_2)F_3- b_2 G_1+4 b_3 G_2)\, ,\nonumber \\
\{D_3, F_1\}&=\frac{{\rm i}}{2}(c_1 C_2^2-2 C_2^3+b_2 C_2 d_1+2 b_2 C_2 D_2+(b_1-b_2)C_2 D_3+2 b_2 C_2 D_4-b_1 b_2 E_1 +4 D_2 E_1)\, ,\nonumber \\
\{D_3, F_2\}&=2 {\rm i}(C_2 F_1-b_1 G_2)\, ,\nonumber \\
\{D_3, F_3\}&= {\rm i} \bigl(\frac{1}{8}b_1 b_2^2 c_1-\bigl(\frac{1}{8}b_1 b_2^2+b_2 b_3^2\bigl)C_2-\bigl(\frac{1}{2}b_2 c_1-b_2 C_2\bigl)D_2+2 b_3 C_2 D_3-b_2 C_2 D_4\nonumber \\
&\hskip 0.5cm+2 D_3 E_1-2 C_2 F_4\bigl)\, ,\nonumber \\
\{D_3, F_4\}&={\rm i} I_1 \, ,\nonumber \\
\{D_4, F_1\}&={\rm i}\bigl(-\frac{1}{4}b_1 b_2 b_3 (c_1-C_2)-\frac{1}{4}c_1 C_2^2+\frac{1}{4}C_2^3+\frac{b_2}{2}(c_1 -C_2)d_1+(b_2c_1-(b_2+b_3)C_2)D_2\nonumber \\
&\hskip 0.5cm+\bigl(\frac{b_1}{4}-\frac{b_2}{2}\bigl)(c_1-C_2)D_3+\bigl(b_2 c_1-\bigl(\frac{b_1}{2}+b_2\bigl)C_2\bigl)D_4+(d_1+4D_2-D_3+2D_4)E_1\nonumber \\
&\hskip 0.5 cm-\frac{1}{2}C_2 F_2\bigl) \, ,\nonumber \\
\{D_4, F_2\}&={\rm i}I_2 \, ,\nonumber \\
\{D_4, F_3\}&=-\{D_2,F_3\}+\frac{1}{2}\{D_3,F_1\}+\frac{1}{2}\{D_3, F_3\} \, ,\nonumber \\
\{D_4, F_4\}&={\rm i}\bigl((c_1-C_2)F_3-4 b_3 G_2+\frac{1}{2}I_1\bigl) \, ,\nonumber \\
\{C_2, G_1\}&=\{D_4,F_1\} \, ,\nonumber \\
\{C_2,G_2\}&=\{D_4, F_3\} \, .
\end{align}
This implies that we find two indecomposable polynomials of degree nine $I_{1}$ and $I_2$ to be added to the set \eqref{generators}. The latter are defined as:
\begin{equation}
I_1:=-{\rm i}\{D_3, F_4\} \, ,\qquad I_2:=-{\rm i}\{D_4, F_2\} \, .
\end{equation}
We remark that the introduction of these new polynomials is needed because the Berezin bracket of $D_3, F_4$ and $D_4, F_2$ turn out to be not expressible as combinations of total degree nine  involving lower-order terms. Therefore, at this level, we are dealing with the following set composed by twenty polynomials up to degree nine:
\begin{equation}
\boldsymbol{Q}^{[9]} =  \{b_1, b_2, b_3, c_1, C_2, d_1, D_2, D_3, D_4, E_1, F_1, F_2, F_3, F_4, G_1, G_2, H_1, H_2, I_1, I_2\} \, . 
\label{generators2}
\end{equation}
At this point, a crucial observation that ensures the correctness of our results is that the twenty polynomials we have obtained are directly connected to the ones reported in \cite{nuclear}. In that paper, in fact, the following twenty polynomials commuting with the generators $\{s_i,t_\alpha\}$ are presented (sum over repeated indices is understood):
\begin{equation}
\begin{split}
C^{(2 0 0)}&=s_i s_i \, , \quad C^{(020)}=t_\alpha t_\alpha \, , \quad C^{(002)}=q_{i \alpha}q_{i \alpha} \, , \quad C^{(111)}=s_i t_\alpha q_{i \alpha} \, , \quad C^{(003)}=\epsilon_{ijk}\epsilon_{\alpha \beta \gamma}q_{i\alpha}q_{j \beta}q_{k \gamma} \, ,\\
C^{(202)}&=s_i s_j q_{i \alpha}q_{j \alpha} \, , \quad C^{(022)}=t_\alpha t_{\beta}q_{i\alpha}q_{i \beta} \, , \quad C^{(004)}=q_{i \alpha}q_{i \beta}q_{j \alpha}q_{j\beta} \, , \quad C^{(112)}=\epsilon_{ijk}\epsilon_{\alpha \beta \gamma}s_i t_\alpha q_{j \beta} q_{k \gamma} \, ,\\
C^{(113)}&=s_i t_\alpha q_{i \beta}q_{j \alpha}q_{j \beta} \, , \quad C^{(204)}=s_is_jq_{i \alpha}q_{j \beta}q_{k \alpha}q_{k \beta} \, , \quad C^{(024)}=t_\alpha t_\beta q_{i \alpha}q_{j \beta} q_{i \gamma}q_{j \gamma} \, ,\\
 C^{(213)}&=\epsilon_{ijk}s_i s_\ell t_{\alpha}q_{j \alpha}q_{k \beta}q_{\ell \beta} \, ,\quad 
C^{(123)}=\epsilon_{\alpha \beta \gamma} s_i t_\alpha t_\delta q_{i\beta}q_{j \gamma}q_{j \delta} \, , \quad C^{(214)}=\epsilon_{\alpha \beta \gamma}s_i s_j t_{\alpha}q_{i \beta}q_{j \delta}q_{k \gamma}q_{k \delta}\, , \\
 C^{(124)}&=\epsilon_{ijk}s_i t_\alpha t_\beta q_{j \alpha}q_{\ell \beta}q_{k \gamma}q_{\ell \gamma} \, , \quad C^{(215)}=\epsilon_{ijk}s_i s_\ell t_{\alpha}q_{j \alpha}q_{k \beta}q_{\ell \gamma}q_{m \beta} q_{m \gamma} \, , \\
 C^{(125)}&=\epsilon_{\alpha \beta \gamma}s_i  t_{\alpha}t_{\delta} q_{i \beta}q_{j \gamma}q_{k \delta}q_{j \rho} q_{k \rho}\,  , \quad C^{(306)}=\epsilon_{ijk}s_i s_\ell s_m q_{j \alpha}q_{k \beta}q_{\ell \alpha}q_{m \gamma}q_{n \beta}q_{n \gamma} \, , \\
 C^{(036)}&=\epsilon_{\alpha \beta \gamma}t_\alpha t_\delta t_\rho q_{i \beta}q_{j \gamma}q_{i \delta}q_{k \rho}q_{j \sigma}q_{k \sigma} \, .
\end{split}
\label{polquesne}
\end{equation}
 Remarkably, the indecomposable generators in $\boldsymbol{Q}^{[9]}$ can be expressed in terms of the above polynomials as:
\begin{equation}
\begin{split}
b_1&=C^{(200)}  \, , \quad b_2=C^{(020)} \, , \quad b_3=C^{(002)} \, , \quad c_1=C^{(111)}-\frac{2}{3}C^{(003)} \, , \quad C_2=C^{(111)}\\
d_1&=C^{(202)}+C^{(022)}-C^{(112)}-2 C^{(004)} \, ,\quad D_2=\frac{1}{2}C^{(112)} \, , \quad D_3=C^{(022)} \, ,\quad D_4=C^{(004)} \, , \\
 E_1&=C^{(113)} \, ,\quad  F_1=C^{(213)} \, , \quad F_2=2 C^{(204)}-C^{(200)}C^{(004)} \, , \quad F_3=C^{(123)} \, , \\
 F_4&=2 C^{(024)}-C^{(020)}C^{(004)} \, , \quad G_1=C^{(214)} \, , \quad  G_2=C^{(124)} \, , \quad H_1= C^{(215)} \, , \quad H_2=C^{(125)} \, , \\
 I_1&=-8C^{(036)} \, , \quad I_2=-4\bigl(C^{(306)}+C^{(002)}C^{(214)}+\frac{1}{6}C^{(003)} C^{(213)}\bigl) \, .
\end{split} 
\label{poldanilo}
\end{equation}
By examining these polynomials, it becomes apparent that simplifying some of the relations is feasible. In principle, it is permissible to omit additional terms that appear in certain polynomials, such as $F_2$ and $F_4$, where we can remove the terms $C^{(200)}C^{(004)}$ and $C^{(020)}C^{(004)}$, respectively. This method retains the linearly independent and indecomposable part of the polynomials. The same procedure applies to the polynomial $I_2$. Therefore, the most suitable generators for this problem are as follows:
\begin{equation}
\begin{split}
\bar{b}_1&:=C^{(200)}  \, , \quad \bar{b}_2:=C^{(020)} \, , \quad \bar{b}_3:=C^{(002)} \, , \quad \bar{c}_1:=C^{(111)}-\frac{2}{3}C^{(003)} \, , \quad \bar{C}_2:=C^{(111)}\\
\bar{d}_1&:=C^{(202)}+C^{(022)}-C^{(112)}-2 C^{(004)} \, ,\quad \bar{D}_2:=C^{(112)} \, , \quad \bar{D}_3:=C^{(022)} \, ,\quad \bar{D}_4:=C^{(004)} \, , \\
 \bar{E}_1&:=C^{(113)} \, , \quad \bar{F}_1:=C^{(213)} \, , \quad \bar{F}_2:=C^{(204)}\, , \quad \bar{F}_3:=C^{(123)} \, , \quad \bar{F}_4:= C^{(024)} \, , \quad \bar{G}_1:=C^{(214)} \, , \\ \bar{G}_2&:=C^{(124)} \, , \quad \bar{H}_1:=C^{(215)} \, , \quad \bar{H}_2:=C^{(125)} \, ,\qquad \bar{I}_1:=C^{(036)} \, , \quad \bar{I}_2:=C^{(306)} \, ,
\end{split} \label{suitgene}
\end{equation}
where we dropped the above terms in the polynomials and eliminated some unessential global factor.   Notice that not all the polynomials have changed, however we preferred to keep the notation uniform for the new basis. Therefore, in what follows, our efforts will focus on constructing the polynomial algebra generated by the following set composed of twenty polynomials up to degree nine (five central elements and fifteen generators):
\begin{equation}
  \mathcal{P}:= \{\bar{b}_1, \bar{b}_2, \bar{b}_3, \bar{c}_1, \bar{C}_2, \bar{d}_1, \bar{D}_2, \bar{D}_3, \bar{D}_4, \bar{E}_1, \bar{F}_1, \bar{F}_2, \bar{F}_3, \bar{F}_4, \bar{G}_1, \bar{G}_2, \bar{H}_1, \bar{H}_2, \bar{I}_1, \bar{I}_2\} \, . 
\label{generatorspoly}
\end{equation} Hence $\textbf{Alg} \left\langle \mathcal{P} \right\rangle $ generates a polynomial algebra with $\dim_{FL} \textbf{Alg} \left\langle \mathcal{P} \right\rangle = 20.$ In the next section, we will show that $\textbf{Alg} \left\langle \mathcal{P} \right\rangle$ is closed under the Poisson bracket $\{\cdot,\cdot\}.$ Additionally, we should highlight that the two missing label operators, which Moshinsky and Nagel discussed in \cite{MR189602} and represented by two commuting $\mathfrak{su}(2) \times \mathfrak{su}(2)$ scalars within the enveloping algebra of $\mathfrak{su}(4)$, can be expressed as follows:
\begin{equation}
\Omega=C^{(111)} =\bar{C}_2\, , \qquad \Phi=C^{(202)}+C^{(022)}-C^{(112)}=\bar{d}_1+2\bar{D}_4 \, .
\end{equation}

\section{The construction of the polynomial algebra}
\label{sec4}

In this Section \ref{sec4}, we  analyze the structure of the polynomial algebra generated by the set $\mathcal{P}$ in detail. In order to simplify the computations, we apply the grading method initially introduced in \cite{campoamor2025On} to reduce  the number of monomials predicted from the compact forms (i.e., just based on the degree).  In order to understand the precise features of the polynomial algebra, 
we inspect the formal expansions that arise as a consequence of the analysis of the degree associated with the Poisson brackets. For each degree, besides the Poisson bracket relations, there can also exist additional algebraic relations involving the above polynomials, i.e. relations of the type
\begin{align}
\text{P}(\bar{\mathbf{B}},\bar{\mathbf{C}},\bar{\mathbf{D}},\bar{\mathbf{E}},\bar{\mathbf{F}},\bar{\mathbf{G}},\bar{\mathbf{H}},\bar{\mathbf{I}})=0 \, ,
\label{algrel}
\end{align}
where we have indicated formally polynomial expansions of a given degree (in the underlying Lie algebra generators) involving the generators in the set \eqref{generatorspoly}. Let us begin by observing that the five elements $\{\bar{b}_1, \bar{b}_2,\bar{b}_3,\bar{c}_1,\bar{d}_1\}$ are central, and thus their Berezin brackets with any generator is zero. Again, the relations begin at degree six, where we can write the formal relation:
\begin{equation}
	\{\bar{\mathbf{C}},\bar{ \mathbf{D}}\} \sim \bar{\mathbf{F}}+\bar{\mathbf{B}}\bar{\mathbf{D}}+\bar{\mathbf{C}}\bar{\mathbf{C}}+\bar{\mathbf{B}}\bar{\mathbf{B}}\bar{\mathbf{B}} \, .
	\label{deg6b}
\end{equation}
At degree seven, for Berezin brackets involving $\{\bar{\mathbf{C}},\bar{\mathbf{E}}\}$ and $\{\bar{\mathbf{D}},\bar{\mathbf{D}}\}$, we have a formal expansion of the type:
\begin{equation}
	\{\bar{\mathbf{C}},\bar{\mathbf{E}}\} \sim \{\bar{\mathbf{D}},\bar{\mathbf{D}}\}  \sim \bar{\mathbf{G}}+\bar{\mathbf{B}}\bar{\mathbf{E}}+\bar{\mathbf{C}}\bar{\mathbf{D}}+\bar{\mathbf{B}}\bar{\mathbf{B}}\bar{\mathbf{C}} \, .
	\label{deg7b}
\end{equation}
Degree eight expansions involve the Berezin brackets $\{\bar{\mathbf{C}},\bar{ \mathbf{F}}\}$ and $\{\bar{\mathbf{D}}, \bar{\mathbf{E}}\}$ through the following formal expansion:
\begin{equation}
	\{\bar{\mathbf{C}}, \bar{\mathbf{F}}\} \sim \{\bar{\mathbf{D}}, \bar{\mathbf{E}}\}\sim  \bar{\mathbf{B}}\bar{\mathbf{F}}+\bar{\mathbf{C}}\bar{\mathbf{E}}+\bar{\mathbf{D}}\bar{\mathbf{D}}+\bar{\mathbf{B}}\bar{\mathbf{C}}\bar{\mathbf{C}}+\bar{\mathbf{B}}\bar{\mathbf{B}}\bar{\mathbf{D}}+\bar{\mathbf{B}}\bar{\mathbf{B}}\bar{\mathbf{B}}\bar{\mathbf{B}} \,  .
	\label{deg8b}
\end{equation}
Poisson brackets involving degree nine expansions are $\{\bar{\mathbf{C}},\bar{\mathbf{G}}\}$,  $\{\bar{\mathbf{D}},\bar{\mathbf{F}} \}$, and involve the following formal relation:
\begin{equation}
	\{\bar{\mathbf{C}}, \bar{\mathbf{G}}\} \sim \{\bar{\mathbf{D}}, \bar{\mathbf{F}}\}\sim  \bar{\mathbf{B}} \bar{\mathbf{G}}+ \bar{\mathbf{C}} \bar{\mathbf{F}}+ \bar{\mathbf{D}} \bar{\mathbf{E}}+ \bar{\mathbf{C}} \bar{\mathbf{C}} \bar{\mathbf{C}}+ \bar{\mathbf{B}} \bar{\mathbf{C}} \bar{\mathbf{D}}+\bar{\mathbf{B}} \bar{\mathbf{B}} \bar{\mathbf{E}}+ \bar{\mathbf{B}} \bar{\mathbf{B}} \bar{\mathbf{B}} \bar{\mathbf{C}}\,  .
	\label{deg9b}
\end{equation}
For degree-ten relations, which involve the Berezin brackets $\{\bar{\mathbf{C}},\bar{\mathbf{H}} \}$,  $\{\bar{\mathbf{D}},\bar{\mathbf{G}} \}$ and $\{\bar{\mathbf{E}},\bar{\mathbf{F}} \}$, we have the following:
\begin{align}
	\{\bar{\mathbf{C}}, \bar{\mathbf{H}}\} \sim \{\bar{\mathbf{D}}, \bar{\mathbf{G}}\}\sim \{\bar{\mathbf{E}},\bar{ \mathbf{F}}\}\  \sim  \bar{\mathbf{B}}\bar{\mathbf{H}}&+\bar{\mathbf{C}}\bar{\mathbf{G}}+\bar{\mathbf{D}}\bar{\mathbf{F}}+\bar{\mathbf{E}}\bar{\mathbf{E}}+\bar{\mathbf{B}}\bar{\mathbf{B}}\bar{\mathbf{F}}+\bar{\mathbf{B}}\bar{\mathbf{C}}\bar{\mathbf{E}}+\bar{\mathbf{B}}\bar{\mathbf{D}}\bar{\mathbf{D}}\nonumber \\
&+\bar{\mathbf{B}}\bar{\mathbf{B}}\bar{\mathbf{C}}\bar{\mathbf{C}}+\bar{\mathbf{C}}\bar{\mathbf{C}}\bar{\mathbf{D}}+\bar{\mathbf{B}}\bar{\mathbf{B}}\bar{\mathbf{B}}\bar{\mathbf{D}}+\bar{\mathbf{B}}\bar{\mathbf{B}}\bar{\mathbf{B}}\bar{\mathbf{B}}\bar{\mathbf{B}} \, .
	\label{deg10b}
\end{align}
 Degree eleven Poisson brackets are given by $\{\bar{\mathbf{C}},\bar{\mathbf{I}} \}$,  $\{\bar{\mathbf{D}},\bar{\mathbf{H}} \}$, $\{\bar{\mathbf{E}},\bar{\mathbf{G}} \}$  and  $\{\bar{\mathbf{F}},\bar{\mathbf{F}} \}$ through the following expansion:
\begin{align}
	\{\bar{\mathbf{C}}, \bar{\mathbf{I}}\} \sim \{\bar{\mathbf{D}}, \bar{\mathbf{H}}\}\sim \{\bar{\mathbf{E}},\bar{ \mathbf{G}}\}  \sim \{\bar{\mathbf{F}},\bar{ \mathbf{F}}\}  \sim \bar{\mathbf{B}}\bar{\mathbf{I}}+\bar{\mathbf{C}}\bar{\mathbf{H}}&+\bar{\mathbf{D}}\bar{\mathbf{G}}+\bar{\mathbf{E}}\bar{\mathbf{F}}+\bar{\mathbf{B}}\bar{\mathbf{C}}\bar{\mathbf{F}}+\bar{\mathbf{B}}\bar{\mathbf{D}}\bar{\mathbf{E}}+\bar{\mathbf{B}}\bar{\mathbf{C}}\bar{\mathbf{C}}\bar{\mathbf{C}}+\bar{\mathbf{C}}\bar{\mathbf{D}}\bar{\mathbf{D}} \nonumber \\
	&+\bar{\mathbf{C}}\bar{\mathbf{C}}\bar{\mathbf{E}}+\bar{\mathbf{B}}\bar{\mathbf{B}}\bar{\mathbf{G}}+\bar{\mathbf{B}}\bar{\mathbf{B}}\bar{\mathbf{C}}\bar{\mathbf{D}}+\bar{\mathbf{B}}\bar{\mathbf{B}}\bar{\mathbf{B}}\bar{\mathbf{E}}+\bar{\mathbf{B}}\bar{\mathbf{B}}\bar{\mathbf{B}}\bar{\mathbf{B}}\bar{\mathbf{C}}\, .
	\label{deg11b}
\end{align}
  Degree twelve expansions involve $\{\bar{\mathbf{D}},\bar{\mathbf{I}} \}$,  $\{\bar{\mathbf{E}},\bar{\mathbf{H}} \}$, $\{\bar{\mathbf{F}},\bar{\mathbf{G}} \}$  through the following formal expansion:
\begin{align}
	\{\bar{\mathbf{D}}, \bar{\mathbf{I}}\} \sim \{\bar{\mathbf{E}}, \bar{\mathbf{H}}\}\sim \{\bar{\mathbf{F}},\bar{ \mathbf{G}}\}   \sim \bar{\mathbf{C}}\bar{\mathbf{I}}&+\bar{\mathbf{D}}\bar{\mathbf{H}}+\bar{\mathbf{E}}\bar{\mathbf{G}}+\bar{\mathbf{F}}\bar{\mathbf{F}}+\bar{\mathbf{C}}\bar{\mathbf{D}}\bar{\mathbf{E}}+\bar{\mathbf{C}}\bar{\mathbf{C}}\bar{\mathbf{F}}+\bar{\mathbf{C}}\bar{\mathbf{C}}\bar{\mathbf{C}}\bar{\mathbf{C}}+\bar{\mathbf{D}}\bar{\mathbf{D}}\bar{\mathbf{D}} \nonumber \\
	&+\bar{\mathbf{B}}\bar{\mathbf{D}}\bar{\mathbf{F}}+\bar{\mathbf{B}}\bar{\mathbf{E}}\bar{\mathbf{E}}+\bar{\mathbf{B}}\bar{\mathbf{C}}\bar{\mathbf{G}}+\bar{\mathbf{B}}\bar{\mathbf{C}}\bar{\mathbf{C}}\bar{\mathbf{D}}+\bar{\mathbf{B}}\bar{\mathbf{B}}\bar{\mathbf{H}}+\bar{\mathbf{B}}\bar{\mathbf{B}}\bar{\mathbf{D}}\bar{\mathbf{D}} \nonumber\\
	&+\bar{\mathbf{B}}\bar{\mathbf{B}}\bar{\mathbf{C}}\bar{\mathbf{E}}+\bar{\mathbf{B}}\bar{\mathbf{B}}\bar{\mathbf{B}}\bar{\mathbf{F}}+\bar{\mathbf{B}}\bar{\mathbf{B}}\bar{\mathbf{B}}\bar{\mathbf{B}}\bar{\mathbf{D}}+\bar{\mathbf{B}}\bar{\mathbf{B}}\bar{\mathbf{B}}\bar{\mathbf{C}}\bar{\mathbf{C}}+\bar{\mathbf{B}}\bar{\mathbf{B}}\bar{\mathbf{B}}\bar{\mathbf{B}}\bar{\mathbf{B}}\bar{\mathbf{B}}\, .
	\label{deg12b}
\end{align}
  For degree thirteen expansions, which involve $\{\bar{\mathbf{E}},\bar{\mathbf{I}} \}$,  $\{\bar{\mathbf{F}},\bar{\mathbf{H}} \}$, $\{\bar{\mathbf{G}},\bar{\mathbf{G}} \}$, we  have the following:
\begin{align}
	\{\bar{\mathbf{E}}, \bar{\mathbf{I}}\} \sim \{\bar{\mathbf{F}}, \bar{\mathbf{H}}\}\sim \{\bar{\mathbf{G}},\bar{ \mathbf{G}}\}   \sim \bar{\mathbf{D}}\bar{\mathbf{I}} &+\bar{\mathbf{E}}\bar{\mathbf{H}}+ \bar{\mathbf{F}}\bar{\mathbf{G}}+ \bar{\mathbf{D}}\bar{\mathbf{D}}\bar{\mathbf{E}}+\bar{\mathbf{C}}\bar{\mathbf{E}}\bar{\mathbf{E}}+\bar{\mathbf{C}}\bar{\mathbf{D}}\bar{\mathbf{F}}+\bar{\mathbf{C}}\bar{\mathbf{C}}\bar{\mathbf{G}}+\bar{\mathbf{C}}\bar{\mathbf{C}}\bar{\mathbf{C}}\bar{\mathbf{D}}+\bar{\mathbf{B}}\bar{\mathbf{E}}\bar{\mathbf{F}} \nonumber \\
	&+\bar{\mathbf{B}}\bar{\mathbf{D}}\bar{\mathbf{G}}+\bar{\mathbf{B}}\bar{\mathbf{C}}\bar{\mathbf{H}}+\bar{\mathbf{B}}\bar{\mathbf{C}}\bar{\mathbf{D}}\bar{\mathbf{D}}+\bar{\mathbf{B}}\bar{\mathbf{C}}\bar{\mathbf{C}}\bar{\mathbf{E}}+\bar{\mathbf{B}}\bar{\mathbf{B}}\bar{\mathbf{I}}+\bar{\mathbf{B}}\bar{\mathbf{B}}\bar{\mathbf{D}}\bar{\mathbf{E}}+\bar{\mathbf{B}}\bar{\mathbf{B}}\bar{\mathbf{C}}\bar{\mathbf{F}}\nonumber \\
	&+\bar{\mathbf{B}}\bar{\mathbf{B}}\bar{\mathbf{C}}\bar{\mathbf{C}}\bar{\mathbf{C}}+\bar{\mathbf{B}}\bar{\mathbf{B}}\bar{\mathbf{B}}\bar{\mathbf{G}}+\bar{\mathbf{B}}\bar{\mathbf{B}}\bar{\mathbf{B}}\bar{\mathbf{C}}\bar{\mathbf{D}}+\bar{\mathbf{B}}\bar{\mathbf{B}}\bar{\mathbf{B}}\bar{\mathbf{B}}\bar{\mathbf{E}}+\bar{\mathbf{B}}\bar{\mathbf{B}}\bar{\mathbf{B}}\bar{\mathbf{B}}\bar{\mathbf{B}}\bar{\mathbf{C}} \, .
\end{align}
Degree fourteen Poisson brackets are $\{\bar{\mathbf{F}},\bar{\mathbf{I}} \}$,  $\{\bar{\mathbf{G}},\bar{\mathbf{H}} \}$, and they are given through the following formal expansion:
\begin{align}
	\{\bar{\mathbf{F}}, \bar{\mathbf{I}}\} \sim \{\bar{\mathbf{G}}, \bar{\mathbf{H}}\} \sim \bar{\mathbf{E}}\bar{\mathbf{I}} &+\bar{\mathbf{F}}\bar{\mathbf{H}}+ \bar{\mathbf{G}}\bar{\mathbf{G}}+\bar{\mathbf{D}}\bar{\mathbf{E}}\bar{\mathbf{E}}+\bar{\mathbf{D}}\bar{\mathbf{D}}\bar{\mathbf{F}}+\bar{\mathbf{C}}\bar{\mathbf{E}}\bar{\mathbf{F}}+\bar{\mathbf{C}}\bar{\mathbf{D}}\bar{\mathbf{G}}+\bar{\mathbf{C}}\bar{\mathbf{C}}\bar{\mathbf{H}}+\bar{\mathbf{C}}\bar{\mathbf{C}}\bar{\mathbf{D}}\bar{\mathbf{D}}+\bar{\mathbf{C}}\bar{\mathbf{C}}\bar{\mathbf{C}}\bar{\mathbf{E}} \nonumber \\
	&+ \bar{\mathbf{B}}\bar{\mathbf{F}}\bar{\mathbf{F}}+ \bar{\mathbf{B}} \bar{\mathbf{E}} \bar{\mathbf{G}}+ \bar{\mathbf{B}} \bar{\mathbf{D}} \bar{\mathbf{H}}+ \bar{\mathbf{B}} \bar{\mathbf{D}} \bar{\mathbf{D}} \bar{\mathbf{D}}+ \bar{\mathbf{B}} \bar{\mathbf{C}} \bar{\mathbf{I}}+ \bar{\mathbf{B}} \bar{\mathbf{C}} \bar{\mathbf{D}} \bar{\mathbf{E}}+ \bar{\mathbf{B}} \bar{\mathbf{C}} \bar{\mathbf{C}} \bar{\mathbf{F}}+ \bar{\mathbf{B}} \bar{\mathbf{C}} \bar{\mathbf{C}} \bar{\mathbf{C}} \bar{\mathbf{C}}\nonumber \\
	&+  \bar{\mathbf{B}} \bar{\mathbf{B}} \bar{\mathbf{E}} \bar{\mathbf{E}}+ \bar{\mathbf{B}} \bar{\mathbf{B}} \bar{\mathbf{D}} \bar{\mathbf{F}}+ \bar{\mathbf{B}} \bar{\mathbf{B}} \bar{\mathbf{C}} \bar{\mathbf{G}}+ \bar{\mathbf{B}} \bar{\mathbf{B}} \bar{\mathbf{C}} \bar{\mathbf{C}} \bar{\mathbf{D}}+ \bar{\mathbf{B}} \bar{\mathbf{B}} \bar{\mathbf{B}} \bar{\mathbf{H}}+ \bar{\mathbf{B}} \bar{\mathbf{B}} \bar{\mathbf{B}} \bar{\mathbf{D}} \bar{\mathbf{D}}+ \bar{\mathbf{B}} \bar{\mathbf{B}} \bar{\mathbf{B}} \bar{\mathbf{C}} \bar{\mathbf{E}}\nonumber \\
	&+   \bar{\mathbf{B}} \bar{\mathbf{B}} \bar{\mathbf{B}} \bar{\mathbf{B}} \bar{\mathbf{F}} +\bar{\mathbf{B}} \bar{\mathbf{B}} \bar{\mathbf{B}} \bar{\mathbf{B}} \bar{\mathbf{C}}\bar{\mathbf{C}}  +\bar{\mathbf{B}} \bar{\mathbf{B}} \bar{\mathbf{B}} \bar{\mathbf{B}} \bar{\mathbf{B}} \bar{\mathbf{D}} +\bar{\mathbf{B}} \bar{\mathbf{B}} \bar{\mathbf{B}} \bar{\mathbf{B}} \bar{\mathbf{B}} \bar{\mathbf{B}} \bar{\mathbf{B}} \, .
\end{align}
Degree fifteen expansions involve $\{\bar{\mathbf{G}},\bar{\mathbf{I}} \}$,  $\{\bar{\mathbf{H}},\bar{\mathbf{H}} \}$  through the following formal expansion:
\begin{align}
	\{\bar{\mathbf{G}}, \bar{\mathbf{I}}\} \sim \{\bar{\mathbf{H}}, \bar{\mathbf{H}}\} \sim \bar{\mathbf{F}}\bar{\mathbf{I}}&+\bar{\mathbf{G}}\bar{\mathbf{H}}+\bar{\mathbf{E}}\bar{\mathbf{E}}\bar{\mathbf{E}}+\bar{\mathbf{D}}\bar{\mathbf{E}}\bar{\mathbf{F}}+\bar{\mathbf{D}}\bar{\mathbf{D}}\bar{\mathbf{G}}+\bar{\mathbf{C}}\bar{\mathbf{F}}\bar{\mathbf{F}}+\bar{\mathbf{C}}\bar{\mathbf{E}}\bar{\mathbf{G}}+\bar{\mathbf{C}}\bar{\mathbf{D}}\bar{\mathbf{H}}+\bar{\mathbf{C}}\bar{\mathbf{D}}\bar{\mathbf{D}}\bar{\mathbf{D}}+\bar{\mathbf{C}}\bar{\mathbf{C}}\bar{\mathbf{I}}\nonumber \\
	&+\bar{\mathbf{C}}\bar{\mathbf{C}}\bar{\mathbf{D}}\bar{\mathbf{E}}+\bar{\mathbf{C}}\bar{\mathbf{C}}\bar{\mathbf{C}}\bar{\mathbf{F}}+\bar{\mathbf{C}}\bar{\mathbf{C}}\bar{\mathbf{C}}\bar{\mathbf{C}}\bar{\mathbf{C}}+\bar{\mathbf{B}}\bar{\mathbf{F}}\bar{\mathbf{G}}+\bar{\mathbf{B}}\bar{\mathbf{E}}\bar{\mathbf{H}}+\bar{\mathbf{B}}\bar{\mathbf{D}}\bar{\mathbf{I}}+\bar{\mathbf{B}}\bar{\mathbf{D}}\bar{\mathbf{D}}\bar{\mathbf{E}}+\bar{\mathbf{B}}\bar{\mathbf{C}}\bar{\mathbf{E}}\bar{\mathbf{E}}\nonumber \\
	&+\bar{\mathbf{B}}\bar{\mathbf{C}}\bar{\mathbf{D}}\bar{\mathbf{F}}+\bar{\mathbf{B}}\bar{\mathbf{C}}\bar{\mathbf{C}}\bar{\mathbf{G}}+\bar{\mathbf{B}}\bar{\mathbf{C}}\bar{\mathbf{C}}\bar{\mathbf{C}}\bar{\mathbf{D}}+\bar{\mathbf{B}}\bar{\mathbf{B}}\bar{\mathbf{E}}\bar{\mathbf{F}}+\bar{\mathbf{B}}\bar{\mathbf{B}}\bar{\mathbf{D}}\bar{\mathbf{G}}+\bar{\mathbf{B}}\bar{\mathbf{B}}\bar{\mathbf{C}}\bar{\mathbf{H}}+\bar{\mathbf{B}}\bar{\mathbf{B}}\bar{\mathbf{C}}\bar{\mathbf{D}}\bar{\mathbf{D}}\nonumber \\
	&+\bar{\mathbf{B}}\bar{\mathbf{B}}\bar{\mathbf{C}}\bar{\mathbf{C}}\bar{\mathbf{E}}+\bar{\mathbf{B}}\bar{\mathbf{B}}\bar{\mathbf{B}}\bar{\mathbf{I}}+\bar{\mathbf{B}}\bar{\mathbf{B}}\bar{\mathbf{B}}\bar{\mathbf{D}}\bar{\mathbf{E}}+\bar{\mathbf{B}}\bar{\mathbf{B}}\bar{\mathbf{B}}\bar{\mathbf{C}}\bar{\mathbf{F}}+\bar{\mathbf{B}}\bar{\mathbf{B}}\bar{\mathbf{B}}\bar{\mathbf{C}}\bar{\mathbf{C}}\bar{\mathbf{C}}+\bar{\mathbf{B}}\bar{\mathbf{B}}\bar{\mathbf{B}}\bar{\mathbf{B}}\bar{\mathbf{G}}\nonumber \\
	&+\bar{\mathbf{B}}\bar{\mathbf{B}}\bar{\mathbf{B}}\bar{\mathbf{B}}\bar{\mathbf{C}}\bar{\mathbf{D}}+\bar{\mathbf{B}}\bar{\mathbf{B}}\bar{\mathbf{B}}\bar{\mathbf{B}}\bar{\mathbf{B}}\bar{\mathbf{E}}+\bar{\mathbf{B}}\bar{\mathbf{B}}\bar{\mathbf{B}}\bar{\mathbf{B}}\bar{\mathbf{B}}\bar{\mathbf{B}}\bar{\mathbf{C}} \, .
\end{align}
Degree sixteen expansions involve $\{\bar{\mathbf{H}},\bar{\mathbf{I}} \}$ through the following formal expansion:
\begin{align}
	\{\bar{\mathbf{H}}, \bar{\mathbf{I}}\} \sim \bar{\mathbf{H}}\bar{\mathbf{H}}&+\bar{\mathbf{G}}\bar{\mathbf{I}}+\bar{\mathbf{E}}\bar{\mathbf{E}}\bar{\mathbf{F}}+\bar{\mathbf{D}}\bar{\mathbf{F}}\bar{\mathbf{F}}+\bar{\mathbf{D}}\bar{\mathbf{E}}\bar{\mathbf{G}}+\bar{\mathbf{D}}\bar{\mathbf{D}}\bar{\mathbf{H}}+\bar{\mathbf{D}}\bar{\mathbf{D}}\bar{\mathbf{D}}\bar{\mathbf{D}}+\bar{\mathbf{C}}\bar{\mathbf{F}}\bar{\mathbf{G}}+\bar{\mathbf{C}}\bar{\mathbf{E}}\bar{\mathbf{H}}+\bar{\mathbf{C}}\bar{\mathbf{D}}\bar{\mathbf{I}}+\bar{\mathbf{C}}\bar{\mathbf{D}}\bar{\mathbf{D}}\bar{\mathbf{E}}\nonumber \\
	&+\bar{\mathbf{C}}\bar{\mathbf{C}}\bar{\mathbf{E}}\bar{\mathbf{E}}+\bar{\mathbf{C}}\bar{\mathbf{C}}\bar{\mathbf{D}}\bar{\mathbf{F}}+\bar{\mathbf{C}}\bar{\mathbf{C}}\bar{\mathbf{C}}\bar{\mathbf{G}}+\bar{\mathbf{C}}\bar{\mathbf{C}}\bar{\mathbf{C}}\bar{\mathbf{C}}\bar{\mathbf{D}}+\bar{\mathbf{B}}\bar{\mathbf{G}}\bar{\mathbf{G}}+\bar{\mathbf{B}}\bar{\mathbf{F}}\bar{\mathbf{H}}+\bar{\mathbf{B}}\bar{\mathbf{E}}\bar{\mathbf{I}}+\bar{\mathbf{B}}\bar{\mathbf{D}}\bar{\mathbf{E}}\bar{\mathbf{E}}+\bar{\mathbf{B}}\bar{\mathbf{D}}\bar{\mathbf{D}}\bar{\mathbf{F}}\nonumber\\
	&+\bar{\mathbf{B}}\bar{\mathbf{C}}\bar{\mathbf{E}}\bar{\mathbf{F}}+\bar{\mathbf{B}}\bar{\mathbf{C}}\bar{\mathbf{D}}\bar{\mathbf{G}}+\bar{\mathbf{B}}\bar{\mathbf{C}}\bar{\mathbf{C}}\bar{\mathbf{H}}+\bar{\mathbf{B}}\bar{\mathbf{C}}\bar{\mathbf{C}}\bar{\mathbf{D}}\bar{\mathbf{D}}+\bar{\mathbf{B}}\bar{\mathbf{C}}\bar{\mathbf{C}}\bar{\mathbf{C}}\bar{\mathbf{E}}+\bar{\mathbf{B}}\bar{\mathbf{B}}\bar{\mathbf{F}}\bar{\mathbf{F}}+\bar{\mathbf{B}}\bar{\mathbf{B}}\bar{\mathbf{E}}\bar{\mathbf{G}}+\bar{\mathbf{B}}\bar{\mathbf{B}}\bar{\mathbf{D}}\bar{\mathbf{H}}\nonumber \\
	&+ \bar{\mathbf{B}}\bar{\mathbf{B}}\bar{\mathbf{D}}\bar{\mathbf{D}}\bar{\mathbf{D}}+\bar{\mathbf{B}}\bar{\mathbf{B}}\bar{\mathbf{C}}\bar{\mathbf{I}}+\bar{\mathbf{B}}\bar{\mathbf{B}}\bar{\mathbf{C}}\bar{\mathbf{D}}\bar{\mathbf{E}}+\bar{\mathbf{B}}\bar{\mathbf{B}}\bar{\mathbf{C}}\bar{\mathbf{C}}\bar{\mathbf{F}}+\bar{\mathbf{B}}\bar{\mathbf{B}}\bar{\mathbf{C}}\bar{\mathbf{C}}\bar{\mathbf{C}}\bar{\mathbf{C}}+\bar{\mathbf{B}}\bar{\mathbf{B}}\bar{\mathbf{B}}\bar{\mathbf{E}}\bar{\mathbf{E}}+\bar{\mathbf{B}}\bar{\mathbf{B}}\bar{\mathbf{B}}\bar{\mathbf{D}}\bar{\mathbf{F}} \nonumber \\
	&+\bar{\mathbf{B}}\bar{\mathbf{B}}\bar{\mathbf{B}}\bar{\mathbf{C}}\bar{\mathbf{G}}+\bar{\mathbf{B}}\bar{\mathbf{B}}\bar{\mathbf{B}}\bar{\mathbf{C}}\bar{\mathbf{C}}\bar{\mathbf{D}}+\bar{\mathbf{B}}\bar{\mathbf{B}}\bar{\mathbf{B}}\bar{\mathbf{B}}\bar{\mathbf{H}}+\bar{\mathbf{B}}\bar{\mathbf{B}}\bar{\mathbf{B}}\bar{\mathbf{B}}\bar{\mathbf{D}}\bar{\mathbf{D}}+\bar{\mathbf{B}}\bar{\mathbf{B}}\bar{\mathbf{B}}\bar{\mathbf{B}}\bar{\mathbf{C}}\bar{\mathbf{E}}+\bar{\mathbf{B}}\bar{\mathbf{B}}\bar{\mathbf{B}}\bar{\mathbf{B}}\bar{\mathbf{B}}\bar{\mathbf{F}}\nonumber \\
	&+\bar{\mathbf{B}}\bar{\mathbf{B}}\bar{\mathbf{B}}\bar{\mathbf{B}}\bar{\mathbf{B}}\bar{\mathbf{C}}\bar{\mathbf{C}}+\bar{\mathbf{B}}\bar{\mathbf{B}}\bar{\mathbf{B}}\bar{\mathbf{B}}\bar{\mathbf{B}}\bar{\mathbf{B}}\bar{\mathbf{D}}+\bar{\mathbf{B}}\bar{\mathbf{B}}\bar{\mathbf{B}}\bar{\mathbf{B}}\bar{\mathbf{B}}\bar{\mathbf{B}}\bar{\mathbf{B}}\bar{\mathbf{B}} \, .
\end{align}
Finally, degree seventeen expansions involve $\{\bar{\mathbf{I}},\bar{\mathbf{I}} \}$ through the following formal expansion:
\begin{align}
	\{\bar{\mathbf{I}}, \bar{\mathbf{I}}\} \sim \bar{\mathbf{H}}\bar{\mathbf{I}}&+\bar{\mathbf{E}}\bar{\mathbf{F}}\bar{\mathbf{F}}+\bar{\mathbf{E}}\bar{\mathbf{E}}\bar{\mathbf{G}}+\bar{\mathbf{D}}\bar{\mathbf{F}}\bar{\mathbf{G}}+\bar{\mathbf{D}}\bar{\mathbf{E}}\bar{\mathbf{H}}+\bar{\mathbf{D}}\bar{\mathbf{D}}\bar{\mathbf{I}}+\bar{\mathbf{D}}\bar{\mathbf{D}}\bar{\mathbf{D}}\bar{\mathbf{E}}+\bar{\mathbf{C}}\bar{\mathbf{G}}\bar{\mathbf{G}}+\bar{\mathbf{C}}\bar{\mathbf{F}}\bar{\mathbf{H}}+\bar{\mathbf{C}}\bar{\mathbf{E}}\bar{\mathbf{I}}+\bar{\mathbf{C}}\bar{\mathbf{D}}\bar{\mathbf{E}}\bar{\mathbf{E}} \nonumber \\
	&+\bar{\mathbf{C}}\bar{\mathbf{D}}\bar{\mathbf{D}}\bar{\mathbf{F}}+\bar{\mathbf{C}}\bar{\mathbf{C}}\bar{\mathbf{E}}\bar{\mathbf{F}}+\bar{\mathbf{C}}\bar{\mathbf{C}}\bar{\mathbf{D}}\bar{\mathbf{G}}+\bar{\mathbf{C}}\bar{\mathbf{C}}\bar{\mathbf{C}}\bar{\mathbf{H}}+\bar{\mathbf{C}}\bar{\mathbf{C}}\bar{\mathbf{C}}\bar{\mathbf{D}}\bar{\mathbf{D}}+\bar{\mathbf{C}}\bar{\mathbf{C}}\bar{\mathbf{C}}\bar{\mathbf{C}}\bar{\mathbf{E}}+\bar{\mathbf{B}}\bar{\mathbf{G}}\bar{\mathbf{H}}+\bar{\mathbf{B}}\bar{\mathbf{F}}\bar{\mathbf{I}}+\bar{\mathbf{B}}\bar{\mathbf{E}}\bar{\mathbf{E}}\bar{\mathbf{E}} \nonumber \\
	&+\bar{\mathbf{B}}\bar{\mathbf{D}}\bar{\mathbf{E}}\bar{\mathbf{F}}+\bar{\mathbf{B}}\bar{\mathbf{D}}\bar{\mathbf{D}}\bar{\mathbf{G}}+\bar{\mathbf{B}}\bar{\mathbf{C}}\bar{\mathbf{F}}\bar{\mathbf{F}}+\bar{\mathbf{B}}\bar{\mathbf{C}}\bar{\mathbf{E}}\bar{\mathbf{G}}+\bar{\mathbf{B}}\bar{\mathbf{C}}\bar{\mathbf{D}}\bar{\mathbf{H}}+\bar{\mathbf{B}}\bar{\mathbf{C}}\bar{\mathbf{D}}\bar{\mathbf{D}}\bar{\mathbf{D}}+\bar{\mathbf{B}}\bar{\mathbf{C}}\bar{\mathbf{C}}\bar{\mathbf{I}}+\bar{\mathbf{B}}\bar{\mathbf{C}}\bar{\mathbf{C}}\bar{\mathbf{D}}\bar{\mathbf{E}}\nonumber \\
	&+ \bar{\mathbf{B}}\bar{\mathbf{C}}\bar{\mathbf{C}}\bar{\mathbf{C}}\bar{\mathbf{F}}+\bar{\mathbf{B}}\bar{\mathbf{C}}\bar{\mathbf{C}}\bar{\mathbf{C}}\bar{\mathbf{C}}\bar{\mathbf{C}}+\bar{\mathbf{B}}\bar{\mathbf{B}}\bar{\mathbf{F}}\bar{\mathbf{G}}+\bar{\mathbf{B}}\bar{\mathbf{B}}\bar{\mathbf{E}}\bar{\mathbf{H}}+\bar{\mathbf{B}}\bar{\mathbf{B}}\bar{\mathbf{D}}\bar{\mathbf{I}}+\bar{\mathbf{B}}\bar{\mathbf{B}}\bar{\mathbf{D}}\bar{\mathbf{D}}\bar{\mathbf{E}}+\bar{\mathbf{B}}\bar{\mathbf{B}}\bar{\mathbf{C}}\bar{\mathbf{E}}\bar{\mathbf{E}}+\bar{\mathbf{B}}\bar{\mathbf{B}}\bar{\mathbf{C}}\bar{\mathbf{D}}\bar{\mathbf{F}}\nonumber \\
	&+\bar{\mathbf{B}}\bar{\mathbf{B}}\bar{\mathbf{C}}\bar{\mathbf{C}}\bar{\mathbf{G}}+\bar{\mathbf{B}}\bar{\mathbf{B}}\bar{\mathbf{C}}\bar{\mathbf{C}}\bar{\mathbf{C}}\bar{\mathbf{D}}+\bar{\mathbf{B}}\bar{\mathbf{B}}\bar{\mathbf{B}}\bar{\mathbf{E}}\bar{\mathbf{F}}+\bar{\mathbf{B}}\bar{\mathbf{B}}\bar{\mathbf{B}}\bar{\mathbf{D}}\bar{\mathbf{G}}+\bar{\mathbf{B}}\bar{\mathbf{B}}\bar{\mathbf{B}}\bar{\mathbf{C}}\bar{\mathbf{H}}+\bar{\mathbf{B}}\bar{\mathbf{B}}\bar{\mathbf{B}}\bar{\mathbf{C}}\bar{\mathbf{D}}\bar{\mathbf{D}}+\bar{\mathbf{B}}\bar{\mathbf{B}}\bar{\mathbf{B}}\bar{\mathbf{C}}\bar{\mathbf{C}}\bar{\mathbf{E}}\nonumber \\
	&+\bar{\mathbf{B}}\bar{\mathbf{B}}\bar{\mathbf{B}}\bar{\mathbf{B}}\bar{\mathbf{I}}+\bar{\mathbf{B}}\bar{\mathbf{B}}\bar{\mathbf{B}}\bar{\mathbf{B}}\bar{\mathbf{D}}\bar{\mathbf{E}}+\bar{\mathbf{B}}\bar{\mathbf{B}}\bar{\mathbf{B}}\bar{\mathbf{B}}\bar{\mathbf{C}}\bar{\mathbf{F}}+\bar{\mathbf{B}}\bar{\mathbf{B}}\bar{\mathbf{B}}\bar{\mathbf{B}}\bar{\mathbf{C}}\bar{\mathbf{C}}\bar{\mathbf{C}}+\bar{\mathbf{B}}\bar{\mathbf{B}}\bar{\mathbf{B}}\bar{\mathbf{B}}\bar{\mathbf{B}}\bar{\mathbf{G}}+\bar{\mathbf{B}}\bar{\mathbf{B}}\bar{\mathbf{B}}\bar{\mathbf{B}}\bar{\mathbf{B}}\bar{\mathbf{C}}\bar{\mathbf{D}} \nonumber \\
	&+ \bar{\mathbf{B}}\bar{\mathbf{B}}\bar{\mathbf{B}}\bar{\mathbf{B}}\bar{\mathbf{B}}\bar{\mathbf{B}}\bar{\mathbf{E}}+\bar{\mathbf{B}}\bar{\mathbf{B}}\bar{\mathbf{B}}\bar{\mathbf{B}}\bar{\mathbf{B}}\bar{\mathbf{B}}\bar{\mathbf{B}}\bar{\mathbf{C}} \, . \label{eq:deg17b}
\end{align}
\vskip 0.5cm
% By observing the terms on the right hand side, we can conclude that the polynomial algebra generated by the twenty polynomials in the commutant can be at most of order $5$.  This is because of the possible appearence of terms coming from the elements $\bar{\mathbf{C}}\bar{\mathbf{C}}\bar{\mathbf{C}}\bar{\mathbf{C}}\bar{\mathbf{D}}$,  $\bar{\mathbf{C}}\bar{\mathbf{C}}\bar{\mathbf{C}}\bar{\mathbf{D}}\bar{\mathbf{D}}$, $\bar{\mathbf{C}}\bar{\mathbf{C}}\bar{\mathbf{C}}\bar{\mathbf{C}}\bar{\mathbf{E}}$, $\bar{\mathbf{B}}\bar{\mathbf{C}}\bar{\mathbf{C}}\bar{\mathbf{C}}\bar{\mathbf{C}}\bar{\mathbf{C}}$, where we recall that the $\bar{\mathbf{B}}$  term is composed by only central elements of degree two.
From the compact forms $\eqref{deg6b}$ to $\eqref{eq:deg17b}$, the number of predicted terms is computable using identity $\eqref{eq:counting}$. Conversely, in Subsection $\ref{4.1},$ we calculate the admissible terms in the selected Poisson brackets using the grading method from Subsection $\ref{2.4}$. The following table compares the number of admissible terms in the compact form with those from the grading method. It is crucial to note that our initial attempt to derive polynomial relations up to degree seventeen demonstrated that expansions beyond degree eleven couldn't be achieved by merely considering the all terms appearing on the right-hand side. We will see in the following subsections, how the grading method significantly reduced the number of generators, as detailed in the accompanying table.

\begin{table}[h]
    \centering
    \begin{tabular}{|c|c|c|c|c|}
        \hline
        Deg & No. from compact forms & No. from gradings & Example of Poisson brackets & $\Delta$ (difference) \\ \hline
        6    & 29     & 2     & $\{\bar{C}_2,\bar{D}_2\}$     & 27     \\ \hline
        7    & 26     & 2     & $ \{\bar{C}_2,\bar{E}_1\}$     & 24    \\ \hline
        8    & 67     & 26     & $\{\bar{C}_2,\bar{F}_1\}$     & 41     \\ \hline
        9    & 65     & 31     & $\{\bar{D}_2,\bar{F}_1\}$     & 34     \\ \hline
        10    & 169     & 57    & $\{\bar{C}_2,\bar{H}_2\}$     & 112     \\ \hline
        11    & 228     & 77     & $\{\bar{F}_1,\bar{F}_2\}$     & 151     \\ \hline
        12    & 398     & 149     & $\{\bar{D}_2,\bar{I}_2\}$     & 249     \\ \hline
        13    & 473     & 173    & $\{\bar{E}_1,\bar{I}_1\}$     & 300     \\ \hline
        14    & 737     & 141     & $\{\bar{F}_2,\bar{I}_1\}$     & 596     \\ \hline
        15   & 1132     & 117     & $\{\bar{G}_1,\bar{I}_2\}$     & 1015     \\ \hline
        16   & 1441     & 191     & $\{\bar{H}_1,\bar{I}_1\}$     & 1250     \\ \hline
        17   & 2064     & 378     & $ \{\bar{I}_1,\bar{I}_2\}$     & 1686     \\ \hline
       % Row 13   & Data     & Data     & Data     & Data     \\ \hline
    \end{tabular}
    
    \qquad
    
    \caption{Comparison of allowed generators in the Poisson brackets from different approaches}
    \label{tab:1}
\end{table}
Table $\ref{tab:1}$ presented herein further demonstrates the significant reduction in the number of permissible  terms achieved through the application of the grading method, prior to performing explicit calculations. In the following subsections, we elaborate on the precise formulas of the generators that are obtained via the grading  method.

\subsection{The allowed monomials from the grading of Poisson brackets}

\label{4.1}

 We now study the specific grading assigned to each generator within the set $\mathcal{P}$. Additionally, we perform a detailed analysis of the grading associated with their respective Poisson brackets, using the terminology and definitions outlined in Subsection $\ref{subsec2.3}$.

Assume that $\mathfrak{g} = \mathfrak{g}_1 \oplus \mathfrak{g}_2 \oplus \mathfrak{g}_3$ is a semisimple Lie algebra  with  \begin{align}
\nonumber
    [\mathfrak{g}_1,\mathfrak{g}_1] & \,\subset \mathfrak{g}_1, \quad  [\mathfrak{g}_2,\mathfrak{g}_2] \subset \mathfrak{g}_2, \quad  [\mathfrak{g}_1,\mathfrak{g}_2] = \{0\},  \\
    [\mathfrak{g}_1,\mathfrak{g}_3] & \,\subset \mathfrak{g}_3, \quad  [\mathfrak{g}_2,\mathfrak{g}_3] \subset \mathfrak{g}_3, \quad  [\mathfrak{g}_3,\mathfrak{g}_3] \subset \mathfrak{g}_1 \oplus \mathfrak{g}_2. \label{eq:c1}
\end{align} 
Then $S(\mathfrak{g}) \cong S(\mathfrak{g}_1) \otimes S(\mathfrak{g}_2) \otimes S(\mathfrak{g}_3)$.  A first step consists of knowing the grading of the non-trivial bracket of the generators of the commutant $S(\mathfrak{g})^{\mathfrak{g}_1 \oplus \mathfrak{g}_2}.$ Without loss of generality, assume that $\mathfrak{g}_1 = \mathrm{span} \{x_1,\ldots,x_u\}$, $\mathfrak{g}_2 = \mathrm{span} \{x_{u+1},\ldots,x_{u+v}\}$, and $ \mathfrak{g}_3  = \mathrm{span} \{x_{u+v+1},\ldots,x_{u+v+w}\} $. Let $\mathfrak{g}^\prime= \mathfrak{g}_1 \oplus \mathfrak{g}_2$ be a subalgebra of $\mathfrak{g}$ such that $\mathcal{Q}_\mathfrak{g}(d)$ is a polynomial algebra consisting of $\mathfrak{g}^\prime$-invariant polynomials. We further assume that a monomial with degree $i_1 + i_2  + i_3$ in $\mathcal{Q}_{i_1 + i_2  + i_3} \subset \mathcal{Q}_\mathfrak{g}(d)$ has the form of \begin{align*}
    p^{i_1 + i_2  + i_3}=   \underbrace{x_1^{a_1} \cdots x_u^{a_u}}_{\text{elements in $\mathfrak{g}_1$}} \underbrace{x_{u+1}^{a_{u+1}} \cdots x_{u+v}^{a_{u+v}}}_{\text{elements in $\mathfrak{g}_2$}}\underbrace{x_{u+v+1}^{a_{u+v+1}} \cdots x_{u+v+w}^{a_{u+v+w}}}_{\text{elements in $\mathfrak{g}_3$}} , \quad \begin{matrix}
           a_1 + \ldots + a_u = i_1 \\
           a_{u+1} + \ldots + a_{u+v} = i_2, \\
           a_{u+v+1} + \ldots + a_{u+v+w} = i_3
      \end{matrix} 
  \end{align*} $ \text{ with } a_1,\ldots,a_{u+v+w} \in \mathbb{N}_0 .$ By definition, for any non-zero $p \in \mathcal{Q}_{i_1 + i_2  + i_3}$,  $\mathcal{G}(p) = (i_1,i_2,i_3).$

\begin{lemma}
\label{3.9}
  Let $\mathfrak{g}^*$ be its dual admitting the same relations as in $\eqref{eq:c1}$, and let $\mathcal{Q}_\mathfrak{g}(d) \subset S(\mathfrak{g})^{\mathfrak{g}_1 \oplus \mathfrak{g}_2 }.$ For any non-zero indecomposable monomials $p \in \mathcal{Q}_{i_1 + i_2  + i_3}$ and $q \in \mathcal{Q}_{i_1' + i_2' + i_3'}$,   \begin{align}
  \nonumber
   \mathcal{G} \left(\{p,q\}\right) =     (i_1+i_1'-1,i_2+i_2',i_3+i_3') \tilde{+}(i_1+i_1',i_2+i_2'-1,i_3+i_3') \tilde{+}    \left\{\begin{matrix}
(i_1+i_1'+1,i_2+i_2',i_3+i_3'-2)  & \text{ if } \text{ }  [\mathfrak{g}_3,\mathfrak{g}_3] \subset \mathfrak{g}_1 \\
        \\
         (i_1+i_1',i_2+i_2'+1,i_3+i_3'-2) & \text{ if }  \text{ } [\mathfrak{g}_3,\mathfrak{g}_3] \subset \mathfrak{g}_2.
    \end{matrix}\right.  \\
\end{align} % Here $1 \leq \alpha,\beta,i,j \leq 3$ are the indices appearing in the Berezin bracket relations of $\mathfrak{su}^*(4)$ in \eqref{classicalrels}.
 
\end{lemma}

 \begin{proof}
     For any non-zero monomials $p  \in \mathcal{Q}_{i_1 + i_2  + i_3}$ and $q \in \mathcal{Q}_{i_1' + i_2' + i_3'}$ without loss of generality, we may write that $p = Y_1Y_2Y_3$ and $q = Z_1Z_2Z_3,$ where  
        \begin{align*}
         &  Y_1 = x_1^{a_1} \cdots x_u^{a_u},    \text{ } Y_2 = x_{u+1}^{a_{u+1}} \cdots x_{u+v}^{a_{u+v}}, \text{ } Y_3 =  x_{u+v+1}^{a_{u+v+1}} \cdots x_{u+v+w}^{a_{u+v+w}} ; \\
      & Z_1 = x_1^{a_1'} \cdots x_u^{a_u'} ,\text{ } Z_2 = x_{u+1}^{a_{u+1}'} \cdots x_{u+v}^{a_{u+v}'} , \text{ }  Z_3 =  x_{u+v+1}^{a_{u+v+1}'} \cdots x_{u+v+w}^{a_{u+v+w}'}.
    \end{align*} 
    Here $a_1 + \ldots + a_u = i_1, a_{u+1} + \ldots + a_{u+v} = i_2$, $a_{u+v+1} + \ldots + a_{u+v+w} = i_3$ and $a_1' + \ldots + a_u' = i_1', a_{u+1}' + \ldots a_{u+v}' = i_2'$, $a_{u+v+1}' + \ldots + a_{u+v+w}' = i_3'.$ A direct computation shows that 
    \begin{align}
     \{p,q\} = \{Y_1,q\}Y_2Y_3 + Y_1\{Z_2,q\}Z_3 + Y_1 Y_2 \{Y_3,q\}. \label{eq:Leibniz2}
 \end{align}    
 By assumption, the equation $\eqref{eq:Leibniz2}$ becomes 
 \begin{align}
     \{p,q\}    =    & \,    Y_1Y_2 \{Y_3,Z_1\}Z_2 Z_3 +  Y_1Y_2 \{Y_3,Z_2\}Z_1 Z_3 +  Y_1Y_2 \{Y_3,Z_3\}Z_2 Z_1 \, . \label{eq:grading3}
 \end{align}   
 By definition, the grading of $\{p,q\}$ in $\eqref{eq:grading3}$ is equal to the grading of each of the components. From the commutator relations, $\{A_1,B_1\} = 0$. We can then discard this term in $\eqref{eq:grading3}.$ For the rest of the components in $\eqref{eq:grading3}$, we  compute them case by case. Starting from the term  $  Y_1Y_2 \{Y_3,Z_1\}Z_2 Z_3$, a direct computation shows that 
 \begin{align*}
     \mathcal{G} \left(  Y_1Y_2 \{Y_3,Z_1\}Z_2 Z_3\right) = & \, \mathcal{G} \left(\{Y_3,Z_1\}\right) + \mathcal{G} (Z_2) + \mathcal{G} (Z_3) + \mathcal{G} (Y_1) +\mathcal{G} (Y_2)  \\
     = & \, (i_1-1,0,i_3) + ( i_1', 0,0 ) + (0,0,i_3')+ (0,i_2,0)+ (0,i_2',0 ) \\
     = & \, (i_1 +i_1'-1,i_2 +i_2',i_3+i_3') .
 \end{align*} Similarly, we have \begin{align*}
      \mathcal{G} \left(Y_1Y_2 \{Y_3,Z_2\}Z_1 Z_3 \right) = & \, (i_1+i_1',i_2+i_2'-1,i_3+i_3')\\
       \mathcal{G} \left(Y_1Y_2 \{Y_3,Z_3\}Z_2 Z_1\right) = & \, \left\{\begin{matrix}
           ( i_1+i_1'+1,i_2+i_2',i_3+i_3'-2) &  \text{ if } \text{ }[\mathfrak{g}_3,\mathfrak{g}_3] \subset \mathfrak{g}_1 \\
           (i_1+i_1',i_2+i_2'+1,i_3+i_3'-2) & \, \text{ if } \text{ }[\mathfrak{g}_3,\mathfrak{g}_3] \subset \mathfrak{g}_2 .
       \end{matrix}\right.   
 \end{align*} 
 Summing all the terms together, we deduce the grading of $\{p,q\}.$  
 \end{proof}

 \begin{remark}
   From the generators provided in $\eqref{suitgene}$, we observe that  
   \begin{align}
       \mathcal{G} \left(\{p,q\}\right) = & \,    (i_1+i_1'-1,i_2+i_2',i_3+i_3') \tilde{+}(i_1+i_1',i_2+i_2'-1,i_3+i_3') \tilde{+}  (i_1+i_1'+1,i_2+i_2',i_3+i_3'-2) \nonumber \\
     & \, \tilde{+}     (i_1+i_1',i_2+i_2'+1,i_3+i_3'-2) .
   \end{align}
 \end{remark}

 We next examine the grading of each generator within $\mathcal{P}.$

\begin{proposition}
\label{grading}
   The grading of each generator in $\mathcal{P}$ are given by \begin{align}
\mathcal{G} (\bar{b}_1) = & \, (2,0,0) ,\quad \mathcal{G} (\bar{b}_2) =  (0,2,0) ,\quad\mathcal{G} (\bar{b}_3) =   (0,0,2) ;\nonumber\\
    \mathcal{G} (\bar{c}_1) = & \, (1,1,1) \tilde{+} (0,0,3) ,\quad  \mathcal{G} (\bar{C}_2) =  (1,1,1) ; \nonumber \\
    \mathcal{G} (\bar{d}_1) = & \, (2,0,2) \tilde{+} (0,2,2) \tilde{+} (1,1,2) \tilde{+} (0,0,4) ;\nonumber \\ 
    \mathcal{G}(\bar{D}_2) = & \,  (1,1,2) , \quad  \mathcal{G}(\bar{D}_3) = (0,2,2) , \quad \mathcal{G}(\bar{D}_4) = (0,0,4) ,\quad  \mathcal{G} (\bar{E}_1) =  (1,1,3) ; \nonumber \\
  \mathcal{G}(\bar{F}_1) = & \,(2,1,3) , \quad \hskip 0.1cm\mathcal{G}(\bar{F}_2) = (2,0,4) , \quad \hskip 0.07cm\mathcal{G}(\bar{F}_3) = (1,2,3), \quad \mathcal{G}(\bar{F}_4) = (0,2,4); \label{eq:generatorsgrad} \\
            \mathcal{G} (\bar{G}_1) = & \, (2,1,4), \quad \hskip 0.05cm\mathcal{G}(\bar{G}_2) = (1,2,4) , \quad \mathcal{G}(\bar{H}_1) = (2,1,5) , \quad \mathcal{G}(\bar{H}_2) = (1,2,5) , \nonumber \\
            \mathcal{G}(\bar{I}_1) = & \,  (0,3,6) , \hskip 0.2cm\quad \mathcal{G}(\bar{I}_2) = (3,0,6).\nonumber
\end{align} 
\end{proposition}

\begin{remark}
\label{regrad}
% We define the \textit{Abelian grading term} by the term that appears in the gradings of the trivial brackets. For instance, the degree $6$ Poisson bracket may contains the following Abelian term  \begin{align*}
 %       \mathcal{G} \left(\{\bar{c}_1,\bar{d}_1\}\right) = & \,\mathcal{G} \left(\{\bar{C}_2,C^{(202)}\}\right) \tilde{+}\mathcal{G} \left(\{\bar{C}_2,\bar{D}_3\}\right) \tilde{+}\mathcal{G} \left(\{\bar{C}_2,\bar{D}_2\}\right) \tilde{+}\mathcal{G} \left(\{\bar{C}_2,\bar{D}_4\}\right) \\
 %       & \, \tilde{+}\mathcal{G} \left(\{C^{(003)},C^{(202)}\}\right) \tilde{+}\mathcal{G} \left(\{C^{(003)},\bar{D}_3\}\right) \tilde{+}\mathcal{G} \left(\{C^{(003)},\bar{D}_2\}\right) \tilde{+}\mathcal{G} \left(\{C^{(003)},\bar{D}_4\}\right) \\
 %       = & \, (3,0,3) \tilde{+}(4,1,1) \tilde{+}(3,2,1) \tilde{+}(1,0,5) \tilde{+} (2,1,3) \tilde{+}(0,1,5) \tilde{+} (1,2,3) \tilde{+} (0,3,3)   \tilde{+}   (1,4,1).  
  %  \end{align*}  From this grading, we conclude that $\mathcal{G} \left(\{\bar{c}_1,\bar{d}_1\}\right) =\alpha F_1 + \beta F_2. $ Hence, in the degree $6$ bracket, we need to insert two more gradings $(1,2,3)$ and $(2,1,3)$ for the Abelian grading term.
    
Using the generators delineated in equations $\eqref{polquesne}$ and $\eqref{suitgene}$, we can determine that the homogeneous components of the generators denoted $\bar{c}_1$ and $\bar{d}_1$ are precisely constituted by the generators $\bar{C}_2$, $\bar{D}_2$, $\bar{D}_3$, and $\bar{D}_4$. We then represent these components in the subsequent manner: \begin{align} 
 \bar{c}_1 =  \bar{C}_2 - \frac{2}{3}C^{(003)}  \text{ and } \bar{d}_1 = C^{(202)} - \bar{D}_2 + \bar{D}_3 - 2\bar{D}_4, \label{eq:changeba1}
 \end{align} where $\mathcal{G} \left(C^{(003)} \right)=(0,0,3)$ and $\mathcal{G} \left(C^{(202)}\right)=(2,0,2)$. As a consequence of these expressions, additional terms emerge within the grading of the Poisson bracket $\{p,q\}$. Consequently, in the initial calculation of the anticipated grading $\mathcal{G}(\{p,q\})$, we will determine the permissible terms consisting of $C^{(003)}$ and $C^{(202)}$ from each homogeneous grading. Then, under the change of generators, all the allowed generators in terms of $\bar{c}_1$ and $\bar{d}_1$ can be obtained. Clearly, this involves the further terms with the grading that do not match $\mathcal{G} (\{p,q\}).$ This implies that the summation of these terms must equal the homogeneous terms $C^{(003)}$ or $C^{(202)}$. This gives us some constraints on determining the coefficients.  This point will become clearer in the following subsections when we will study explicit examples. 
\end{remark}

 In the following, using Proposition $\ref{grading}$ together with Lemma $\ref{3.9}$, we shall elucidate the grading associated with each Poisson bracket in detail. Additionally, we will provide the gradings of all the compact forms, and demonstrate all allowable terms derived from the homogeneous gradings. As $\bar{b}_1,\bar{b}_2, \bar{b}_3$ and $\bar{c}_1$ are central elements, we will start with the degree $6$ brackets.

\vskip 0.5cm
\subsection{Expansions in the Degree $6$ and $7$ brackets}

\label{subsec4.2}

We examine the expansions of the compact forms $\{\bar{\textbf{C}},\bar{\textbf{D}}\}$, $\{\bar{\textbf{C}},\bar{\textbf{E}}\}$, and $\{\bar{\textbf{D}},\bar{\textbf{D}}\} $. We now present the explicit formula for the polynomial relations within these significant brackets, utilizing the grading for clarification. Initially, we examine the non-trivial bracket involving a non-central element from  $\bar{C}_2$ and  $\bar{D}_2,\bar{D}_3,\bar{D}_4$. Through straightforward computation, it is obvious that 
\begin{align}
    \mathcal{G} \left(\{\bar{C}_2,\bar{D}_2\}\right) = &\, (1,2,3) \tilde{+}(2,1,3) \tilde{+}   (3,2,1)  
      \tilde{+}    (2,3,1)    ; \nonumber  \\
    \mathcal{G} \left(\{\bar{C}_2,\bar{D}_3\}\right) = &\,  (0,3,3) \tilde{+}(1,2,3) \tilde{+}   
         (1,4,1) \tilde{+} (2,3,1) ;  \label{eq:grad1}  \\
    \mathcal{G} \left(\{\bar{C}_2,\bar{D}_4\}\right) = &\,  (0,1,5) \tilde{+}(1,0,5) \tilde{+} 
     (2,1,3)   \tilde{+}
         (1,2,3)  \, . \nonumber
\end{align} 
From the grading of each generator in $\eqref{eq:generatorsgrad},$ we can exclude the terms $(2,3,1), (3,2,1)$, $ (0,3,3), (1,4,1)$, $(0,1,5)$ and $(1,0,5)$ as there are no polynomials in $\textbf{Alg} \langle \mathcal{P}\rangle$ with the grading above. Then the gradings in $\eqref{eq:grad1}$ become 
\begin{align*}
    &\mathcal{G}  \left(\{\bar{C}_2,\bar{D}_2\}\right) =    (1,2,3) \tilde{+}(2,1,3) = \mathcal{G}  \left(\{\bar{C}_2,\bar{D}_4\}\right) =
    \mathcal{G}  \left(\{\bar{C}_2,\bar{D}_3\}\right) . 
\end{align*} 
Hence, the allowed terms in each Poisson bracket above are given by \begin{align}
    \left\{\bar{C}_2,\bar{D}_2\right\} = \Gamma_{34}^1 \bar{F}_1 + \Gamma_{34}^2 \bar{F}_3, \text{ }  \left\{\bar{C}_2,\bar{D}_3\right\} =     \Gamma_{34}^3 \bar{F}_1 + \Gamma_{34}^4\bar{F}_3   
    , \text{ } \left\{\bar{C}_2,\bar{D}_4\right\} =      \Gamma_{34}^5 \bar{F}_1  + \Gamma_{34}^6\bar{F}_3 . \label{eq:expect} 
\end{align} 
Here $\Gamma_{34}^k$,  for $1 \leq k \leq 6$, are some constants to be determined. Let us remark that the subindices $``34 \textquotedblright$ in the coefficients indicate the corresponding degree of the generators appearing in the Poisson brackets on the left hand side. We will use this notation throughout the paper. Back to the compact forms, at the degree six, the formal relations from $\eqref{deg6b}$ reduce to 
\begin{equation}
  \left\{\bar{\textbf{C}},\bar{\textbf{D}}\right\} \sim \bar{\textbf{F}}.
\end{equation}  With the help of the relations in $\eqref{classicalrels}$, for degree-six expansions, we explicitly get:
\begin{equation}
\{\bar{C}_2, \bar{D}_2\}=-2{\rm i} (\bar{F}_1+\bar{F}_3) \, , \qquad \{\bar{C}_2, \bar{D}_3\}= -2 {\rm i} \bar{F}_3 \, , \qquad \{\bar{C}_2, \bar{D}_4\}=0 \, .
\end{equation}

Now, we determine the expected polynomials in the degree-seven brackets. This involves the brackets in the compact forms of $\{\bar{\textbf{C}},\bar{\textbf{E}}\}$ and $\{\bar{\textbf{D}},\bar{\textbf{D}}\}.$  We first consider the grading of the compact form  $\{\bar{\textbf{C}},\bar{\textbf{E}}\}$. As $\bar{c}_1$ is a central element, we only need to consider the expansion of the Poisson bracket $\{\bar{C}_2,\bar{E}_1\}.$  A direct calculation shows that  \begin{align}
    \mathcal{G} \left( \{\bar{C}_2, \bar{E}_1\}\right) = (1,2,4) \tilde{+} (2,1,4) \tilde{+}    (3,2,2)   \tilde{+}  (2,3,2) \, . \label{eq:deg7}
\end{align}   It is clear that the allowed polynomials from the grading above are contained in $\eqref{eq:deg7}.$ Hence, $\{\bar{C}_2,\bar{E}_1\} = \Gamma_{35}^1 \bar{G}_1 + \Gamma_{35}^2 \bar{G}_2.$ On the other hand, the grading from the compact form $\{\bar{\textbf{D}},\bar{\textbf{D}}\}$ is given by \begin{align*}
    \mathcal{G} \left( \{\bar{D}_2,\bar{D}_3\} \right) = & \, (0,3,4) \tilde{+}  (1,2,4) \tilde{+} (2,3,2) \tilde{+} (1,4,2) ; \\
    \mathcal{G} \left( \{\bar{D}_3,\bar{D}_4\}\right) = &  \, (0,1,6) \tilde{+} (1,2,4) \tilde{+} (0,3,4) \,. 
\end{align*}
Degree-seven Berezin brackets $\{\bar{\textbf{D}},\bar{\textbf{D}}\}$ read explicitly as follows:
\begin{equation}
\{\bar{D}_2, \bar{D}_3\}=2\{\bar{D}_3, \bar{D}_4\}=4{\rm i}\bar{G}_2 \, , \qquad \{\bar{D}_2, \bar{D}_4\}=2\{\bar{E}_1, \bar{C}_2\}=2{\rm i}\bigl(\bar{G}_1+\bar{G}_2\bigl) \, .
\end{equation} These relations conclude all the degree-seven Berezin brackets.

\vskip 0.5cm

\subsection{Expansions in the degree $8$ brackets}

\label{subsec4.3}

We now determine the non-trivial brackets with degree $8.$ This involves the brackets in the compact forms of $\{\bar{\textbf{C}},\bar{\textbf{F}}\}$ and $\{\bar{\textbf{D}},\bar{\textbf{E}}\}.$ We first consider the grading of all the terms in $\{\bar{\textbf{C}},\bar{\textbf{F}}\}.$ A direct computation shows that \begin{align}
    \mathcal{G} \left(\{\bar{C}_2,\bar{F}_1\}\right) = & \, (2,2,4) \tilde{+} (3,1,4) \tilde{+}   (4,2,2)  \tilde{+}   (3,3,2) \label{eq:4.22} ; \\
       \mathcal{G} \left(\{\bar{C}_2,\bar{F}_2\}\right) = & \, (2,1,5) \tilde{+} (3,0,5) \tilde{+}    (4,1,3) \tilde{+}     (3,2,3)  \label{eq:4.23} ;\\ 
       \mathcal{G} \left(\{\bar{C}_2,\bar{F}_3\}\right) = & \, (1,3,4) \tilde{+} (2,2,4) \tilde{+}     (3,3,2)\tilde{+}  (2,4,2) \label{eq:4.24}; \\
       \mathcal{G} \left(\{\bar{C}_2,\bar{F}_4\}\right) = & \, (0,3,5) \tilde{+} (1,2,4) \tilde{+}    (2,3,3)  \tilde{+}   (1,4,3) \label{eq:4.25}.
\end{align} From Proposition $\ref{grading}$, the allowed polynomials in each homogeneous grading of $\eqref{eq:4.22}$, are given by \begin{align*}
    (2,2,4) : = & \, \left\{\bar{b}_1\bar{b}_2\bar{b}_3^2, \text{ }  \bar{C}_2^2 \bar{b}_3, \text{ }  \bar{D}_2^2, \text{ }  \bar{b}_1 \bar{b}_3\bar{D}_3, \text{ }  \bar{b}_1\bar{b}_2\bar{D}_4, \text{ }  \bar{C}_2\bar{E}_1, \text{ }  \bar{b}_2\bar{F}_2, \text{ }  \bar{b}_1\bar{F}_4, \text{ } C^{(202)}\bar{D}_3, \text{ } \bar{b}_2\bar{b}_3C^{(202)}   \right\} ,\\
    (3,1,4) : = & \, \left\{ \bar{b}_1\bar{b}_3 \bar{D}_2, \text{ } C^{(202)}\bar{D}_2, \text{ } \bar{b}_1 \bar{c}_2 C^{(003)}   \right\}  , \\
    (4,2,2) : = & \,  \left\{\bar{b}_1^2 \bar{b}_2\bar{b}_3, \text{ }  \bar{b}_1\bar{C}_2^2, \text{ }  \bar{b}_1^2\bar{D}_3  , \text{ } \bar{b}_1\bar{b}_2C^{(202)} \right\}  , \\
    (3,3,2) : = & \,  \left\{ \bar{b}_1\bar{b}_2\bar{D}_2\right\}. % \\
   % (2,2,4) \tilde{+} &  (3,1,4) \tilde{+}   (4,2,2)  \tilde{+}   (2,4,2) := \left\{\bar{b}_1\bar{b}_2\bar{d}_1 \right\} , \\
   % (3,1,4) \tilde{+}  &  (4,2,2) =: \left\{\bar{b_1}\bar{c}_1\bar{C}_2\right\}.
\end{align*} Here, the permissible polynomials are coming from the homogeneous gradings. As discussed in Remark $\ref{regrad}$, we simply replace the non-central elements $C^{(202)}$ and $C^{(003)}$ by $\bar{d}_1$ and $\bar{c}_1$ respectively.  We then easily deduce that the expected polynomials, under the new choice of the generators, are given by \begin{align*}
  & \,  \bar{b}_1 \bar{F}_4, \text{ } \bar{b}_2 \bar{F}_2,\text{ }\bar{C}_2 \bar{E}_1,\text{ }\bar{d}_1 \bar{D}_3,\text{ }\bar{D}_2 \bar{D}_3,\text{ }\bar{D}_3^2,\text{ }\bar{D}_3 \bar{D}_4,\text{ }\bar{D}_2^2,\bar{b}_1 \bar{b}_2 \bar{D}_4,\\
  & \,\bar{b}_1 \bar{b}_3 \bar{D}_3,\text{ }\bar{b}_2 \bar{b}_3 \bar{d}_1,\text{ }\bar{b}_2 \bar{b}_3 \bar{D}_2,\text{ }\bar{b}_2 \bar{b}_3 \bar{D}_3,\text{ }\bar{b}_2 \bar{b}_3 \bar{D}_4,\text{ }\bar{b}_3 \bar{C}_2^2,\text{ }\bar{b}_1 \bar{b}_2 \bar{b}_3^2,\\
  & \, \bar{d}_1 \bar{D}_2,\text{ }\bar{D}_2 \bar{D}_4,\text{ }\bar{b}_1 \bar{b}_3 \bar{D}_2,\text{ }\bar{b}_1 \bar{c}_1 \bar{C}_2,\text{ }\bar{b}_1 \bar{C}_2^2,\text{ }\bar{b}_1^2 \bar{D}_3,\\
  & \,\bar{b}_1 \bar{b}_2 \bar{d}_1,\text{ }\bar{b}_1 \bar{b}_2 \bar{D}_2,\text{ }\bar{b}_1 \bar{b}_2 \bar{D}_3,\text{ }\bar{b}_1^2 \bar{b}_2 \bar{b}_3.
\end{align*}  In what follows, we will perform a calculation of the total number of expected generators derived from the compact forms to various degrees. This serves to demonstrate the manner in which the grading technique can effectively streamline the preliminary computational processes. Considering the next case, the polynomials from the homogeneous grading in $ \eqref{eq:4.23}$ are    \begin{align*}
     (2,1,5) : = & \, \left\{\bar{b}_3\bar{F}_1, \bar{H}_1   \right\}  , \text{  } (3,0,5) : =  \emptyset , \text{ } (4,1,3) : = \left\{\bar{b}_1  \bar{F}_1\right\} \text{ and }  (3,2,3) : =   \emptyset .
\end{align*} From the discussion above, we may write \begin{align}
    \{\bar{C}_2, \bar{F}_{1}\}=& \Gamma_{36}^1\bar{b}_1\bar{b}_2\bar{b}_3^2+ \Gamma_{36}^2\bar{C}_2^2 \bar{b}_3+ \Gamma_{36}^3 \bar{D}_2^2+\Gamma_{36}^4 \bar{b}_1 \bar{b}_3\bar{D}_3+ \Gamma_{36}^5 \bar{b}_1\bar{b}_2\bar{D}_4+\Gamma_{36}^6 \bar{C}_2\bar{E}_1+  \Gamma_{36}^7 \bar{b}_2\bar{F}_2 \nonumber \\
    & \, +\Gamma_{36}^8 \bar{b}_1\bar{F}_4+\Gamma_{36}^9  \bar{b}_1\bar{b}_3 \bar{D}_2 + \Gamma_{36}^{10}\bar{b}_1^2 \bar{b}_2\bar{b}_3+\Gamma_{36}^{11} \bar{b}_1\bar{C}_2^2+ \Gamma_{36}^{12} \bar{b}_1^2\bar{D}_3 +\Gamma_{36}^{13}\bar{b}_1\bar{b}_2\bar{D}_2 \nonumber \\
    & \,+\underbrace{\left(\Gamma_{36}^{14}\bar{b_1}\bar{c}_1\bar{C}_2 +\Gamma_{36}^{15} \bar{b_1} \bar{C}_2^2 + \Gamma_{36}^{16} \bar{d}_1\bar{D}_2 + \ldots + \Gamma_{36}^{26} \bar{b}_1 \bar{b}_2 \bar{d}_1 \right)}_{\text{New terms under the change of the generators }}; \label{eq:newt}\\
    \{\bar{C}_2,\bar{F}_2\} = & \, \Gamma_{36}^{  27} \bar{b}_3\bar{F}_1 +  \Gamma_{36}^{ 28}  \bar{H}_1   +  \Gamma_{36}^{ 29} \bar{b}_1  \bar{F}_1. \nonumber
\end{align} Here $\Gamma_{36}^k  $, for $1 \leq k \leq 29$,  are constant  coefficients. According to Remark $\ref{regrad}$, to ensure that the grading of the left-hand side of the given expansion aligns with the grading associated with the bracket $\{\bar{C}_2, \bar{F}_1\}$, we need to implement certain constraints in the new terms mentioned in $\eqref{eq:newt}$. For instance, there exists some $\alpha \in \mathbb{C}$ such that $\left(\Gamma_{36}^{14} \bar{c}_1 + \Gamma_{36}^{15} \bar{C}_2\right)\bar{C}_2 \bar{b_1} = \alpha \bar{C}_2 \bar{b_1} C^{(003)}$. This condition guarantees consistent grading between the constructs involved. With the help of the  classical analog of the relations of $\mathfrak{su}(4)$, we deduce that 
\begin{align}
\{\bar{C}_2, \bar{F}_{1}\}& = \,\frac{{\rm i}}{4}\bigl(-4 \bar{b}_3 \bar{C}_2^2+8 \bar{b}_2 \bar{b}_3\underbrace{(\bar{d}_1+\bar{D}_2-\bar{D}_3+2 \bar{D}_4)}_{\text{equals to $C^{(202)}$}}+ 2 \bar{D}_2\underbrace{(\bar{d}_1 +2 \bar{D}_2- \bar{D}_3+2 \bar{D}_4)}_{\text{equals to $C^{(202)} +\bar{D}_2$}} +4 \bar{C}_2 \bar{E}_1-8 \bar{b}_2 \bar{F}_2 \nonumber\\
& \hskip 0.9cm +\bar{b}_1\bigl(\underbrace{(\bar{c}_1-\bar{C}_2)}_{\text{equals to $-\frac{2}{3}C^{(003)}$} }\bar{C}_2+4 \bar{b}_3 \bar{D}_3+4 \bar{b}_2(\bar{D}_4-\bar{b}_3^2)-4 \bar{F}_4\bigl) \bigl) \, ,\nonumber \\
\{\bar{C}_2, \bar{F}_{2}\}&= \, -2{\rm i}\bar{H}_1 \, .
\end{align} 
Using a similar argument, we can deduce the allowed monomials from the brackets in $\eqref{eq:4.24}$ and $\eqref{eq:4.25}$.  The rest of the Poisson brackets in the compact form $\{\bar{\textbf{C}},\bar{\textbf{F}}\}$ are then given by
\begin{align} 
\{\bar{C}_2, \bar{F}_{3}\}&=\frac{{\rm i}}{4}\bigl( 2\bar{D}_2(\bar{D}_2+\bar{D}_3)+4 \bar{C}_2(\bar{E}_1-\bar{b}_3 \bar{C}_2)+\bar{b}_2\bigl((\bar{c}_1-\bar{C}_2)\bar{C}_2+4 \bar{b}_3(\bar{d}_1+\bar{D}_2-\bar{D}_3+2 \bar{D}_4)-4 \bar{F}_2 \bigl)  \nonumber\\
& \hskip 0.9 cm  +\bar{b}_1\bigl(8 \bar{b}_3 \bar{D}_3+4 \bar{b}_2(\bar{D}_4-\bar{b}_3^2)-8 \bar{F}_4\bigl)\bigl) \, , \nonumber\\
\{\bar{C}_2, \bar{F}_{4}\}&=-2 {\rm i} \bar{H}_2 \, .
\end{align} 
 Additionally, consider all the gradings in the compact form $\{\bar{\textbf{D}},\bar{\textbf{E}}\}$. Since $\bar{d}_1$ is a central element, we will start from $\bar{D}_2.$ A direct calculation shows that \begin{align}
    \mathcal{G} (\{\bar{D}_2,\bar{E}_1\}) = & \, (1,2,5) \tilde{+} (2,1,5) \tilde{+} (3,2,3) \tilde{+} (2,3,3); \label{eq:4.28}\\
    \mathcal{G} (\{\bar{D}_3,\bar{E}_1\}) = & \, (0,3,5) \tilde{+} (1,2,5) \tilde{+} (2,3,3) \tilde{+} (1,4,3) ; \label{eq:4.29}\\ 
   \mathcal{G} (\{\bar{D}_4,\bar{E}_1\}) = & \, (0,1,7) \tilde{+} (1,0,7) \tilde{+} (2,1,5) \tilde{+} (1,2,5).  \label{eq:4.30}
\end{align} The way of choosing permissible polynomials will be the same as presented above. Hence, we only provide one of these cases. Consider the grading from the bracket $\{\bar{D}_2,\bar{E}_1\},$ the allowed polynomials in each component are \begin{align}
    (1,2,5) :=   \left\{\bar{H}_2\right\}, \quad   (2,1,5) :=  \left\{\bar{b}_3\bar{F}_1,\bar{H}_1\right\}, \quad    (3,2,3) :=  \left\{\bar{b}_1\bar{F}_3\right\}  , \quad    (2,3,3):=  \left\{\bar{b}_2\bar{F}_1\right\}. \label{eq:gend2e1}
\end{align} Then the expected generators in the Poisson bracket $\eqref{eq:4.28}$ are given by 
\begin{align*}
    \{\bar{D}_2,\bar{E}_1\} =\Gamma_{45}^1 \bar{H}_1 + \Gamma_{45}^2 \bar{H}_2 + \left(\Gamma_{45}^3\bar{b}_3  + \Gamma_{45}^{4}\bar{b}_2\right)\bar{F}_1 + \Gamma_{45}^{5}\bar{b}_1\bar{F}_3.
\end{align*} Here, $\Gamma_{45}^k$ are constant coefficients for all $1 \leq k \leq 5$.  In the explicit calculation, the explicit values of the coefficients are $\Gamma_{45}^1 = \Gamma_{45}^2 = 2 {\rm i}$, $  \Gamma_{45}^3 = 0$, and $ \Gamma_{45}^4 = \Gamma_{45}^{5}=\frac{{\rm i}}{2}$, leading to
\begin{align}
    \{\bar{D}_2, \bar{E}_{1}\}&={\rm i}\left(2\bar{H}_1+2 \bar{H}_2+\frac{1}{2}(\bar{b}_2 \bar{F}_1+\bar{b}_1 \bar{F}_3)\right) \,.
\end{align} 
Analogously, from the grading $\eqref{eq:4.29}$ and $\eqref{eq:4.30}$, we also find
\begin{align}
\{\bar{D}_3, \bar{E}_{1}\}&={\rm i}\bigl(-2 \bar{H}_2+\frac{1}{2}\bar{b}_2 \bar{F}_1+2 \bar{b}_3 \bar{F}_3\bigl)\, ,\nonumber \\
\{\bar{D}_4, \bar{E}_{1}\}&={\rm i}\bigl( -2 \bar{H}_2-2\bar{H}_1+\bar{b}_3(\bar{F}_1+\bar{F}_3)\bigl)\, .
\end{align} 
This conclude the computations of the Poisson brackets up to degree eight.
\vskip 0.5cm

\subsection{Expansions in the degree $9$ brackets}

We now focus on the compact form corresponding to the degree $9$ bracket relations. The compact forms are $\{\bar{\textbf{C}},\bar{\textbf{G}}\}$ and $\{\bar{\textbf{D}},\bar{\textbf{F}}\}$. We start by analyzing the compact form $\{\bar{\textbf{D}},\bar{\textbf{F}}\}$.  The gradings in the term $\{\bar{D}_2,\bar{F}_s\}$ for all $1 \leq s \leq 4$ are given by \begin{align}
    \mathcal{G}\left(\{\bar{D}_2,\bar{F}_1\}\right) = & \, (2,2,5) \tilde{+} (3,1,5) \tilde{+} (4,2,3) \tilde{+}(3,3,3) ; \nonumber\\
     \mathcal{G}\left(\{\bar{D}_2,\bar{F}_2\}\right) = & \, (2,1,6) \tilde{+} (3,0,6) \tilde{+} (4,1,4) \tilde{+}(3,2,4) ;  \nonumber\\
     \mathcal{G}\left(\{\bar{D}_2,\bar{F}_3\}\right) = & \, (1,3,5) \tilde{+} (2,2,5) \tilde{+} (3,3,3) \tilde{+}(2,4,3) ;  \nonumber\\ 
      \mathcal{G}\left(\{\bar{D}_2,\bar{F}_4\}\right) = & \, (1,2,6) \tilde{+} (0,3,6) \tilde{+} (2,3,4) \tilde{+}(1,4,4).
      \end{align} Moreover, for the Poisson brackets $\{\bar{D}_3,\bar{F}_s\}$ and $\{\bar{D}_4,\bar{F}_s\}$ are given by
      \begin{align}
         \mathcal{G}\left(\{\bar{D}_3,\bar{F}_1\}\right) = & \, (1,3,5) \tilde{+} (2,2,5) \tilde{+} (3,3,3) \tilde{+}(2,4,3) ;  \nonumber \\
          \mathcal{G}\left(\{\bar{D}_3,\bar{F}_2\}\right) = & \, (1,2,6) \tilde{+} (2,1,6) \tilde{+} (3,2,4) \tilde{+}(2,3,4) ;  \nonumber \\
             \mathcal{G}\left(\{\bar{D}_3,\bar{F}_3\}\right) = & \, (0,4,5) \tilde{+} (1,3,5) \tilde{+} (2,4,3) \tilde{+}(1,5,3) ;  \nonumber\\
           \mathcal{G}\left(\{\bar{D}_3,\bar{F}_4\}\right) = & \,   (0,3,6) \tilde{+} (1,4,4) \tilde{+}(0,5,4);  \nonumber \\
       \mathcal{G}\left(\{\bar{D}_4,\bar{F}_1\}\right) = & \, (1,1,7) \tilde{+} (2,0,7) \tilde{+} (3,1,5) \tilde{+}(2,2,5) ;  \nonumber \\
          \mathcal{G}\left(\{\bar{D}_4,\bar{F}_2\}\right) = & \, (1,0,8) \tilde{+}   (3,0,6) \tilde{+}(2,1,6) ; \nonumber \\
             \mathcal{G}\left(\{\bar{D}_4,\bar{F}_3\}\right) = & \, (0,2,7) \tilde{+} (1,1,7) \tilde{+} (2,2,5) \tilde{+}(1,3,5) ;  \nonumber\\
           \mathcal{G}\left(\{\bar{D}_4,\bar{F}_4\}\right) = & \,   (0,1,8) \tilde{+} (1,2,6) \tilde{+}(0,3,6).      \label{eq:4.41}
\end{align} Note that, using Lemma \ref{3.9}, the grading of $\{\bar{D}_3,\bar{F}_4\}$ contains the homogeneous grading $(-1,4,6)$. It is clear that, from Proposition \ref{grading}, there is no generator or product of generators in $\mathcal{P}$ satisfying such homogeneous grading. Hence, we will omit the terms like that.

The procedure for finding the admissible monomials is, in fact, the same as before. Thus, we outline the permissible terms within certain gradings. Consider the grading in the bracket $\{\bar{D}_2,\bar{F}_1\}.$ Note that there will be extra polynomial terms by considering the change of the generators. This can be seen as follows. The admissible monomials for each homogeneous gradings are given by  \begin{align}
\nonumber
 (2,2,5) := & \, \left\{\bar{D}_2\bar{E}_1, \text{ } \bar{b}_3 \bar{C}_2 \bar{D}_2, \text{ }C^{(003)} \bar{C}_2^2, \text{ }\bar{b}_1 C^{(003)} \bar{D}_3, \text{ }\bar{b}_2 C^{(003)} C^{(202)}, \text{ }\bar{b}_1 \bar{b}_2 \bar{b}_3 C^{(003)}\right\}; \\
 (3,1,5):= & \, \left\{\bar{b}_1\bar{b}_3\bar{E}_1, \text{ }\bar{b}_1\bar{b}_3^2\bar{C}_2,    \text{ }\bar{b}_1\bar{C}_2\bar{D}_4, \text{ }\bar{C}_2 \bar{F}_2, \text{ }C^{(202)} \bar{E}_1 , \text{ }\bar{b}_1 C^{(003)} \bar{D}_2, \text{ } \bar{b}_3 \bar{C}_2 C^{(202)} \right\} ; \label{eq:gradingd2f1} \\ 
 \nonumber
 (4,2,3) := & \, \left\{\bar{b}_1\bar{D}_2\bar{C}_2, \text{ }\bar{b}_1^2\bar{b}_2C^{(003)}\right\}; \\
 \nonumber
 (3,3,3) := & \, \left\{\bar{b}_1\bar{b}_2\bar{E}_1, \text{ }\bar{b}_1\bar{b}_2\bar{b}_3\bar{C}_2, \text{ }\bar{b}_1\bar{C}_2\bar{D}_3, \text{ } \bar{C}_2^3, \text{ }\bar{b}_2C^{(202)}\bar{C}_2\right\} .
 % (4,2,3) \, \tilde{+}  & \,(3,1,5)   =: \{\bar{b}_1\bar{c}_1\bar{D}_2\}\\
%  (2,2,5) \,\tilde{+}  & \,(3,3,3)   =: \{\bar{b}_1\bar{b}_2\bar{b}_3\bar{c}_1,\bar{b}_1\bar{c}_1\bar{D}_3,\bar{c}_1\bar{C}_2^2\}\\
 %    (2,2,5)\, \tilde{+}  & \,(3,1,5) \tilde{+} (4,2,3) \tilde{+}(3,3,3) =: \emptyset.
\end{align}  As mentioned before, the permissible monomials that arise from these gradings can be described as follows: \begin{align*}
   & \,  \bar{D}_2 \bar{E}_1,\bar{b}_1 \bar{c}_1 \bar{D}_3,\text{ } \bar{b}_1 \bar{C}_2 \bar{D}_3\text{ } ,\bar{b}_2 \bar{c}_1 \bar{d}_1,\text{ } \bar{b}_2 \bar{c}_1 \bar{D}_2,\text{ } \bar{b}_2 \bar{c}_1 \bar{D}_3,\text{ } \bar{b}_2 \bar{c}_1 \bar{D}_4,\bar{b}_2 \bar{C}_2 \bar{d}_1, \\
   & \, \bar{b}_2 \bar{C}_2 \bar{D}_2,\bar{b}_2 \bar{C}_2 \bar{D}_3,\text{ } \bar{b}_2 \bar{C}_2 \bar{D}_4,\text{ } \bar{b}_3 \bar{C}_2 \bar{D}_2,\text{ } \bar{c}_1 \bar{C}_2^2,\text{ } \bar{C}_2^3,\bar{b}_1 \bar{b}_2 \bar{b}_3 \bar{c}_1, \text{ } \bar{b}_1 \bar{C}_2 \bar{D}_4,\\
   & \,\bar{b}_1 \bar{b}_2  \bar{b}_3 \bar{C}_2,\bar{C}_2 \bar{F}_2,\bar{d}_1 \bar{E}_1,\text{ } \bar{D}_3 \bar{E}_1,\bar{D}_4 \bar{E}_1,\text{ } \bar{b}_1 \bar{b}_3 \bar{E}_1,\text{ } \bar{b}_1 \bar{c}_1 \bar{D}_2,\text{ } \bar{b}_1 \bar{C}_2 \bar{D}_2,\\
   & \,\bar{b}_3 \bar{C}_2 \bar{d}_1,\text{ } \bar{b}_3 \bar{C}_2 \bar{D}_3,\bar{b}_3 \bar{C}_2 \bar{D}_4,\text{ } \bar{b}_1 \bar{b}_3^2 \bar{C}_2,\text{ } \bar{b}_1^2 \bar{b}_2 \bar{c}_1,\text{ } \bar{b}_1^2 \bar{b}_2 \bar{C}_2,\text{ } \bar{b}_1 \bar{b}_2 \bar{E}_1.
\end{align*}  From this, we conclude that there are some coefficients $\Gamma_{46}^k$, with $1 \leq k \leq 31$, such that  \begin{align*}
    \{\bar{D}_2, \bar{F}_1\} = & \,\Gamma_{46}^1 \bar{D}_2\bar{E}_1  +   \ldots+ \Gamma_{46}^{30} \bar{C}_2^3   + \Gamma_{46}^{31} \bar{b}_1\bar{b}_2\bar{E}_1 .
\end{align*} Again, these coefficients can be determined via  the classical analog of the relations of $\mathfrak{su}(4)$ and the constraints from the homogeneous gradings. Eventually, we deduce that \begin{align}
    \{\bar{D}_2, \bar{F}_1\}&= \frac{{\rm i}}{8}\bigl( \bar{b}_1^2 \bar{b}_2(\bar{c}_1-\bar{C}_2) +\bar{b}_1\bigl(4 \bar{b}_2 \bar{b}_3(\bar{c}_1-\bar{C}_2)-8 \bar{b}_3^2 \bar{C}_2-2(\bar{c}_1-2\bar{C}_2)(\bar{D}_2+2\bar{D}_3)+8 \bar{C}_2 \bar{D}_4-4 \bar{b}_2 \bar{E}_1\bigl)\nonumber\\
& \hskip 0.7cm+4\bigl(-3\bar{C}_2^3+  2\bar{c}_1\bigl(\bar{C}_2^2-\bar{b}_2(\bar{d}_1+\bar{D}_2-\bar{D}_3+2\bar{D}_4)\bigl) -2 \bar{D}_2 \bar{E}_1+ \bar{C}_2(3 \bar{b}_2(\bar{d}_1 +\bar{D}_2-\bar{D}_3+2\bar{D}_4)\nonumber\\
& \hskip 0.7cm +\bar{b}_3(4\bar{d}_1+6 \bar{D}_2-4 \bar{D}_3+8 \bar{D}_4) -4 \bar{F}_2)\bigl)\bigl) \, .
\end{align} 

Similarly, the rest of the degree-nine relations are provided as follows:  %{\color{blue}No algebraic relations appears in this degree? Why there are still extra term below??}  
\begin{align}  
\{\bar{D}_2, \bar{F}_2\}&={\rm i}((-\bar{c}_1+2\bar{C}_2)\bar{F}_1+(\bar{b}_1+4 \bar{b}_3)\bar{G}_1) \, , \nonumber\\
\{\bar{D}_2, \bar{F}_3\}&=\frac{{\rm i}}{8}\bigl(\bar{b}_1\bigl((\bar{b}_2^2+4 \bar{b}_2\bar{b}_3)(\bar{c}_1-\bar{C}_2) -4(2\bar{c}_1-3 \bar{C}_2)\bar{D}_3-4 \bar{b}_2 \bar{E}_1\bigl)+4(2\bar{c}_1-3\bar{C}_2)\bar{C}_2^2-8\bar{D}_2 \bar{E}_1 \nonumber\\
&\hskip 0.9 cm +8 \bar{b}_3 \bar{C}_2(\bar{D}_2+2\bar{D}_3)-2 \bar{b}_2 \bigl(4 \bar{b}_3^2 \bar{C}_2^2+(\bar{c}_1-2\bar{C}_2)(2\bar{d}_1+3\bar{D}_2-2\bar{D}_3)+4(\bar{c}_1-3\bar{C}_2)\bar{D}_4\bigl) \nonumber\\
&\hskip 0.9cm -16 \bar{C}_2 \bar{F}_4 
\bigl) \, , \nonumber\\
\{\bar{D}_2, \bar{F}_4\}&={\rm i}\bigl((-\bar{c}_1+2\bar{C}_2)\bar{F}_3+(\bar{b}_2+4\bar{b}_3)\bar{G}_2\bigl) \, .
\end{align} 

On the other hand, using the grading in \eqref{eq:4.41}, we can determine the allowed polynomials in the Poisson brackets $\{\bar{D}_3,\bar{F}_s\}$ and $\{\bar{D}_4,\bar{F}_s\}$ for all $1 \leq s \leq 4.$ Taking into account the Poisson relations of $\mathfrak{su}^*(4)$ given in \eqref{classicalrels}, we then deduce 
\begin{align}
\{\bar{D}_3, \bar{F}_1\}&=\frac{\rm i}{2}\bigl(((\bar{c}_1-2 \bar{C}_2)\bar{C}_2+\bar{b}_1 \bar{D}_3+\bar{b}_2(\bar{d}_1+\bar{D}_2-\bar{D}_3+2\bar{D}_4)) \bar{C}_2-(\bar{b}_1\bar{b}_2-2 \bar{D}_2)\bar{E}_1\bigl) \, , \nonumber\\
\{\bar{D}_3, \bar{F}_2\}&={\rm i}\bar{C}_2 \bar{F}_1 \, ,\nonumber\\
\{\bar{D}_3, \bar{F}_3\}&=\frac{{\rm i}}{8}\bigl(\bar{b}_1 \bar{b}_2^2(\bar{c}_1-\bar{C}_2)-2\bar{b}_2\bar{c}_1\bar{D}_2-4 \bar{b}_2 \bar{C}_2(2 \bar{b}_3^2-\bar{D}_2-2\bar{D}_4)+16((\bar{b}_3\bar{C}_2+\bar{E}_1)\bar{D}_3-2\bar{C}_2\bar{F}_4)\bigl)\, , \nonumber\\
\{\bar{D}_3, \bar{F}_4\}&={\rm i}\bigl(-4 \bar{I}_1+\bar{b}_2 \bar{G}_2\bigl) \, , \nonumber \\
\{\bar{D}_4, \bar{F}_1\}&=\frac{{\rm i}}{4}\bigl( \bar{C}_2^3+2\bigl(\bar{b}_2(\bar{c}_1-\bar{C}_2)+2\bar{E}_1\bigl)(\bar{d}_1+\bar{D}_2-\bar{D}_3+2\bar{D}_4)-2\bar{b}_3\bar{C}_2\bar{D}_2-\bar{c}_1 \bar{C}_2^2 \nonumber\\
& \hskip 0.9cm-\bar{b}_1(\bar{c}_1-\bar{C}_2)(\bar{b}_2 \bar{b}_3-\bar{D}_3)+4(\bar{D}_2 \bar{E}_1-\bar{C}_2 \bar{F}_2)\bigl) \, ,\nonumber\\
\{\bar{D}_4, \bar{F}_2\}&={\rm i}\bigl(-2(\bar{I}_2+\bar{b}_3 \bar{G}_1)+\frac{1}{2}(\bar{c}_1-\bar{C}_2)\bar{F}_1\bigl) \, ,\nonumber\\
\{\bar{D}_4,\bar{F}_3\}&=\frac{1}{2}(\{\bar{D}_3,\bar{F}_1\}+\{\bar{D}_3,\bar{F}_3\}-\{\bar{D}_2,\bar{F}_3\}) \, ,\nonumber\\
\{\bar{D}_4, \bar{F}_4\}&={\rm i}\bigl(\frac{1}{2}(\bar{c}_1-\bar{C}_2)\bar{F}_3-2\bar{b}_3 \bar{G}_2-2 \bar{I}_1 \bigl)\, .
\end{align}  
\noindent  These computations clarify how the absence of a certain grading affects the admissible polynomials in the Poisson brackets. From now on, our analysis concentrates on presenting data on the allowable generators that arise from homogeneous gradings, selected from one case associated with each compact form. This approach will enable us to examine the differences that arise by transitioning from $C^{(003)}$ and $C^{(202)}$ to $\bar{c}_1$ and $\bar{d}_1$, respectively.

Now, for the grading in the compact form $\{\bar{\textbf{C}},\bar{\textbf{G}}\},$ which includes the Poisson brackets $\{\bar{C}_2,\bar{G}_1\}$ and $\{\bar{C}_2,\bar{G}_2\}.$ A direct computation shows that \begin{align}
    \mathcal{G} \left(\{\bar{C}_2, \bar{G}_1\}\right) =  (2,2,5) \tilde{+} (3,1,5) \tilde{+}(4,2,3) \tilde{+}(4,3,3)  ; \nonumber\\
      \mathcal{G} \left(\{\bar{C}_2, \bar{G}_2\}\right) =  (1,3,5) \tilde{+} (2,2,5) \tilde{+}(3,3,3) \tilde{+}(2,4,3) .
\end{align} Here, we present the allowed polynomials from each homogeneous grading in $\mathcal{G}\left(\{\bar{C}_2, \bar{G}_2\}\right) .$ A direct calculations shows that  \begin{align}
\nonumber
 (1,3,5): =   & \, \left\{\bar{C}_2 \bar{F}_4, \text{ } \bar{D}_3 \bar{E}_1, \text{ } \bar{b}_2 \bar{b}_3^2 \bar{C}_2, \text{ } \bar{b}_2 \bar{b}_3 \bar{E}_1, \text{ } \bar{b}_2 \bar{C}_2 \bar{D}_4, \text{ } \bar{b}_3 \bar{C}_2 \bar{D}_3, \text{ } \bar{b}_2 C^{(003)} \bar{D}_2\right\}; \\
    (2,2,5) : = & \, \left\{\bar{D}_2\bar{E}_1, \text{ } \bar{b}_3\bar{C}_2\bar{D}_2 , \text{ } \bar{D}_2 \bar{E}_1, \text{ } \bar{b}_1 C^{(003)} \bar{D}_3, \text{ } \bar{b}_2 C^{(003)} C^{(202)}, \text{ }  C^{(003)} \bar{C}_2^2, \text{ } \bar{b}_1 \bar{b}_2 \bar{b}_3 C^{(003)}\right\} ; \label{eq:gradingc2g2}\\
    \nonumber
     (3,3,3) : = & \,  \left\{\bar{b}_1\bar{b}_2\bar{E}_1, \text{ }   \bar{b}_1 \bar{C}_2 \bar{D}_3, \text{ } \bar{C}_2^3, \text{ } \bar{b}_1 \bar{b}_2 \bar{b}_3 \bar{C}_2, \text{ } \bar{b}_2 \bar{C}_2 C^{(202)}\right\}; \\
     \nonumber
    (2,4,3) : =  & \,  \left\{\bar{b}_2\bar{C}_2\bar{D}_2 , \text{ } \bar{b}_1\bar{b}_2^2C^{(003)}\right\}. 
   %  (2,2,5) \tilde{+} & \,  (3,1,5) \tilde{+}(4,2,3) \tilde{+}(4,3,3) = \emptyset
\end{align}  Hence, together with the polynomials found in $\eqref{eq:gradingc2g2}$, all the admissible polynomials in terms of the extended terms $\bar{c}_1$ and $\bar{d}_1$ are given by \begin{align*}
    & \, \bar{C}_2 \bar{F}_4, \text{ } \bar{D}_3 \bar{E}_1, \text{ } \bar{b}_2 \bar{b}_3 \bar{E}_1, \text{ } \bar{b}_2 \bar{c}_1 \bar{D}_2, \text{ } \bar{b}_2 \bar{C}_2 \bar{D}_2, \text{ } \bar{b}_2 \bar{C}_2 \bar{D}_4,\text{ } \bar{b}_3 \bar{C}_2 \bar{D}_3, \text{ } \bar{b}_2 \bar{b}_3^2 \bar{C}_2, \text{ } \bar{b}_3 \bar{C}_2 \bar{D}_2, \\ 
    & \,  \bar{D}_2 \bar{E}_1, \text{ } \bar{b}_1 \bar{c}_1 \bar{D}_3, \text{ } \bar{b}_1 \bar{C}_2 \bar{D}_3, \text{ } \bar{b}_2 \bar{c}_1 \bar{d}_1, \text{ } \bar{b}_2 \bar{c}_1 \bar{D}_3, \text{ }\bar{b}_2 \bar{c}_1 \bar{D}_4,  \text{ } \bar{b}_2 \bar{C}_2 \bar{d}_1, \text{ } \bar{b}_2 \bar{C}_2 \bar{D}_3, \\
    & \, \bar{c}_1 \bar{C}_2^2,\text{ } \bar{C}_2^3, \text{ } \bar{b}_1 \bar{b}_2 \bar{b}_3 \bar{c}_1, \text{ } \bar{b}_1 \bar{b}_2 \bar{b}_3 \bar{C}_2, \text{ } \bar{b}_1 \bar{b}_2 \bar{E}_1, \text{ } \bar{b}_1 \bar{b}_2^2 \bar{c}_1, \text{ } \bar{b}_1 \bar{b}_2^2 \bar{C}_2.
\end{align*} There are $24$ coefficients such that $\{\bar{C}_2,\bar{G}_2\} = \Gamma_{37}^1 \bar{C}_2 \bar{F}_4 +\ldots + \Gamma_{37}^{24}\bar{b}_1 \bar{b}_2^2 \bar{C}_2.$ Moreover, from $\eqref{deg9b}$, the compact forms $\{\bar{\textbf{C}},\bar{\textbf{G}}\}$ and $\{\bar{\textbf{D}},\bar{\textbf{F}}\}$ have the same expansion. The grading method indicates that the permissible terms can be reduced from $65$ to $24.$ Together with the information above, after direct inspection we can further conclude that $\{\bar{C}_2, \bar{G}_2\} =\{\bar{D}_4, \bar{F}_3\} \, . $  Similarly, we also deduce that
\begin{align}
 \{\bar{C}_2, \bar{G}_1\} =\{\bar{D}_4, \bar{F}_1\} \, .
\end{align}  This concludes all the expansion in degree nine Poisson brackets.
\vskip 0.5cm

\subsection{Expansions in the degree $10$ brackets}

 For the case of brackets having degree $10$, the reduced forms can be grouped into three distinct compact forms: $\{\bar{\textbf{C}},\bar{\textbf{H}}\},$ $ \{\bar{\textbf{D}},\bar{\textbf{G}}\},$ and $\{\bar{\textbf{E}},\bar{\textbf{F}}\}$.   We first examine the compact form of type $\{\bar{\textbf{C}},\bar{\textbf{H}}\}$, which includes the Poisson brackets $\{\bar{C}_2,\bar{H}_1\}$ and $\{\bar{C}_2,\bar{H}_2\}.$ The gradations assigned to each Poisson brackets are as follows
 \begin{align}
    \mathcal{G} \left(\{\bar{C}_2, \bar{H}_1\}\right) = (2,2,6) \tilde{+} (3,1,6) \tilde{+} (4,2,4) \tilde{+}(3,3,4); \nonumber \\
     \mathcal{G} \left(\{\bar{C}_2, \bar{H}_2\}\right) =  {(1,3,6)}\tilde{+} (2,2,6)  \tilde{+}(3,3,4)\tilde{+} (2,4,4) ; 
\end{align} As an example, the allowed monomials in each homogeneous grading for the second bracket are \begin{align*}
    (1,3,6) := &\, \{\bar{D}_2 \bar{F}_4, \text{ } \bar{b}_2 \bar{D}_2 \bar{D}_4, \text{ } \bar{b}_3 \bar{D}_2 \bar{D}_3, \text{ } \bar{b}_2 \bar{b}_3^2 \bar{D}_2, \text{ }  \bar{b}_2 C^{(003)} \bar{E}_1, \text{ } \bar{b}_2 \bar{b}_3 C^{(003)} \bar{C}_2, \text{ } C^{(003)} \bar{C}_2 \bar{D}_3\}; \\ 
    (2,2,6)  : = & \,\{\bar{D}_3 \bar{F}_2, \text{ } \bar{E}_1^2, \text{ } \bar{b}_1 \bar{b}_3 \bar{F}_4, \text{ } \bar{b}_1 \bar{D}_3 \bar{D}_4, \text{ } \bar{b}_2 \bar{b}_3 \bar{F}_2, \text{ } \bar{b}_3 \bar{C}_2 \bar{E}_1, \text{ } \bar{b}_3 \bar{D}_2^2, \text{ } \bar{b}_1 \bar{b}_2 \bar{b}_3^3, \text{ }  \bar{C}_2^2 \bar{D}_4, \text{ } \bar{b}_1 \bar{b}_2 \bar{b}_3 \bar{D}_4, \text{ } \bar{b}_3^2 \bar{C}_2^2, \text{ } \bar{b}_1 \bar{b}_3^2 \bar{D}_3,\\
    & \, \bar{b}_2 C^{(202)} \bar{D}_4, \text{ } C^{(202)} \bar{F}_4, \text{ } \bar{b}_3 C^{(202)} \bar{D}_3, \text{ } C^{(003)} \bar{C}_2 \bar{D}_2, \text{ } \bar{b}_1 \bar{b}_2 \left(C^{(003)}\right)^2, \text{ } \bar{b}_2 \bar{b}_3^2 C^{(202)}\} ;\\ 
    (3,3,4) := &\, \{\bar{b}_1\bar{b}_2\bar{b}_3\bar{D}_2, \text{ } \bar{b}_1 \bar{D}_2 \bar{D}_3, \text{ } \bar{C}_2^2 \bar{D}_2, \text{ }  \bar{b}_2 C^{(202)} \bar{D}_2, \text{ } \bar{b}_1 \bar{b}_2 C^{(003)} \bar{C}_2\} ;\\
    (2,4,4) := &\, \{\bar{b}_1 \bar{b}_2 \bar{F}_4, \text{ } \bar{b}_1 \bar{D}_3^2, \text{ } \bar{b}_2^2 \bar{F}_2, \text{ } \bar{b}_2 \bar{C}_2 \bar{E}_1, \text{ } \bar{b}_1 \bar{b}_2^2 \bar{b}_3^2, \text{ } \bar{b}_2 \bar{D}_2^2, \text{ } \bar{C}_2^2 \bar{D}_3, \text{ } \bar{b}_1 \bar{b}_2^2 \bar{D}_4, \text{ } \bar{b}_1 \bar{b}_2 \bar{b}_3 \bar{D}_3, \text{ } \bar{b}_2 \bar{b}_3 \bar{C}_2^2, \text{ } \bar{b}_2 C^{(202)} \bar{D}_3, \text{ } \bar{b}_2^2 \bar{b}_3 C^{(202)}\}  .   
\end{align*} 
Again, the allowed polynomials underline the corresponding one with the homogeneous gradings, and we deduce that all the allowed polynomials in the expansion $\{\bar{C}_2,\bar{H}_2\}$ are \begin{align*}
    & \, \bar{D}_2 \bar{F}_4,\text{ } \bar{b}_2 \bar{c}_1 \bar{E}_1,\text{ }\bar{b}_2 \bar{C}_2 \bar{E}_1,\text{ }\bar{b}_2 \bar{D}_2 \bar{D}_4,\bar{b}_3 \bar{D}_2 \bar{D}_3,\text{ }\bar{c}_1 \bar{C}_2 \bar{D}_3,\text{ }\bar{C}_2^2 \bar{D}_3,\text{ }\bar{b}_2 \bar{b}_3^2 \bar{D}_2,\text{ }\bar{b}_2 \bar{b}_3 \bar{c}_1 \bar{C}_2,\text{ }\bar{b}_1 \bar{b}_2 \bar{c}_1^2,  \text{ }\bar{b}_2 \bar{D}_3^2,\\
    & \, \bar{b}_2 \bar{b}_3 \bar{C}_2^2,\text{ }\bar{d}_1 \bar{F}_4,\text{ }\bar{D}_3 \bar{F}_4,\text{ }\bar{D}_4 \bar{F}_4,\text{ }\bar{D}_3 \bar{F}_2,\text{ }\bar{E}_1^2,\bar{b}_1 \bar{b}_3 \bar{F}_4,\text{ }\bar{b}_1 \bar{D}_3 \bar{D}_4,\text{ }\bar{b}_2 \bar{b}_3 \bar{F}_2,\text{ }\bar{b}_2 \bar{d}_1 \bar{D}_4,\text{ }\bar{b}_2 \bar{D}_3 \bar{D}_4, \\
    & \, \bar{b}_2 \bar{D}_4^2,\text{ }\bar{b}_3 \bar{C}_2 \bar{E}_1,\text{ }\bar{b}_3 \bar{d}_1 \bar{D}_3,\text{ }\bar{b}_3 \bar{D}_3^2, \text{ } \bar{b}_3 \bar{D}_3 \bar{D}_4,\text{ }\bar{b}_3 \bar{D}_2^2,\text{ }\bar{c}_1 \bar{C}_2 \bar{D}_2,\text{ }\bar{C}_2^2 \bar{D}_2,\text{ }\bar{C}_2^2 \bar{D}_4,\text{ }\bar{b}_1 \bar{b}_2 \bar{b}_3 \bar{D}_4, \\
    & \,  \bar{b}_1 \bar{b}_2 \bar{c}_1 \bar{C}_2,\bar{b}_1 \bar{b}_2 \bar{C}_2^2,\text{ }\bar{b}_1 \bar{b}_3^2 \bar{D}_3,\text{ }\bar{b}_2 \bar{b}_3^2 \bar{d}_1,\text{ }\bar{b}_2 \bar{b}_3^2 \bar{D}_3,\text{ }\bar{b}_2 \bar{b}_3^2 \bar{D}_4,\text{ }\text{ }\bar{b}_3^2 \bar{C}_2^2,\text{ }\bar{b}_1 \bar{b}_2 \bar{b}_3^3,\text{ }\bar{b}_1 \bar{D}_2 \bar{D}_3,\\
    & \, \bar{b}_2 \bar{d}_1 \bar{D}_2,\text{ }\bar{b}_2 \bar{D}_2^2,\bar{b}_2 \bar{D}_2 \bar{D}_3,\text{ }\bar{b}_1 \bar{b}_2 \bar{b}_3 \bar{D}_2,\text{ }\bar{b}_1 \bar{b}_2 \bar{c}_1 \bar{C}_2,\text{ }\bar{b}_1 \bar{b}_2 \bar{F}_4,\text{ }\bar{b}_1 \bar{D}_3^2,\text{ }\bar{b}_2^2 \bar{F}_2,\text{ }\bar{b}_2 \bar{d}_1 \bar{D}_3,\\
    & \, \bar{b}_1 \bar{b}_2^2 \bar{D}_4,\text{ }\bar{b}_1 \bar{b}_2 \bar{b}_3 \bar{D}_3,\text{ }\bar{b}_2^2 \bar{b}_3 \bar{d}_1,\text{ }\bar{b}_2^2 \bar{b}_3 \bar{D}_2,\text{ }\bar{b}_2^2 \bar{b}_3 \bar{D}_3,\text{ }\bar{b}_2^2 \bar{b}_3 \bar{D}_4,\text{ }\bar{b}_1 \bar{b}_2^2 \bar{b}_3^2.
\end{align*} 
The maximal possible count of coefficients denoted as $\Gamma_{38}^1,\ldots,\Gamma_{38}^{57}$ reaches $57$, wherein the expression $\{\bar{C}_2,\bar{H}_2\}$ can be represented as $\Gamma_{38}^1\bar{D}_2 \bar{F}_4 + \ldots + \Gamma_{38}^{57} \bar{b}_1 \bar{b}_2^2 \bar{b}_3^2.$ By meticulously evaluating and identifying each of these coefficients, we derive
\begin{align*}
    \{\bar{C}_2, \bar{H}_2\}&=-\frac{{\rm i}}{32}\bigl(-24 \bar{b}_3^2 \bar{C}_2^2-4 \bar{b}_3(\bar{D}_2^2-16 \bar{C}_2 \bar{E}_1+4 \bar{b}_2 \bar{F}_2) +8\bigl( \bar{C}_2^2(\bar{D}_2+\bar{D}_4)+\bar{b}_2\bar{C}_2 \bar{E}_1 \nonumber \\
&\hskip 1.25cm+2\bar{b}_2 \bar{D}_4(\bar{d}_1+\bar{D}_2-\bar{D}_3+2\bar{D}_4)-2 \bar{E}_1^2-\bar{c}_1(\bar{C}_2 \bar{D}_2+\bar{b}_2 \bar{E}_1)+2\bar{D}_3 \bar{F}_2 \nonumber\\
&\hskip 1.25cm -2(3\bar{d}_1+4 \bar{D}_2-3 \bar{D}_3+6 \bar{D}_4)\bar{F}_4\bigl)+\bar{b}_1\bigl(\bar{b}_2(8 \bar{b}_3^3-(\bar{c}_1-\bar{C}_2)^2-8 \bar{b}_3 \bar{D}_4) \nonumber\\
&\hskip 1.25cm + 16 \bar{b}_3(\bar{F}_4-\bar{b}_3 \bar{D}_3) \bigl)\bigl) \, .
\end{align*} 
With a similar approach, the other bracket relation in degree-ten expansions of the compact form $\{\bar{\textbf{C}},\bar{\textbf{H}}\}$ reads explicitly: 
\begin{align}
\{\bar{C}_2, \bar{H}_1\}&=\frac{{\rm i}}{16}\bigl(2 \bar{C}_2\bigl((12 \bar{b}_3^2-3\bar{D}_2-4 \bar{D}_4)\bar{C}_2+3 \bar{c}_1 \bar{D}_2-24 \bar{b}_3 \bar{E}_1\bigl)+8(\bar{D}_2+4 \bar{D}_3)\bar{F}_2 \nonumber\\
& \hskip 0.9cm +\bar{b}_1\bigl(\bar{b}_2(\bar{c}_1-\bar{C}_2)^2-8 \bar{D}_3(\bar{b}_3^2+\bar{D}_4)+4 (\bar{c}_1 - \bar{C}_2)\bar{E}_1+16 \bar{b}_3 \bar{F}_4\bigl) \bigl) \, .
\end{align}

Now, in the compact form $ \{\bar{\textbf{D}},\bar{\textbf{G}}\}$, since $d_1$ is the central element, we will consider the Poisson brackets $\{\bar{D}_2, \bar{G}_1\}, \ldots, \{\bar{D}_4, \bar{G}_2\}$. The gradings in each of these Poisson brackets are given by \begin{align}
        \mathcal{G} \left(\{\bar{D}_2, \bar{G}_1\}\right) = (2,2,6) \tilde{+} (3,1,6) \tilde{+} (4,2,4) \tilde{+}(3,3,4); \nonumber \\
       \mathcal{G} \left(\{\bar{D}_2, \bar{G}_2\}\right) = (1,3,6)  \tilde{+} (2,2,6)\tilde{+} (3,3,4) \tilde{+}(4,2,4); \nonumber\\
      \mathcal{G} \left(\{\bar{D}_3, \bar{G}_1\}\right) = (1,3,6) \tilde{+} (2,2,6) \tilde{+} (3,3,4) \tilde{+}(4,2,4); \nonumber \\
     \mathcal{G} \left(\{\bar{D}_3, \bar{G}_2\}\right) = (0,4,6) \tilde{+} (1,3,6) \tilde{+} (2 ,4,4) \tilde{+}(1,5,4); \nonumber \\
     \mathcal{G} \left(\{\bar{D}_4, \bar{G}_1\}\right) =  (1,1,8)\tilde{+}(2,0,8)  \tilde{+} (3,1,6) \tilde{+}(2,2,6); \nonumber \\
     \mathcal{G} \left(\{\bar{D}_4, \bar{G}_2\}\right) = (0,2,8) \tilde{+} (1,1,8) \tilde{+} (2,2,6) \tilde{+}(1,3,6).
\end{align}  
We take as a guiding example the bracket $\{\bar{D}_2, \bar{G}_1\}$. The predicted polynomials for each homogeneous grading are as follows \begin{align*}
    (2,2,6) := &\, \left\{\bar{D}_3 \bar{F}_2, \text{ } \bar{E}_1^2, \text{ } \bar{b}_1 \bar{b}_3 \bar{F}_4, \text{ } \bar{b}_1 \bar{D}_3 \bar{D}_4, \text{ } \bar{b}_2 \bar{b}_3 \bar{F}_2, \text{ } \bar{b}_3 \bar{C}_2 \bar{E}_1\bar{b}_3 \bar{D}_2^2, \text{ } \bar{C}_2^2 \bar{D}_4, \text{ } \bar{b}_1 \bar{b}_2 \bar{b}_3 \bar{D}_4, \text{ } \bar{b}_3^2 \bar{C}_2^2, \text{ } \bar{b}_1 \bar{b}_2 \bar{b}_3^3, \text{ } \bar{b}_1 \bar{b}_3^2 \bar{D}_3,\right. \\
    &\, \left.\bar{b}_1 \bar{b}_2 \left(C^{(003)}\right)^2, \text{ } \bar{b}_2 \bar{b}_3^2 C^{(202)}, \text{ } C^{(003)} \bar{C}_2 \bar{D}_2, \text{ } C^{(202)} \bar{F}_4, \text{ }  \bar{b}_2 C^{(202)} \bar{D}_4, \text{ } \bar{b}_3 C^{(202)} \bar{D}_3 \right\} ; \\
    (3,1,6) := &\, \left\{\bar{D}_2 \bar{F}_2, \text{ } \bar{b}_1 \bar{D}_2 \bar{D}_4, \text{ } \bar{b}_1 \bar{b}_3^2 \bar{D}_2, \text{ } \bar{b}_1 C^{(003)} \bar{E}_1, \text{ } C^{(003)} \bar{C}_2 C^{(202)}, \text{ } \bar{b}_1 \bar{b}_3 C^{(003)} \bar{C}_2, \text{ } \bar{b}_3 C^{(202)} \bar{D}_2\right\}; \\
    (4,2,4) := &\, \left\{\bar{b}_1^2 \bar{F}_4, \text{ } \bar{b}_1 \bar{b}_2 \bar{F}_2, \text{ } \bar{b}_1 \bar{C}_2 \bar{E}_1, \text{ } \bar{b}_1 \bar{D}_2^2, \text{ } \bar{b}_1^2 \bar{b}_2 \bar{D}_4, \text{ } \bar{b}_1^2 \bar{b}_3 \bar{D}_3, \text{ } \bar{b}_1 \bar{b}_3 \bar{C}_2^2, \text{ } \bar{b}_1^2 \bar{b}_2 \bar{b}_3^2, \text{ } \bar{C}_2^2 C^{(202)}, \text{ } \bar{b}_1 \bar{b}_2 \bar{b}_3 C^{(202)}, \right. \\
  & \,    \left. \bar{b}_2 \left(C^{(202)}\right)^2, \text{ } \bar{b}_1 C^{(202)} \bar{D}_3\right\} ; \\
    (3,3,4) := &\, \left\{\bar{b}_1 \bar{D}_2 \bar{D}_3,\bar{C}_2^2 \bar{D}_2,\bar{b}_1 \bar{b}_2 \bar{b}_3 \bar{D}_2,\bar{b}_1 \bar{b}_2 C^{(003)} \bar{C}_2,\bar{b}_2 C^{(202)} \bar{D}_2\right\}.
\end{align*} The change of generators, i.e. $C^{(003)}$ to $\bar{c}_1$ and $C^{(202)}$ to $\bar{d}_1,$ make all the allowed polynomials in the Poisson bracket $\{\bar{D}_2,\bar{G}_1\}$ be extended  to \begin{align*}
   &\, \bar{d}_1 \bar{F}_4, \text{ }\bar{D}_2 \bar{F}_4,\text{ }\bar{D}_3 \bar{F}_4,\text{ }\bar{D}_4 \bar{F}_4,\text{ }\bar{D}_3 \bar{F}_2,\text{ }\bar{E}_1^2,\bar{b}_1 \bar{b}_3 \bar{F}_4,\text{ }\bar{b}_1 \bar{D}_3 \bar{D}_4,\text{ }\bar{b}_2 \bar{b}_3 \bar{F}_2,\text{ }\bar{b}_2 \bar{d}_1 \bar{D}_4,\text{ }\bar{b}_2 \bar{D}_2 \bar{D}_4,\text{ }\bar{b}_2 \bar{D}_3 \bar{D}_4,\text{ }\bar{b}_2 \bar{D}_4^2,\bar{b}_3 \bar{C}_2 \bar{E}_1,\text{ }\bar{b}_3 \bar{d}_1 \bar{D}_3,\\
    &\, \bar{b}_3 \bar{D}_2 \bar{D}_3,\text{ }\bar{b}_3 \bar{D}_3^2,\text{ }\bar{b}_3 \bar{D}_3 \bar{D}_4,\text{ }\bar{b}_3 \bar{D}_2^2,\text{ }\bar{c}_1 \bar{C}_2 \bar{D}_2,\text{ }\bar{C}_2^2 \bar{D}_2,\text{ }\bar{C}_2^2 \bar{D}_4,\bar{b}_1 \bar{b}_2 \bar{b}_3 \bar{D}_4,\text{ }\bar{b}_1 \bar{b}_2 \bar{c}_1^2,\text{ } \bar{b}_1 \bar{b}_2 \bar{c}_1 \bar{C}_2,\text{ }\bar{b}_1 \bar{b}_2 \bar{C}_2^2,\text{ }\bar{b}_1 \bar{b}_3^2 \bar{D}_3,\text{ }\bar{b}_2 \bar{b}_3^2 \bar{d}_1,\\
    &\, \bar{b}_2 \bar{b}_3^2 \bar{D}_2,\text{ }\bar{b}_2 \bar{b}_3^2 \bar{D}_3,\text{ }\bar{b}_2 \bar{b}_3^2 \bar{D}_4,\text{ }\bar{b}_3^2 \bar{C}_2^2,\text{ }\bar{b}_1 \bar{b}_2 \bar{b}_3^3,\text{ }\bar{D}_2 \bar{F}_2,\text{ }\bar{b}_1 \bar{c}_1 \bar{E}_1,\text{ }\text{ }\bar{b}_1 \bar{C}_2 \bar{E}_1,\text{ }\bar{b}_1 \bar{D}_2 \bar{D}_4, \text{ }\bar{b}_3 \bar{d}_1 \bar{D}_2,\text{ }\bar{b}_3 \bar{D}_2 \bar{D}_4,\text{ }\bar{c}_1 \bar{C}_2 \bar{d}_1,\\
    &\,\text{ }\bar{c}_1 \bar{C}_2 \bar{D}_3,\text{ }\bar{c}_1 \bar{C}_2 \bar{D}_4,\bar{C}_2^2 \bar{d}_1,\text{ }\bar{C}_2^2 \bar{D}_3,\text{ }\bar{b}_1 \bar{b}_3^2 \bar{D}_2, \text{ } \bar{b}_1 \bar{b}_3 \bar{c}_1 \bar{C}_2,\bar{b}_1 \bar{b}_3 \bar{C}_2^2,\text{ }\bar{b}_1^2 \bar{F}_4,\bar{b}_1 \bar{b}_2 \bar{F}_2,\text{ }\bar{b}_1 \bar{d}_1 \bar{D}_3,\text{ }\bar{b}_2 \bar{D}_3^2,\bar{b}_2 \bar{D}_3 \bar{D}_4,\\
    &\, \bar{b}_1 \bar{D}_2 \bar{D}_3,\text{ }\bar{b}_1 \bar{D}_3^2,\text{ }\bar{b}_1 \bar{D}_2^2,\text{ }\bar{b}_2 \bar{d}_1^2,\text{ } \bar{b}_2 \bar{d}_1 \bar{D}_2,\text{ } \bar{b}_2 \bar{d}_1 \bar{D}_3, \text{ }  \bar{b}_2 \bar{d}_1 \bar{D}_4,\text{ }\bar{b}_2 \bar{D}_2^2,2 \bar{b}_2 \bar{D}_2 \bar{D}_3,\text{ } \bar{b}_2 \bar{D}_2 \bar{D}_4,\text{ }\bar{b}_1^2 \bar{b}_2 \bar{D}_4 \\
    &\,\bar{b}_1^2 \bar{b}_3 \bar{D}_3,\text{ }\bar{b}_1 \bar{b}_2 \bar{b}_3 \bar{d}_1,\text{ }\bar{b}_1 \bar{b}_2 \bar{b}_3 \bar{D}_2,\text{ }\bar{b}_1 \bar{b}_2 \bar{b}_3 \bar{D}_3,\text{ }\bar{b}_1^2 \bar{b}_2 \bar{b}_3^2,\text{ }\bar{b}_2 \bar{d}_1 \bar{D}_2,\text{ }\bar{b}_2 \bar{D}_2 \bar{D}_3,\text{ }\bar{b}_1 \bar{b}_2 \bar{c}_1 \bar{C}_2.
\end{align*} 
In this case, the grading method reduces the number of the admissible terms from $169$ to just $71$. Hence, there are $71$ distinct polynomials that could appear in the expansion of the Poisson bracket $\{\bar{D}_2,\bar{G}_1\}$. Then taking into  account the Poisson relations of $\mathfrak{su}^*(4)$, we deduce that 
\begin{align}
    \{\bar{D}_2, \bar{G}_1\}&=\frac{{\rm i}}{8}\bigl(2 \bar{b}_1^2(2 \bar{F}_4-2 \bar{b}_3 \bar{D}_3+\bar{b}_2(\bar{b}_3^2-\bar{D}_4))+2(\bar{c}_1 \bar{C}_2-2 \bar{C}_2^2)(2\bar{d}_1+\bar{D}_2-2 \bar{D}_3)-8 \bar{b}_3^2 \bar{C}_2^2 \nonumber\\ & \hskip 0.7cm +8(\bar{c}_1 \bar{C}_2-3 \bar{C}_2^2)\bar{D}_4+32 \bar{b}_3 \bar{C}_2 \bar{E}_1-16 \bar{E}_1^2+16 \bar{D}_3 \bar{F}_2+8 \bar{b}_2^2(\bar{F}_2-\bar{b}_3(\bar{d}_1+\bar{D}_2-\bar{D}_3+2 \bar{D}_4) \nonumber \\
& \hskip 0.7cm +2 \bar{b}_2(2\bar{d}_1^2+\bar{D}_2^2+16  \bar{D}_2 \bar{D}_4+4 \bar{d}_1(\bar{D}_2+4 \bar{D}_4)-2(\bar{D}_3-2\bar{D}_4)(\bar{D}_3+6 \bar{D}_4)-8 \bar{C}_2 \bar{E}_1 \nonumber \\
& \hskip 0.7cm +4 \bar{b}_3(\bar{C}_2^2-2\bar{F}_2)) -16 (\bar{d}_1+\bar{D}_2-\bar{D}_3+2\bar{D}_4)\bar{F}_4 +\bar{b}_1\bigl(-\bar{D}_2^2+4 \bar{b}_2^2(\bar{b}_3^2-\bar{D}_4)\nonumber\\
&\hskip 0.7cm -4 (\bar{c}_1-2\bar{C}_2)\bar{E}_1+\bar{b}_2(\bar{c}_1 \bar{C}_2-\bar{C}_2^2-8\bar{b}_3 \bar{D}_3-4 \bar{F}_2+8\bar{F}_4)\bigl) \bigl) \,  .
\end{align} 
Eventually, with the same process above, the rest of the additional degree-ten bracket relations are obtained: 
\begin{align}
\{\bar{D}_2, \bar{G}_2\}&=-\frac{{\rm i}}{8}\bigl(8 \bar{b}_3^2 \bar{C}_2^2+2 (\bar{c}_1 \bar{C}_2-2\bar{C}_2^2)(\bar{D}_2-2\bar{D}_3)+\bar{b}_2 \bar{D}_2^2 +8 \bar{C}_2^2 \bar{D}_4+4(\bar{b}_2\bar{c}_1-2\bar{b}_2\bar{C}_2)\bar{E}_1+16 \bar{E}_1^2\nonumber\\
&\hskip 0.7cm +4 \bar{b}_3(\bar{b}_2^2\bigl(\bar{d}_1+\bar{D}_2-\bar{D}_3+2\bar{D}_4)-8 \bar{C}_2 \bar{E}_1\bigl)-4 (\bar{b}_2^2-4\bar{D}_3)\bar{F}_2-16(\bar{d}_1+\bar{D}_2-\bar{D}_3+2\bar{D}_4)\bar{F}_4\nonumber\\
& \hskip 0.7cm +\bar{b}_1\bigl(2 \bar{b}_2^2(\bar{D}_4-\bar{b}_3^2)-4 \bar{D}_3(\bar{D}_3+4\bar{D}_4)+4(\bar{b}_2+4\bar{b}_3)\bar{F}_4+\bar{b}_2(\bar{C}_2^2-\bar{c}_1 \bar{C}_2) \bigl)\bigl) \, , \nonumber\\
\{\bar{D}_3, \bar{G}_1\}&=\frac{{\rm i}}{8}\bigl(\bar{C}_2\bigl(\bar{b}_1 \bar{b}_2(\bar{c}_1-\bar{C}_2)+8 \bar{b}_3^2 \bar{C}_2+2\bar{c}_1\bar{D}_2-8\bar{C}_2 \bar{D}_4\bigl)+16 \bar{E}_1^2-16\bar{b}_3 \bar{C}_2 \bar{E}_1\bigl) \, , \nonumber\\
\{\bar{D}_3, \bar{G}_2\}&=\frac{{\rm i}}{8}\bigl(2 \bar{b}_1\bigl(2 \bar{D}_3^2+\bar{b}_2^2(\bar{b}_3^2-\bar{D}_4)-2\bar{b}_2\bar{F}_4\bigl)-\bar{b}_2\bigl(4 \bar{b}_3(\bar{c}_1-\bar{C}_2)\bar{C}_2+\bar{D}_2^2-4\bar{c}_1 \bar{E}_1 \bigl) \nonumber\\
& \hskip 0.7cm +4 \bar{b}_2^2\bigl(\bar{F}_2-\bar{b}_3(\bar{d}_1+\bar{D}_2-\bar{D}_3+2\bar{D}_4)\bigl)+4 \bigl( (\bar{c}_1 \bar{C}_2 -2\bar{C}_2^2-2\bar{b}_3\bar{D}_2)\bar{D}_3+4\bar{D}_2\bar{F}_4\bigl)\bigl) \, , \nonumber\\
\{\bar{D}_4, \bar{G}_1\}&=\frac{{\rm i}}{4}\bigl(4 \bar{b}_3^2\bar{C}_2^2+(\bar{c}_1 \bar{C}_2-\bar{C}_2^2)\bar{D}_2-4(\bar{b}_2 \bar{d}_1+\bar{b}_2 \bar{D}_2-\bar{b}_2 \bar{D}_3)\bar{D}_4-8 \bar{b}_2 \bar{D}_4^2 +8 \bar{E}_1^2+4(\bar{D}_2-\bar{D}_3)\bar{F}_2 \nonumber\\
&\hskip 0.7cm -\bar{b}_1(\bar{c}_1-\bar{C}_2)(\bar{b}_3\bar{C}_2-2\bar{E}_1) -2 \bar{b}_3\bigl(\bar{D}_2(\bar{d}_1+\bar{D}_2-\bar{D}_3+2\bar{D}_4)+6\bar{C}_2 \bar{E}_1-2\bar{b}_2 \bar{F}_2 \bigl) \nonumber\\
&\hskip 0.7cm +4(\bar{d}_1+\bar{D}_2-\bar{D}_3+2\bar{D}_4)\bar{F}_4\bigl) \, , \nonumber\\
\{\bar{D}_4, \bar{G}_2\}&= \frac{1}{2}(\{\bar{D}_3,\bar{G}_1\}+\{\bar{D}_3,\bar{G}_2\}-\{\bar{D}_2, \bar{G}_2\})\, .%,\nonumber \\
\end{align} 
As last case for the degree-ten expansion, we consider the grading in the compact form $\{\bar{\textbf{E}} ,\bar{\textbf{F}}\}$ for which we have  \begin{align}
    \mathcal{G} (\{\bar{E}_1, \bar{F}_1\}) = & \, (2,2,6) \tilde{+} (3,1,6) \tilde{+}   (3,3,4) \tilde{+} (4,2,4) ; \nonumber \\
     \mathcal{G} (\{\bar{E}_1, \bar{F}_2\}) = & \, (2,1,7) \tilde{+} (3,0,7) \tilde{+} (4,1,5) \tilde{+} (3,2,5); \nonumber \\
      \mathcal{G} (\{\bar{E}_1, \bar{F}_3\}) = & \, (1,3,6) \tilde{+} (2,2,6) \tilde{+} (3,3,4) \tilde{+} (2,4,4); \nonumber  \\
       \mathcal{G} (\{\bar{E}_1, \bar{F}_4\}) = & \, (0,3,7) \tilde{+} (1,2,7) \tilde{+} (2,3,5) \tilde{+} (1,4,5).  
\end{align} 
 We consider the derivation of permissible monomials in the context of the Poisson bracket $\{\bar{E}_1,\bar{F}_1\}$, the remaining cases being analogous. A straightforward but routine calculation reveals that the admissible generators associated with each of the homogeneous gradings are given by 
\begin{align*}
% (1,3,6) :=&\, \left\{\bar{D}_2 \bar{F}_4,\bar{b}_2 \bar{c}_1 \bar{E}_1,\bar{b}_2 \bar{D}_2 \bar{D}_4,\bar{b}_3 \bar{D}_2 \bar{D}_3,\bar{c}_1 \bar{C}_2 \bar{D}_3,\bar{b}_2 \bar{b}_3^2 \bar{D}_2,\bar{b}_2 \bar{b}_3 \bar{c}_1 \bar{C}_2\right\} \\ 
    (2,2,6) :=&\, \left\{ \bar{d}_1 \bar{F}_4,\bar{D}_3 \bar{F}_2,\bar{E}_1^2, \text{ } \bar{b}_1 \bar{b}_3 \bar{F}_4, \text{ } \bar{b}_1 \bar{D}_3 \bar{D}_4, \text{ } \bar{b}_2 \bar{b}_3 \bar{F}_2, \text{ } \bar{b}_3 \bar{C}_2 \bar{E}_1\bar{b}_3 \bar{D}_2^2, \text{ } \bar{C}_2^2 \bar{D}_4, \text{ } \bar{b}_1 \bar{b}_2 \bar{b}_3 \bar{D}_4, \text{ } \bar{b}_3^2 \bar{C}_2^2, \text{ } \bar{b}_1 \bar{b}_2 \bar{b}_3^3, \text{ } \bar{b}_1 \bar{b}_3^2 \bar{D}_3,\right. \\
    &\, \left.\bar{b}_1 \bar{b}_2 \left(C^{(003)}\right)^2, \text{ } \bar{b}_2 \bar{b}_3^2 C^{(202)}, \text{ } C^{(003)} \bar{C}_2 \bar{D}_2, \text{ } C^{(202)} \bar{F}_4, \text{ }  \bar{b}_2 C^{(202)} \bar{D}_4, \text{ } \bar{b}_3 C^{(202)} \bar{D}_3 \right\} ; \\
    (3,1,6) := &\, \left\{\bar{D}_2 \bar{F}_2, \text{ } \bar{b}_1 \bar{D}_2 \bar{D}_4, \text{ } \bar{b}_1 \bar{b}_3^2 \bar{D}_2, \text{ } \bar{b}_1 C^{(003)} \bar{E}_1, \text{ } C^{(003)} \bar{C}_2 C^{(202)}, \text{ } \bar{b}_1 \bar{b}_3 C^{(003)} \bar{C}_2, \text{ } \bar{b}_3 C^{(202)} \bar{D}_2\right\} \\   
    (3,3,4) :=&\, \left\{\bar{b}_1 \bar{D}_2 \bar{D}_3, \text{ } \bar{C}_2^2 \bar{D}_2, \text{ } \bar{b}_1 \bar{b}_2 \bar{b}_3 \bar{D}_2, \text{ } \bar{b}_1 \bar{b}_2 C^{(003)} \bar{C}_2, \text{ } \bar{b}_2 C^{(202)} \bar{D}_2\right\}; \\ 
      (4,2,4) :=&\, \left\{\bar{b}_1^2 \bar{F}_4, \text{ } \bar{b}_1 \bar{b}_2 \bar{F}_2, \text{ } \bar{b}_1 \bar{C}_2 \bar{E}_1, \text{ } \bar{b}_1 \bar{d}_1 \bar{D}_3, \text{ } \bar{b}_1 \bar{D}_2^2, \text{ } \bar{b}_2 \bar{d}_1^2, \text{ } \bar{C}_2^2 \bar{d}_1, \text{ } \bar{b}_1^2 \bar{b}_2 \bar{D}_4,\bar{b}_1^2 \bar{b}_3 \bar{D}_3, \text{ } \bar{b}_1 \bar{b}_2 \bar{b}_3 \bar{d}_1, \text{ } \bar{b}_1 \bar{b}_3 \bar{C}_2^2, \text{ } \bar{b}_1^2 \bar{b}_2 \bar{b}_3^2\right\} .
\end{align*} 
 After considering the change of generators, there are $71$ predicted generators such that $\{\bar{E}_1,\bar{F}_1\} = \Gamma_{56}^1 \bar{d}_1\bar{F}_4 +\ldots +\Gamma_{56}^{71} \bar{b}_1^2\bar{b}_2 \bar{b}_3^2.$ Here $\Gamma_{56}^1,\ldots,\Gamma_{56}^{71}$ are constant coefficients that,  once determined, lead us to the following result:  \begin{align}
 \{\bar{E}_1, \bar{F}_1\}&=\frac{{\rm i}}{16}\bigl(16 \bar{b}_3^2 \bar{C}_2^2+4(\bar{c}_1 \bar{C}_2-2\bar{C}_2^2)\bar{d}_1+2(4 \bar{c}_1 \bar{C}_2-5 \bar{C}_2^2)\bar{D}_2-4(\bar{c}_1 \bar{C}_2-2\bar{C}_2^2)\bar{D}_3-16 \bar{b}_3 \bar{C}_2 \bar{E}_1 \nonumber\\
& \hskip 0.7cm + 8(\bar{c}_1 \bar{C}_2-4\bar{C}_2^2)\bar{D}_4+2\bar{b}_2\bigl(2 \bar{d}_1^2+4 \bar{d}_1 \bar{D}_2+\bar{D}_2^2-2\bar{D}_3^2+8(\bar{d}_1+\bar{D}_2)\bar{D}_4+8 \bar{D}_4^2 \nonumber\\
& \hskip 0.7cm +4 \bar{C}_2(\bar{b}_3 \bar{C}_2-2\bar{E}_1)\bigl)-16 \bar{E}_1^2+8 \bar{D}_2 \bar{F}_2+8 \bar{b}_2^2\bigl(\bar{F}_2-\bar{b}_3(\bar{d}_1+\bar{D}_2-\bar{D}_3+2\bar{D}_4)\bigl)\nonumber\\ 
& \hskip 0.7cm +32 (\bar{d}_1+\bar{D}_2-\bar{D}_3+2\bar{D}_4)\bar{F}_4+\bar{b}_1\bigl(4 \bar{b}_2^2(\bar{b}_3^2-\bar{D}_4)+2\bar{D}_3(-4\bar{b}_3^2+2\bar{d}_1+3\bar{D}_2+8\bar{D}_4)\nonumber\\
& \hskip 0.7cm +\bar{b}_2(\bar{c}_1^2-\bar{C}_2^2-8 \bar{b}_3 \bar{D}_3-4 \bar{F}_2+8\bar{F}_4) \bigl) \bigl) \, . 
\end{align}

We then list the rest of the expansions:\begin{align} 
\{\bar{E}_1, \bar{F}_2\}&=\frac{{\rm i}}{4}\bigl( (2\bar{d}_1+\bar{D}_2+8\bar{D}_4-4\bar{b}_3^2-2\bar{b}_2\bar{b}_3)\bar{F}_1-(4 \bar{b}_1 \bar{b}_3+\bar{D}_2)\bar{F}_3+2(\bar{c}_1-\bar{C}_2)\bar{G}_1-2\bar{C}_2\bar{G}_2\nonumber\\
& \hskip 0.7cm +2\bar{b}_2\bar{H}_1+4\bar{b}_1 \bar{H}_2\bigl) \, , \nonumber\\
 \{\bar{E}_1, \bar{F}_3\}&=\frac{{\rm i}}{32}\bigl(8 \bar{b}_2^2 \bar{b}_3(\bar{d}_1+\bar{D}_2-\bar{D}_3+2\bar{D}_4)-2\bar{b}_2\bigl(4 \bar{b}_3 \bar{C}_2^2-\bar{D}_2(2 \bar{d}_1+3\bar{D}_2-2\bar{D}_3) -16 \bar{D}_4^2-8 \bar{C}_2 \bar{E}_1\nonumber\\
& \hskip 0.7cm -4(2\bar{d}_1+3\bar{D}_2-2\bar{D}_3)\bar{D}_4 \bigl) +8\bar{b}_3^2 \bar{C}_2^2+4\bar{b}_3 \bar{D}_2^2+4 \bar{C}_2(\bar{c}_1\bar{D}_2+2(\bar{c}_1-2\bar{C}_2)\bar{D}_3-6\bar{C}_2 \bar{D}_4)\nonumber\\
&  \hskip 0.7cm -16 \bar{E}_1^2-8(\bar{b}_2^2+2\bar{b}_2\bar{b}_3-6\bar{D}_3)\bar{F}_2-16(\bar{d}_1-\bar{D}_3+2\bar{D}_4)\bar{F}_4-\bar{b}_1\bigl( -8 \bar{D}_3(2 \bar{b}_3^2+\bar{D}_3)\nonumber\\
&  \hskip 0.7cm+4\bar{b}_2^2(\bar{b}_3^2-\bar{D}_4)+16 \bar{b}_3 \bar{F}_4+\bar{b}_2(8 \bar{b}_3^3-(\bar{c}_1-\bar{C}_2)(\bar{c}_1+3\bar{C}_2)-8\bar{b}_3(\bar{D}_3+\bar{D}_4)+16 \bar{F}_4)\bigl)\bigl) \,,\nonumber \\
\{\bar{E}_1, \bar{F}_4\}&=-\frac{{\rm i}}{4}\bigl((4 \bar{b}_3^2+\bar{D}_2-2\bar{D}_3-4\bar{D}_4)\bar{F}_3-2(\bar{c}_1-\bar{C}_2)\bar{G}_2+4 \bar{b}_2(\bar{b}_3 \bar{F}_1-\bar{H}_1) \bigl) \, .
\end{align}

\vskip 0.5cm

All the expansions above conclude the degree-ten Poisson brackets. We now look at the expansions for all the degree-eleven brackets.
\subsection{Expansions in the degree $11$ brackets}

In the context of degree eleven brackets, the corresponding compact forms can be identified as $  \{\bar{\textbf{F}},\bar{\textbf{F}}\} $, $  \{\bar{\textbf{C}},\bar{\textbf{I}}\} $, $\{\bar{\textbf{D}},\bar{\textbf{H}}\}$, and $\{\bar{\textbf{E}},\bar{\textbf{G}}\}$. Our initial focus is to compute all the expansions within the compact form $  \{\bar{\textbf{F}},\bar{\textbf{F}}\} $, which includes the Poisson brackets $\{\bar{F}_1, \bar{F}_2\},\{\bar{F}_1, \bar{F}_3\},\{\bar{F}_1, \bar{F}_4\},\{\bar{F}_2, \bar{F}_3\},\{\bar{F}_2, \bar{F}_4\}$ and $ \{\bar{F}_3, \bar{F}_4\}.$ Following the same procedure above, it becomes evident that the gradings applicable to all these brackets are 
 \begin{align}
       \mathcal{G} (\{\bar{F}_1, \bar{F}_2\}) = & \, (3,1,7) \tilde{+} (4,0,7) \tilde{+} (5,1,5) \tilde{+} (4,2,5) ; \nonumber \\
     \mathcal{G} (\{\bar{F}_1, \bar{F}_3\}) = & \, (2,3,6) \tilde{+} (3,2,6) \tilde{+} (4,3,4) \tilde{+} (3,4,4);  \nonumber \\
      \mathcal{G} (\{\bar{F}_1, \bar{F}_4\}) = & \, (1,3,7) \tilde{+} (2,2,7) \tilde{+} (3,3,5) \tilde{+} (2,4,5);  \nonumber \\
       \mathcal{G} (\{\bar{F}_2, \bar{F}_3\}) = & \, (2,2,7) \tilde{+} (3,1,7) \tilde{+} (4,2,5) \tilde{+} (3,3,5); \nonumber \\
         \mathcal{G} (\{\bar{F}_2, \bar{F}_4\}) = & \, (1,2,8) \tilde{+} (2,1,8) \tilde{+} (3,2,6) \tilde{+} (2,3,6) ;  \nonumber \\
     \mathcal{G} (\{\bar{F}_3, \bar{F}_4\}) = & \, (0,4,7) \tilde{+} (1,3,7) \tilde{+} (2,4,5) \tilde{+} (1,5,5). 
 \end{align} The allowed polynomials for each homogeneous grading in $\{\bar{F}_1,\bar{F}_2\}$,  which we consider as the guiding example, are \begin{align*}
     (3,1,7) = &\, \left\{\bar{E}_1 \bar{F}_2, \text{ } \bar{b}_1 \bar{D}_4 \bar{E}_1, \text{ } \bar{b}_3 \bar{C}_2 \bar{F}_2, \text{ } \bar{b}_1 \bar{b}_3^3 \bar{C}_2, \text{ } \bar{b}_1 \bar{b}_3^2 \bar{E}_1, \text{ } \bar{b}_3 C^{(202)} \bar{E}_1, \text{ } \bar{b}_1 \bar{b}_3 \bar{C}_2 \bar{D}_4, \text{ } C^{(003)} C^{(202)} \bar{D}_2, \text{ }  \bar{C}_2 C^{(202)} \bar{D}_4, \right. \\
     & \,\left.\bar{b}_1 \bar{b}_3 C^{(003)} \bar{D}_2, \text{ } \bar{b}_1 \left(C^{(003)}\right)^2 \bar{C}_2, \text{ } \bar{b}_3^2 \bar{C}_2 C^{(202)} \right\}; \\
     (4,0,7)= &\, \left\{\bar{b}_1 C^{(003)} \bar{F}_2, \text{ } C^{(003)} \left(C^{(202)}\right)^2, \text{ } \bar{b}_1^2 C^{(003)} \bar{D}_4, \text{ } \bar{b}_1 \bar{b}_3 C^{(003)} C^{(202)}, \text{ } \bar{b}_1^2 \bar{b}_3^2 C^{(003)} \right\} \\   
     (5,1,5)= &\, \left\{ \bar{b}_1 \bar{C}_2 \bar{F}_2, \text{ }  \bar{b}_1^2 \bar{b}_3^2 \bar{C}_2, \text{ } \bar{b}_1^2 \bar{C}_2 \bar{D}_4, \text{ } \bar{b}_1 C^{(202)} \bar{E}_1, \text{ } \bar{C}_2 \left(C^{(202)}\right)^2, \text{ } \bar{b}_1^2 \bar{b}_3 \bar{E}_1, \text{ } \bar{b}_1^2 C^{(003)} \bar{D}_2, \text{ } \bar{b}_1 \bar{b}_3 \bar{C}_2 C^{(202)}\right\} ;\\
     (4,2,5) = &\, \left\{\bar{b}_1 \bar{D}_2 \bar{E}_1, \text{ } \bar{C}_2 C^{(202)} \bar{D}_2, \text{ } \bar{b}_1^2 C^{(003)} \bar{D}_3, \text{ } \bar{b}_1 \bar{b}_2 C^{(003)} C^{(202)}, \text{ } \bar{b}_1 \bar{b}_3 \bar{C}_2 \bar{D}_2, \text{ } \bar{b}_1 C^{(003)} \bar{C}_2^2, \text{ } \bar{b}_1^2 \bar{b}_2 \bar{b}_3 C^{(003)} \right\} .
 \end{align*} 
 Taking into account the usual change of the generators from $C^{(003)}$ and $C^{(202)}$ to $\bar{c}_1$ and $\bar{d}_1$,  the previous polynomials are reformulated as 
 \begin{align*}
  &\,    \bar{E}_1 \bar{F}_2, \text{ }\bar{b}_1 \bar{D}_4 \bar{E}_1,\text{ }\bar{b}_3 \bar{C}_2 \bar{F}_2,\bar{b}_3 \bar{d}_1 \bar{E}_1,\text{ }\bar{b}_3 \bar{D}_2 \bar{E}_1,\text{ }\bar{b}_3 \bar{D}_3 \bar{E}_1,\text{ }\bar{b}_3 \bar{D}_4 \bar{E}_1,\text{ }\bar{c}_1 \bar{d}_1 \bar{D}_2,\bar{c}_1 \bar{D}_2^2,\text{ }\bar{c}_1 \bar{D}_2 \bar{D}_3,\bar{c}_1 \bar{D}_2 \bar{D}_4, \\
     &\, \bar{C}_2 \bar{d}_1 \bar{D}_2,\text{ }\bar{C}_2 \bar{D}_2^2,\bar{C}_2 \bar{D}_2 \bar{D}_3, \text{ }\bar{C}_2 \bar{D}_2 \bar{D}_4,\text{ }\bar{C}_2 \bar{d}_1 \bar{D}_4, \text{ }\bar{C}_2 \bar{D}_3 \bar{D}_4, \text{ }\bar{C}_2 \bar{D}_4^2,\text{ }\bar{b}_1 \bar{b}_3^2 \bar{E}_1, \text{ }\bar{b}_1 \bar{b}_3 \bar{c}_1 \bar{D}_2, \text{ }\bar{b}_1 \bar{b}_3 \bar{C}_2 \bar{D}_2,   \\
     &\, \bar{b}_1 \bar{c}_1^2 \bar{C}_2,\text{ }  \bar{b}_1 \bar{c}_1 \bar{C}_2^2,\text{ }\bar{b}_1 \bar{C}_2^3, \text{ }\bar{b}_3^2 \bar{C}_2 \bar{d}_1, \text{ }\bar{b}_3^2 \bar{C}_2 \bar{D}_2, \text{ }\bar{b}_3^2 \bar{C}_2 \bar{D}_3, \text{ }\bar{b}_3^2 \bar{C}_2 \bar{D}_4, \text{ }\bar{b}_1 \bar{b}_3^3 \bar{C}_2, \text{ }\bar{b}_1 \bar{c}_1 \bar{F}_2, \text{ }\bar{b}_1 \bar{C}_2 \bar{F}_2,   \\
     &\, \bar{c}_1 \bar{d}_1^2,\text{ }  \bar{c}_1 \bar{d}_1 \bar{D}_2, \text{ }  \bar{c}_1 \bar{d}_1 \bar{D}_3, \text{ }  \bar{c}_1 \bar{d}_1 \bar{D}_4, \text{ }  \bar{c}_1 \bar{D}_2 \bar{D}_3,\text{ }  \bar{c}_1 \bar{D}_2 \bar{D}_4,\text{ }\bar{c}_1 \bar{D}_3^2,\text{  }  \bar{c}_1 \bar{D}_3 \bar{D}_4, \text{ }\bar{c}_1 \bar{D}_4^2, \text{ }\bar{C}_2 \bar{d}_1^2,  \\
     &\,   \bar{C}_2 \bar{d}_1 \bar{D}_3, \text{ }  \bar{C}_2 \bar{d}_1 \bar{D}_4, \text{ }  \bar{C}_2 \bar{D}_2 \bar{D}_3, \text{ }  \bar{C}_2 \bar{D}_2 \bar{D}_4, \text{ }\bar{C}_2 \bar{D}_3^2, \text{ }  \bar{C}_2 \bar{D}_3 \bar{D}_4, \text{ }\bar{b}_1^2 \bar{c}_1 \bar{D}_4, \text{ }\bar{b}_1^2 \bar{C}_2 \bar{D}_4, \\
    &\, \bar{b}_1 \bar{b}_3 \bar{C}_2 \bar{d}_1, \text{ }\bar{b}_1 \bar{b}_3 \bar{C}_2 \bar{D}_3, \text{ }\bar{b}_1^2 \bar{b}_3^2 \bar{c}_1, \text{ }\bar{b}_1^2 \bar{b}_3^2 \bar{C}_2, \text{ }\bar{b}_1 \bar{d}_1 \bar{E}_1, \text{ }\bar{b}_1 \bar{D}_2 \bar{E}_1, \text{ }\bar{b}_1 \bar{D}_3 \bar{E}_1, \text{ }\bar{b}_1^2 \bar{b}_3 \bar{E}_1, \text{ }\bar{b}_1^2 \bar{c}_1 \bar{D}_2, \\
     &\, \bar{b}_1^2 \bar{C}_2 \bar{D}_2, \text{ }\bar{b}_1^2 \bar{c}_1 \bar{D}_3, \text{ }\bar{b}_1^2 \bar{C}_2 \bar{D}_3, \text{ }\bar{b}_1 \bar{b}_2 \bar{c}_1 \bar{d}_1, \text{ }\bar{b}_1 \bar{b}_2 \bar{c}_1 \bar{D}_2, \text{ }\bar{b}_1 \bar{b}_2 \bar{c}_1 \bar{D}_3, \text{ }\bar{b}_1 \bar{b}_2 \bar{c}_1 \bar{D}_4, \text{ }   \bar{C}_2 \bar{d}_1 \bar{D}_2,\\
    &\, \bar{b}_1 \bar{b}_2 \bar{C}_2 \bar{d}_1, \text{ }\bar{b}_1 \bar{b}_2 \bar{C}_2 \bar{D}_2, \text{ }\bar{b}_1 \bar{b}_2 \bar{C}_2 \bar{D}_3, \text{ }\bar{b}_1 \bar{b}_2 \bar{C}_2 \bar{D}_4, \text{ }\bar{b}_1 \bar{c}_1 \bar{C}_2^2, \text{ }\bar{b}_1^2 \bar{b}_2 \bar{b}_3 \bar{c}_1, \text{ }\bar{b}_1^2 \bar{b}_2 \bar{b}_3 \bar{C}_2 \\
    & \, \text{ }\bar{b}_1 \bar{b}_3 \bar{c}_1 \bar{D}_4,\text{ }\bar{b}_1 \bar{b}_3 \bar{c}_1 \bar{D}_3,\text{ }\bar{b}_1 \bar{b}_3 \bar{c}_1 \bar{d}_1 .
 \end{align*} In other words, there are some coefficients $\Gamma_{66}^1,\ldots,\Gamma_{66}^{77}$ such that $$\{\bar{F}_1,\bar{F}_2\} =  \Gamma_{66}^1 \bar{E}_1 \bar{F}_2 + \ldots +  \Gamma_{66}^{76}  \bar{b}_1 \bar{b}_3 \bar{c}_1 \bar{D}_3 + \Gamma_{66}^{77}  \bar{b}_1 \bar{b}_3 \bar{c}_1 \bar{d}_1 .$$  Through meticulous determination of these coefficients, we infer that \begin{align*}
     \{\bar{F}_1, \bar{F}_2\}&=\frac{{\rm i}}{8}\bigl( 2(\bar{d}_1+\bar{D}_2-\bar{D}_3+2\bar{D}_4)(2\bar{c}_1\bar{D}_2-2\bar{C}_2\bar{D}_2-16 \bar{b}_3\bar{E}_1)+16(\bar{b}_3\bar{C}_2+\bar{E}_1)\bar{F}_2 \nonumber \\
& \hskip 0.7 cm +\bar{b}_1 \bigl(\bar{c}_1^2\bar{C}_2+\bar{C}_2^3+2\bar{b}_2\bar{C}_2(\bar{d}_1+\bar{D}_2-\bar{D}_3+2\bar{D}_4)-2 \bar{c}_1\bigl(\bar{C}_2^2+\bar{b}_2(\bar{d}_1+\bar{D}_2-\bar{D}_3+2\bar{D}_4)\bigl) \nonumber\\
& \hskip 0.7 cm +2 (4 \bar{b}_3^2-2\bar{d}_1-3\bar{D}_2+2\bar{D}_3-8\bar{D}_4)\bar{E}_1+4 \bar{C}_2 \bar{F}_2\bigl) \bigl) \, .
 \end{align*} 
 
Following the same approach, we deduce that the expansion of the rest of the brackets are given by 
\begin{align}
\{\bar{F}_1, \bar{F}_3\}&=\frac{{\rm i}}{4}\bigl( \bar{b}_2 \bar{C}_2 \bar{F}_1-4\bar{E}_1(\bar{F}_1+2\bar{F}_3)-2\bar{D}_2 \bar{G}_1+\bar{b}_1 \bar{b}_2(\bar{G}_1-\bar{G}_2)+2(2\bar{d}_1+5 \bar{D}_2-2\bar{D}_3+4\bar{D}_4)\bar{G}_2\nonumber \\
& \hskip 0.7 cm +4 \bar{C}_2(\bar{H}_2+\bar{b}_3 \bar{F}_3)-\bar{b}_1(8\bar{I}_1+\bar{c}_1\bar{F}_3)   \bigl) \, ,\nonumber \\
\{\bar{F}_1, \bar{F}_4\}&=-\frac{{\rm i}}{16}\bigl(\bar{b}_1 \bar{b}_2(4\bar{b}_3^2\bar{C}_2+\bar{c}_1 \bar{D}_2-\bar{C}_2(\bar{D}_2+4\bar{D}_4)-8\bar{b}_3 \bar{E}_1) +2\bar{b}_2(\bar{c}_1-\bar{C}_2)\bar{C}_2^2-2\bar{c}_1\bar{D}_2^2 \nonumber\\
 & \hskip 0.7 cm +2 \bar{C}_2 \bigl( \bar{D}_2^2-4\bar{b}_3^2 \bar{D}_2+2(2\bar{d}_1+3\bar{D}_2)\bar{D}_3-4\bar{D}_3^2+4(\bar{D}_2+2\bar{D}_3)\bar{D}_4+4(\bar{c}_1-3\bar{C}_2)\bar{E}_1\bigl) \nonumber\\
 & \hskip 0.7 cm +16 \bar{b}_3\bar{D}_2\bar{E}_1+8\bar{b}_2(\bar{d}_1+\bar{D}_2-\bar{D}_3+2\bar{D}_4)\bar{E}_1+8\bar{b}_1 \bar{C}_2 \bar{F}_4\bigl) \, , \nonumber
 \end{align} and 
 \begin{align} 
\{\bar{F}_2, \bar{F}_3\}&=\frac{{\rm i}}{16}\bigl( 2 \bar{D}_2\bigl(-4 \bar{b}_3^2 \bar{C}_2-\bar{c}_1 \bar{D}_2+2\bar{C}_2(\bar{d}_1+2\bar{D}_2-\bar{D}_3+4 \bar{D}_4)\bigl)+8(\bar{c}_1-\bar{C}_2)\bar{C}_2\bar{E}_1 \nonumber \\
 & \hskip 0.7 cm +8 \bar{b}_3\bigl( -\bar{C}_2^3+\bar{b}_2 \bar{C}_2(\bar{d}_1+\bar{D}_2-\bar{D}_3+2\bar{D}_4)+2\bar{D}_2 \bar{E}_1\bigl)+\bar{b}_1 \bigl(2 \bar{c}_1 \bar{C}_2^2-2\bar{C}_2^3+\bar{b}_2(\bar{c}_1-\bar{C}_2)\bar{D}_2 \nonumber\\
 & \hskip 0.7 cm +8 \bar{b}_3 \bar{C}_2 \bar{D}_3+8(\bar{D}_3-\bar{b}_2 \bar{b}_3)\bar{E}_1-8\bar{C}_2 \bar{F}_4   \bigl)\bigl) \, ,\nonumber\\
\{\bar{F}_2, \bar{F}_4\}&=\frac{{\rm i}}{4}\bigl( 4 \bar{b}_3 \bar{C}_2 \bar{F}_3+4(\bar{d}_1+\bar{D}_2-\bar{D}_3+2\bar{D}_4)\bar{G}_2-4 \bar{E}_1(2\bar{F}_1-\bar{F}_3)-\bar{b}_1(\bar{c}_1-\bar{C}_2)\bar{F}_3-4\bar{C}_2 \bar{H}_2\bigl) \, ,\nonumber\\
\{\bar{F}_3, \bar{F}_4\}&=\frac{{\rm i}}{8}\bigl( 2 \bar{D}_3((2\bar{c}_1-3\bar{C}_2)\bar{D}_2-16 \bar{b}_3 \bar{E}_1)+16 (\bar{E}_1+\bar{b}_3 \bar{C}_2)\bar{F}_4+\bar{b}_2\bigl(\bar{c}_1^2 \bar{C}_2+\bar{C}_2^3+2\bar{b}_1\bar{C}_2\bar{D}_3\nonumber\\
 & \hskip 0.7 cm -2 \bar{c}_1(\bar{C}_2^2+\bar{b}_1\bar{D}_3)+2(4\bar{b}_3^2-\bar{D}_2-2\bar{D}_3-4\bar{D}_4)\bar{E}_1+4\bar{C}_2 \bar{F}_4 \bigl)\bigl) \, ,\nonumber\\
 & \hskip 0.7 cm+8(\bar{D}_4-\bar{b}_3^2)\bar{E}_1\bigl)+24 \bar{E}_1 \bar{F}_4-8 \bar{b}_3(3 \bar{D}_3\bar{E}_1+\bar{C}_2\bar{F}_4) \, ,%\nonumber\\
 \end{align} 

The compact form $\{\bar{\textbf{C}},\bar{\textbf{I}}\}$ contains only two brackets, as $\bar{c}_1$ is a central element. A direct computation shows that 
  \begin{align} 
 \{\bar{C}_2, \bar{I}_1\}&=\frac{{\rm i}}{8}\bigl(\bigl( 3\bar{c}_1 \bar{D}_2+\bar{C}_2(12 \bar{b}_3^2-3\bar{D}_2-4\bar{D}_4)\bigl)\bar{D}_3-\bar{b}_2\bigl( \bar{b}_3( \bar{c}_1 \bar{D}_2-\bar{C}_2(\bar{D}_2+4\bar{D}_4-4\bar{b}_3^2))\nonumber \\
  &  \hskip 0.7 cm+8(\bar{D}_4-\bar{b}_3^2)\bar{E}_1\bigl)+24 \bar{E}_1 \bar{F}_4-8 \bar{b}_3(3 \bar{D}_3\bar{E}_1+\bar{C}_2\bar{F}_4) \, ,\nonumber\\
\{\bar{C}_2, \bar{I}_2\}&=\frac{{\rm i}}{8}\bigl((\bar{d}_1+\bar{D}_2-\bar{D}_3+2\bar{D}_4)(12 \bar{b}_3^2 \bar{C}_2+3(\bar{c}_1-\bar{C}_2)\bar{D}_2-4\bar{C}_2 \bar{D}_4-24 \bar{b}_3 \bar{E}_1)+8(\bar{b}_3 \bar{C}_2-3 \bar{E}_1)\bar{F}_2 \nonumber\\
& \hskip 0.7 cm -\bar{b}_1 (\bar{b}_3 \bar{c}_1 \bar{D}_2-\bar{b}_3 \bar{C}_2(\bar{D}_2+4\bar{D}_4-4\bar{b}_3^2)+8(\bar{D}_4-\bar{b}_3^2)\bar{E}_1) \bigl) \,.
 \end{align}

We proceed to examine the gradings from the compact form $\{\bar{\textbf{D}},\bar{\textbf{H}}\}$. In particular, $\bar{d}_1$ is a central element, and thus our analysis begins with $\bar{D}_2$. For the specific gradings for each respective term, we have
\begin{align}
\mathcal{G} \left( \{\bar{D}_2,\bar{H}_1\}\right)= &\, (2,2,7) \tilde{+}  (3,1,7) \tilde{+} (4,2,5) \tilde{+} (3,3,5) ; \nonumber \\
 \mathcal{G} \left( \{\bar{D}_2,\bar{H}_2\}\right)= &\, (1,3,7) \tilde{+}  (2,2,7) \tilde{+}  (3,3,5) \tilde{+} (2,4,5) ; \nonumber \\
 \mathcal{G} \left( \{\bar{D}_3,\bar{H}_1\}\right)= &\, (1,3,7) \tilde{+}  (2,2,7) \tilde{+}  (3,3,5) \tilde{+} (2,4,5)  ; \\
 \mathcal{G} \left( \{\bar{D}_3,\bar{H}_2\}\right)= &\, (0,4,7) \tilde{+} (1,3,7) \tilde{+} (2,4,5) \tilde{+} (1,5,5) ; \nonumber \\
 \mathcal{G} \left( \{\bar{D}_4,\bar{H}_1\}\right)= &\, (1,1,9) \tilde{+}  (2,0,9) \tilde{+} (3,1,7) \tilde{+} (2,2,7) ; \nonumber \\
 \mathcal{G} \left( \{\bar{D}_4,\bar{H}_2\}\right)= &\, (0,2,9) \tilde{+}  (1,1,9) \tilde{+}  (2,2,7) \tilde{+} (1,3,7) . \nonumber
\end{align} 
Now, we present all the predicted polynomials from the homogeneous grading associated with the first bracket. They are \begin{align*}
     (2,2,7) =&\, \left\{\bar{E}_1 \bar{F}_4,\bar{b}_2 \bar{D}_4 \bar{E}_1,\bar{b}_2 \bar{b}_3 \bar{C}_2 \bar{D}_4,\bar{b}_3 \bar{C}_2 \bar{F}_4,\bar{b}_3 \bar{D}_3 \bar{E}_1,\bar{C}_2 \bar{D}_3 \bar{D}_4,\bar{b}_3^2 \bar{C}_2 \bar{D}_3,\bar{b}_2 \bar{b}_3^3 \bar{C}_2,\bar{b}_2 \bar{b}_3^2 \bar{E}_1,C^{(003)} \bar{D}_2 \bar{D}_3,\bar{b}_2 \bar{b}_3 C^{(003)} \bar{D}_2,\right. \\
     & \, \left.\bar{b}_2 \left(C^{(003)}\right)^2 \bar{C}_2\right\}; \\ 
     (3,1,7)=&\, \left\{\bar{C}_2 \bar{D}_2 \bar{D}_4,\bar{b}_3 \bar{D}_2 \bar{E}_1,\bar{b}_3^2 \bar{C}_2 \bar{D}_2,C^{(003)} C^{(202)} \bar{D}_3,\bar{b}_1 \bar{b}_2 C^{(003)} \bar{D}_4,\bar{b}_1 \bar{b}_3 C^{(003)} \bar{D}_3,\bar{b}_2 \bar{b}_3 C^{(003)} C^{(202)},\bar{b}_3 C^{(003)} \bar{C}_2^2,\right. \\
     &\,\left.\bar{b}_1 \bar{b}_2 \bar{b}_3^2 C^{(003)},\bar{b}_1 C^{(003)} \bar{F}_4,\bar{b}_2 C^{(003)} \bar{F}_2,C^{(003)} \bar{C}_2 \bar{E}_1,C^{(003)} \bar{D}_2^2 \right\}; \\  
     (4,2,5) =&\, \left\{\bar{b}_1 \bar{C}_2 \bar{F}_4,\bar{b}_1 \bar{D}_3 \bar{E}_1,\bar{b}_1 \bar{b}_3 \bar{C}_2 \bar{D}_3,\bar{b}_2 \bar{C}_2 \bar{F}_2,\bar{C}_2^2 \bar{E}_1,\bar{C}_2 \bar{D}_2^2,\bar{b}_1 \bar{b}_2 \bar{b}_3 \bar{E}_1,\bar{b}_3 \bar{C}_2^3,\bar{b}_1 \bar{b}_2 \bar{C}_2 \bar{D}_4,\bar{b}_1 \bar{b}_2 \bar{b}_3^2 \bar{C}_2,\bar{b}_2 C^{(202)} \bar{E}_1,\right. \\
     &\, \left. \bar{b}_1 \bar{b}_2 C^{(003)} \bar{D}_2,\bar{C}_2 C^{(202)} \bar{D}_3,\bar{b}_2 \bar{b}_3 \bar{C}_2 C^{(202)}\right\}; \\ 
     (3,3,5)=&\, \left\{\bar{b}_2 \bar{D}_2 \bar{E}_1,\bar{C}_2 \bar{D}_2 \bar{D}_3,\bar{b}_2 \bar{b}_3 \bar{C}_2 \bar{D}_2,\bar{b}_2^2 C^{(003)} C^{(202)},\bar{b}_1 \bar{b}_2 C^{(003)} \bar{D}_3,\bar{b}_2 C^{(003)} \bar{C}_2^2,\bar{b}_1 \bar{b}_2^2 \bar{b}_3 C^{(003)}\right\}. 
\end{align*} 
 Replacing the generators $C^{(003)}$ by $\bar{c}_1$ and $C^{(202)}$ by $\bar{d}_1$ there are, after some computations, some coefficients  $\Gamma_{48}^1,\ldots,\Gamma_{48}^{65}$ such that $\{\bar{D}_2,\bar{H}_1\} = \Gamma_{48}^1 \bar{E}_1\bar{F}_4 + \ldots + \Gamma_{48}^{64} \bar{b}_2+\Gamma_{48}^{65} \bar{b}_1\bar{b}_2^2 \bar{b}_3 \bar{C}_2$.
Determining these coefficients provides us with the following expression 
\begin{align}
    \{\bar{D}_2, \bar{H}_1\}&=\frac{{\rm i}}{16}\bigl(4(2\bar{b}_3^2\bar{C}_2-(\bar{c}_2-2\bar{C}_2)\bar{d}_1)\bar{D}_2 +2(2\bar{C}_2-\bar{c}_1)\bar{D}_2^2+2\bar{b}_1^2(\bar{c}_1-\bar{C}_2)(\bar{b}_2 \bar{b}_3-\bar{D}_3)\nonumber\\
& \hskip 0.7 cm+4(\bar{c}_1-2\bar{C}_2)\bar{D}_2 \bar{D}_3-8(\bar{c}_1-3 \bar{C}_2)\bar{D}_2 \bar{D}_4-8 \bar{b}_2 \bar{b}_3 \bar{C}_2(\bar{d}_1+\bar{D}_2-\bar{D}_3+2\bar{D}_4) \nonumber \\
& \hskip 0.7 cm+8(2\bar{c}_1-3 \bar{C}_2)\bar{C}_2 \bar{E}_1+16(2\bar{d}_1+\bar{D}_2-2\bar{D}_3+4\bar{D}_4)\bar{b}_3 \bar{E}_1-16\bar{b}_2(\bar{c}_1-2\bar{C}_2)\bar{F}_2 \nonumber \\
& \hskip 0.7 cm -32 \bar{E}_1 \bar{F}_2+\bar{b}_1 \bar{b}_2\bigl(8 \bar{b}_3^2 \bar{C}_2+\bar{c}_1(\bar{D}_2+8\bar{D}_4)-\bar{C}_2(\bar{D}_2+16 \bar{D}_4)-\bar{b}_3 \bar{E}_1 \bigl) \nonumber\\
& \hskip 0.7 cm +2 \bar{b}_1\bigl(-\bar{C}_2^3-2(4\bar{b}_3^2+\bar{D}_2-2 \bar{D}_3-4\bar{D}_4)\bar{E}_1 +\bar{c}_1(\bar{C}_2^2-4\bar{F}_4)+2\bar{C}_2(\bar{b}_3 \bar{D}_2+2\bar{F}_4)\bigl) \bigl) \, .
\end{align} 
A similar approach shows that the expansion of the rest of the brackets are
 \begin{align} 
\{\bar{D}_2, \bar{H}_2\}&=\frac{{\rm i}}{16}\bigl( 8 \bar{b}_3 \bar{C}_2 (\bar{C}_2^2+\bar{b}_3 \bar{D}_2)+2(\bar{c}_1-3\bar{C}_2)\bar{D}_2^2+4(2\bar{C}_2 \bar{d}_1-\bar{c}_1\bar{D}_2+4\bar{C}_2\bar{D}_2)\bar{D}_3 \nonumber \\
& \hskip 0.7 cm -8 \bar{C}_2(\bar{D}_3^2+\bar{D}_2 \bar{D}_4-2 \bar{D}_3 \bar{D}_4)-2\bar{b}_2^2(\bar{c}_1-\bar{C}_2)(\bar{d}_1+\bar{D}_2-2\bar{D}_3+2\bar{D}_4)+16 \bar{c}_1 \bar{C}_2 \bar{E}_1 \nonumber \\
& \hskip 0.7 cm -40 \bar{C}_2^2 \bar{E}_1-16 \bar{b}_3 \bar{D}_2 \bar{E}_1+32 \bar{b}_3 \bar{D}_3 \bar{E}_1-2\bar{b}_2 \bigl( \bar{C}_2^3+2 \bar{b}_3 \bar{C}_2(2 \bar{d}_1+\bar{D}_2-2\bar{D}_3+4\bar{D}_4)\nonumber\\
& \hskip 0.7 cm+2(4 \bar{b}_3^2-2\bar{d}_1-\bar{D}_2+2\bar{D}_3-8\bar{D}_4)\bar{E}_1-\bar{c}_1(\bar{C}_2^2-4\bar{F}_2)-8\bar{C}_2 \bar{F}_2\bigl)-32 \bar{E}_1 \bar{F}_4\nonumber \\
& \hskip 0.7 cm +\bar{b}_1 \bigl( 2 \bar{b}_2^2 \bar{b}_3(\bar{c}_1-\bar{C}_2)+\bar{b}_2 (12 \bar{b}_3^2 \bar{C}_2+\bar{c}_1(\bar{D}_2+8\bar{D}_4)-\bar{C}_2(\bar{D}_2+20 \bar{D}_4)-8 \bar{b}_3 \bar{E}_1)\nonumber \\
& \hskip 0.7 cm -8(2 \bar{b}_3 \bar{C}_2 \bar{D}_3+2 \bar{c}_1 \bar{F}_4-5 \bar{C}_2 \bar{F}_4)\bigl)\bigl) \, ,\nonumber \\
\{\bar{D}_3, \bar{H}_1\}&=\frac{{\rm i}}{16}\bigl(8(\bar{D}_4-\bar{b}_3^2)\bar{C}_2 \bar{D}_2-2 \bar{c}_1\bar{D}_2^2+8(\bar{c}_1\bar{C}_2-2\bar{C}_2^2+2\bar{b}_3 \bar{D}_2+\bar{b}_1\bar{D}_3)\bar{E}_1+\bar{b}_1\bar{b}_2\bigl( 8 \bar{b}_3^2\bar{C}_2+\bar{c}_1 \bar{D}_2 \nonumber\\
& \hskip 0.7 cm-\bar{C}_2(\bar{D}_2+8\bar{D}_4)-8\bar{b}_3 \bar{E}_1\bigl)+8 \bar{b}_2\bar{C}_2\bigl( \bar{b}_3(\bar{d}_1+\bar{D}_2-\bar{D}_3+2\bar{D}_4)-2\bar{F}_2\bigl) \bigl) \, ,\nonumber \\
\{\bar{D}_3, \bar{H}_2\}&=\frac{{\rm i}}{8}\bigl( \bar{b}_2^2(\bar{c}_1-\bar{C}_2)(\bar{b}_1 \bar{b}_3-(\bar{d}_1+\bar{D}_2-\bar{D}_3+2\bar{D}_4))+2 \bar{D}_3 (12 \bar{b}_3^2 \bar{C}_2+\bar{c}_1 \bar{D}_2-4\bar{C}_2 \bar{D}_4-8\bar{b}_3 \bar{E}_1)\nonumber \\
& \hskip 0.7 cm-\bar{b}_2(8\bar{b}_3^3 \bar{C}_2+(\bar{c}_1-2\bar{C}_2)(\bar{c}_1-\bar{C}_2)\bar{C}_2+2\bar{b}_3 \bar{c}_1\bar{D}_2-4 \bar{b}_3 \bar{C}_2(\bar{D}_2+2\bar{D}_4)+2(\bar{D}_2+4\bar{D}_4) \bar{E}_1 \nonumber \\
& \hskip 0.7 cm-8 \bar{b}_3^2 \bar{E}_1)+32 (\bar{E}_1-\bar{b}_3 \bar{C}_2)\bar{F}_4 \bigl) \, ,
\end{align} and
\begin{align}
\{\bar{D}_4, \bar{H}_1\}&=\frac{{\rm i}}{16}\bigl(4 \bar{C}_2^2 \bar{E}_1-2 \bar{C}_2 \bar{D}_2(2 \bar{d}_1+\bar{D}_2-2 \bar{D}_3)-8 \bar{C}_2 \bar{D}_4(\bar{d}_1+2\bar{D}_2-\bar{D}_3)-16 \bar{C}_2 \bar{D}_4^2-8 \bar{b}_2 \bar{C}_2 \bar{F}_2\nonumber \\
& \hskip 0.7 cm +8 \bar{b}_3^2 \bar{C}_2(3 \bar{d}_1+2\bar{D}_2-3 \bar{D}_3+6 \bar{D}_4)+48 \bar{E}_1 \bar{F}_2+2 \bar{c}_1\bigl( \bar{D}_2(2\bar{d}_1+\bar{D}_2-2 \bar{D}_3+4\bar{D}_4)-2\bar{C}_2 \bar{E}_1 \nonumber \\
& \hskip 0.7 cm +4\bar{b}_2 \bar{F}_2\bigl) -16 \bar{b}_3\bigl((2 \bar{d}_1+\bar{D}_2-2 \bar{D}_3+4\bar{D}_4)\bar{E}_1+2\bar{C}_2 \bar{F}_2 \bigl) -\bar{b}_1 \bigl((8\bar{b}_3^3+\bar{c}_1^2-4\bar{b}_2 \bar{D}_4 \nonumber\\
& \hskip 0.7 cm -2\bar{b}_3(\bar{D}_2+4\bar{D}_4))\bar{C}_2 +\bar{C}_2^3+2\bar{b}_3 \bar{c}_1 \bar{D}_2+16(\bar{D}_4-\bar{b}_3^2)\bar{E}_1+4 \bar{C}_2 \bar{F}_4-2\bar{c}_1(\bar{C}_2^2-2\bar{b}_2 \bar{D}_4+2\bar{F}_4)\bigl)\bigl) \, ,\nonumber \\
\{\bar{D}_4, \bar{H}_2\}&=\frac{{\rm i}}{16} \bigl(2\bar{C}_2\bigl(\bar{D}_2^2-4\bar{b}_3^2(\bar{D}_2-3\bar{D}_3)-2\bar{D}_3(\bar{D}_2+2\bar{D}_4)\bigl)+4(\bar{C}_2^2+4\bar{b}_3 (\bar{D}_2-2\bar{D}_3)+12 \bar{F}_4)\bar{E}_1\nonumber\\
& \hskip 0.7 cm -8 \bar{b}_1 \bar{C}_2 \bar{F}_4-2 \bar{c}_1(\bar{D}_2^2-2\bar{D}_2 \bar{D}_3+2\bar{C}_2 \bar{E}_1-4\bar{b}_1 \bar{F}_4) -32 \bar{b}_3 \bar{C}_2 \bar{F}_4-\bar{b}_2\bigl(8 \bar{b}_3^2 \bar{C}_2+\bar{c}_1^2 \bar{C}_2\nonumber\\
& \hskip 0.7 cm +\bar{C}_2^3+2 \bar{b}_3 \bar{c}_1 \bar{D}_2-4 \bar{b}_1 \bar{C}_2 \bar{D}_4-2 \bar{b}_3 \bar{C}_2(\bar{D}_2+4\bar{D}_4)-16 \bar{b}_3^2\bar{E}_1+16 \bar{D}_4 \bar{E}_1 +4 \bar{C}_2 \bar{F}_2\nonumber\\
& \hskip 0.7 cm-2 \bar{c}_1(\bar{C}_2^2-2\bar{b}_1\bar{D}_4+2\bar{F}_2)\bigl)\bigl) \, .%\nonumber\\
\end{align}

Finally, we consider the compact form $\{\bar{\textbf{E}},\bar{\textbf{G}}\},$ which contains only two Poisson brackets. The gradings of these terms are given by \begin{align*}
      \mathcal{G} (\{\bar{E}_1, \bar{G}_1\}) = & \,  (2,2,7) \tilde{+}  (3,1,7) \tilde{+} (4,2,5) \tilde{+} (3,3,5)   ; \\
     \mathcal{G} (\{\bar{E}_1, \bar{G}_2\}) = & \,  (1,3,7) \tilde{+}  (2,2,7) \tilde{+}  (3,3,5) \tilde{+} (2,4,5).
\end{align*} Note that, from the computation above, we deduce that \begin{align*}
     \mathcal{G} (\{\bar{E}_1, \bar{G}_1\}) = \mathcal{G} \left( \{\bar{D}_2,\bar{H}_1\}\right)\text{ and } \mathcal{G} (\{\bar{E}_1, \bar{G}_2\}) = \mathcal{G} \left( \{\bar{D}_2,\bar{H}_2\}\right).
\end{align*} 
However,   we observe that, even if the admissible polynomials from these gradings are the same, the coefficients in the Poisson brackets heavily depend on the structure tensor of the given Lie algebra. In this case, we deduce that the explicit expression in the expansion are  
\begin{align}
\{\bar{E}_1, \bar{G}_1\}&=\frac{{\rm i}}{16}\bigl(\bar{C}_2(\bar{D}_2^2+4\bar{d}_1(\bar{D}_2+\bar{D}_3)-4\bar{D}_3^2+8(\bar{D}_2-\bar{d}_1+2\bar{D}_3)\bar{D}_4-16 \bar{D}_4^2)+8 \bar{b}_3^2 \bar{C}_2(\bar{d}_1-\bar{D}_3\nonumber\\
& \hskip 0.7 cm+2\bar{D}_4)+4 \bar{E}_1\bigl(3 \bar{c}_1 \bar{C}_2-7 \bar{C}_2^2+2\bar{b}_2(\bar{d}_1+\bar{D}_2-\bar{D}_3+2\bar{D}_4)\bigl)+4(\bar{b}_2(3\bar{C}_2-2\bar{c}_1)+4\bar{E}_1)\bar{F}_2\nonumber\\
& \hskip 0.7 cm+8 \bar{b}_3\bigl( 2\bar{C}_2^3+\bar{c}_1(\bar{b}_2(\bar{d}_1+\bar{D}_2-\bar{D}_3+2\bar{D}_4)-\bar{C}_2^2)+\bar{D}_2 \bar{E}_1-2\bar{C}_2(\bar{b}_2(\bar{d}_1+\bar{D}_2-\bar{D}_3+2\bar{D}_4)\nonumber\\
& \hskip 0.7 cm+\bar{F}_2)\bigl)-\bar{b}_1 \bigl(\bar{c}_1^2 \bar{C}_2-3\bar{c}_1 \bar{C}_2^2+2\bar{C}_2^3+4(\bar{b}_3(5\bar{C}_2-\bar{c}_1)-\bar{E}_1)\bar{D}_3+\bar{b}_2 \bigl( 2 \bar{b}_3^2(2\bar{c}_1-7\bar{C}_2)\nonumber\\
& \hskip 0.7 cm-(\bar{c}_1-\bar{C}_2)(\bar{d}_1+2\bar{D}_2-\bar{D}_3)+2\bar{D}_4(8 \bar{C}_2-3 \bar{c}_1)+4 \bar{b}_3 \bar{E}_1\bigl)+4(\bar{c}_1-5\bar{C}_2)\bar{F}_4 \bigl) \bigl) \,
\end{align}
and \begin{align}
\{\bar{E}_1, \bar{G}_2\}&=\frac{{\rm i}}{16}\bigl(8 \bar{b}_3 \bar{C}_2(2\bar{C}_2^2-\bar{c}_1 \bar{C}_2-\bar{b}_3 \bar{D}_2)-3 \bar{C}_2 \bar{D}_2^2+4\bar{C}_2 \bar{D}_3(\bar{d}_1+2(\bar{D}_2+\bar{b}_3^2))-4 \bar{C}_2 \bar{D}_3^2+8 \bar{C}_2 \bar{D}_2 \bar{D}_4\nonumber\\
& \hskip 0.7 cm +4(3 \bar{c}_1 \bar{C}_2-7 \bar{C}_2^2+2\bar{b}_3 \bar{D}_2)\bar{E}_1+16 (\bar{E}_1-\bar{b}_3 \bar{C}_2)\bar{F}_4-\bar{b}_1 \bigl( -8 \bar{D}_3(\bar{E}_1+\bar{b}_3(\bar{c}_1-2\bar{C}_2))\nonumber\\
& \hskip 0.7 cm +\bar{b}_2(2\bar{b}_3^2(2\bar{c}_1-7\bar{C}_2)-\bar{c}_1(\bar{D}_2+\bar{D}_3+4\bar{D}_4)+\bar{C}_2(\bar{D}_2+\bar{D}_3+14\bar{D}_4)+4\bar{b}_3 \bar{E}_1)\nonumber\\
& \hskip 0.7 cm +4(2\bar{c}_1-3\bar{C}_2)\bar{F}_4\bigl)-\bar{b}_2 \bigl( \bar{c}_1^2\bar{C}_2-3\bar{c}_1\bar{C}_2^2+2\bar{C}_2^3-4(\bar{d}_1+\bar{D}_2-\bar{D}_3+2\bar{D}_4)(\bar{E}_1+\bar{b}_3 \bar{c}_1)\nonumber\\
& \hskip 0.7 cm+20 \bar{C}_2(\bar{b}_3(\bar{d}_1+\bar{D}_2-\bar{D}_3+2\bar{D}_4)-\bar{F}_2)+4\bar{c}_1 \bar{F}_2\bigl) \bigl) \, .
\end{align} These expansions conclude the degree-eleven non-trivial Poisson brackets.
\vskip 0.5cm

\subsection{Expansions in the degree $12$ brackets}

For the expansions associated with degree brackets $12$, the compact forms are characterized by distinct pairings of brackets, which include pairs such as $\{\bar{\mathbf{D}},\bar{\mathbf{I}} \}$, $\{\bar{\mathbf{E}},\bar{\mathbf{H}} \}$, and $\{\bar{\mathbf{F}},\bar{\mathbf{G}} \}$. We begin with the compact form represented by $\{\bar{\mathbf{D}},\bar{\mathbf{I}} \}$. Thorough computational analysis reveals that the gradations resulting from these respective brackets are characterized by \begin{align}
    \{\bar{D}_2,\bar{I}_1\} = &\, (0,4,8) \tilde{+}(1,3,8) \tilde{+}(2,4,6) \tilde{+}(1,5,6)  ; \nonumber \\
   \{\bar{D}_3,\bar{I}_1\} = &\,  (0,4,8) \tilde{+}(1,5,6) \tilde{+}(0,6,6)  ; \nonumber \\
   \{\bar{D}_4,\bar{I}_1\} = &\,  (0,2,10) \tilde{+}(1,3,8) \tilde{+}(0,4,8)  ; \nonumber\\
   \{\bar{D}_2,\bar{I}_2\} = &\, (3,1,8) \tilde{+}(4,0,8) \tilde{+}(5,1,6) \tilde{+}(4,2,6)  ; \nonumber \\
   \{\bar{D}_3,\bar{I}_2\} = &\, (2,2,8) \tilde{+}(3,1,8) \tilde{+}(4,2,6) \tilde{+}(3,3,6)  ; \\
   \{\bar{D}_4,\bar{I}_2\} = &\, (2,0,10) \tilde{+}(4,0,8)\tilde{+}(3,1,8)  . \nonumber
\end{align} As an example, we will provide the allowed terms from the homogeneous grading of $\{\bar{D}_2,\bar{I}_2\}$. A direct calculation shows that \begin{align*}
    (3,1,8) = & \, \left\{\bar{b}_3 \bar{D}_2 \bar{F}_2,\bar{b}_1 \bar{b}_3 \bar{D}_2 \bar{D}_4,\bar{b}_1 \bar{b}_3^3 \bar{D}_2,C^{(003)} \bar{C}_2 \bar{F}_2,C^{(003)} C^{(202)} \bar{E}_1,C^{(202)} \bar{D}_2 \bar{D}_4,\bar{b}_1 \bar{b}_3 C^{(003)} \bar{E}_1,\bar{b}_1 \left(C^{(003)}\right)^2 \bar{D}_2,\bar{b}_1 C^{(003)} \bar{C}_2 \bar{D}_4, \right. \\
    &\, \left. \bar{b}_3^2 C^{(202)} \bar{D}_2,\bar{b}_3 C^{(003)} \bar{C}_2 C^{(202)},\bar{b}_1 \bar{b}_3^2 C^{(003)} \bar{C}_2\right\} ; \\ 
    (4,0,8)= & \, \left\{\bar{F}_2^2,\bar{b}_1 \bar{D}_4 \bar{F}_2,\bar{b}_1^2 \bar{b}_3^2 \bar{D}_4,\bar{b}_1 \bar{b}_3^2 \bar{F}_2,\bar{b}_1^2 \bar{b}_3^4, \bar{b}_1^2 \bar{D}_4^2,\bar{b}_1 \bar{b}_3 C^{(202)} \bar{D}_4,\bar{b}_1 \left(C^{(003)}\right)^2 C^{(202)},\bar{b}_3^2 \left(C^{(202)}\right)^2,\bar{b}_3 C^{(202)} \bar{F}_2, \right. \\
    &\, \left. \left(C^{(202)}\right)^2 \bar{D}_4,\bar{b}_1^2 \bar{b}_3 \left(C^{(003)}\right)^2,\bar{b}_1 \bar{b}_3^3 C^{(202)}\right\} ;\\
    (5,1,6) = & \, \left\{\bar{b}_1 \bar{D}_2 \bar{F}_2,\left(C^{(202)}\right)^2 \bar{D}_2,\bar{b}_1^2 C^{(003)} \bar{E}_1,\bar{b}_1^2 \bar{D}_2 \bar{D}_4,\bar{b}_1 \bar{b}_3 C^{(202)} \bar{D}_2,\bar{b}_1 C^{(003)} \bar{C}_2 C^{(202)},\bar{b}_1^2 \bar{b}_3^2 \bar{D}_2,\bar{b}_1^2 \bar{b}_3 C^{(003)} \bar{C}_2\right\} ; \\
    (4,2,6) = & \, \left\{\bar{F}_1^2\bar{b}_1 \bar{D}_3 \bar{F}_2,\bar{b}_1 \bar{E}_1^2\bar{C}_2^2 \bar{F}_2,\bar{b}_1^2 \bar{b}_3 \bar{F}_4,\bar{b}_1^2 \bar{D}_3 \bar{D}_4,\bar{b}_1 \bar{b}_2 \bar{b}_3 \bar{F}_2,\bar{b}_1^2 \bar{b}_2 \bar{b}_3 \bar{D}_4,\bar{b}_1 \bar{b}_3 \bar{C}_2 \bar{E}_1,\bar{b}_1 \bar{C}_2^2 \bar{D}_4,\bar{b}_1 \bar{b}_3 \bar{D}_2^2,\bar{b}_1 \bar{b}_3^2 \bar{C}_2^2,\bar{b}_1^2 \bar{b}_2 \bar{b}_3^3, \bar{b}_1^2 \bar{b}_3^2 \bar{D}_3,\right.\\
    &\,\left.  \bar{b}_2 C^{(202)} \bar{F}_2,,\bar{b}_1 \bar{b}_2 C^{(202)} \bar{D}_4,\bar{b}_1 C^{(003)} \bar{C}_2 \bar{D}_2,\bar{b}_1 \bar{b}_3 C^{(202)} \bar{D}_3,\bar{b}_2 \bar{b}_3 \left(C^{(202)}\right)^2,\bar{b}_3 \bar{C}_2^2 C^{(202)},\bar{b}_1^2 \bar{b}_2 \left(C^{(003)}\right)^2,\bar{b}_1 \bar{b}_2 \bar{b}_3^2 C^{(202)},\right. \\
    &\, \left.  \bar{b}_1 C^{(202)} \bar{F}_4, \bar{C}_2 C^{(202)} \bar{E}_1,\left(C^{(202)}\right)^2 \bar{D}_3,C^{(202)} \bar{D}_3^2\right\} . 
\end{align*}% The replacement of the generators will change the permissible generators above to the following list %\begin{align*}
After the usual change of the generators, we have that the bracket $\{\bar{D}_2,\bar{I}_2\}$ contains $149$ admissible terms, which means there are  coefficients $\Gamma_{49}^1 ,\ldots , \Gamma_{49}^{149}$ such that  
\begin{align*}
    \{\bar{D}_2,\bar{I}_2\} = \Gamma_{49}^1 \bar{b}_3 \bar{D}_2 \bar{F}_2 + \ldots + \Gamma_{49}^{148}   \bar{D}_2^2  \bar{D}_3+ \Gamma_{49}^{149}  \bar{D}_2     \bar{D}_3^2.
\end{align*} 
Determining these coefficients will lead to the explicit expansion of this brackets, which is given by 
\begin{align}
    \{\bar{D}_2, \bar{I}_2\}&=\frac{{\rm i}}{32}\bigl( 4(\bar{d}_1+\bar{D}_2-\bar{D}_3+2\bar{D}_4)\bigl( 2 \bar{d}_1 \bar{D}_2+5\bar{D}_2^2-2(4 \bar{d}_1 +5 \bar{D}_2)\bar{D}_3+8\bar{D}_3^2+4(3\bar{D}_2-4\bar{D}_3)\bar{D}_4 \nonumber \\
& \hskip 0.7 cm -16 \bar{b}_3 \bar{C}_2^2+12 \bar{c}_1 \bar{E}_1	\bigl)-192 \bar{b}_2^2 \bar{b}_3(\bar{b}_3(\bar{d}_1+\bar{D}_2-\bar{D}_3+2\bar{D}_4)-\bar{F}_2)-48(\bar{c}_1-2\bar{C}_2)\bar{C}_2\bar{F}_2 \nonumber \\
& \hskip 0.7 cm -32\bar{b}_3\bar{D}_2\bar{F}_2-48 \bar{b}_2\bigl((\bar{d}_1+\bar{D}_2-\bar{D}_3+2\bar{D}_4)\bar{F}_2-\bar{b}_3 \bigl(4 \bar{b}_3 \bar{C}_2^2+\bar{d}_1^2+2\bar{d}_1(\bar{D}_2+\bar{D}_3)+2\bar{D}_2 \bar{D}_3  \nonumber \\
& \hskip 0.7 cm-3\bar{D}_3^2 +4(\bar{d}_1+\bar{D}_2+\bar{D}_3)\bar{D}_4+4\bar{D}_4^2-8 \bar{C}_2 \bar{E}_1\bigl) \bigl)+4 \bar{b}_1\bigl( 18 \bar{b}_3^2 \bar{C}_2^2+24 \bar{b}_2^2(\bar{b}_3^3-\bar{b}_3 \bar{D}_4) \nonumber \\
& \hskip 0.7 cm+(\bar{c}_1-\bar{C}_2)\bar{C}_2(\bar{d}_1+5\bar{D}_2-\bar{D}_3)+2(3\bar{c}_1-8\bar{C}_2)\bar{C}_2\bar{D}_4-4\bar{E}_1^2+2\bar{b}_3\bigl( \bar{d}_1(\bar{D}_2+4\bar{D}_3)+3\bar{D}_2 \bar{D}_3  \nonumber \\
& \hskip 0.7 cm-4\bar{D}_3^2+2(\bar{D}_2+4\bar{D}_3) \bar{D}_4-2(\bar{c}_1+4\bar{C}_2)\bar{E}_1\bigl)-4(\bar{D}_2-2\bar{D}_3)\bar{F}_2+2\bar{b}_2 \bigl(4\bar{b}_3(\bar{F}_2+6\bar{F}_4)  \nonumber \\
& \hskip 0.7 cm+3 \bar{D}_4(\bar{d}_1+\bar{D}_2-\bar{D}_3+2\bar{D}_4)-\bar{b}_3^2(7(\bar{d}_1+\bar{D}_2)+17 \bar{D}_3+14 \bar{D}_4)\bigl)\bigl)\bigl) \, .
\end{align} 
 Repeating the same argument, the rest of the  degree-twelve expansions in the compact form $\{\bar{\textbf{D}},\bar{\textbf{I}}\}$ close as follows:
\begin{align}
\{\bar{D}_2, \bar{I}_1\}&=\frac{{\rm i}}{32}\bigl( \bar{b}_1 \bigl( \bar{b}_2^2(\bar{c}_1-\bar{C}_2)^2+24 \bar{b}_2\bar{D}_3(\bar{D}_4-\bar{b}_3^2)+48\bar{D}_3(\bar{b}_3 \bar{D}_3-\bar{F}_4)\bigl)+4 \bar{b}_2^2(\bar{c}_1-\bar{C}_2)(\bar{b}_3\bar{C}_2-2\bar{E}_1)  \nonumber \\
& \hskip 0.7 cm +4\bar{D}_3\bigl(3\bar{D}_2^2-16\bar{b}_3\bar{C}_2^2-8\bar{D}_3(\bar{d}_1-\bar{D}_3+2\bar{D}_4)+\bar{D}_2(8\bar{D}_4-6\bar{D}_3)+12 \bar{c}_1 \bar{E}_1 \bigl)-16 \bar{F}_4 (2 \bar{b}_3 \bar{D}_2  \nonumber \\
& \hskip 0.7 cm+3(\bar{c}_1-2 \bar{C}_2)\bar{C}_2)+4 \bar{b}_2 \bigl((6 \bar{b}_3^2-2\bar{D}_2-\bar{D}_3-10 \bar{D}_4) \bar{C}_2^2+(2\bar{D}_2+\bar{D}_3+4\bar{D}_4)\bar{c}_1 \bar{C}_2  \nonumber \\
& \hskip 0.7 cm+2\bar{b}_3 \bar{D}_3(4\bar{d}_1+5\bar{D}_2-4\bar{D}_3+8\bar{D}_4)-4\bar{b}_3\bar{c}_1\bar{E}_1+4\bar{E}_1^2-8\bar{D}_3\bar{F}_2-4\bar{D}_2\bar{F}_4 \bigl) \bigl) \, ,  \nonumber \\
\{\bar{D}_3, \bar{I}_1\}&= \frac{{\rm i}}{8} \bigl(2 \bar{D}_3^2(12 \bar{b}_3^2+\bar{D}_2-4\bar{D}_4)-\bar{b}_2 \bar{D}_3\bigl(8 \bar{b}_3^3+3\bar{c}_1^2-7 \bar{c}_1 \bar{C}_2+4\bar{C}_2^2-2\bar{b}_3(\bar{D}_2+4\bar{D}_4)\bigl)-64 \bar{b}_3 \bar{D}_3 \bar{F}_4  \nonumber \\
& \hskip 0.7 cm+\bar{b}_2^2(\bar{c}_1-\bar{C}_2)(\bar{b}_3 \bar{c}_1-\bar{E}_1)+4\bar{b}_2(4\bar{b}_3^2-\bar{D}_2-4\bar{D}_4)\bar{F}_4+48 \bar{F}_4^2 \bigl) \, ,  \nonumber \\
\{\bar{D}_4, \bar{I}_1\}&=\frac{{\rm i}}{32}\bigl(\bar{D}_3 \bigl( 12(\bar{d}_1+\bar{D}_2-\bar{D}_3)\bar{D}_3+48 \bar{b}_3^2 \bar{D}_3+12\bar{b}_3(\bar{C}_2^2-\bar{b}_1 \bar{D}_3)+8(\bar{D}_3-2\bar{D}_2)\bar{D}_4 \nonumber \\
& \hskip 0.7 cm-3(\bar{D}_2^2+8\bar{c}_1\bar{E}_1)\bigl)+2\bar{b}_2^2 \bar{b}_3(\bar{c}_1-\bar{C}_2)^2+4 \bar{F}_4\bigl(6(\bar{c}_1-\bar{C}_2)\bar{C}_2+4\bar{b}_3(\bar{D}_2-8\bar{D}_3)+3\bar{b}_1\bar{D}_3 \bigl)  \nonumber \\
& \hskip 0.7 cm +96 \bar{F}_4^2-2\bar{b}_2 \bigl( (8 \bar{b}_3^3+3\bar{c}_1^2-6 \bar{c}_1 \bar{C}_2+3\bar{C}_2^2)\bar{D}_3+4(\bar{c}_1 \bar{C}_2-\bar{C}_2^2)\bar{D}_4+3\bar{b}_1\bar{D}_3(\bar{D}_4-\bar{b}_3^2) \nonumber \\
& \hskip 0.7 cm +\bar{b}_3(6(\bar{d}_1+\bar{D}_2-\bar{D}_3)\bar{D}_3+4\bar{D}_3 \bar{D}_4-4(\bar{c}_1-\bar{C}_2)\bar{E}_1)-6 \bar{D}_3 \bar{F}_2-16 (\bar{b}_3^2-\bar{D}_4)\bar{F}_4	\bigl)\bigl) \, , \nonumber \\
\{\bar{D}_3, \bar{I}_2\}&=\frac{\rm i}{64}  \bigl(4 \bar{b}_1 \bigl(\bar{C}_2 \bigl(14 \bar{b}_3^2 \bar{C}_2+5 \bar{c}_1 \bar{D}_2-5 \bar{C}_2 \bar{D}_2-14 \bar{C}_2 \bar{D}_4\bigl)-8 \bar{E}_1 \bar{b}_3 \bar{C}_2+6 \bar{b}_2 \bar{D}_4 \bigl(\bar{d}_1+\bar{D}_2-\bar{D}_3+2 \bar{D}_4\bigl) \nonumber \\
& \hskip 0.7 cm-6 \bar{b}_3^2 \bar{b}_2 \bigl(\bar{d}_1+\bar{D}_2+7 \bar{D}_3+2 \bar{D}_4\bigl)+24 \bar{b}_2^2 \bigl(\bar{b}_3^3-\bar{b}_3 \bar{D}_4\bigl)+48 \bar{b}_3 \bar{b}_2 \bar{F}_4+4 \bar{D}_3 \bar{F}_2+4 \bar{E}_1^2\bigl) \nonumber \\
& \hskip 0.7 cm+3 \bar{b}_1^2 \bar{b}_2 \bigl(\bar{c}_1-\bar{C}_2\bigl)^2-4 \bigl(12 \bar{b}_2 \bigl(\bar{F}_2 \bigl(\bar{d}_1+\bar{D}_2-\bar{D}_3+2 \bar{D}_4\bigl)-\bar{b}_3 \bigl(4 \bar{b}_3 \bar{C}_2^2-8 \bar{E}_1 \bar{C}_2+2 \bar{d}_1 \bar{D}_2  \nonumber \\
& \hskip 0.7 cm+2 \bar{d}_1 \bar{D}_3+4 \bar{D}_4 \bigl(\bar{d}_1+\bar{D}_2+\bar{D}_3\bigl)+\bar{d}_1^2-3 \bar{D}_3^2+4 \bar{D}_4^2+2 \bar{D}_2 \bar{D}_3\bigl)\bigl)+\bigl(\bar{d}_1+\bar{D}_2-\bar{D}_3+2 \bar{D}_4\bigl)  \nonumber \\
& \hskip 0.7 cm \times	\bigl(20 \bar{b}_3 \bar{C}_2^2+4 \bar{D}_3 \bigl(\bar{d}_1-\bar{D}_3+2 \bar{D}_4\bigl)-3 \bar{D}_2^2+4 \bar{D}_2 \bar{D}_3\bigl)+48 \bar{b}_3 \bar{b}_2^2 \bigl(\bar{b}_3 \bigl(\bar{d}_1+\bar{D}_2-\bar{D}_3+2 \bar{D}_4\bigl)  \nonumber \\
& \hskip 0.7 cm-\bar{F}_2\bigl)-24 \bar{C}_2^2 \bar{F}_2\bigl)\bigl) \,  ,  \end{align} and
\begin{align}
\{\bar{D}_4, \bar{I}_2\}&=-\frac{{\rm i}}{128}  \bigl(\bar{b}_1^2 \bigl(3 \bar{b}_2-8 \bar{b}_3\bigl) \bigl(\bar{c}_1-\bar{C}_2\bigl)^2-4 \bigl(192 \bar{C}_2^2 \bar{b}_3^3+48\bar{b}_3^2 \bigl(\bar{d}_1^2+2 \bigl(\bar{D}_2+\bar{D}_3+2 \bar{D}_4\bigl) \bar{d}_1+4 \bar{D}_4^2  \nonumber \\
& \hskip 0.7 cm	+\bar{D}_3 \bigl(2 \bar{D}_2-3 \bar{D}_3\bigl)+4 \bigl(\bar{D}_2+\bar{D}_3\bigl) \bar{D}_4-8 \bar{C}_2 \bar{E}_1\bigl) +48 \bar{b}_2^2\bar{b}_3 \bigl(\bar{b}_3 \bigl(\bar{d}_1+\bar{D}_2-\bar{D}_3+2 \bar{D}_4\bigl)-\bar{F}_2\bigl)  \nonumber \\
& \hskip 0.7 cm +4 \bigl(3 \bar{C}_2^2 \bigl(\bar{d}_1+\bar{D}_2-\bar{D}_3+2 \bar{D}_4\bigl)-4 \bigl(8 \bar{d}_1+7 \bar{D}_2-8 \bar{D}_3+16 \bar{D}_4\bigl) \bar{F}_2\bigl) \bar{b}_3+96 \bar{F}_2^2  \nonumber \\
& \hskip 0.7 cm+\bigl(\bar{d}_1+\bar{D}_2-\bar{D}_3+2 \bar{D}_4\bigl) \bigl(-3 \bar{D}_2^2+4 \bigl(3 \bar{D}_3-8 \bar{D}_4\bigl) \bar{D}_2+4 \bigl(\bar{d}_1-\bar{D}_3+2 \bar{D}_4\bigl) \bigl(3 \bar{D}_3-4 \bar{D}_4\bigl) \nonumber \\
& \hskip 0.7 cm-24 \bar{c}_1 \bar{E}_1\bigl)+24 \bigl(\bar{c}_1-\bar{C}_2\bigl) \bar{C}_2 \bar{F}_2+12 \bar{b}_2 \bigl(\bigl(16 \bar{b}_3^2+\bar{d}_1+\bar{D}_2-\bar{D}_3+2 \bar{D}_4\bigl) \bar{F}_2-\bar{b}_3 \bigl(16 \bigl(\bar{d}_1+\bar{D}_2 \nonumber \\
& \hskip 0.7 cm-\bar{D}_3+2 \bar{D}_4\bigl) \bar{b}_3^2+4 \bar{C}_2^2 \bar{b}_3+\bar{d}_1^2-3 \bar{D}_3^2+4 \bar{D}_4^2+2 \bar{d}_1 \bar{D}_2+2 \bigl(\bar{d}_1+\bar{D}_2\bigl) \bar{D}_3+4 \bigl(\bar{d}_1+\bar{D}_2+\bar{D}_3\bigl) \bar{D}_4 \nonumber \\
& \hskip 0.7 cm-8 \bar{C}_2 \bar{E}_1\bigl)\bigl)\bigl)+4 \bar{b}_1 \bigl(16 \bigl(\bar{d}_1+\bar{D}_2+11 \bar{D}_3+2 \bar{D}_4\bigl) \bar{b}_3^3+2 \bar{b}_3^2\bigl(3 \bar{C}_2^2-16 \bigl(\bar{F}_2+6 \bar{F}_4\bigl)\bigl) +8 \bigl(\bigl(\bar{C}_2-\bar{c}_1\bigl) \bar{E}_1  \nonumber \\
& \hskip 0.7 cm-2 \bar{D}_4 \bigl(\bar{d}_1+\bar{D}_2-\bar{D}_3+2 \bar{D}_4\bigl)\bigl) \bar{b}_3-12 \bar{E}_1^2+\bigl(6 \bar{c}_1^2-12 \bar{C}_2 \bar{c}_1+6 \bar{C}_2^2\bigl) \bar{d}_1+\bigl(6 \bar{c}_1^2-9 \bar{C}_2 \bar{c}_1+3 \bar{C}_2^2\bigl) \bar{D}_2  \nonumber \\
& \hskip 0.7 cm+6 \bigl(-\bar{c}_1^2+2 \bar{C}_2 \bar{c}_1-\bar{C}_2^2\bigl) \bar{D}_3+\bigl(12 \bar{c}_1^2-16 \bar{C}_2 \bar{c}_1-2 \bar{C}_2^2\bigl) \bar{D}_4+24 \bar{b}_2^2 \bigl(\bar{b}_3^3-\bar{b}_3 \bar{D}_4\bigl)+4 \bigl(3 \bar{D}_3+8 \bar{D}_4\bigl) \bar{F}_2 \nonumber \\
& \hskip 0.7 cm-6 \bar{b}_2 \bigl(16 \bar{b}_3^4+\bigl(\bar{d}_1+\bar{D}_2+7 \bar{D}_3-14 \bar{D}_4\bigl) \bar{b}_3^2-8 \bar{F}_4 \bar{b}_3-\bar{D}_4 \bigl(\bar{d}_1+\bar{D}_2-\bar{D}_3+2 \bar{D}_4\bigl)\bigl)\bigl)\bigl) \, .
\end{align} 

Next, we consider the grading from the compact form $\{\bar{\textbf{F}},\bar{\textbf{G}}\}$. We first look at the gradings of $\{\bar{F}_s, \bar{G}_1\}$, for $s=1,2,3,4$. Certain gradings can be determined through a direct calculation as follows: \begin{align}
    \mathcal{G} \left(\{\bar{F}_1, \bar{G}_1\}\right)  = & \, (3,2,7) \tilde{+} (4,1,7)\tilde{+} (5,2,5)\tilde{+} (4,3,5); \nonumber\\
    \mathcal{G} \left(\{\bar{F}_2, \bar{G}_1\}\right)  = & \, (3,1,8) \tilde{+} (4,0,8)\tilde{+} (5,1,6)\tilde{+} (4,2,6); \nonumber \\
    \mathcal{G} \left(\{\bar{F}_3, \bar{G}_1\}\right)  = & \,(2,3,7) \tilde{+} (3,2,7)\tilde{+} (4,3,5)\tilde{+} (3,4,5); \nonumber \\
    \mathcal{G} \left(\{\bar{F}_4, \bar{G}_1\}\right)  = & \, (1,3,8) \tilde{+} (2,2,8)\tilde{+} (3,3,6)\tilde{+} (2,4,6). 
    \end{align}  
 All the predicted terms from the grading of $\{\bar{F}_1, \bar{G}_1\}$,  for example, are listed as 
 \begin{align*}
    (3,2,7) = &\, \left\{\bar{b}_1 \bar{b}_3^2 \bar{F}_3, \text{ } \bar{b}_3 \bar{D}_2 \bar{F}_1 , \text{ } \bar{D}_2 \bar{H}_1, \text{ } \bar{E}_1 \bar{G}_1, \text{ } \bar{F}_2 \bar{F}_3, \text{ } \bar{b}_1 \bar{b}_3 \bar{H}_2, \text{ } \bar{b}_1 \bar{D}_4 \bar{F}_3, \text{ } \bar{b}_3 \bar{C}_2 \bar{G}_1, \text{ } \bar{b}_1 C^{(003)} \bar{G}_2,  \right. \\
    & \, \left. \bar{b}_3 C^{(202)} \bar{F}_3, \text{ } C^{(003)} \bar{C}_2 \bar{F}_1, \text{ } C^{(202)} \bar{H}_2\right\} ; \\
    (4,1,7)= &\, \left\{ \bar{C}_2 \bar{I}_2, \text{ } \bar{F}_1 \bar{F}_2, \text{ } \bar{b}_1 \bar{b}_3^2 \bar{F}_1, \text{ } \bar{b}_1 \bar{D}_4 \bar{F}_1, \text{ } \bar{b}_1 \bar{b}_3 \bar{H}_1, \text{ } \bar{b}_1 C^{(003)} \bar{G}_1, \text{ } \bar{b}_3 C^{(202)} \bar{F}_1, \text{ } C^{(202)} \bar{H}_1\right\} ; \\
    (5,2,5)= &\, \left\{ \bar{b}_1^2 \bar{H}_2, \text{ } \bar{b}_1 \bar{D}_2 \bar{F}_1, \text{ } \bar{b}_1^2 \bar{b}_3 \bar{F}_3, \text{ } \bar{b}_1 \bar{C}_2 \bar{G}_1, \text{ } \bar{b}_1 C^{(202)} \bar{F}_3\right\} ; \\
    (4,3,5)= &\, \left\{\bar{b}_1 \bar{b}_2 \bar{H}_1, \text{ } \bar{b}_1 \bar{C}_2 \bar{G}_2, \text{ } \bar{b}_1 \bar{D}_2 \bar{F}_3, \text{ } \bar{b}_1 \bar{b}_2 \bar{b}_3 \bar{F}_1, \text{ } \bar{b}_1 \bar{D}_3 \bar{F}_1, \text{ } \bar{C}_2^2 \bar{F}_1, \text{ } \bar{b}_2 C^{(202)} \bar{F}_1 \right\}.
\end{align*} 
Taking into account the replacement of the generators, the set of permissible monomials reads  \begin{align*}
        & \,     \bar{d}_1 \bar{H}_2, \text{ } \bar{D}_2 \bar{H}_2, \text{ } \bar{D}_3 \bar{H}_2, \text{ } \bar{D}_4 \bar{H}_2, \text{ } \bar{D}_2 \bar{H}_1, \text{ } \bar{E}_1 \bar{G}_1, \text{ } \bar{F}_2 \bar{F}_3, \text{ } \bar{b}_1 \bar{b}_3 \bar{H}_2, \text{ } \bar{b}_1 \bar{c}_1 \bar{G}_2, \text{ } \bar{b}_1 \bar{C}_2 \bar{G}_2, \text{ } \bar{b}_1 \bar{D}_4 \bar{F}_3, \text{ } \bar{b}_3 \bar{C}_2 \bar{G}_1, \text{ } \bar{b}_3 \bar{d}_1 \bar{F}_3, \text{ } \bar{b}_3 \bar{D}_2 \bar{F}_3, \\
        & \,   \bar{b}_3 \bar{D}_3 \bar{F}_3, \text{ } \bar{b}_3 \bar{D}_4 \bar{F}_3, \text{ } \bar{b}_3 \bar{D}_2 \bar{F}_1, \text{ } \bar{c}_1 \bar{C}_2 \bar{F}_1, \text{ } \bar{C}_2^2 \bar{F}_1, \text{ } \bar{b}_1 \bar{b}_3^2 \bar{F}_3, \text{ } \bar{C}_2 \bar{I}_2, \text{ } \bar{d}_1 \bar{H}_1, \text{ } \bar{D}_3 \bar{H}_1, \text{ } \bar{D}_4 \bar{H}_1, \text{ } \bar{F}_1 \bar{F}_2, \text{ } \bar{b}_1 \bar{b}_3 \bar{H}_1, \text{ } \bar{b}_1 \bar{c}_1 \bar{G}_1, \text{ } \bar{b}_1 \bar{C}_2 \bar{G}_1, \\
        & \, \bar{b}_1 \bar{D}_4 \bar{F}_1, \text{ } \bar{b}_3 \bar{d}_1 \bar{F}_1, \text{ } \bar{b}_3 \bar{D}_3 \bar{F}_1, \text{ } \bar{b}_3 \bar{D}_4 \bar{F}_1, \text{ } \bar{b}_1 \bar{b}_3^2 \bar{F}_1, \text{ } \bar{b}_1^2 \bar{H}_2, \text{ } \bar{b}_1 \bar{d}_1 \bar{F}_3, \text{ } \bar{b}_1 \bar{D}_2 \bar{F}_3, \text{ } \bar{b}_1 \bar{D}_3 \bar{F}_3, \text{ } \bar{b}_1 \bar{D}_2 \bar{F}_1, \text{ } \bar{b}_1^2 \bar{b}_3 \bar{F}_3, \\
        & \,\bar{b}_1 \bar{b}_2 \bar{H}_1, \text{ } \bar{b}_1 \bar{D}_3 \bar{F}_1, \text{ } \bar{b}_2 \bar{d}_1 \bar{F}_1, \text{ } \bar{b}_2 \bar{D}_2 \bar{F}_1, \text{ } \bar{b}_2 \bar{D}_3 \bar{F}_1, \text{ } \bar{b}_2 \bar{D}_4 \bar{F}_1, \text{ } \bar{b}_1 \bar{b}_2 \bar{b}_3 \bar{F}_1 .   
\end{align*} This indicates the existence of $49$ coefficients $\Gamma_{67}^{1},\ldots,\Gamma_{67}^{49}$ such that $ \{\bar{F}_1,\bar{G}_1\} = \Gamma_{67}^1 \bar{d}_1\bar{H}_2 + \ldots +  \Gamma_{67}^{49}  \bar{b}_1 \bar{b}_2 \bar{b}_3 \bar{F}_1  $. In detail, by determining these coefficients, we find that 
\begin{align}
    \{\bar{F}_1, \bar{G}_1\}&=-\frac{\rm i}{8}  \bigl(-4 \bar{b}_2 \bigl(\bar{F}_1 \bigl(4 \bar{b}_3^2+\bar{d}_1+\bar{D}_2-\bar{D}_3+6 \bar{D}_4\bigl)-8 \bar{b}_3 \bar{H}_1+2 \bar{G}_1 \bigl(\bar{c}_1-\bar{C}_2\bigl)\bigl)+\bar{b}_1 \bigl(\bar{F}_1 \bigl(-2 \bar{b}_2 \bar{b}_3+\bar{D}_2+2 \bar{D}_3\bigl) \nonumber \\
&\hskip 0.7cm+8 \bar{F}_3 \bigl(\bar{b}_3^2+2 \bar{D}_4\bigl)+4 \bar{b}_2 \bar{H}_1-24 \bar{b}_3 \bar{H}_2+2 \bar{c}_1 \bigl(\bar{G}_1+3 \bar{G}_2\bigl)-2 \bar{C}_2 \bigl(\bar{G}_1+4 \bar{G}_2\bigl)\bigl)+16 \bar{b}_3 \bar{C}_2 \bigl(\bar{G}_1-2 \bar{G}_2\bigl)  \nonumber \\
&	\hskip 0.7cm+2 \bar{b}_1^2 \bigl(\bar{b}_3 \bar{F}_3-\bar{H}_2\bigl)+\bar{F}_1 \bigl(-4 \bar{c}_1 \bar{C}_2+6 \bar{C}_2^2+8 \bar{F}_2\bigl)+8 \bar{F}_3 \bigl(\bar{c}_1 \bar{C}_2-\bar{C}_2^2-4 \bar{F}_2\bigl)-8 \bigl(\bar{H}_1-2 \bar{H}_2\bigl) \nonumber \\
&	\hskip 0.7cm \times	\bigl(\bar{d}_1-\bar{D}_3+2 \bar{D}_4\bigl)+32 \bar{F}_1 \bar{F}_4\bigl) \, .
\end{align}  Then, after some heavy computations, we find that the rest of the expansions are given by 
\begin{align}
\{\bar{F}_2, \bar{G}_1\}&=\frac{{\rm i}}{64}  \bigl(4 \bar{b}_1 \bigl(2 \bar{b}_3^2 \bar{C}_2 \bigl(17 \bar{C}_2-2 \bar{c}_1\bigl)+2 \bar{b}_3 \bigl(-24 \bar{E}_1 \bar{C}_2+8 \bar{D}_3 \bigl(\bar{d}_1-\bar{D}_3+2 \bar{D}_4\bigl)-3 \bar{D}_2^2+8 \bar{D}_2 \bar{D}_3\bigl)\nonumber\\
&	\hskip 0.7cm -2 \bar{b}_2 \bigl(\bar{b}_3^2 \bigl(11 \bar{d}_1+11 \bar{D}_2+21 \bar{D}_3+22 \bar{D}_4\bigl)-8 \bar{b}_3 \bigl(\bar{F}_2+4 \bar{F}_4\bigl)+\bar{D}_4 \bigl(\bar{d}_1+\bar{D}_2-\bar{D}_3+2 \bar{D}_4\bigl)\bigl)\nonumber\\
&	\hskip 0.7cm +32 \bar{b}_2^2 \bigl(\bar{b}_3^3-\bar{b}_3 \bar{D}_4\bigl)+2 \bar{c}_1 \bar{C}_2 \bar{d}_1+4 \bar{c}_1 \bar{C}_2 \bar{D}_2-2 \bar{c}_1 \bar{C}_2 \bar{D}_3+8 \bar{c}_1 \bar{C}_2 \bar{D}_4+\bar{c}_1^2 \bar{D}_2-2 \bar{C}_2^2 \bar{d}_1-5 \bar{C}_2^2 \bar{D}_2\nonumber\\
&	\hskip 0.7cm +2 \bar{C}_2^2 \bar{D}_3-14 \bar{C}_2^2 \bar{D}_4+8 \bar{F}_4 \bigl(\bar{d}_1+\bar{D}_2-\bar{D}_3+2 \bar{D}_4\bigl)+4 \bar{D}_3 \bar{F}_2+4 \bar{E}_1^2\bigl)+4 \bigl(\bigl(\bar{d}_1+\bar{D}_2-\bar{D}_3+2 \bar{D}_4\bigl)\nonumber\\
&	\hskip 0.7cm \times \bigl(-20 \bar{b}_3 \bar{C}_2^2-16 \bar{b}_3^2 \bar{D}_2+8 \bar{E}_1 \bigl(2 \bar{c}_1+\bar{C}_2\bigl)+4 \bar{d}_1 \bigl(\bar{D}_2-3 \bar{D}_3\bigl)+7 \bar{D}_2^2+12 \bar{D}_3^2-16 \bar{D}_2 \bar{D}_3\nonumber\\
&	\hskip 0.7cm +24 \bigl(\bar{D}_2-\bar{D}_3\bigl) \bar{D}_4\bigl)+8 \bar{F}_2 \bigl(2 \bar{b}_3 \bar{D}_2-\bar{c}_1 \bar{C}_2+\bar{C}_2^2\bigl)+4 \bar{b}_2 \bigl(\bar{b}_3 \bigl(16 \bar{C}_2 \bigl(\bar{b}_3 \bar{C}_2-2 \bar{E}_1\bigl)+10 \bar{d}_1 \bar{D}_2+6 \bar{d}_1 \bar{D}_3\nonumber\\
&	\hskip 0.7cm +4 \bar{D}_4 \bigl(5 \bigl(\bar{d}_1+\bar{D}_2\bigl)+3 \bar{D}_3\bigl)+5 \bar{d}_1^2+\bar{D}_2^2-11 \bar{D}_3^2+20 \bar{D}_4^2+6 \bar{D}_2 \bar{D}_3\bigl)-3 \bar{F}_2 \bigl(\bar{d}_1+\bar{D}_2-\bar{D}_3+2 \bar{D}_4\bigl)\bigl)\nonumber\\
&	\hskip 0.7cm+64 \bar{b}_3 \bar{b}_2^2 \bigl(\bar{F}_2-\bar{b}_3 \bigl(\bar{d}_1+\bar{D}_2-\bar{D}_3+2 \bar{D}_4\bigl)\bigl)+8 \bar{F}_1^2\bigl)+\bar{b}_1^2 \bigl(\bar{b}_2 \bigl(-48 \bar{b}_3 \bar{D}_4+48 \bar{b}_3^3+\bigl(\bar{c}_1-\bar{C}_2\bigl)^2\bigl)\nonumber\\
&	\hskip 0.7cm+96 \bar{b}_3 \bigl(\bar{F}_4-\bar{b}_3 \bar{D}_3\bigl)\bigl)\bigl)\, ,\nonumber \\
\{\bar{F}_3, \bar{G}_1\}&=\frac{\rm i}{4} \bigl(-4 \bar{b}_3 \bar{C}_2 \bar{G}_1+\bar{b}_2 \bar{d}_1 \bar{F}_1-2 \bar{b}_2 \bar{d}_1 \bar{F}_3+4 \bar{b}_3 \bar{d}_1 \bar{F}_3+2 \bar{b}_2 \bar{D}_2 \bar{F}_1-\bar{b}_2 \bar{D}_3 \bar{F}_1+2 \bar{b}_2 \bar{D}_4 \bar{F}_1-2 \bar{b}_2 \bar{D}_2 \bar{F}_3+6 \bar{b}_3 \bar{D}_2 \bar{F}_3 \nonumber\\
&	\hskip 0.7cm-\bar{b}_1 \bar{D}_3 \bar{F}_3+2 \bar{b}_2 \bar{D}_3 \bar{F}_3-4 \bar{b}_3 \bar{D}_3 \bar{F}_3-4 \bar{b}_2 \bar{D}_4 \bar{F}_3+8 \bar{b}_3 \bar{D}_4 \bar{F}_3+\bar{b}_1 \bar{b}_2 \bar{b}_3 \bar{F}_3+\bar{c}_1 \bar{C}_2 \bigl(\bar{F}_1-\bar{F}_3\bigl)+\bar{C}_2^2 \bigl(3 \bar{F}_3-2 \bar{F}_1\bigl) \nonumber\\
&	\hskip 0.7cm-4 \bar{H}_2 \bigl(\bar{d}_1+2 \bar{D}_2-\bar{D}_3+2 \bar{D}_4\bigl)-4 \bar{F}_2 \bar{F}_3+8 \bar{E}_1 \bar{G}_1\bigl) \, ,\nonumber\\
\{\bar{F}_4, \bar{G}_1\}&=\frac{\rm i}{8}  \bigl(\bar{b}_1 \bar{c}_1 \bar{C}_2 \bar{D}_3+\bar{b}_2 \bigl(\bar{C}_2 \bigl(\bar{C}_2 \bigl(4 \bar{b}_3^2-\bar{D}_2-4 \bar{D}_4\bigl)+\bar{c}_1 \bar{D}_2\bigl)+\bar{E}_1 \bigl(\bar{b}_1 \bar{c}_1-\bigl(\bar{b}_1+4 \bar{b}_3\bigl) \bar{C}_2\bigl)+4 \bar{E}_1^2\bigl)\nonumber\\
&	\hskip 0.7cm+2 \bar{b}_3 \bigl(\bar{c}_1 \bar{C}_2 \bar{D}_2-\bigl(\bar{C}_2^2 \bigl(\bar{D}_2+2 \bar{D}_3+4 \bar{D}_4\bigl)\bigl)+8 \bar{E}_1^2\bigl)-\bar{b}_1 \bar{C}_2^2 \bar{D}_3+8 \bar{b}_3^3 \bar{C}_2^2-16 \bar{E}_1 \bar{b}_3^2 \bar{C}_2+2 \bar{c}_1 \bar{C}_2^3\nonumber\\
&	\hskip 0.7cm-\bar{c}_1^2 \bar{C}_2^2-2 \bar{E}_1 \bar{c}_1 \bar{D}_2+6 \bar{E}_1 \bar{C}_2 \bar{D}_2+4 \bar{C}_2^2 \bar{F}_4-\bar{C}_2^4+4 \bigl(\bar{F}_1-\bar{F}_3\bigl) \bar{F}_3\bigl) \, .
\end{align} 

Moreover, using Proposition \ref{grading} again, the gradings for the brackets $\{\bar{F}_s,\bar{G}_2\}$, for $s=1,2,3,4$, are given by 
 \begin{align}
    \mathcal{G} \left(\{\bar{F}_1, \bar{G}_2\}\right)  = & \, (2,3,7) \tilde{+} (3,2,7)\tilde{+} (4,3,5)\tilde{+} (3,4,5); \nonumber \\
    \mathcal{G} \left(\{\bar{F}_2, \bar{G}_2\}\right)  = & \, (2,2,8) \tilde{+} (3,1,8)\tilde{+} (4,2,6)\tilde{+} (3,3,6) ; \nonumber  \\
    \mathcal{G} \left(\{\bar{F}_3, \bar{G}_2\}\right)  = & \, (1,4,7) \tilde{+} ((2,3,7)\tilde{+} (3,4,5)\tilde{+} (2,5,5); \nonumber \\
    \mathcal{G} \left(\{\bar{F}_4, \bar{G}_2\}\right)  = & \, (0,4,8) \tilde{+} (1,3,8)\tilde{+} (2,4,6)\tilde{+} (1,5,6).
\end{align}  
In this case,  as an example, we deduce all the permissible polynomials from each homogeneous grading of $\{\bar{F}_4,\bar{G}_2\} $ as follows: \begin{align*}
    (0,4,8) = & \, \left\{\bar{F}_4^2, \text{ } \bar{b}_2 \bar{D}_4 \bar{F}_4, \text{ } \bar{b}_2^2 \bar{b}_3^4, \text{ } \bar{b}_2^2 \bar{b}_3^2 \bar{D}_4, \text{ } \bar{b}_3^2 \bar{D}_3^2, \text{ } \bar{b}_3 \bar{D}_3 \bar{F}_4, \text{ } \bar{b}_2 \bar{b}_3^3 \bar{D}_3, \text{ } \bar{D}_3^2 \bar{D}_4, \text{ } \bar{b}_2^2 \bar{D}_4^2, \text{ } \bar{b}_2 \bar{b}_3^2 \bar{F}_4, \text{ } \bar{b}_2 \bar{b}_3 \bar{D}_3 \bar{D}_4, \text{ } \bar{b}_2 \left(C^{(003)}\right)^2 \bar{D}_3, \right. \\
    & \, \left. \bar{b}_2^2 \bar{b}_3 \left(C^{(003)}\right)^2\right\}; \\
    (1,3,8) = & \, \left\{\bar{b}_3 \bar{D}_2 \bar{F}_4, \text{ } \bar{b}_3^2 \bar{D}_2 \bar{D}_3, \text{ } \bar{b}_2 \bar{b}_3 \bar{D}_2 \bar{D}_4, \text{ } \bar{D}_2 \bar{D}_3 \bar{D}_4, \text{ } \bar{b}_2 \bar{b}_3^3 \bar{D}_2, \text{ } C^{(003)} \bar{C}_2 \bar{F}_4, \text{ } C^{(003)} \bar{D}_3 \bar{E}_1, \text{ } \bar{b}_2 \bar{b}_3 C^{(003)} \bar{E}_1, \text{ } \bar{b}_2 \left(C^{(003)}\right)^2 \bar{D}_2,\right.\\
    & \, \left.\bar{b}_2 C^{(003)} \bar{C}_2 \bar{D}_4, \text{ } \bar{b}_3 C^{(003)} \bar{C}_2 \bar{D}_3, \text{ } \bar{b}_2 \bar{b}_3^2 C^{(003)} \bar{C}_2\right\}; \\ 
    (2,4,6) = & \, \left\{\bar{F}_3^2, \text{ } \bar{b}_1 \bar{D}_3 \bar{F}_4, \text{ } \bar{b}_2 \bar{D}_3 \bar{F}_2, \text{ } \bar{b}_2 \bar{E}_1^2, \text{ } \bar{C}_2^2 \bar{F}_4, \text{ } \bar{C}_2 \bar{D}_3 \bar{E}_1, \text{ } \bar{D}_2^2 \bar{D}_3, \text{ } \bar{b}_1 \bar{b}_2 \bar{b}_3 \bar{F}_4, \text{ } \bar{b}_1 \bar{b}_2 \bar{D}_3 \bar{D}_4, \text{ } \bar{b}_1 \bar{b}_3 \bar{D}_3^2, \text{ } \bar{b}_2 \bar{b}_3^2 \bar{C}_2^2, \text{ } \bar{b}_1 \bar{b}_2^2 \bar{b}_3^3, \text{ } \bar{b}_2 \bar{b}_3 \bar{C}_2 \bar{E}_1,\right.\\
    & \, \left. \bar{b}_2^2 \bar{b}_3 \bar{F}_2, \text{ } \bar{b}_2 \bar{b}_3 \bar{D}_2^2, \text{ } \bar{b}_1 \bar{b}_2 \bar{b}_3^2 \bar{D}_3, \text{ }  \bar{b}_2 \bar{C}_2^2 \bar{D}_4, \text{ } \bar{b}_3 \bar{C}_2^2 \bar{D}_3, \text{ } \bar{b}_1 \bar{b}_2^2 \bar{b}_3 \bar{D}_4, \text{ } \bar{b}_2 C^{(202)} \bar{F}_4\bar{b}_2^2 C^{(202)} \bar{D}_4, \text{ } \bar{b}_2 \bar{b}_3 C^{(202)} \bar{D}_3, \text{ }   \bar{b}_2 C^{(003)} \bar{C}_2 \bar{D}_2 ,\right. \\
    & \, \left. \bar{b}_1 \bar{b}_2^2\left(C^{(003)}\right)^2, \text{ } C^{(202)} \bar{D}_3^2, \text{ } \bar{b}_2^2 \bar{b}_3^2 C^{(202)}\right\}; \\ 
    (1,5,6)  = & \, \left\{\bar{b}_2 \bar{D}_2 \bar{F}_4, \text{ } \bar{b}_2^2 \bar{b}_3^2 \bar{D}_2, \text{ } \bar{D}_2 \bar{D}_3^2, \text{ } \bar{b}_2^2 \bar{D}_2 \bar{D}_4, \text{ } \bar{b}_2 \bar{b}_3 \bar{D}_2 \bar{D}_3, \text{ } \bar{b}_2^2 C^{(003)} \bar{E}_1, \text{ } \bar{b}_2 C^{(003)} \bar{C}_2 \bar{D}_3, \text{ } \bar{b}_2^2 \bar{b}_3 C^{(003)} \bar{C}_2\right\}  .
\end{align*} 
After the change of the generators, $73$ distinct allowed polynomials contained in the expansion of $  \{\bar{F}_4,\bar{G}_2\}$ arise, which are expanded in terms of the coefficients $\Gamma_{67}^1,\ldots,\Gamma_{67}^{73}.$ In other words, \begin{align*}
    \{\bar{F}_4,\bar{G}_2\} = \Gamma_{67}^1 \bar{F}_4^2 + \ldots +\Gamma_{67}^{72} \bar{b}_2\bar{c}_1\bar{C}_2\bar{D}_3 +  \Gamma_{67}^{73} \bar{b}_2^2 \bar{b}_3 \bar{c}_1 \bar{C}_2.
\end{align*} %\begin{align*}
After determining all the coefficients above, we deduce that 
\begin{align}
    	\{\bar{F}_4, \bar{G}_2\}&=\frac{{\rm i}}{16}  \bigl(\bar{b}_2 \bigl(4 \bar{b}_3^2 \bar{C}_2 \bigl(4 \bar{C}_2-\bar{c}_1\bigl)-2 \bar{b}_3 \bigl(8 \bar{E}_1 \bar{C}_2+8 \bar{D}_3 \bigl(-\bar{d}_1+\bar{D}_3-2 \bar{D}_4\bigl)+\bar{D}_2^2-8 \bar{D}_3 \bar{D}_2\bigl)+\bar{c}_1 \bar{C}_2 \bigl(\bar{D}_2+2 \bar{D}_3  \nonumber\\
	& \hskip 0.7cm +4 \bar{D}_4\bigl)+\bar{c}_1^2 \bar{D}_2-2 \bar{C}_2^2 \bigl(\bar{D}_2+\bar{D}_3+4 \bar{D}_4\bigl)-8 \bar{F}_4 \bigl(\bar{d}_1+\bar{D}_2-\bar{D}_3+2 \bar{D}_4\bigl)-8 \bar{D}_3 \bar{F}_2+16 \bar{E}_1^2\bigl) \nonumber\\
	& \hskip 0.7cm+2 \bar{D}_3 \bigl(-16 \bar{b}_3 \bar{C}_2^2-8 \bar{b}_3^2 \bar{D}_2+8 \bar{E}_1 \bigl(\bar{c}_1+\bar{C}_2\bigl)+8 \bar{D}_3 \bigl(-\bar{d}_1+\bar{D}_3-2 \bar{D}_4\bigl)+3 \bar{D}_2^2+\bar{D}_2 \bigl(8 \bar{D}_4-6 \bar{D}_3\bigl)\bigl) \nonumber\\
	& \hskip 0.7cm+8 \bar{F}_4 \bigl(2 \bigl(\bar{b}_3 \bar{D}_2+\bar{C}_2^2\bigl)-\bar{c}_1 \bar{C}_2\bigl)+8 \bar{b}_3 \bar{b}_2^2 \bigl(\bar{F}_2-\bar{b}_3 \bigl(\bar{d}_1+\bar{D}_2-\bar{D}_3+2 \bar{D}_4\bigl)\bigl)+4 \bar{b}_1 \bigl(2 \bar{D}_3 \bigl(4 \bar{b}_3 \bar{D}_3-3 \bar{F}_4\bigl) \nonumber\\
	& \hskip 0.7cm+\bar{b}_2^2 \bigl(\bar{b}_3^3-\bar{b}_3 \bar{D}_4\bigl)+\bar{b}_2 \bar{D}_3 \bigl(\bar{D}_4-3 \bar{b}_3^2\bigl)\bigl)+16 \bar{F}_3^2\bigl) \, . 
\end{align}

Using the same strategy, the remaining Poisson brackets have the following structures
\begin{align}
\{\bar{F}_1, \bar{G}_2\}&=\frac{\rm i}{8}  \bigl(-2 \bigl(\bar{b}_2 \bigl(8 \bar{b}_3^2 \bar{F}_1-16 \bar{b}_3 \bar{H}_1+2 \bar{G}_1 \bigl(2 \bar{c}_1-\bar{C}_2\bigl)+\bar{F}_1 \bigl(\bar{d}_1+2 \bar{D}_2-\bar{D}_3+10 \bar{D}_4\bigl)-2 \bar{F}_3 \bigl(\bar{d}_1+\bar{D}_2-\bar{D}_3+2 \bar{D}_4\bigl)\bigl)\nonumber\\
&	\hskip 0.7cm	-12 \bar{b}_3 \bar{C}_2 \bar{G}_1+20 \bar{b}_3 \bar{C}_2 \bar{G}_2+2 \bar{c}_1 \bar{C}_2 \bigl(2 \bar{F}_1-3 \bar{F}_3\bigl)-6 \bar{C}_2^2 \bar{F}_1+7 \bar{C}_2^2 \bar{F}_3-8 \bar{C}_2 \bar{I}_2+4 \bar{H}_1 \bigl(2 \bar{d}_1+\bar{D}_2-2 \bar{D}_3+4 \bar{D}_4\bigl)\nonumber\\
&	\hskip 0.7cm+4 \bar{H}_2 \bigl(-\bar{d}_1+\bar{D}_2+\bar{D}_3-2 \bar{D}_4\bigl)-8 \bar{F}_1 \bar{F}_2+12 \bar{F}_2 \bar{F}_3-16 \bar{F}_1 \bar{F}_4+4 \bar{E}_1 \bar{G}_1-4 \bar{E}_1 \bar{G}_2\bigl)+\bar{b}_1 \bigl(\bar{b}_2 \bigl(4 \bigl(2 \bar{H}_1+\bar{H}_2\bigl)\nonumber\\
&	\hskip 0.7cm-2 \bar{b}_3 \bigl(3 \bar{F}_1+2 \bar{F}_3\bigl)\bigl)+8 \bar{b}_3^2 \bar{F}_3-24 \bar{b}_3 \bar{H}_2+4 \bar{c}_1 \bigl(\bar{G}_1+2 \bar{G}_2\bigl)-2 \bar{C}_2 \bigl(\bar{G}_1+8 \bar{G}_2\bigl)-2 \bar{d}_1 \bar{F}_3+\bigl(\bar{D}_2+4 \bar{D}_3\bigl) \bar{F}_1\nonumber\\
&	\hskip 0.7cm+4 \bigl(-\bar{D}_2+\bar{D}_3+3 \bar{D}_4\bigl) \bar{F}_3\bigl)+2 \bar{b}_1^2 \bigl(\bar{b}_3 \bar{F}_3-\bar{H}_2\bigl)\bigl) \, ,\nonumber\\
\{\bar{F}_2, \bar{G}_2\}&=\frac{{\rm i}}{64}  \bigl(-4 \bigl(2 \bar{b}_2 \bigl(-6 \bar{b}_3 \bigl(-8 \bar{E}_1 \bar{C}_2+2 \bar{d}_1 \bigl(\bar{D}_2+\bar{D}_3+2 \bar{D}_4\bigl)+\bar{d}_1^2+4 \bar{D}_4^2+\bar{D}_3 \bigl(2 \bar{D}_2-3 \bar{D}_3\bigl)+4 \bigl(\bar{D}_2+\bar{D}_3\bigl) \bar{D}_4\bigl)\nonumber\\
& \hskip 0.7cm-24 \bar{b}_3^2 \bar{C}_2^2-\bigl(\bar{d}_1+\bar{D}_2-\bar{D}_3+2 \bar{D}_4\bigl) \bigl(\bar{C}_2 \bigl(\bar{c}_1-\bar{C}_2\bigl)-6 \bar{F}_2\bigl)\bigl)+4 \bar{b}_3 \bigl(-\bar{c}_1 \bar{C}_2 \bar{D}_2+\bar{C}_2^2 \bigl(5 \bar{d}_1+6 \bar{D}_2-5 \bar{D}_3\nonumber\\
& \hskip 0.7cm	+14 \bar{D}_4\bigl)-8 \bar{E}_1^2\bigl)-16 \bar{b}_3^3 \bar{C}_2^2+32 \bar{E}_1 \bar{b}_3^2 \bar{C}_2+48 \bar{b}_2^2 \bar{b}_3 \bigl(\bar{b}_3 \bigl(\bar{d}_1+\bar{D}_2-\bar{D}_3+2 \bar{D}_4\bigl)-\bar{F}_2\bigl)+2 \bigl(-2 \bar{c}_1 \bigl(\bar{C}_2^3-\bar{E}_1 \bar{D}_2\bigl)\nonumber\\
& \hskip 0.7cm+\bar{c}_1^2 \bar{C}_2^2+\bar{C}_2 \bigl(-8 \bar{C}_2 \bar{F}_2+\bar{C}_2^3-2 \bar{E}_1 \bigl(2 \bar{d}_1+5 \bar{D}_2-2 \bar{D}_3+4 \bar{D}_4\bigl)\bigl)\bigl)+\bigl(\bar{d}_1+\bar{D}_2-\bar{D}_3+2 \bar{D}_4\bigl)\nonumber\\
& \hskip 0.7cm \times \bigl(4 \bar{D}_3 \bigl(\bar{d}_1-\bar{D}_3+2 \bar{D}_4\bigl)-3 \bar{D}_2^2+4 \bar{D}_3 \bar{D}_2\bigl)-8 \bar{F}_1 \bar{F}_3\bigl)+4 \bar{b}_1 \bigl(-2 \bar{b}_2 \bigl(3 \bar{b}_3^2 \bigl(\bar{d}_1+\bar{D}_2+7 \bar{D}_3+2 \bar{D}_4\bigl)\nonumber\\
& \hskip 0.7cm-24 \bar{b}_3 \bar{F}_4+\bar{E}_1 \bigl(\bar{C}_2-\bar{c}_1\bigl)-3 \bar{D}_4 \bigl(\bar{d}_1+\bar{D}_2-\bar{D}_3+2 \bar{D}_4\bigl)\bigl)+\bar{C}_2 \bigl(\bar{C}_2 \bigl(14 \bar{b}_3^2-5 \bar{D}_2-14 \bar{D}_4\bigl)+5 \bar{c}_1 \bar{D}_2\bigl)\nonumber\\
& \hskip 0.7cm-8 \bar{E}_1 \bar{b}_3 \bar{C}_2+24 \bar{b}_2^2 \bigl(\bar{b}_3^3-\bar{b}_3 \bar{D}_4\bigl)+4 \bar{D}_3 \bar{F}_2+4 \bar{E}_1^2\bigl)+3 \bar{b}_1^2 \bar{b}_2 \bigl(\bar{c}_1-\bar{C}_2\bigl)^2\bigl) \,, \nonumber \\
\{\bar{F}_3, \bar{G}_2\}&=\frac{{\rm i}}{8}  \bigl(\bar{b}_2 \bigl(4 \bar{b}_3^2 \bar{F}_1-4 \bar{b}_1 \bar{b}_3 \bar{F}_3+2 \bar{b}_1 \bar{H}_2+2 \bar{c}_1 \bar{G}_2-4 \bar{C}_2 \bar{G}_1+4 \bar{d}_1 \bar{F}_3-\bigl(3 \bar{D}_2+2 \bar{D}_3+4 \bar{D}_4\bigl) \bar{F}_1 \nonumber\\
& \hskip 0.7cm+4 \bigl(\bar{D}_2-\bar{D}_3+2 \bar{D}_4\bigl) \bar{F}_3\bigl)+2 \bigl(\bar{b}_3 \bigl(4 \bar{C}_2 \bar{G}_2-4 \bar{D}_3 \bar{F}_1\bigl)+2 \bar{b}_1 \bar{D}_3 \bar{F}_3-\bar{C}_2^2 \bar{F}_3+8 \bar{C}_2 \bar{I}_1+2 \bar{D}_2 \bar{H}_2 \nonumber\\
& \hskip 0.7cm-4 \bar{D}_3\bigl(\bar{H}_1+\bar{H}_2\bigl)+4 \bigl(2 \bar{F}_1+\bar{F}_3\bigl) \bar{F}_4-4 \bar{E}_1 \bar{G}_2\bigl)\bigl) \,.  %\nonumber \\
 \end{align}

 Finally, from the compact form $\{\bar{\textbf{E}},\bar{\textbf{H}}\},$ which contains the Poisson brackets $\{\bar{E}_1,\bar{H}_1\}$ and $\{\bar{E}_1,\bar{H}_2\}.$ We then conclude the gradings of these terms are given by \begin{align}
     \mathcal{G} \left(\{\bar{E}_1,\bar{H}_1\}\right) = &\, (2,2,8) \tilde{+}(3,1,8) \tilde{+}(4,2,6) \tilde{+}(3,3,6)  , \nonumber \\
      \mathcal{G} \left(\{\bar{E}_1,\bar{H}_2\}\right) = &\, (1,3,8) \tilde{+}(2,2,8) \tilde{+}(3,3,6) \tilde{+}(2,4,6)  .
 \end{align}   Again, by considering all the predicted generators from each of the homogeneous gradings above, we find that the expansion of these degree $12$ brackets are as follows: 
 \begin{align}
	\{\bar{E}_1, \bar{H}_1\}&=-\frac{\rm i}{64}  \bigl(\bar{b}_1^2 \bar{b}_2 \bigl(\bar{c}_1-\bar{C}_2\bigl)^2-2 \bigl(-3 \bar{D}_2^3-2 \bar{d}_1 \bar{D}_2^2+4 \bar{D}_3^2 \bar{D}_2-8 \bar{D}_3^3+16 \bar{d}_1 \bar{D}_3^2+16 \bigl(\bar{D}_2-3 \bar{D}_3\bigl) \bar{D}_4^2+\bigl(-8 \bar{d}_1^2 \nonumber\\
	& \hskip 0.7 cm-4 \bar{D}_2 \bar{d}_1+6 \bar{D}_2^2\bigl) \bar{D}_3+\bigl(-2 \bar{D}_2^2+8 \bar{d}_1 \bar{D}_2-24 \bar{D}_3 \bar{D}_2+40 \bar{D}_3^2-40 \bar{d}_1 \bar{D}_3\bigl) \bar{D}_4+4 \bar{c}_1 \bigl(2 \bar{d}_1-\bar{D}_2-2 \bar{D}_3 \nonumber\\
	& \hskip 0.7 cm+4 \bar{D}_4\bigl) \bar{E}_1-8 \bar{b}_3^2 \bigl(\bar{D}_2 \bigl(\bar{d}_1-\bar{D}_3+2 \bar{D}_4\bigl)+2 \bar{C}_2 \bar{E}_1\bigl)+48 \bar{b}_2^2 \bar{b}_3 \bigl(\bar{b}_3 \bigl(\bar{d}_1+\bar{D}_2-\bar{D}_3+2 \bar{D}_4\bigl)-\bar{F}_2\bigl) \nonumber\\
	& \hskip 0.7 cm+8 \bar{b}_3 \bigl(-\bigl(\bigl(\bar{d}_1+\bar{D}_2-\bar{D}_3+3 \bar{D}_4\bigl) \bar{C}_2^2\bigl)+\bar{c}_1 \bar{D}_2 \bar{C}_2+4 \bar{E}_1^2+2 \bar{D}_2 \bar{F}_2\bigl)+2 \bar{b}_2 \bigl(16\bar{b}_3^3 \bigl(\bar{d}_1+\bar{D}_2-\bar{D}_3+2 \bar{D}_4\bigl) \nonumber\\
	& \hskip 0.7 cm -8 \bigl(3 \bar{C}_2^2+2 \bar{F}_2\bigl) \bar{b}_3^2-2 \bigl(2 \bar{d}_1^2+\bigl(5 \bar{D}_2+8 \bar{D}_3+14 \bar{D}_4\bigl) \bar{d}_1-10 \bigl(\bar{D}_3-2 \bar{D}_4\bigl) \bigl(\bar{D}_3+\bar{D}_4\bigl)+\bar{D}_2 \bigl(7 \bar{D}_3+16 \bar{D}_4\bigl) \nonumber\\
	& \hskip 0.7 cm-24 \bar{C}_2 \bar{E}_1\bigl) \bar{b}_3-\bigl(2 \bar{c}_1-3 \bar{C}_2\bigl) \bigl(\bar{c}_1-\bar{C}_2\bigl) \bigl(\bar{d}_1+\bar{D}_2-\bar{D}_3+2 \bar{D}_4\bigl)+2 \bigl(4 \bar{d}_1+5 \bar{D}_2-4 \bar{D}_3+14 \bar{D}_4\bigl) \bar{F}_2\bigl)\bigl) \nonumber\\
	& \hskip 0.7 cm+2 \bar{b}_1 \bigl(24 \bigl(\bar{b}_3^3-\bar{b}_3 \bar{D}_4\bigl) \bar{b}_2^2+2 \bigl(8 \bar{b}_3^4-\bigl(4 \bar{d}_1+5 \bar{D}_2+20 \bar{D}_3+22 \bar{D}_4\bigl) \bar{b}_3^2+\bigl(-\bar{c}_1^2+\bar{C}_2^2+4 \bar{F}_2+24 \bar{F}_4\bigl) \bar{b}_3  \nonumber\\
	& \hskip 0.7 cm+\bar{D}_4 \bigl(4 \bar{d}_1+5 \bar{D}_2-4 \bar{D}_3+14 \bar{D}_4\bigl)+\bigl(\bar{c}_1-\bar{C}_2\bigl) \bar{E}_1\bigl) \bar{b}_2+\bigl(\bar{C}_2^2-\bar{c}_1^2\bigl) \bar{D}_2+2 \bigl(-8 \bar{b}_3^3+\bar{c}_1^2+\bar{C}_2^2-2 \bar{c}_1 \bar{C}_2\bigl) \bar{D}_3 \nonumber\\
	& \hskip 0.7 cm+8 \bar{C}_2^2 \bar{D}_4+8 \bar{b}_3 \bigl(\bar{D}_3 \bar{D}_4+\bar{C}_2 \bar{E}_1\bigl)-16 \bar{D}_3 \bar{F}_2-8 \bar{D}_4 \bar{F}_4-8 \bar{b}_3^2 \bigl(\bar{C}_2^2-2 \bar{F}_4\bigl)\bigl)\bigl) \, ,
 \end{align} and
 \begin{align} 
	\{\bar{E}_1, \bar{H}_2\}&=-\frac{{\rm i}}{64}  \bigl(2 \bigl(8 \bar{b}_3 \bigl(\bar{E}_1 \bar{b}_2 \bar{C}_2-\bar{b}_2 \bigl(2 \bar{D}_3-\bar{D}_4\bigl) \bigl(\bar{d}_1+\bar{D}_2-\bar{D}_3+2 \bar{D}_4\bigl)-\bar{c}_1 \bar{C}_2 \bar{D}_2+\bar{C}_2^2 \bigl(\bar{D}_3+\bar{D}_4\bigl) \nonumber\\
	&\hskip 0.7cm-2 \bar{D}_2 \bar{F}_4-4 \bar{E}_1^2\bigl)+\bar{b}_2 \bigl(\bigl(\bar{c}_1-\bar{C}_2\bigl) \bigl(\bar{c}_1 \bigl(2 \bar{d}_1+\bar{D}_2-2 \bar{D}_3+4 \bar{D}_4\bigl)+\bar{C}_2 \bigl(-2 \bar{d}_1+\bar{D}_2+2 \bar{D}_3  \nonumber\\
	&\hskip 0.7cm-4 \bar{D}_4\bigl)\bigl)+8 \bigl(2 \bar{D}_3-\bar{D}_4\bigl) \bar{F}_2-16 \bar{E}_1^2\bigl)-8 \bar{b}_3^2 \bigl(2 \bigl(\bar{b}_2 \bar{F}_2-5 \bar{E}_1 \bar{C}_2\bigl)+\bar{D}_3 \bigl(4 \bar{d}_1+3 \bar{D}_2-4 \bar{D}_3+8 \bar{D}_4\bigl)\bigl) \nonumber\\
	&\hskip 0.7cm+16 \bar{b}_3^3 \bigl(\bar{b}_2 \bigl(\bar{d}_1+\bar{D}_2-\bar{D}_3+2 \bar{D}_4\bigl)-2 \bar{C}_2^2\bigl)+12 \bar{E}_1 \bar{c}_1 \bar{D}_2-8 \bar{E}_1 \bar{c}_1 \bar{D}_3+8 \bar{d}_1 \bar{D}_3^2-12 \bar{d}_1 \bar{D}_2 \bar{D}_3 \nonumber\\
	&\hskip 0.7cm+8 \bar{d}_1 \bar{D}_3 \bar{D}_4+\bar{D}_2^3-8 \bar{D}_3^3+20 \bar{D}_2 \bar{D}_3^2+16 \bar{D}_3 \bar{D}_4^2-10 \bar{D}_2^2 \bar{D}_3+6 \bar{D}_2^2 \bar{D}_4+8 \bar{D}_3^2 \bar{D}_4-24 \bar{D}_2 \bar{D}_3 \bar{D}_4\bigl) \nonumber\\
	&\hskip 0.7cm+\bar{b}_1 \bigl(-4 \bar{b}_2 \bigl(\bar{b}_3 \bigl(\bar{c}_1^2-\bar{C}_2^2-4 \bar{F}_4\bigl)+\bar{b}_3^2 \bigl(\bar{D}_2+4 \bar{D}_3+6 \bar{D}_4\bigl)+\bar{E}_1 \bigl(\bar{C}_2-\bar{c}_1\bigl)-\bar{D}_4 \bigl(\bar{D}_2+4 \bar{D}_3+6 \bar{D}_4\bigl)\bigl) \nonumber\\
	&\hskip 0.7cm+3 \bar{b}_2^2 \bigl(\bar{c}_1-\bar{C}_2\bigl)^2+8 \bar{b}_3 \bar{D}_3 \bigl(\bar{D}_2+2 \bar{D}_3+6 \bar{D}_4\bigl)+4 \bar{D}_3 \bigl(2 \bar{c}_1-3 \bar{C}_2\bigl) \bigl(\bar{c}_1-\bar{C}_2\bigl)-8 \bigl(\bar{D}_2+4 \bar{D}_3+6 \bar{D}_4\bigl) \bar{F}_4\bigl)\bigl) \, .
\end{align}

Up to here, from all these relations given above, we have provided all the expansions in degree-twelve Poisson brackets. We then study the expansions in the degree-thirteen Poisson brackets.

\vskip 0.5cm

\subsection{Expansions in the degree $13$ brackets}

  Considering Poisson brackets of degree-thirteen, an evaluation of the compact forms reveals an organization into sets such as $\{\bar{\textbf{E}},\bar{\textbf{I}}\}$, $\{\bar{\textbf{F}},\bar{\textbf{H}}\}$, and $\{\bar{\textbf{G}},\bar{\textbf{G}}\}$. Starting our analysis with the compact forms $\{\bar{\textbf{E}},\bar{\textbf{I}}\}$, which exhibits uniform properties across all its elements, we derive the following result:
\begin{align}
    \mathcal{G}\left(\{\bar{E}_1,\bar{I}_1\}\right) = & \, (0,4,9) \tilde{+} (1,3,9) \tilde{+} (1,5,7) \tilde{+}(2,4,7); \nonumber \\
     \mathcal{G}\left(\{\bar{E}_1,\bar{I}_2\}\right) = & \, (3,1,9) \tilde{+} (4,0,9) \tilde{+} (5,1,7) \tilde{+}(4,2,7).
\end{align} 
  Expanding these elements, the following components in $ \mathcal{G}\left(\{\bar{E}_1,\bar{I}_1\}\right) $ are deduced 
\begin{align*}
    (0,4,9) := & \, \left\{\bar{b}_3 \bar{E}_1 \bar{F}_2, \text{ } \bar{b}_1 \bar{b}_3^4 \bar{C}_2, \text{ } \bar{b}_1 \bar{b}_3^3 \bar{E}_1, \text{ } \bar{C}_2 \bar{D}_4 \bar{F}_2, \text{ } \bar{b}_1 \bar{b}_3 \bar{D}_4 \bar{E}_1, \text{ } \bar{b}_1 \bar{C}_2 \bar{D}_4^2, \text{ } \bar{b}_1 \bar{b}_3^2 \bar{C}_2 \bar{D}_4, \text{ } \bar{b}_3^2 \bar{C}_2 \bar{F}_2, \text{ } \bar{b}_1 \left(C^{(003)}\right)^2 \bar{E}_1, \text{ } \bar{b}_1 C^{(003)} \bar{D}_2 \bar{D}_4,\right.\\
    & \, \left. \bar{b}_3^2 C^{(202)} \bar{E}_1, \text{ } \bar{b}_3 C^{(003)} C^{(202)} \bar{D}_2, \text{ } \left(C^{(003)}\right)^2 \bar{C}_2 C^{(202)}, \text{ }   \bar{b}_3 \bar{C}_2 C^{(202)} \bar{D}_4, \text{ } \bar{b}_1 \bar{b}_3^2 C^{(003)} \bar{D}_2, \text{ } \bar{b}_1 \bar{b}_3 \left(C^{(003)}\right)^2 \bar{C}_2, \right. \\
    & \, \left. \bar{b}_3^3 \bar{C}_2 C^{(202)}, \text{ } C^{(202)} \bar{D}_4 \bar{E}_1, \text{ }C^{(003)} \bar{D}_2 \bar{F}_2\right\} ;\\ 
    (1,3,9) := & \, \left\{C^{(003)} C^{(202)} \bar{F}_2, \text{ }\bar{b}_1 \bar{b}_3 C^{(003)} \bar{F}_2, \text{ }\bar{b}_1 C^{(003)} C^{(202)} \bar{D}_4, \text{ }\bar{b}_3 C^{(003)} \left(C^{(202)}\right)^2, \text{ }\bar{b}_1^2 \bar{b}_3 C^{(003)} \bar{D}_4, \text{ }\bar{b}_1^2 \left(C^{(003)}\right)^3,\right. \\
    &\, \left. \bar{b}_1 \bar{b}_3^2 C^{(003)} C^{(202)}, \text{ }\bar{b}_1^2 \bar{b}_3^3 C^{(003)} \right\} ;\\
     (1,5,7) := & \, \left\{\bar{b}_1 \bar{E}_1 \bar{F}_2, \text{ }\bar{b}_1^2 \bar{D}_4 \bar{E}_1, \text{ }\bar{b}_1^2 \bar{b}_3 \bar{C}_2 \bar{D}_4, \text{ }\bar{b}_1^2 \bar{b}_3^3 \bar{C}_2, \text{ }\bar{b}_1 \bar{b}_3 \bar{C}_2 \bar{F}_2, \text{ }\bar{b}_1^2 \bar{b}_3^2 \bar{E}_1, \text{ }\bar{C}_2 C^{(202)} \bar{F}_2, \text{ }\left(C^{(202)}\right)^2 \bar{E}_1, \text{ }\bar{b}_1 \bar{b}_3 C^{(202)} \bar{E}_1,  \right. \\
    & \left.\bar{b}_1 C^{(003)} C^{(202)} \bar{D}_2, \text{ } \bar{b}_1 \bar{C}_2 C^{(202)} \bar{D}_4, \text{ }\bar{b}_3 \bar{C}_2 \left(C^{(202)}\right)^2, \text{ }\bar{b}_1^2 \bar{b}_3 C^{(003)} \bar{D}_2, \text{ }\bar{b}_1^2 \left(C^{(003)}\right)^2 \bar{C}_2, \text{ }\bar{b}_1 \bar{b}_3^2 \bar{C}_2 C^{(202)}\right\} ;\\ 
     (2,4,7) := & \,\left\{\bar{F}_1 \bar{G}_1, \text{ }\bar{C}_2 \bar{D}_2 \bar{F}_2, \text{ }\bar{b}_1 \bar{b}_3 \bar{D}_2 \bar{E}_1, \text{ }\bar{b}_1 \bar{b}_3^2 \bar{C}_2 \bar{D}_2, \text{ }\bar{b}_1 \bar{C}_2 \bar{D}_2 \bar{D}_4, \text{ }\bar{b}_1 \bar{b}_2 C^{(003)} \bar{F}_2, \text{ }C^{(202)} \bar{D}_2 \bar{E}_1, \text{ }\bar{b}_1 C^{(003)} \bar{C}_2 \bar{E}_1, \text{ }\bar{b}_1 C^{(003)} C^{(202)} \bar{D}_3,\right.\\
    &\, \bar{b}_1 C^{(003)} \bar{D}_2^2, \text{ }\bar{b}_2 C^{(003)} \left(C^{(202)}\right)^2, \text{ }\bar{b}_3 \bar{C}_2 C^{(202)} \bar{D}_2, \text{ }C^{(003)} \bar{C}_2^2 C^{(202)}, \text{ }\bar{b}_1^2 \bar{b}_2 C^{(003)} \bar{D}_4, \text{ }\bar{b}_1^2 \bar{b}_3 C^{(003)} \bar{D}_3, \\
    & \,\left.\bar{b}_1 \bar{b}_2 \bar{b}_3 C^{(003)} C^{(202)}, \text{ }\bar{b}_1^2 C^{(003)} \bar{F}_4, \text{ }\bar{b}_1 \bar{b}_3 C^{(003)} \bar{C}_2^2, \text{ }\bar{b}_1^2 \bar{b}_2 \bar{b}_3^2 C^{(003)}\right\} .
\end{align*} After replacement of the new generators $\bar{c}_1$ and $\bar{d}_1$, %all the allowed monomials have the form of \begin{align*}
we conclude that there cannot exist more than $173$ polynomials in the Poisson bracket $\{\bar{E}_1,\bar{I}_1\}$. In other words, we may write \begin{align}
     \{\bar{E}_1,\bar{I}_1\} = \Gamma_{59}^{1} \bar{b}_3\bar{E}_1\bar{F}_2 +  \ldots + \Gamma_{59}^{172} \bar{b}_1\bar{b}_3 \bar{c}_1 \bar{C}_2^2 + \Gamma_{59}^{173} \bar{b}_1^2 \bar{b}_2 \bar{b}_3^2 \bar{c}_1,
\end{align} where $\Gamma_{59}^{1},\ldots,\Gamma_{59}^{173}$ are arbitrary coefficients.   Specifically, after determining the coefficients, we end up with the expression 
 \begin{align}
	\{\bar{E}_1, \bar{I}_1\}&=-\frac{\rm i}{64}  \bigl(2 \bar{b}_2^2 \bigl(-2 \bar{c}_1 \bigl(-\bar{b}_3 \bigl(\bar{d}_1+\bar{D}_2-\bar{D}_3+2 \bar{D}_4\bigl)+\bar{C}_2^2+\bar{F}_2\bigl)+\bar{C}_2 \bigl(2 \bar{b}_3 \bigl(\bar{d}_1+\bar{D}_2-\bar{D}_3+2 \bar{D}_4\bigl)+\bar{C}_2^2-2 \bar{F}_2\bigl) \nonumber\\
	& \hskip 0.7cm +\bar{c}_1^2 \bar{C}_2\bigl)+\bar{b}_2 \bigl(4 \bar{b}_3 \bigl(-\bar{c}_1 \bar{C}_2 \bigl(\bar{c}_1+\bar{C}_2\bigl)+4 \bar{C}_2 \bar{F}_4+2 \bar{E}_1 \bigl(\bar{D}_2+2 \bar{D}_3+4 \bar{D}_4\bigl)\bigl)+8 \bar{b}_3^2 \bigl(\bar{c}_1 \bar{D}_2-2 \bar{C}_2 \bigl(\bar{D}_2+2 \bar{D}_3 \nonumber\\
	& \hskip 0.7cm +4 \bar{D}_4\bigl)\bigl)+32 \bar{b}_3^4 \bar{C}_2-32 \bar{E}_1 \bar{b}_3^3-\bar{c}_1 \bigl(-16 \bar{E}_1 \bar{C}_2+4 \bar{D}_3 \bigl(\bar{d}_1-\bar{D}_3+2 \bar{D}_4\bigl)+\bar{D}_2^2+8 \bigl(\bar{D}_3+\bar{D}_4\bigl) \bar{D}_2\bigl) \nonumber\\
	& \hskip 0.7cm +\bar{C}_2 \bigl(-4 \bar{d}_1 \bar{D}_3+3 \bar{D}_2^2+16 \bar{D}_4 \bar{D}_2+4 \bigl(\bar{D}_3+2 \bar{D}_4\bigl) \bigl(\bar{D}_3+4 \bar{D}_4\bigl)\bigl)-32 \bar{E}_1 \bar{F}_4\bigl)+8 \bar{F}_4 \bigl(12 \bar{b}_3^2 \bar{C}_2-16 \bar{E}_1 \bar{b}_3  \nonumber\\
	& \hskip 0.7cm+3 \bar{c}_1 \bar{D}_2-6 \bar{C}_2 \bigl(\bar{D}_3+2 \bar{D}_4\bigl)\bigl)+2 \bar{b}_1 \bigl(-\bigl(\bar{b}_2^2 \bigl(\bar{b}_3^2-\bar{D}_4\bigl) \bigl(\bar{c}_1+\bar{C}_2\bigl)\bigl)+2 \bar{b}_3 \bar{b}_2 \bar{D}_3 \bigl(\bar{c}_1+\bar{C}_2\bigl)-4 \bar{b}_2 \bar{C}_2 \bar{F}_4 \nonumber\\
	& \hskip 0.7cm+2 \bar{D}_3^2 \bigl(\bar{C}_2-\bar{c}_1\bigl)\bigl)+4 \bar{D}_3 \bigl(-6 \bar{b}_3 \bar{c}_1 \bar{D}_2+4 \bar{b}_3 \bar{C}_2 \bigl(\bar{D}_2+2 \bar{D}_3+6 \bar{D}_4\bigl)-24 \bar{b}_3^3 \bar{C}_2+\bar{E}_1 \bigl(24 \bar{b}_3^2-6 \bar{D}_2+4 \bar{D}_3 \nonumber\\
	& \hskip 0.7cm+8 \bar{D}_4\bigl)+3 \bar{C}_2 \bigl(\bar{c}_1-\bar{C}_2\bigl)^2\bigl)\bigl) . 
 \end{align} Then, the last expression in this compact form is given by \begin{align}
		\{\bar{E}_1, \bar{I}_2\}&=-\frac{\rm i}{32}  \bigl(-2 \bigl(2 \bar{F}_2 \bigl(6 \bar{C}_2 \bigl(-2 \bar{b}_3^2+\bar{d}_1-\bar{D}_3+4 \bar{D}_4\bigl)+16 \bar{E}_1 \bar{b}_3-3 \bar{c}_1 \bar{D}_2\bigl)+\bigl(\bar{d}_1+\bar{D}_2-\bar{D}_3+2 \bar{D}_4\bigl) \bigl(-\bigl(\bar{c}_1-\bar{C}_2\bigl) \nonumber\\
		& \hskip 0.7cm \times \bigl(3 \bar{C}_2 \bigl(\bar{c}_1+\bar{C}_2\bigl)-7 \bar{b}_2 \bigl(\bar{d}_1+\bar{D}_2-\bar{D}_3+2 \bar{D}_4\bigl)\bigl)+6 \bar{b}_3 \bar{c}_1 \bar{D}_2-4 \bar{b}_3 \bar{C}_2 \bigl(2 \bar{d}_1+3 \bar{D}_2-2 \bar{D}_3+10 \bar{D}_4\bigl) \nonumber\\
		& \hskip 0.7cm +24 \bar{b}_3^3 \bar{C}_2-24 \bar{E}_1 \bar{b}_3^2+2 \bar{E}_1 \bigl(-2 \bar{d}_1+7 \bar{D}_2+2 \bar{D}_3-8 \bar{D}_4\bigl)\bigl)+24 \bar{F}_1 \bar{G}_1\bigl)+\bar{b}_1 \bigl(\bar{b}_3 \bigl(-4 \bar{b}_2 \bar{C}_2 \bigl(\bar{d}_1+\bar{D}_2-\bar{D}_3+2 \bar{D}_4\bigl) \nonumber\\
		& \hskip 0.7cm -2 \bar{c}_1^2 \bar{C}_2+8 \bar{C}_2 \bar{F}_2+6 \bar{C}_2^3+4 \bar{E}_1 \bigl(2 \bar{d}_1+3 \bar{D}_2-2 \bar{D}_3+8 \bar{D}_4\bigl)\bigl)-\bar{c}_1 \bigl(-14 \bar{b}_2 \bar{F}_2+8 \bar{E}_1 \bar{C}_2+\bar{D}_2 \bigl(2 \bar{d}_1+3 \bar{D}_2 \nonumber\\
		& \hskip 0.7cm -2 \bar{D}_3+8 \bar{D}_4\bigl)\bigl)+4 \bar{b}_3^2 \bigl(\bar{c}_1 \bar{D}_2-2 \bar{C}_2 \bigl(2 \bar{d}_1+3 \bar{D}_2-2 \bar{D}_3+8 \bar{D}_4\bigl)\bigl)-2 \bar{F}_2 \bigl(5 \bar{b}_2 \bar{C}_2+8 \bar{E}_1\bigl)+16 \bar{b}_3^4 \bar{C}_2-16 \bar{E}_1 \bar{b}_3^3  \nonumber\\
		& \hskip 0.7cm +2 \bar{C}_2 \bigl(2 \bar{D}_4 \bigl(4 \bar{d}_1+7 \bar{D}_2-2 \bar{D}_3\bigl)+\bigl(\bar{D}_2+2 \bar{D}_3\bigl) \bigl(\bar{d}_1+\bar{D}_2-\bar{D}_3\bigl)+24 \bar{D}_4^2\bigl)\bigl)+\bar{b}_1^2 \bar{C}_2 \bigl(-4 \bar{b}_3 \bar{D}_3+2 \bar{b}_2 \bigl(\bar{b}_3^2-\bar{D}_4\bigl) \nonumber\\
		& \hskip 0.7cm +\bigl(\bar{c}_1-\bar{C}_2\bigl)^2+4 \bar{F}_4\bigl)\bigl) \, .
  \end{align} 

 Concerning the compact form $\{\bar{\textbf{F}},\bar{\textbf{H}}\}$,  certain gradings are determined through a direct calculation as follows: \begin{align}
    \mathcal{G} \left(\{\bar{F}_1, \bar{H}_1\}\right)  = & \, (3,2,8) \tilde{+} (4,1,8)\tilde{+} (5,2,6)\tilde{+} (4,3,6); \nonumber \\
    \mathcal{G} \left(\{\bar{F}_2, \bar{H}_1\}\right)  = & \, (3,1,9) \tilde{+} (4,0,9)\tilde{+} (5,1,7)\tilde{+} (4,2,7);  \nonumber \\
    \mathcal{G} \left(\{\bar{F}_3, \bar{H}_1\}\right)  = & \, (2,3,8) \tilde{+} (3,2,8)\tilde{+} (4,3,6)\tilde{+} (3,4,6); \nonumber \\
    \mathcal{G} \left(\{\bar{F}_4, \bar{H}_1\}\right)  = & \, (1,3,9) \tilde{+} (2,2,9)\tilde{+} (3,3,7)\tilde{+} (2,4,7).
    \end{align} 
   All the allowed polynomials for the first example are given by 
    \begin{align*}
     (3,2,8) = & \, \left\{\bar{D}_3 \bar{I}_2,\bar{E}_1 \bar{H}_1, \text{ }\bar{F}_2 \bar{G}_2, \text{ }\bar{b}_1 \bar{D}_4 \bar{G}_2, \text{ }\bar{b}_2 \bar{b}_3 \bar{I}_2, \text{ }\bar{b}_3 \bar{C}_2 \bar{H}_1, \text{ }\bar{b}_3 \bar{D}_2 \bar{G}_1, \text{ }\bar{b}_3 \bar{E}_1 \bar{F}_1, \text{ }\bar{b}_3^2 \bar{C}_2 \bar{F}_1, \text{ }\bar{C}_2 \bar{D}_4 \bar{F}_1, \text{ }\bar{b}_1 \bar{b}_3^2 \bar{G}_2, \right. \\
     & \, \left. \bar{b}_1 C^{(003)} \bar{H}_2, \text{ }\bar{b}_3 C^{(202)} \bar{G}_2, \text{ }C^{(003)} \bar{C}_2 \bar{G}_1, \text{ }C^{(003)} C^{(202)} \bar{F}_3, \text{ }C^{(003)} \bar{D}_2 \bar{F}_1, \text{ }\bar{b}_1 \bar{b}_3 C^{(003)} \bar{F}_3\right\} ;\\
     (4,1,8)= & \, \left\{\bar{D}_2 \bar{I}_2, \text{ }\bar{F}_2 \bar{G}_1, \text{ }\bar{b}_1 \bar{D}_4 \bar{G}_1, \text{ }\bar{b}_1 \bar{b}_3^2 \bar{G}_1, \text{ }\bar{b}_1 C^{(003)} \bar{H}_1, \text{ }\bar{b}_3 C^{(202)} \bar{G}_1, \text{ }C^{(003)} C^{(202)} \bar{F}_1, \text{ }\bar{b}_1 \bar{b}_3 C^{(003)} \bar{F}_1\right\} ;\\
     (5,2,6)= & \, \left\{\bar{b}_1^2 \bar{b}_3 \bar{G}_2, \text{ }\bar{b}_1 \bar{b}_2 \bar{I}_2, \text{ }\bar{b}_1 \bar{C}_2 \bar{H}_1, \text{ }\bar{b}_1 \bar{D}_2 \bar{G}_1, \text{ }\bar{b}_1 \bar{E}_1 \bar{F}_1, \text{ }\bar{b}_1 \bar{b}_3 \bar{C}_2 \bar{F}_1, \text{ }\bar{C}_2 C^{(202)} \bar{F}_1, \text{ }\bar{b}_1^2 C^{(003)} \bar{F}_3, \text{ }\bar{b}_1 C^{(202)} \bar{G}_2\right\} ;\\
     (4,3,6) = & \, \left\{\bar{b}_1^2 \bar{I}_1, \text{ }\bar{b}_1 \bar{C}_2 \bar{H}_2, \text{ }\bar{b}_1 \bar{b}_3 \bar{C}_2 \bar{F}_3, \text{ }\bar{b}_1 \bar{D}_2 \bar{G}_2, \text{ }\bar{b}_1 \bar{D}_3 \bar{G}_1, \text{ }\bar{b}_1 \bar{E}_1 \bar{F}_3, \text{ }\bar{C}_2 \bar{D}_2 \bar{F}_1, \text{ }\bar{b}_1 \bar{b}_2 \bar{b}_3 \bar{G}_1, \text{ }\bar{C}_2^2 \bar{G}_1,\right. \\
     & \, \left. \bar{b}_2 C^{(202)} \bar{G}_1, \text{ }\bar{C}_2 C^{(202)} \bar{F}_3, \text{ }\bar{b}_1 \bar{b}_2 C^{(003)} \bar{F}_1\right\}.
\end{align*}In this case, considering the newly introduced generators $\bar{c}_1$ and $\bar{d}_1$, we arrive at the conclusion that there is a possibility of having up to a maximum of $66$ distinct terms with coefficients $\Gamma_{68}^1,\ldots,\Gamma_{68}^{66}$: $$ \{\bar{F}_1,\bar{H}_1\} = \Gamma_{68}^1\bar{D}_3 \bar{I}_2 +   \ldots + \Gamma_{68}^{65} \bar{C}_2\bar{d}_1 \bar{F}_3+\Gamma_{68}^{66}\bar{b}_1 \bar{b}_2 \bar{c}_1 \bar{F}_1,$$ thereby establishing the relations. %\begin{align*}
 After explicit computations, we deduce that 
\begin{align}
    \{\bar{F}_1 , \bar{H}_1\}&=\frac{{\rm i}}{16}  \bigl(\bar{b}_1 \bigl(\bar{b}_2 \bigl(\bar{c}_1 \bar{F}_1-\bar{C}_2 \bar{F}_1-4 \bar{I}_2\bigl)+8 \bar{b}_3 \bar{C}_2 \bar{F}_1+8 \bar{b}_3^2 \bar{G}_2+2 \bar{G}_2 \bigl(2 \bar{d}_1+\bar{D}_2-2 \bar{D}_3\bigl)+4 \bigl(\bar{D}_2-\bar{D}_3\bigl) \bar{G}_1 \nonumber\\
		& \hskip 0.7cm+8 \bar{E}_1 \bigl(\bar{F}_3-2 \bar{F}_1\bigl)\bigl)+2 \bigl(-16 \bar{b}_3 \bar{G}_2 \bigl(\bar{d}_1+\bar{D}_2-\bar{D}_3+2 \bar{D}_4\bigl)+\bar{C}_2 \bigl(2 \bar{c}_1 \bar{G}_1+\bar{F}_1 \bigl(2 \bar{d}_1+3 \bar{D}_2-2 \bar{D}_3+4 \bar{D}_4\bigl) \nonumber\\
	& \hskip 0.7cm	-6 \bar{F}_3 \bigl(\bar{d}_1+\bar{D}_2-\bar{D}_3+2 \bar{D}_4\bigl)\bigl)+2 \bar{c}_1 \bar{F}_3 \bigl(\bar{d}_1+\bar{D}_2-\bar{D}_3+2 \bar{D}_4\bigl)+\bar{c}_1 \bar{D}_2 \bar{F}_1-4 \bar{D}_2 \bar{I}_2+16 \bar{F}_2 \bar{G}_2\bigl) \nonumber\\
	& \hskip 0.7cm	+2 \bar{b}_1^2 \bar{F}_3 \bigl(\bar{C}_2-\bar{c}_1\bigl)\bigl) \, .
\end{align} 
A similar approach reveals that the rest of the explicit expansions of each of the Poisson brackets are  \begin{align}
	\{\bar{F}_2, \bar{H}_1\}&=-\frac{\rm i}{32}  \bigl(\bar{b}_1 \bigl(4 \bar{b}_3^2 \bigl(\bar{c}_1 \bar{D}_2-2 C_2 \bigl(2 \bar{d}_1+3 \bar{D}_2-2 \bar{D}_3+8 \bar{D}_4\bigl)\bigl)+4 \bar{b}_3 C_2 \bigl(-4 \bar{c}_1 C_2+\bar{c}_1^2+3 C_2^2+4 \bar{F}_2\bigl)+4 \bar{c}_1 \bigl(2 \bar{b}_2 \bar{F}_2 \nonumber\\
& \hskip 0.7cm	+\bar{E}_1 C_2\bigl)-8 \bar{F}_2 \bigl(\bar{b}_2 C_2+2 \bar{E}_1\bigl)+16 \bar{b}_3^4 C_2-\bar{c}_1 \bar{D}_2 \bigl(2 \bar{d}_1+\bar{D}_2-2 \bar{D}_3+8 \bar{D}_4\bigl)-4 \bar{E}_1 \bar{c}_1^2+C_2 \bigl(4 \bar{D}_4 \bigl(4 \bar{d}_1 \nonumber\\
& \hskip 0.7cm+7 \bar{D}_2-4 \bar{D}_3\bigl)+\bar{D}_2 \bigl(2 \bar{d}_1+\bar{D}_2-2 \bar{D}_3\bigl)+48 \bar{D}_4^2\bigl)\bigl)+8 \bar{F}_2 \bigl(2 C_2 \bigl(6 \bar{b}_3^2-\bar{d}_1+\bar{D}_3-4 \bar{D}_4\bigl)-8 \bar{E}_1 \bar{b}_3+\bar{c}_1 \bar{D}_2\bigl) \nonumber\\
& \hskip 0.7cm-4 \bigl(\bar{d}_1+\bar{D}_2-\bar{D}_3+2 \bar{D}_4\bigl) \bigl(-\bigl(\bar{c}_1-C_2\bigl) \bigl(-4 \bar{b}_2 \bigl(\bar{d}_1+\bar{D}_2-\bar{D}_3+2 \bar{D}_4\bigl)+2 \bar{c}_1 C_2+C_2^2\bigl)-2 \bar{b}_3 C_2 \bigl(2 \bar{d}_1 \nonumber\\
& \hskip 0.7cm+\bar{D}_2-2 \bar{D}_3+12 \bar{D}_4\bigl)+16 \bar{b}_3^3 C_2+8 \bar{E}_1 \bar{D}_2\bigl)+\bar{b}_1^2 \bigl(\bar{c}_1-C_2\bigl) \bigl(8 \bar{b}_3 \bar{D}_3+4 \bar{b}_2 \bigl(\bar{D}_4-\bar{b}_3^2\bigl)+C_2 \bigl(\bar{c}_1-C_2\bigl)-8 \bar{F}_4\bigl) \nonumber\\
& \hskip 0.7cm-64 \bar{F}_1 \bar{G}_1\bigl) \, , \nonumber\\
\{\bar{F}_3, \bar{H}_1\}&=-\frac{{\rm i}}{16}  \bigl(\bar{b}_1 \bigl(\bar{b}_2 \bigl(\bar{F}_1+4 \bar{F}_3\bigl) \bigl(\bar{c}_1-\bar{C}_2\bigl)-4 \bar{b}_3 \bar{c}_1 \bar{F}_3+4 \bar{b}_2 \bar{b}_3 \bigl(\bar{G}_1-\bar{G}_2\bigl)+6 \bar{c}_1 \bar{H}_2+\bar{D}_2 \bar{G}_2+4 \bar{D}_3 \bigl(\bar{G}_2-4 \bar{G}_1\bigl) \nonumber\\
& \hskip 0.7cm +6 \bar{E}_1 \bar{F}_3\bigl)+4 \bar{b}_2 \bigl(\bar{b}_3 \bar{c}_1 \bar{F}_1-3 \bar{G}_2 \bigl(\bar{d}_1+\bar{D}_2-\bar{D}_3+2 \bar{D}_4\bigl)+3 \bar{E}_1 \bar{F}_1\bigl)-8 \bar{b}_3^2 \bar{C}_2 \bar{F}_3+8 \bar{b}_3 \bar{C}_2 \bar{H}_1-8 \bar{b}_3 \bar{d}_1 \bar{G}_2 \nonumber\\
& \hskip 0.7cm-8 \bar{b}_3 \bar{D}_2 \bar{G}_1-16 \bar{b}_3 \bar{D}_3 \bar{G}_1+8 \bar{b}_3 \bar{D}_3 \bar{G}_2-16 \bar{b}_3 \bar{D}_4 \bar{G}_2-8 \bar{E}_1 \bar{b}_3 \bar{F}_1+16 \bar{E}_1 \bar{b}_3 \bar{F}_3-2 \bar{b}_1^2 \bar{I}_1-6 \bar{c}_1 \bar{C}_2 \bar{G}_1 \nonumber\\
& \hskip 0.7cm+8 \bar{c}_1\bar{C}_2 \bar{G}_2+2 \bar{c}_1 \bar{d}_1 \bar{F}_3+3 \bar{c}_1 \bar{D}_2 \bar{F}_1-2 \bar{c}_1 \bar{D}_3 \bar{F}_3+4 \bar{c}_1 \bar{D}_4 \bar{F}_3-10 \bar{C}_2 \bar{d}_1 \bar{F}_3+\bar{C}_2 \bar{D}_2 \bar{F}_1-12 \bar{C}_2 \bar{D}_3 \bar{F}_1 \nonumber \\
& \hskip 0.7cm-8 \bar{C}_2 \bar{D}_4 \bar{F}_1+10 \bar{C}_2 \bar{D}_3 \bar{F}_3-12 \bar{C}_2 \bar{D}_4 \bar{F}_3+18 \bar{C}_2^2 \bar{G}_1+16 \bar{F}_2 \bar{G}_2\bigl) \,   ,
\end{align} and 
\begin{align} 
\{\bar{F}_4, \bar{H}_1\}&=\frac{{\rm i}}{64}  \bigl(-4 \bigl(2 \bar{b}_3 \bar{c}_1 \bar{D}_2^2+\bar{b}_3 \bar{C}_2 \bigl(32 \bar{E}_1 \bar{C}_2+4 \bar{D}_3 \bigl(\bar{d}_1-\bar{D}_3+2 \bar{D}_4\bigl)+\bar{D}_2^2-8 \bar{D}_4 \bar{D}_2\bigl)+8 \bar{b}_3^3 \bar{C}_2 \bar{D}_2-8 \bar{b}_3^2 \bigl(2 \bar{C}_2^3  \nonumber\\
& \hskip 0.7cm+\bar{E}_1 \bar{D}_2\bigl)-\bar{c}_1^2 \bar{C}_2 \bar{D}_2-\bar{c}_1 \bar{C}_2^2 \bigl(\bar{D}_2+2 \bar{D}_3\bigl)-8 \bar{E}_1^2 \bar{c}_1+8 \bar{C}_2 \bigl(2 \bar{E}_1^2-3 \bar{D}_3 \bar{F}_2\bigl)+2 \bar{C}_2^3 \bigl(\bar{D}_2+\bar{D}_3+2 \bar{D}_4\bigl) \nonumber\\
& \hskip 0.7cm-8 \bar{E}_1 \bar{D}_2 \bar{D}_4+16 \bigl(\bar{F}_3-\bar{F}_1\bigl) \bar{G}_2\bigl)-\bar{b}_2 \bar{C}_2 \bigl(4 \bar{b}_3 \bar{C}_2 \bigl(\bar{C}_2-2 \bar{c}_1\bigl)+48 \bar{b}_3^2 \bigl(\bar{d}_1+\bar{D}_2-\bar{D}_3+2 \bar{D}_4\bigl)+8 \bar{E}_1 \bar{c}_1 \nonumber\\
& \hskip 0.7cm-4 \bar{d}_1 \bar{D}_3+\bigl(\bar{D}_2-2 \bar{D}_3\bigl)^2-8 \bar{D}_3 \bar{D}_4\bigl)+2 \bar{b}_1 \bigl(2 \bar{b}_2 \bigl(\bar{b}_3 \bar{c}_1 \bar{D}_2-\bar{b}_3 \bar{C}_2 \bigl(\bar{D}_2+\bar{D}_3+10 \bar{D}_4\bigl)+10 \bar{b}_3^3 \bar{C}_2 \nonumber\\
& \hskip 0.7cm-4 \bar{E}_1 \bar{b}_3^2+2 \bar{c}_1 \bar{F}_4-\bar{C}_2 \bar{F}_4-4 \bar{E}_1 \bar{D}_4\bigl)+\bar{b}_2^2 \bar{C}_2 \bigl(\bar{b}_3^2-\bar{D}_4\bigl)-2 \bar{C}_2 \bar{D}_3 \bigl(4 \bar{b}_3^2+\bar{D}_2-2 \bar{D}_3+8 \bar{D}_4\bigl) \nonumber\\
& \hskip 0.7cm+8 \bar{F}_4 \bigl(3 \bar{b}_3 \bar{C}_2+2 \bar{E}_1\bigl)+2 \bar{c}_1 \bar{D}_3 \bigl(\bar{D}_2-2 \bar{D}_3\bigl)\bigl)+4 \bar{b}_2^2 \bar{C}_2 \bigl(\bar{F}_2-\bar{b}_3 \bigl(\bar{d}_1+\bar{D}_2-\bar{D}_3+2 \bar{D}_4\bigl)\bigl) \nonumber\\
& \hskip 0.7cm+16 \bar{b}_2 \bar{F}_2 \bigl(3 \bar{b}_3 \bar{C}_2+2 \bar{E}_1\bigl)\bigl) \, .
\end{align} 
The rest of the grading in the compact form $\{\bar{\textbf{F}},\bar{\textbf{H}}\}$ are given by  
    \begin{align}
    \mathcal{G} \left(\{\bar{F}_1, \bar{H}_2\}\right)  = & \, (2,3,8) \tilde{+} (3,2,8)\tilde{+} (4,3,6)\tilde{+} (3,4,6); \nonumber \\
    \mathcal{G} \left(\{\bar{F}_2, \bar{H}_2\}\right)  = & \, (2,2,9) \tilde{+} (3,1,9)\tilde{+} (4,2,7)\tilde{+} (3,3,7); \nonumber \\
    \mathcal{G} \left(\{\bar{F}_3, \bar{H}_2\}\right)  = & \, (1,4,8) \tilde{+} (2,3,8)\tilde{+} (3,4,6)\tilde{+} (2,5,6); \nonumber \\
    \mathcal{G} \left(\{\bar{F}_4, \bar{H}_2\}\right)  = & \, (0,4,9) \tilde{+} (1,3,9)\tilde{+} (2,4,7)\tilde{+} (1,5,7).
\end{align} 
By evaluating the  admissible terms from each homogeneous gradings,  and proceeding as in the previous cases, we  deduce 
\begin{align}
\{\bar{F}_1, \bar{H}_2\}&=-\frac{{\rm i}}{32}  \bigl(2 \bigl(-8 \bar{b}_3 \bigl(3 \bar{C}_2 \bar{H}_1+\bar{G}_2 \left(-\bar{d}_1+\bar{D}_3-2 \bar{D}_4\right)-\bar{D}_2 \bar{G}_1+\bar{E}_1 \bigl(\bar{F}_1+2 \bar{F}_3\bigl)\bigl)+4 \bar{G}_1 \bigl(\bar{b}_2 \bigl(\bar{d}_1+\bar{D}_2-\bar{D}_3+2 \bar{D}_4\bigl) \nonumber\\
& \hskip 0.7cm-5 \bar{C}_2^2+4 \bar{F}_4\bigl)+8 \bigl(-\bar{E}_1 \bar{b}_2 \bar{F}_1+\bar{C}_2 \bar{F}_3 \bigl(\bar{d}_1-\bar{D}_3+\bar{D}_4\bigl)+\bar{C}_2 \bigl(\bar{D}_3+\bar{D}_4\bigl) \bar{F}_1\bigl)+8 \bar{b}_3^2 \bar{C}_2 \bigl(2 \bar{F}_1+\bar{F}_3\bigl) \nonumber\\
& \hskip 0.7cm+\bar{c}_1 \bigl(14 \bar{C}_2 \bar{G}_1-6 \bar{d}_1 \bar{F}_3+\bigl(\bar{D}_2-4 \bar{D}_3\bigl) \bar{F}_1+6 \bigl(\bar{D}_3-2 \bar{D}_4\bigl) \bar{F}_3\bigl)\bigl)+\bar{b}_1 \bigl(\bar{b}_2 \bigl(\bar{F}_1-6 \bar{F}_3\bigl) \bigl(\bar{c}_1-\bar{C}_2\bigl)+8 \bar{F}_3 \bigl(\bar{b}_3 \bar{c}_1 \nonumber\\
& \hskip 0.7cm+\bar{E}_1\bigl)+8 \bar{b}_2 \bar{b}_3 \bigl(\bar{G}_2-\bar{G}_1\bigl)-12 \bar{c}_1 \bar{H}_2+4 \bar{D}_3 \bar{G}_1-6 \bar{D}_2 \bar{G}_2\bigl)+12 \bar{b}_1^2 \bar{I}_1\bigl) \, , \nonumber\\
\{\bar{F}_2, \bar{H}_2\}&=\frac{\rm i}{32}  \bigl(6 \bar{C}_2 \bigl(-2 \bar{b}_3 \bar{D}_3+\bar{b}_2 \bigl(\bar{b}_3^2-\bar{D}_4\bigl)+2 \bar{F}_4\bigl) \bar{b}_1^2+\bigl(24 \bar{C}_2 \bar{D}_3 \bar{b}_3^2+4 \bar{C}_2^2 \bigl(\bar{c}_1+2 \bar{C}_2\bigl) \bar{b}_3-8 \bar{D}_3 \bar{E}_1 \bar{b}_3-3 \bar{C}_2 \bar{D}_2^2 \nonumber\\
& \hskip 0.7cm-12 \bar{C}_2 \bar{D}_3^2+12 \bar{C}_2 \bar{d}_1 \bar{D}_3+6 \bar{c}_1 \bar{D}_2 \bar{D}_3+6 \bar{C}_2 \bar{D}_2 \bar{D}_3-16 \bar{C}_2^2 \bar{E}_1-8 \bar{c}_1 \bar{C}_2 \bar{E}_1+2 \bar{b}_2 \bigl(4 \bar{C}_2 \bar{b}_3^3-14 \bar{E}_1 \bar{b}_3^2 \nonumber\\
& \hskip 0.7cm+\bar{c}_1 \bar{D}_2 \bar{b}_3-\bar{C}_2 \bigl(6 \bar{d}_1+7 \bar{D}_2-6 \bar{D}_3+16 \bar{D}_4\bigl) \bar{b}_3+6 \bar{D}_4 \bar{E}_1+4 \bar{c}_1 \bar{F}_2+2 \bar{C}_2 \bar{F}_2\bigl)+24 \bar{E}_1 \bar{F}_4\bigl) \bar{b}_1 \nonumber\\
& \hskip 0.7cm+2 \bigl(-8 \bar{C}_2 \bar{D}_2 \bar{b}_3^3+8 \bigl(-\bar{C}_2^3+\bar{b}_2 \bigl(\bar{d}_1+\bar{D}_2-\bar{D}_3+2 \bar{D}_4\bigl) \bar{C}_2+\bar{D}_2 \bar{E}_1\bigl) \bar{b}_3^2+2 \bigl(-2 \bar{E}_1 \bar{C}_2^2+2 \bar{d}_1 \bigl(\bar{D}_2 \nonumber\\
& \hskip 0.7cm - 4 \bar{D}_3\bigl) \bar{C}_2+\bigl(\bar{D}_2-2 \bar{D}_3\bigl) \bigl(3 \bar{D}_2-4 \bar{D}_3+8 \bar{D}_4\bigl) \bar{C}_2-\bar{c}_1 \bar{D}_2^2+10 \bar{b}_2 \bigl(\bar{d}_1+\bar{D}_2-\bar{D}_3+2 \bar{D}_4\bigl) \bar{E}_1\bigl) \bar{b}_3  \nonumber\\
& \hskip 0.7cm +4 \bar{b}_2 \bar{C}_2 \bar{d}_1^2+3 \bar{b}_2 \bar{C}_2 \bar{D}_2^2+4 \bar{b}_2 \bar{C}_2 \bar{D}_3^2-4 \bar{C}_2^3 \bar{d}_1-3 \bar{C}_2^3 \bar{D}_2+\bar{c}_1^2 \bar{C}_2 \bar{D}_2+7 \bar{b}_2 \bar{C}_2 \bar{d}_1 \bar{D}_2+4 \bar{C}_2^3 \bar{D}_3 \nonumber\\
& \hskip 0.7cm -8 \bar{b}_2 \bar{C}_2 \bar{d}_1 \bar{D}_3-7 \bar{b}_2 \bar{C}_2 \bar{D}_2 \bar{D}_3+8 \bar{b}_2 \bar{C}_2 \bar{d}_1 \bar{D}_4+6 \bar{b}_2 \bar{C}_2 \bar{D}_2 \bar{D}_4-8 \bar{b}_2 \bar{C}_2 \bar{D}_3 \bar{D}_4+\bar{D}_2^2 \bar{E}_1-4 \bar{D}_3^2 \bar{E}_1 \nonumber\\
& \hskip 0.7cm -4 \bar{d}_1 \bar{D}_2 \bar{E}_1+4 \bar{d}_1 \bar{D}_3 \bar{E}_1+8 \bar{D}_2 \bar{D}_3 \bar{E}_1+8 \bar{D}_3 \bar{D}_4 \bar{E}_1+\bar{c}_1 \bigl(2 \bigl(2 \bar{d}_1+\bar{D}_2-2 \bar{D}_3+4 \bar{D}_4\bigl) \bar{C}_2^2+8 \bar{E}_1^2  \nonumber\\
& \hskip 0.7cm -\bar{b}_2 \bigl(\bar{d}_1+\bar{D}_2-\bar{D}_3+2 \bar{D}_4\bigl) \bigl(4 \bar{d}_1+3 \bar{D}_2-4 \bar{D}_3+8 \bar{D}_4\bigl)\bigl)+4 \bar{C}_2 \bar{D}_2 \bar{F}_2-12 \bar{b}_2 \bar{E}_1 \bar{F}_2\bigl)+32 \bar{F}_3 \bar{G}_1  \nonumber\\
& \hskip 0.7cm +48 \bar{F}_1 \bigl(\bar{G}_2-\bar{G}_1\bigl)\bigl) \, , \nonumber \\
\{\bar{F}_3, \bar{H}_2\}&=\frac{\rm i}{16}   \bigl(\bar{b}_2 \bigl(6 \bar{b}_3 \bar{C}_2 \bar{F}_1+8 \bar{b}_3^2 \bar{G}_1+4 \bar{b}_1 \bar{I}_1+\bigl(\bar{D}_2-4 \bar{D}_3-8 \bar{D}_4\bigl) \bar{G}_1-2 \bar{E}_1 \bar{F}_1\bigl)-8 \bar{b}_3^2 \bar{C}_2 \bar{F}_3-32 \bar{b}_3 \bar{D}_3 \bar{G}_1  \nonumber\\
& \hskip 0.7cm+8 \bar{E}_1 \bar{b}_3 \bar{F}_3-6 \bar{b}_2^2 \bar{I}_2+6 \bar{c}_1 \bigl(\bar{C}_2 \bar{G}_2+\bar{D}_3 \bar{F}_1\bigl)-16 \bar{C}_2 \bar{D}_3 \bar{F}_1+(5 \bar{D}_2 +4  \bar{D}_3+8 \bar{D}_4) \bar{C}_2\bar{F}_3+24 \bar{F}_4 \bar{G}_1  \nonumber\\
& \hskip 0.7cm+8 (\bar{d}_1- \bar{D}_3 +2 \bar{D}_4) \bar{I}_1\bigl) \, ,  \nonumber\\
\{\bar{F}_4, \bar{H}_2\}&=-\frac{\rm i}{32}  \bigl(\bar{b}_2^2 \bigl(-2 \bar{c}_1 \bigl(-4 \bar{b}_3 \bigl(\bar{d}_1+\bar{D}_2-\bar{D}_3+2 \bar{D}_4\bigl)+\bar{C}_2^2+4 \bar{F}_2\bigl)+\bar{c}_1^2 \bar{C}_2+\bar{C}_2^3\bigl)+\bar{b}_2 \bigl(4 \bar{b}_3^2 \bigl(\bar{c}_1 \bar{D}_2  \nonumber\\
& \hskip 0.7cm -2 \bar{C}_2 \bigl(\bar{D}_2+2 \bar{D}_3+4 \bar{D}_4\bigl)\bigl)+4 \bar{b}_3 \bar{C}_2 \bigl(-4 \bar{c}_1 \bar{C}_2+\bar{c}_1^2+\bar{C}_2^2+4 \bar{F}_4\bigl)+16 \bar{b}_3^4 \bar{C}_2+12 \bar{E}_1 \bar{c}_1 \bar{C}_2 \nonumber\\
& \hskip 0.7cm+\bar{c}_1 \bar{D}_2 \bigl(\bar{D}_2-2 \bar{D}_3-4 \bar{D}_4\bigl)-4 \bar{E}_1 \bar{c}_1^2+\bar{C}_2 \bigl(8 \bar{E}_1 \bar{C}_2-8 \bar{d}_1 \bar{D}_3+\bar{D}_2^2-6 \bar{D}_3 \bar{D}_2+8 \bar{D}_4 \bar{D}_2 \nonumber\\
& \hskip 0.7cm+8 \bar{D}_3^2+16 \bar{D}_4^2\bigl)-16 \bar{E}_1 \bar{F}_4\bigl)+\bar{b}_1 \bigl(4 \bar{b}_2^2 \bar{c}_1 \bigl(\bar{D}_4-\bar{b}_3^2\bigl)+8 \bar{b}_2 \bar{C}_2 \bigl(\bar{b}_3 \bar{D}_3-\bar{F}_4\bigl)+8 \bar{D}_3^2 \bigl(\bar{C}_2-\bar{c}_1\bigl)\bigl) \nonumber\\
& \hskip 0.7cm+8 \bar{F}_4 \bigl(12 \bar{b}_3^2 \bar{C}_2-8 \bar{E}_1 \bar{b}_3+\bar{c}_1 \bar{D}_2-2 \bar{C}_2 \bigl(\bar{D}_3+2 \bar{D}_4\bigl)\bigl)+4 \bar{D}_3 \bigl(\bar{C}_2 \bigl(-2 \bar{b}_3 \bigl(8 \bar{b}_3^2+\bar{D}_2  \nonumber\\
& \hskip 0.7cm-2 \bigl(\bar{D}_3+4 \bar{D}_4\bigl)\bigl)-3 \bar{c}_1\bar{C}_2+2 \bar{c}_1^2+\bar{C}_2^2\bigl)-4 \bar{E}_1 \bar{D}_2\bigl)-32 \bar{F}_3 \bar{G}_2\bigl) .
\end{align}

Finally, we will look at the compact form $\{\bar{\textbf{G}},\bar{\textbf{G}}\}$, which contains the Poisson bracket $\{\bar{G}_1,\bar{G}_2\}$ only. By evaluating the grading, we have
 \begin{align*}
    \mathcal{G} \left(\{\bar{G}_1,\bar{G}_2\}\right) = (2,3,8) \tilde{+}(3,2,8) \tilde{+}(4,3,6) \tilde{+}(3,4,6) ,
\end{align*} 
 and the explicit expression reads:
\begin{align}
\{\bar{G}_1,\bar{G}_2\}&=\frac{{\rm i}}{16}  \bigl(2 \bigl(2 \bar{b}_2 \bar{F}_1 \bigl(\bar{b}_3 \bar{c}_1-4 \bar{E}_1\bigl)+8 \bar{b}_3 \bigl(\bar{C}_2 \bar{H}_1-\bar{D}_3 \bar{G}_1+\bar{E}_1 \bar{F}_1\bigl)+8 \bar{b}_3^2 \bar{C}_2 \bigl(\bar{F}_3-2 \bar{F}_1\bigl)+4 \bar{b}_2 \bar{G}_2 \bigl(\bar{d}_1+\bar{D}_2 \nonumber \\
& \hskip 0.7cm-\bar{D}_3+2 \bar{D}_4\bigl)-2 \bar{G}_2 \bigl(\bar{c}_1 \bar{C}_2+8 \bar{F}_2\bigl)+8 \bar{c}_1 \bar{C}_2 \bar{G}_1-2 \bar{c}_1 \bar{d}_1 \bar{F}_3+\bar{c}_1 \bar{D}_2 \bar{F}_1-4 \bar{c}_1 \bar{D}_3 \bar{F}_1+2 \bar{c}_1 \bar{D}_3 \bar{F}_3 \nonumber \\
& \hskip 0.7cm-4 \bar{c}_1 \bar{D}_4 \bar{F}_3-4 \bar{C}_2 \bar{D}_2 \bar{F}_1+8 \bar{C}_2 \bar{D}_3 \bar{F}_1+16 \bar{C}_2 \bar{D}_4 \bar{F}_1-16 \bar{C}_2 \bar{D}_4 \bar{F}_3-8 \bar{C}_2^2 \bar{G}_1+16 \bar{F}_4 \bar{G}_1\bigl) \nonumber \\
& \hskip 0.7cm+\bar{b}_1 \bigl(5 \bar{b}_2 \bar{F}_1 \bigl(\bar{C}_2-\bar{c}_1\bigl)+2 \bigl(6 \bar{b}_3 \bar{c}_1 \bar{F}_3-6 \bar{c}_1 \bar{H}_2+6 \bar{D}_3 \bar{G}_1+\bar{D}_2 \bar{G}_2-4 \bar{E}_1 \bar{F}_3\bigl)\bigl)-4 \bar{b}_1^2 \bar{I}_1\bigl) \, .
\end{align}

All these relations conclude the degree thirteen Poisson brackets. We now focus on the expansions in the degree-fourteen Poisson brackets.

\subsection{Expansions in the degree $14$ brackets}

In this case, we consider the compact forms $\{\bar{\textbf{F}},\bar{\textbf{I}}\}$ and $\{\bar{\textbf{G}},\bar{\textbf{H}}\}.$ Starting with the first of these, we deduce that the grading is given as follows: 
\begin{align}
    \mathcal{G} \left(\{\bar{F}_1, \bar{I}_1\}\right)  = & \, (1,4,9) \tilde{+} (2,3,9)\tilde{+} (3,4,7)\tilde{+} (2,5,7); \nonumber \\
    \mathcal{G} \left(\{\bar{F}_2, \bar{I}_1\}\right)  = & \, (1,3,10) \tilde{+} (2,2,10)\tilde{+} (3,3,8)\tilde{+} (2,4,8); \nonumber \\
    \mathcal{G} \left(\{\bar{F}_3, \bar{I}_1\}\right)  = & \, (0,5,9) \tilde{+} (1,4,9)\tilde{+} (2,5,7)\tilde{+} (1,6,7); \nonumber \\
    \mathcal{G} \left(\{\bar{F}_4, \bar{I}_1\}\right)  = & \,  (0,4,10)\tilde{+} (1,5,8)\tilde{+} (0,6,8).
\end{align} 
As an illustration, we elaborate on the monomial prediction from the grading of $\{\bar{F}_2, \bar{I}_1\}.$ The allowed generators from each homogeneous grading in $\mathcal{G} \left(\{\bar{F}_2, \bar{I}_1\}\right)$ are given by \begin{align*}
     (1,3,10) = & \, \left\{\bar{b}_3 \bar{D}_2 \bar{D}_3 \bar{D}_4, \text{ } \bar{b}_2 \bar{D}_2 \bar{D}_4^2, \text{ } \bar{b}_3^2 \bar{D}_2 \bar{F}_4, \text{ } \bar{D}_2 \bar{D}_4 \bar{F}_4, \text{ } \bar{b}_2 \bar{b}_3^2 \bar{D}_2 \bar{D}_4, \text{ } \bar{b}_2 \bar{b}_3^4 \bar{D}_2, \text{ } \bar{b}_3^3 \bar{D}_2 \bar{D}_3, \text{ } \bar{b}_2 C^{(003)} \bar{D}_4 \bar{E}_1, \text{ } C^{(003)} \bar{E}_1 \bar{F}_4,  \right. \\
     &\,   \bar{b}_3 C^{(003)} \bar{C}_2 \bar{F}_4, \text{ }     \bar{b}_3 C^{(003)} \bar{D}_3 \bar{E}_1, \text{ } \left(C^{(003)}\right)^2 \bar{D}_2 \bar{D}_3, \text{ } C^{(003)} \bar{C}_2 \bar{D}_3 \bar{D}_4, \text{ } \bar{b}_2 \bar{b}_3^2 C^{(003)} \bar{E}_1, \text{ } \bar{b}_2 \bar{b}_3 \left(C^{(003)}\right)^2 \bar{D}_2, \\
     & \, \left.    \bar{b}_2 \bar{b}_3 C^{(003)} \bar{C}_2 \bar{D}_4, \text{ } \bar{b}_2 \left(C^{(003)}\right)^3 \bar{C}_2, \text{ }\bar{b}_3^2 C^{(003)} \bar{C}_2 \bar{D}_3,\bar{b}_2 \bar{b}_3^3 C^{(003)} \bar{C}_2\right\}; \\
      (2,2,10)= & \, \left\{\bar{b}_3 \bar{F}_2 \bar{F}_4, \text{ }\bar{D}_3 \bar{D}_4 \bar{F}_2, \text{ }\bar{D}_4 \bar{E}_1^2, \text{ }\bar{b}_1 \bar{b}_3 \bar{D}_4 \bar{F}_4, \text{ }\bar{b}_1 \bar{D}_3 \bar{D}_4^2, \text{ }\bar{b}_2 \bar{b}_3 \bar{D}_4 \bar{F}_2, \text{ }\bar{b}_3 \bar{D}_2^2 \bar{D}_4, \text{ }\bar{b}_1 \bar{b}_3^4 \bar{D}_3, \text{ }\bar{b}_3^3 \bar{C}_2 \bar{E}_1, \text{ }\bar{b}_3^3 \bar{D}_2^2,  \text{ } \bar{b}_1 \bar{b}_2 \bar{b}_3 \bar{D}_4^2, \right.\\
     & \, \bar{b}_1 \bar{b}_2 \bar{b}_3^3 \bar{D}_4,\text{ }\bar{b}_1 \bar{b}_3^3 \bar{F}_4, \text{ }\bar{b}_2 \bar{b}_3^3 \bar{F}_2, \text{ }\bar{b}_3^4 \bar{C}_2^2, \text{ }\bar{b}_1 \bar{b}_3^2 \bar{D}_3 \bar{D}_4, \text{ }\bar{b}_3^2 \bar{D}_3 \bar{F}_2, \text{ }\bar{b}_3^2 \bar{E}_1^2, \text{ }\bar{C}_2^2 \bar{D}_4^2, \text{ } \bar{b}_3 \bar{C}_2 \bar{D}_4 \bar{E}_1, \text{ }\bar{b}_3^2 \bar{C}_2^2 \bar{D}_4,\\
     & \, \left.\bar{b}_2 C^{(202)} \bar{D}_4^2, \text{ }\bar{b}_3^2 C^{(202)} \bar{F}_4,  \text{ } \bar{b}_3 C^{(202)} \bar{D}_3 \bar{D}_4, \text{ }\bar{b}_3 C^{(003)} \bar{D}_2 \bar{E}_1, \text{ }\left(C^{(003)}\right)^2 \bar{C}_2 \bar{E}_1, \text{ }\bar{b}_1 \left(C^{(003)}\right)^2 \bar{F}_4,  \right. \\
     &\, \bar{b}_2 \left(C^{(003)}\right)^2 \bar{F}_2, \text{ }\left(C^{(003)}\right)^2 C^{(202)} \bar{D}_3, \text{ } \left(C^{(003)}\right)^2 \bar{D}_2^2, \text{ }C^{(003)} \bar{C}_2 \bar{D}_2 \bar{D}_4, \text{ }\bar{b}_1 \bar{b}_2 \left(C^{(003)}\right)^2 \bar{D}_4, \\
     & \, \bar{b}_1 \bar{b}_3 \left(C^{(003)}\right)^2 \bar{D}_3, \text{ }C^{(202)} \bar{D}_4 \bar{F}_4\bar{b}_2, \text{ } \bar{b}_3^2 C^{(202)} \bar{D}_4,\text{ } \bar{b}_2 \bar{b}_3 \left(C^{(003)}\right)^2 C^{(202)}, \text{ }\bar{b}_3^3 C^{(202)} \bar{D}_3, \\
     & \, \left. \bar{b}_3^2 C^{(003)} \bar{C}_2 \bar{D}_2, \text{ } \bar{b}_3 \left(C^{(003)}\right)^2 \bar{C}_2^2, \text{ }\bar{b}_1 \bar{b}_2 \bar{b}_3^2 \left(C^{(003)}\right)^2, \text{ }\bar{b}_2 \bar{b}_3^4 C^{(202)}\right\}; \\ 
     (3,3,8)= & \, \left\{\bar{F}_1 \bar{H}_2,\text{ }\bar{b}_3^2 \bar{C}_2^2 \bar{D}_2, \text{ }\bar{F}_3 \bar{H}_1, \text{ }\bar{G}_1 \bar{G}_2, \text{ }\bar{b}_3 \bar{F}_1 \bar{F}_3, \text{ }\bar{D}_2 \bar{D}_3 \bar{F}_2, \text{ }\bar{D}_2 \bar{E}_1^2, \text{ }\bar{b}_1 \bar{b}_3 \bar{D}_2 \bar{F}_4, \text{ }\bar{b}_1 \bar{D}_2 \bar{D}_3 \bar{D}_4, \text{ }\bar{b}_2 \bar{b}_3 \bar{D}_2 \bar{F}_2, \text{ }\bar{b}_1 \bar{b}_2 \bar{b}_3 \bar{D}_2 \bar{D}_4,  \right. \\
     & \, \left. \bar{b}_1 \bar{b}_3^2 \bar{D}_2 \bar{D}_3, \text{ }\bar{b}_3 \bar{D}_2^3, \text{ }\bar{b}_1 \bar{b}_2 \bar{b}_3^3 \bar{D}_2, \text{ } \bar{b}_3 \bar{C}_2 \bar{D}_2 \bar{E}_1, \text{ }\bar{C}_2^2 \bar{D}_2 \bar{D}_4, \text{ }   \bar{b}_2 C^{(003)} \bar{C}_2 \bar{F}_2, \text{ }\bar{b}_2 C^{(003)} C^{(202)} \bar{E}_1, \text{ }\bar{b}_2 C^{(202)} \bar{D}_2 \bar{D}_4,  \right. \\
     & \, \bar{b}_3 C^{(202)} \bar{D}_2 \bar{D}_3, \text{ }C^{(003)} \bar{C}_2^2 \bar{E}_1,  \text{ } \bar{b}_1 C^{(003)}\bar{D}_3 \bar{E}_1, \text{ }C^{(202)} \bar{D}_2 \bar{F}_4, \text{ }\bar{b}_1 \bar{b}_2 \bar{b}_3 C^{(003)} \bar{E}_1, \text{ } C^{(003)} \bar{C}_2 \bar{F}_4,  \\ 
     & \,   \bar{b}_1 \bar{b}_2 \left(C^{(003)}\right)^2 \bar{D}_2, \text{ }    \bar{b}_1 \bar{b}_2 C^{(003)} \bar{C}_2 \bar{D}_4, \text{ }\bar{b}_1 \bar{b}_3 C^{(003)} \bar{C}_2 \bar{D}_3, \text{ }\bar{b}_2 \bar{b}_3^2 C^{(202)} \bar{D}_2, \text{ } C^{(003)} \bar{C}_2 \bar{D}_2^2\bar{b}_1, \\
     & \, \left.\bar{b}_2 \bar{b}_3 C^{(003)} \bar{C}_2 C^{(202)},  \text{ }\bar{b}_3 C^{(003)} \bar{C}_2^3, \text{ }\bar{b}_1 \bar{b}_2 \bar{b}_3^2 C^{(003)} \bar{C}_2 C^{(003)} \bar{C}_2 C^{(202)} \bar{D}_3\right\} 
     \end{align*}   
     
and, in addition, the following $47$ generators for the last homogeneous grading
     \begin{align*} 
     (2,4,8)= & \, \left\{\bar{F}_3 \bar{H}_2, \text{ }\bar{G}_2^2, \text{ }\bar{b}_1 \bar{F}_4^2, \text{ }\bar{b}_2 \bar{F}_2 \bar{F}_4, \text{ }\bar{b}_3 \bar{F}_3^2, \text{ }\bar{C}_2 \bar{E}_1 \bar{F}_4, \text{ }\bar{D}_2^2 \bar{F}_4, \text{ }\bar{D}_3^2 \bar{F}_2, \text{ }\bar{D}_3 \bar{E}_1^2, \text{ }\bar{b}_1 \bar{b}_2 \bar{D}_4 \bar{F}_4, \text{ }\bar{b}_1 \bar{b}_3 \bar{D}_3 \bar{F}_4, \text{ }\bar{b}_1 \bar{D}_3^2 \bar{D}_4, \text{ }\bar{b}_3 \bar{C}_2 \bar{D}_3 \bar{E}_1,   \right. \\
     & \,  \bar{b}_2 \bar{b}_3^2 \bar{C}_2 \bar{E}_1,  \text{ } \bar{b}_2 \bar{D}_2^2 \bar{D}_4, \text{ }\bar{b}_2 \bar{b}_3 \bar{D}_3 \bar{F}_2, \text{ }\bar{b}_3 \bar{C}_2^2 \bar{F}_4, \text{ }\bar{b}_2 \bar{b}_3^3 \bar{C}_2^2, \text{ }\bar{b}_3^2 \bar{C}_2^2 \bar{D}_3, \text{ }\bar{b}_1 \bar{b}_2^2 \bar{b}_3^2 \bar{D}_4, \text{ }\bar{b}_3 \bar{D}_2^2 \bar{D}_3, \text{ }\bar{b}_2^2 \bar{D}_4 \bar{F}_2, \text{ }\bar{b}_1 \bar{b}_2^2 \bar{D}_4^2, \text{ }\bar{b}_1 \bar{b}_3^2 \bar{D}_3^2, \text{ }\bar{b}_2^2 \bar{b}_3^2 \bar{F}_2,\\
     & \,\bar{C}_2^2 \bar{D}_3 \bar{D}_4, \text{ }\bar{b}_2 \bar{b}_3 \bar{C}_2^2 \bar{D}_4,  \text{ }\bar{b}_1 \bar{b}_2 \bar{b}_3^3 \bar{D}_3,  \text{ }\bar{b}_2 \bar{b}_3 \bar{E}_1^2,\text{ }\bar{b}_2 \bar{C}_2 \bar{D}_4 \bar{E}_1,\text{ }\bar{b}_1 \bar{b}_2 \bar{b}_3^2 \bar{F}_4, \text{ }\bar{b}_1 \bar{b}_2 \bar{b}_3 \bar{D}_3 \bar{D}_4, \text{ }\bar{b}_2 \bar{b}_3^2 \bar{D}_2^2,   \text{ }\bar{b}_2 \bar{b}_3 C^{(202)} \bar{F}_4, \\
     & \,   \bar{b}_2 C^{(202)} \bar{D}_3 \bar{D}_4, \text{ }\bar{b}_1 \bar{b}_2 \left(C^{(003)}\right)^2 \bar{D}_3, \text{ }\bar{b}_2^2 \bar{b}_3 C^{(202)} \bar{D}_4, \text{ }\bar{b}_2^2 \left(C^{(003)}\right)^2 C^{(202)}, \text{ }C^{(003)} \bar{C}_2 \bar{D}_2 \bar{D}_3, \text{ }\bar{b}_2 \bar{b}_3^2 C^{(202)} \bar{D}_3\\
     & \,    \left.\bar{b}_2 \bar{b}_3 C^{(003)} \bar{C}_2 \bar{D}_2, \text{ }\bar{b}_2 \left(C^{(003)}\right)^2 \bar{C}_2^2, \text{ }    \bar{b}_1 \bar{b}_2^2 \bar{b}_3 \left(C^{(003)}\right)^2, \text{ }\bar{b}_2^2 \bar{b}_3^3 C^{(202)}, \text{ }C^{(202)} \bar{D}_3 \bar{F}_4, \text{ }\bar{b}_3 C^{(202)} \bar{D}_3^2 ,\text{ }\bar{b}_2 C^{(003)} \bar{D}_2 \bar{E}_1 \right\}.
\end{align*}

After the replacement of $C^{(003)}$ and $C^{(202)}$ to $\bar{c}_1$ and $\bar{d}_1$, we conclude that there are $141$ permissible terms such that  
\begin{align*}
    \{\bar{F}_2,\bar{I}_1\} = \Gamma_{69}^1   \bar{F}_3 \bar{H}_2 +\ldots + \Gamma_{69}^{140} \bar{b}_3 \bar{d}_1 \bar{D}_3^2   +\Gamma_{69}^{141}    \bar{b}_2 \bar{c}_1 \bar{C}_2 \bar{D}_2 \bar{E}_1  
\end{align*} 
  for coefficients $\Gamma_{69}^k$  with $1 \leq k \leq 141$.   After a heavy but routine computation, the coefficients are determined, leading to 
\begin{align}
    \{\bar{F}_2, \bar{I}_1\}&=-\frac{{\rm i}}{64}  \bigl(-2 \bigl(4 \bar{b}_3 \bar{c}_1 \bar{C}_2^3-2 \bar{b}_3 \bigl(4 \bar{C}_2^4+\bar{D}_2^2 \bar{D}_3+32 \bar{F}_3^2\bigl)+4 \bar{b}_3^2 \bar{C}_2^2 \bar{D}_2+\bar{c}_1 \bar{C}_2 \bar{D}_2 \bigl(\bar{D}_2-4 \bar{D}_3\bigl)-4 \bar{C}_2^2 \bar{d}_1 \bar{D}_3\nonumber \\
& \hskip 0.7cm+4 \bar{C}_2^2 \bar{D}_3^2-4 \bar{C}_2^2 \bar{D}_2 \bar{D}_4+8 \bar{E}_1 \bar{C}_2^3-24 \bar{D}_3^2 \bar{F}_2-8 \bar{D}_2 \bar{D}_3 \bar{F}_2-6 \bar{D}_2^2 \bar{F}_4+8 \bar{E}_1^2 \bar{D}_2+16 \bar{E}_1^2 \bar{D}_3+32 \bar{F}_3 \bar{H}_2 \nonumber \\
& \hskip 0.7cm+8 \bar{G}_2 \bigl(\bar{G}_2-4 \bar{G}_1\bigl)\bigl)+\bar{b}_2 \bigl(-8 \bar{b}_3 \bigl(\bar{c}_1 \bar{C}_2 \bar{D}_2+\bar{C}_2^2 \bigl(\bar{d}_1-\bar{D}_3-3 \bar{D}_4\bigl)+4 \bar{F}_4 \bigl(\bar{d}_1+\bar{D}_2+2 \bar{D}_4\bigl)-6 \bar{E}_1^2 \nonumber \\
& \hskip 0.7cm+4 \bar{D}_3 \bigl(\bar{F}_2-\bar{F}_4\bigl)\bigl)-24 \bar{b}_3^3 \bar{C}_2^2+32 \bar{b}_3^2 \bar{D}_3 \bigl(\bar{d}_1+\bar{D}_2-\bar{D}_3+2 \bar{D}_4\bigl)+8 \bar{c}_1 \bar{C}_2 \bar{F}_2-6 \bar{c}_1 \bar{C}_2^3+3 \bar{c}_1^2 \bar{C}_2^2 \nonumber \\
& \hskip 0.7cm+3 \bar{C}_2^4+4 \bar{D}_4 \bigl(4 \bar{d}_1 \bar{D}_3-\bigl(\bar{D}_2-2 \bar{D}_3\bigl)^2+8 \bar{D}_3 \bar{D}_4\bigl)-48 \bar{F}_2 \bar{F}_4\bigl)+\bar{b}_1 \bigl(\bar{b}_2 \bigl(-8 \bar{b}_3 \bigl(\bar{E}_1 \bigl(\bar{c}_1-\bar{C}_2\bigl)+\bar{D}_3 \bar{D}_4\bigl) \nonumber \\
& \hskip 0.7cm+4 \bar{b}_3^2 \bigl(\bar{c}_1 \bar{C}_2-6 \bar{F}_4\bigl)-8 \bar{b}_3^3 \bar{D}_3-2 \bar{c}_1 \bar{C}_2 \bigl(\bar{D}_2+4 \bar{D}_3+2 \bar{D}_4\bigl)+\bar{c}_1^2 \bigl(\bar{D}_2+4 \bar{D}_3\bigl)+\bar{C}_2^2 \bigl(\bar{D}_2+4 \bar{D}_3\bigl)\nonumber \\
& \hskip 0.7cm+40 \bar{D}_4 \bar{F}_4\bigl)+8 \bigl(\bar{D}_3 \bigl(-\bar{b}_3 \bar{C}_2^2+\bar{b}_3^2 \bar{D}_3+\bar{E}_1 \bigl(\bar{c}_1-\bar{C}_2\bigl)+\bar{D}_3 \bar{D}_4\bigl)+\bar{F}_4 \bigl(6 \bar{b}_3 \bar{D}_3+\bar{C}_2^2\bigl)-8 \bar{F}_4^2\bigl)\nonumber \\
& \hskip 0.7cm+\bar{b}_2^2 \bigl(-\bar{b}_3 \bigl(\bar{c}_1-\bar{C}_2\bigl)^2+8 \bar{b}_3^2 \bar{D}_4-8 \bar{D}_4^2\bigl)\bigl)-16 \bar{b}_2^2 \bar{D}_4 \bigl(\bar{b}_3 \bigl(\bar{d}_1+\bar{D}_2-\bar{D}_3+2 \bar{D}_4\bigl)-\bar{F}_2\bigl)\bigl) \,.
\end{align}

  Proceeding similarly, the remaining degree-fourteen relations close as follows:
\begin{align}
\{\bar{F}_1, \bar{I}_1\}&=\frac{{\rm i}}{16}  \bigl(\bar{b}_2 \bigl(-2 \bar{c}_1 \bigl(3 \bar{b}_1 \bar{G}_2+\bar{C}_2 \bar{F}_1\bigl)+2 \bar{b}_3 \bigl(2 \bar{c}_1 \bar{G}_1+8 \bar{D}_3 \bar{F}_1+4 \bar{D}_4 \bar{F}_1+\bar{D}_2 \bigl(10 \bar{F}_1+\bar{F}_3\bigl)\bigl)+6 \bar{G}_2 \bigl(\bar{b}_1 \bar{C}_2+4 \bar{E}_1\bigl) \nonumber \\
& \hskip 0.7cm-8 \bar{b}_3^3 \bar{F}_1+3 \bar{c}_1^2 \bar{F}_1-\bar{F}_1 \bigl(\bar{C}_2^2+16 \bar{F}_4\bigl)-2 \bigl(9 \bar{D}_2 \bar{H}_1+10 \bar{D}_3 \bar{H}_1+4 \bar{E}_1 \bar{G}_1\bigl)\bigl)-4 \bar{b}_3 \bigl(-\bar{b}_1 \bar{D}_3 \bar{F}_3+3 \bar{C}_2 \bar{F}_3 \bigl(\bar{c}_1 \nonumber \\
& \hskip 0.7cm+\bar{C}_2\bigl)+4 \bar{F}_1 \bar{F}_4\bigl)+8 \bar{b}_3^2 \bigl(\bar{C}_2 \bar{G}_2+2 \bar{D}_3 \bar{F}_1-\bar{D}_2 \bar{F}_3\bigl)+16 \bar{b}_1 \bar{D}_3 \bar{H}_2+4 \bar{b}_3 \bar{b}_2^2 \bigl(\bar{b}_3 \bar{F}_1-\bar{H}_1\bigl)-20 \bar{b}_1 \bar{F}_3 \bar{F}_4 \nonumber \\
& \hskip 0.7cm-12 \bar{c}_1 \bar{D}_3 \bar{G}_1+6 \bar{c}_1 \bar{D}_2 \bar{G}_2+12 \bar{E}_1 \bar{c}_1 \bar{F}_3-8 \bar{C}_2 \bar{D}_4 \bar{G}_2+28 \bar{E}_1 \bar{C}_2 \bar{F}_3-4 \bar{D}_3^2 \bar{F}_1-20 \bar{D}_2 \bar{D}_3 \bar{F}_1+11 \bar{D}_2^2 \bar{F}_3\nonumber \\
& \hskip 0.7cm+6 \bar{D}_2 \bar{D}_3 \bar{F}_3+8 \bar{D}_2 \bar{D}_4 \bar{F}_3\bigl) \, ,\nonumber \\
\{\bar{F}_3, \bar{I}_1\}&=-\frac{{\rm i}}{4}  \bigl(\bar{b}_2 \bigl(\bar{b}_3 \bar{c}_1 \bar{G}_2-2 \bar{b}_3 \bar{D}_4 \bar{F}_3-\bar{b}_3 \bar{D}_3 \bigl(6 \bar{F}_1+\bar{F}_3\bigl)+2 \bar{b}_3^3 \bar{F}_3-4 \bar{b}_3^2 \bar{H}_2-2 \bar{C}_2 \bar{I}_1+5 \bar{D}_3 \bar{H}_1+\bigl(\bar{D}_2+2 \bar{D}_3 \nonumber \\
& \hskip 0.7cm+4 \bar{D}_4\bigl) \bar{H}_2+\bar{F}_1 \bar{F}_4-2 \bar{E}_1 \bar{G}_2\bigl)-\bar{D}_3 \bigl(6 \bar{b}_3^2 \bar{F}_3-12 \bar{b}_3 \bar{H}_2+3 \bar{c}_1 \bar{G}_2+3 \bar{D}_2 \bar{F}_3+2 \bar{D}_4 \bar{F}_3\bigl)+\bar{b}_3 \bar{b}_2^2 \bigl(\bar{b}_3 \bar{F}_1 \nonumber \\
& \hskip 0.7cm-\bar{H}_1\bigl)+4 \bar{F}_4 \bigl(2 \bar{b}_3 \bar{F}_3-3 \bar{H}_2\bigl)+\bar{D}_3^2 \bigl(4 \bar{F}_1-\bar{F}_3\bigl)\bigl) \, ,\nonumber \\
\{\bar{F}_4, \bar{I}_1\}&=\frac{{\rm i}}{32} \bigl(4 \bar{b}_2 \bigl(\bar{F}_4 \bigl(-2 \bar{b}_3 \bigl(\bar{D}_2+2 \bar{D}_3+4 \bar{D}_4\bigl)+8 \bar{b}_3^3+\bigl(\bar{c}_1-\bar{C}_2\bigl) \bigl(3 \bar{c}_1-2 \bar{C}_2\bigl)\bigl)-2 \bar{D}_3 \bigl(\bar{b}_3 \bigl(\bar{c}_1-\bar{C}_2\bigl) \bigl(\bar{c}_1-2 \bar{C}_2\bigl) \nonumber \\
& \hskip 0.7cm-\bar{b}_3^2 \bigl(\bar{D}_2+2 \bar{D}_3+8 \bar{D}_4\bigl)+4 \bar{b}_3^4+2 \bar{E}_1 \bigl(\bar{c}_1-\bar{C}_2\bigl)+\bar{D}_4 \bigl(\bar{D}_2+2 \bar{D}_3+4 \bar{D}_4\bigl)\bigl)+4 \bar{F}_4^2\bigl)+\bar{b}_2^2 \bigl(4 \bar{b}_3^2-\bar{D}_2 \nonumber \\
& \hskip 0.7cm-2 \bar{D}_3-4 \bar{D}_4\bigl) \bigl(\bar{c}_1-\bar{C}_2\bigl)^2-4 \bar{D}_3^2 \bigl(\bigl(3 \bar{c}_1-4 \bar{C}_2\bigl) \bigl(\bar{c}_1-\bar{C}_2\bigl)-2 \bar{b}_3 \bigl(12 \bar{b}_3^2+\bar{D}_2-2 \bigl(\bar{D}_3+6 \bar{D}_4\bigl)\bigl)\bigl) \nonumber \\
& \hskip 0.7cm+16 \bar{D}_3 \bar{F}_4 \bigl(-12 \bar{b}_3^2+\bar{D}_3+4 \bar{D}_4\bigl)+128 \bar{b}_3 \bar{F}_4^2\bigl) \, . 
\end{align} 

Subsequently, we proceed to assess the gradings associated with the remaining components of the Poisson brackets. From $\{\bar{\textbf{F}},\bar{\textbf{I}}\}$, these are given by  \begin{align}
     \mathcal{G} \left(\{\bar{F}_1, \bar{I}_2\}\right)  = & \,  (4,1,9) \tilde{+} (5,0,9)\tilde{+} (6,1,7)\tilde{+} (5,2,7); \nonumber \\
    \mathcal{G} \left(\{\bar{F}_2, \bar{I}_2\}\right)  = & \, (4,0,10) \tilde{+} (6,0,8)\tilde{+} (5,1,8); \nonumber \\
    \mathcal{G} \left(\{\bar{F}_3, \bar{I}_2\}\right)  = & \, (3,2,9) \tilde{+} (4,1,9)\tilde{+} (5,2,7)\tilde{+} (4,3,7); \nonumber \\
    \mathcal{G} \left(\{\bar{F}_4, \bar{I}_2\}\right)  = & \,(2,2,10) \tilde{+} (3,1,10)\tilde{+} (4,2,8)\tilde{+} (3,3,8).
\end{align} 
  The allowed polynomials in these Poisson brackets are deduced as before, from which we  obtain that 
\begin{align} 
\{\bar{F}_1, \bar{I}_2\}&=\frac{{\rm i}}{8}  \bigl(-2 \bar{b}_1 \bigl(\bar{b}_3 \bar{c}_1 \bar{G}_1-\bar{b}_3 \bar{F}_1 \bigl(\bar{d}_1+2 \bar{D}_2-\bar{D}_3+4 \bar{D}_4\bigl)-6 \bar{b}_3 \bar{F}_3 \bigl(\bar{d}_1+\bar{D}_2-\bar{D}_3+2 \bar{D}_4\bigl)+2 \bar{b}_3^3 \bar{F}_1-4 \bar{b}_3^2 \bar{H}_1 \nonumber \\
& \hskip 0.7cm+2 \bar{C}_2 \bar{I}_2+2 \bar{H}_1 \bigl(\bar{D}_3-\bar{d}_1\bigl)+7 \bar{H}_2 \bigl(\bar{d}_1+\bar{D}_2-\bar{D}_3+2 \bar{D}_4\bigl)+\bar{F}_2 \bigl(4 \bar{F}_1+\bar{F}_3\bigl)-4 \bar{E}_1 \bar{G}_1\bigl)+\bar{b}_1^2 \bigl(2 \bar{b}_3 \bar{H}_2 \nonumber \\
& \hskip 0.7cm-\bigl(2 \bar{G}_1+\bar{G}_2\bigl) \bigl(\bar{c}_1-\bar{C}_2\bigl)-2 \bar{D}_4 \bar{F}_3\bigl)+2 \bigl(6 \bar{b}_3^2 \bar{F}_1 \bigl(\bar{d}_1+\bar{D}_2-\bar{D}_3+2 \bar{D}_4\bigl)-12 \bar{b}_3 \bar{H}_1 \bigl(\bar{d}_1+\bar{D}_2-\bar{D}_3 \nonumber \\
& \hskip 0.7cm+2 \bar{D}_4\bigl)-8 \bar{b}_3 \bar{F}_1 \bar{F}_2+\bigl(\bar{d}_1+\bar{D}_2-\bar{D}_3+2 \bar{D}_4\bigl) \bigl(3 \bar{c}_1 \bar{G}_1+\bar{F}_1 \bigl(\bar{d}_1+4 \bar{D}_2-\bar{D}_3+4 \bar{D}_4\bigl)-4 \bar{F}_3 \bigl(\bar{d}_1+\bar{D}_2 \nonumber \\
& \hskip 0.7cm-\bar{D}_3+2 \bar{D}_4\bigl)\bigl)+12 \bar{F}_2 \bar{H}_1\bigl)\bigl) \, , \nonumber \\
\{\bar{F}_2, \bar{I}_2\}&=\frac{{\rm i}}{32}\bigl(4 \bigl(-\bigl(\bar{d}_1+\bar{D}_2-\bar{D}_3+2 \bar{D}_4\bigl)^2 \bigl(2 \bar{b}_3 \bigl(2 \bar{d}_1+\bar{D}_2-2 \bar{D}_3+16 \bar{D}_4\bigl)-24 \bar{b}_3^3+\bigl(\bar{c}_1-\bar{C}_2\bigl) \bigl(3 \bar{c}_1-4 \bar{C}_2\bigl)\bigl) \nonumber \\
& \hskip 0.7cm+4 \bar{F}_2 \bigl(\bar{d}_1+\bar{D}_2-\bar{D}_3+2 \bar{D}_4\bigl) \bigl(-12 \bar{b}_3^2+\bar{d}_1+\bar{D}_2-\bar{D}_3+6 \bar{D}_4\bigl)+32 \bar{b}_3 \bar{F}_2^2\bigl)-4 \bar{b}_1 \bigl(\bar{F}_2 \bigl(2 \bar{b}_3 \bigl(2 \bar{d}_1+3 \bar{D}_2 \nonumber \\
& \hskip 0.7cm-2 \bar{D}_3+8 \bar{D}_4\bigl)-8 \bar{b}_3^3-\bigl(3 \bar{c}_1-2 \bar{C}_2\bigl) \bigl(\bar{c}_1-\bar{C}_2\bigl)\bigl)+2 \bigl(\bar{d}_1+\bar{D}_2-\bar{D}_3+2 \bar{D}_4\bigl) \bigl(\bar{b}_3 \bigl(\bar{c}_1-\bar{C}_2\bigl) \bigl(\bar{c}_1-2 \bar{C}_2\bigl) \nonumber \\
& \hskip 0.7cm-\bar{b}_3^2 \bigl(2 \bar{d}_1+3 \bar{D}_2-2 \bar{D}_3+12 \bar{D}_4\bigl)+4 \bar{b}_3^4+2 \bar{E}_1 \bigl(\bar{c}_1-\bar{C}_2\bigl)+\bar{D}_4 \bigl(2 \bar{d}_1+3 \bar{D}_2-2 \bar{D}_3+8 \bar{D}_4\bigl)\bigl)-4 \bar{F}_2^2\bigl) \nonumber \\
& \hskip 0.7cm+\bar{b}_1^2 \bigl(\bar{c}_1-\bar{C}_2\bigl)^2 \bigl(4 \bar{b}_3^2-2 \bar{d}_1-3 \bar{D}_2+2 \bar{D}_3-8 \bar{D}_4\bigl)\bigl) \, ,\\
\{\bar{F}_3, \bar{I}_2\}&=\frac{{\rm i}}{32}  \bigl(\bar{b}_1 \bigl(2 \bar{b}_2 \bigl(-11 \bar{b}_3^2 \bar{F}_1+4 \bar{b}_3 \bar{H}_1+\bar{c}_1 \bar{G}_1-\bar{C}_2 \bar{G}_1+7 \bar{D}_4 \bar{F}_1\bigl)+2 \bar{c}_1 \bigl(-2 \bar{b}_3 \bar{G}_1+7 \bar{C}_2 \bar{F}_1-3 \bar{C}_2 \bar{F}_3\bigl)-8 \bar{b}_3 \bar{d}_1 \bar{F}_3 \nonumber \\
& \hskip 0.7cm-12 \bar{b}_3 \bar{D}_3 \bar{F}_1+8 \bar{b}_3 \bar{D}_3 \bar{F}_3-3 \bar{c}_1^2 \bigl(\bar{F}_1-2 \bar{F}_3\bigl)-11 \bar{C}_2^2 \bar{F}_1+8 \bar{d}_1 \bar{H}_2+2 \bar{D}_2 \bar{H}_1+48 \bar{D}_3 \bar{H}_1-8 \bar{D}_3 \bar{H}_2 \nonumber \\
& \hskip 0.7cm+20 \bar{F}_1 \bar{F}_4+12 \bar{E}_1 \bar{G}_1-64 \bar{E}_1 \bar{G}_2\bigl)-4 \bar{b}_3 \bigl(5 \bar{b}_2 \bar{F}_1 \bigl(\bar{d}_1+\bar{D}_2-\bar{D}_3+2 \bar{D}_4\bigl)+4 \bar{c}_1 \bar{C}_2 \bar{F}_1+\bar{C}_2^2 \bar{F}_1+4 \bar{C}_2 \bar{I}_2 \nonumber \\
& \hskip 0.7cm-4 \bar{H}_1 \bigl(\bar{d}_1-\bar{D}_3+2 \bar{D}_4\bigl)+16 \bar{F}_1 \bar{F}_2\bigl)+\bar{b}_1^2 \bigl(2 \bar{F}_3 \bigl(\bar{b}_3^2+\bar{D}_4\bigl)-4 \bar{b}_3 \bar{H}_2+\bar{G}_2 \bigl(\bar{c}_1-\bar{C}_2\bigl)\bigl)+8 \bar{b}_3^2 \bigl(2 \bar{C}_2 \bar{G}_1 \nonumber \\
& \hskip 0.7cm +3 \bar{F}_1 \bigl(\bar{d}_1+\bar{D}_2-\bar{D}_3+2 \bar{D}_4\bigl)\bigl)+32 \bar{b}_2 \bar{d}_1 \bar{H}_1+32 \bar{b}_2 \bar{D}_2 \bar{H}_1-32 \bar{b}_2 \bar{D}_3 \bar{H}_1+64 \bar{b}_2 \bar{D}_4 \bar{H}_1-12 \bar{b}_2 \bar{F}_1 \bar{F}_2 \nonumber \\
& \hskip 0.7cm-12 \bar{c}_1 \bar{d}_1 \bar{G}_1-24 \bar{c}_1 \bar{d}_1 \bar{G}_2+12 \bar{c}_1 \bar{D}_3 \bar{G}_1-24 \bar{c}_1 \bar{D}_4 \bar{G}_1-24 \bar{c}_1 \bar{D}_2 \bar{G}_2+24 \bar{c}_1 \bar{D}_3 \bar{G}_2-48 \bar{c}_1 \bar{D}_4 \bar{G}_2+24 \bar{E}_1 \bar{c}_1 \bar{F}_1 \nonumber \\
& \hskip 0.7cm+6 \bar{d}_1 \bar{D}_2 \bar{F}_1+28 \bar{d}_1 \bar{D}_3 \bar{F}_1-24 \bar{d}_1 \bar{D}_4 \bar{F}_1-32 \bar{d}_1 \bar{D}_2 \bar{F}_3-8 \bar{d}_1 \bar{D}_3 \bar{F}_3+4 \bar{d}_1^2 \bar{F}_3+21 \bar{D}_2^2 \bar{F}_1-28 \bar{D}_3^2 \bar{F}_1-48 \bar{D}_4^2 \bar{F}_1 \nonumber \\
& \hskip 0.7cm+22 \bar{D}_2 \bar{D}_3 \bar{F}_1+12 \bar{D}_2 \bar{D}_4 \bar{F}_1+80 \bar{D}_3 \bar{D}_4 \bar{F}_1-36 \bar{D}_2^2 \bar{F}_3+4 \bar{D}_3^2 \bar{F}_3-16 \bar{D}_4^2 \bar{F}_3+32 \bar{D}_2 \bar{D}_3 \bar{F}_3-80 \bar{D}_2 \bar{D}_4 \bar{F}_3 \nonumber \\
& \hskip 0.7cm+48 \bar{F}_2 \bar{H}_1-48 \bar{E}_1 \bar{I}_2\bigl) \, , \nonumber 
\end{align} and
\begin{align}
\{\bar{F}_4, \bar{I}_2\}&=-\frac{{\rm i}}{64} \bigl(8 \bigl(-\bigl(\bar{c}_1-\bar{C}_2\bigl)^2 \bar{D}_3-8 \bar{b}_3 \bar{D}_3 \bar{D}_4+\bar{b}_2 \bigl(4 \bar{D}_4 \bar{b}_3^2+\bigl(\bar{c}_1-\bar{C}_2\bigl)^2 \bar{b}_3-4 \bar{D}_4^2\bigl)+8 \bar{D}_4 \bar{F}_4\bigl) \bar{b}_1^2+\bigl(10 \bar{C}_2^4 \nonumber \\
& \hskip 0.7cm-2 \bar{b}_2 \bar{d}_1 \bar{C}_2^2-\bar{b}_2 \bar{D}_2 \bar{C}_2^2+2 \bar{b}_2 \bar{D}_3 \bar{C}_2^2+20 \bar{D}_2 \bar{E}_1 \bar{C}_2-16 \bar{D}_4 \bar{E}_1 \bar{C}_2+\bar{c}_1^2 \bigl(10 \bar{C}_2^2-\bar{b}_2 \bigl(2 \bar{d}_1+\bar{D}_2-2 \bar{D}_3\bigl)\bigl) \nonumber \\
& \hskip 0.7cm-16 \bar{D}_2^2 \bar{D}_4+4 \bar{b}_3^3 \bigl(12 \bar{C}_2^2-\bar{b}_2 \bigl(16 \bar{d}_1+15 \bar{D}_2-16 \bar{D}_3+32 \bar{D}_4\bigl)\bigl)+\bar{c}_1 \bigl(-20 \bar{C}_2^3+2 \bar{b}_2 \bar{C}_2\bigl(2 \bar{d}_1+\bar{D}_2\nonumber \\
& \hskip 0.7cm-2 \bar{D}_3\bigl) -20 \bar{D}_2 \bar{E}_1\bigl)+112 \bar{b}_2 \bar{D}_4 \bar{F}_2+8 \bar{b}_3^2 \bigl(8 \bar{d}_1 \bar{D}_3+7 \bar{D}_2 \bar{D}_3-2 \bigl(4 \bar{D}_3 \bigl(\bar{D}_3-2 \bar{D}_4\bigl)+5 \bar{C}_2 \bar{E}_1+3 \bar{b}_2 \bar{F}_2\bigl)\bigl) \nonumber \\
& \hskip 0.7cm-128 \bar{F}_2 \bar{F}_4+4 \bar{b}_3 \bigl(\bigl(4 \bar{D}_4-7 \bar{D}_2\bigl) \bar{C}_2^2+7 \bar{c}_1 \bar{D}_2 \bar{C}_2+24 \bar{E}_1^2-\bar{b}_2 \bar{D}_2 \bar{D}_4+2 \bar{D}_2 \bar{F}_4\bigl)\bigl) \bar{b}_1+8 \bar{b}_3^2 \bigl(\bar{b}_2 \bigl(\bar{d}_1+\bar{D}_2 \nonumber \\
& \hskip 0.7cm-\bar{D}_3+2 \bar{D}_4\bigl) \bigl(14 \bar{d}_1+13 \bar{D}_2-14 \bar{D}_3+28 \bar{D}_4\bigl)-2 \bar{C}_2^2 \bigl(9 \bigl(\bar{d}_1+\bar{D}_2-\bar{D}_3\bigl)+16 \bar{D}_4\bigl)\bigl)+2 \bar{b}_3 \bigl(15 \bar{D}_2^3 \nonumber \\
& \hskip 0.7cm-4 \bigl(\bar{D}_3-8 \bar{D}_4\bigl) \bar{D}_2^2+4 \bar{D}_3 \bigl(6 \bar{D}_4-7 \bar{D}_3\bigl) \bar{D}_2+16 \bar{d}_1^2 \bar{D}_3+4 \bar{d}_1 \bigl(4 \bar{D}_2^2+7 \bar{D}_3 \bar{D}_2+8 \bar{D}_3 \bigl(\bar{D}_4-\bar{D}_3\bigl)\bigl) \nonumber \\
& \hskip 0.7cm+4 \bigl(\bar{C}_2^4-\bar{c}_1 \bar{C}_2^3+12 \bar{F}_2 \bar{C}_2^2+\bar{b}_2 \bar{D}_2 \bar{F}_2\bigl)+16 \bigl(\bar{D}_3^3-2 \bar{D}_4 \bar{D}_3^2+\bar{F}_1 \bigl(6 \bar{F}_1-\bar{F}_3\bigl)\bigl)\bigl)+8 \bigl(\bigl(\bigl(2 \bar{d}_1+\bar{D}_2 \nonumber \\
& \hskip 0.7cm-2 \bar{D}_3\bigl) \bar{D}_4+\bar{c}_1 \bar{E}_1\bigl) \bar{C}_2^2+2 \bigl(\bar{c}_1 \bar{D}_2 \bar{D}_4+2 \bar{E}_1 \bar{F}_2\bigl) \bar{C}_2+4 \bar{d}_1 \bar{E}_1^2-4 \bar{D}_3 \bar{E}_1^2-\bar{C}_2^3 \bar{E}_1+3 \bar{D}_2^2 \bar{F}_2+4 \bar{D}_3^2 \bar{F}_2 \nonumber \\
& \hskip 0.7cm-4 \bar{d}_1 \bar{D}_3 \bar{F}_2-2 \bar{D}_2 \bar{D}_3 \bar{F}_2+2 \bar{b}_2 \bigl(\bar{D}_4 \bigl(\bar{d}_1+\bar{D}_2-\bar{D}_3+2 \bar{D}_4\bigl){}^2-8 \bar{F}_2^2\bigl)-2 \bigl(2 \bar{d}_1+\bar{D}_2-2 \bar{D}_3\bigl) \nonumber \\
& \hskip 0.7cm \times \bigl(\bar{d}_1+\bar{D}_2-\bar{D}_3+2 \bar{D}_4\bigl) \bar{F}_4+4 \bar{G}_1 \bigl(\bar{G}_2-\bar{G}_1\bigl)+4 \bigl(\bar{F}_1+\bar{F}_3\bigl) \bar{H}_1\bigl)\bigl) \, .%\nonumber \\
\end{align}

Ultimately, in relation to the compact form denoted by $\{\bar{\textbf{G}},\bar{\textbf{H}}\}$, we can derive the gradings associated with the terms as follows:
 \begin{align}
    \mathcal{G} \left(\{\bar{G}_1,\bar{H}_1\}\right) =& \, (3,2,9) \tilde{+} (4,1,9) \tilde{+}(5,2,7) \tilde{+} (4,3,7)  ; \nonumber \\
    \mathcal{G} \left(\{\bar{G}_1,\bar{H}_2\}\right) =& \, (2,3,9) \tilde{+} (3,2,9) \tilde{+}(4,3,7) \tilde{+} (3,4,7)  ; \nonumber \\
    \mathcal{G} \left(\{\bar{G}_2,\bar{H}_1\}\right) =& \,  (2,3,9) \tilde{+} (3,2,9) \tilde{+}(4,3,7) \tilde{+} (3,4,7)   ; \nonumber \\
      \mathcal{G} \left(\{\bar{G}_2,\bar{H}_2\}\right) =& \, (1,4,9) \tilde{+} (2,3,9) \tilde{+}(3,4,7) \tilde{+} (2,5,7),
\end{align} 
  from which the explicit expressions are obtained as 
\begin{align}
\{\bar{G}_1, \bar{H}_1\}&=-\frac{{\rm i}}{32}  \bigl(4 \bar{b}_3 \bigl(9 \bar{b}_2 \bar{F}_1\bigl(\bar{d}_1+\bar{D}_2-\bar{D}_3+2 \bar{D}_4\bigl)-4 \bar{c}_1 \bar{C}_2 \bar{F}_1+\bar{C}_2^2 \bar{F}_1+8 \bar{C}_2 \bar{I}_2-8 \bar{H}_1 \bigl(\bar{d}_1-\bar{D}_3+2 \bar{D}_4\bigl) \nonumber \\
& \hskip 0.7cm-16 \bar{E}_1 \bar{G}_1\bigl)+2 \bar{b}_1 \bigl(\bar{b}_2 \bigl(\bar{b}_3^2 \bar{F}_1-8 \bar{b}_3 \bar{H}_1+2 \bar{c}_1 \bar{G}_1-2 \bar{C}_2 \bar{G}_1+3 \bar{D}_4 \bar{F}_1\bigl)+4 \bar{b}_3 \bar{c}_1 \bigl(\bar{G}_1+\bar{G}_2\bigl) \nonumber \\
& \hskip 0.7cm-2 \bigl(\bar{b}_3 \bigl(2 \bar{d}_1 \bar{F}_3-\bar{D}_2 \bar{F}_1+5 \bar{D}_3 \bar{F}_1+\bar{D}_2 \bar{F}_3-2 \bar{D}_3 \bar{F}_3\bigl)+4 \bar{b}_3^3 \bar{F}_3+\bar{D}_2 \bar{H}_1-8 \bar{D}_3 \bar{H}_1+4 \bar{F}_2 \bar{F}_3 \nonumber \\
& \hskip 0.7cm-3 \bar{F}_1 \bar{F}_4-2 \bar{E}_1 \bar{G}_1+6 \bar{E}_1 \bar{G}_2\bigl)+\bar{c}_1 \bar{C}_2 \bigl(\bar{F}_1-\bar{F}_3\bigl)+\bar{c}_1^2 \bigl(\bar{F}_3-\bar{F}_1\bigl)\bigl)+2 \bar{b}_1^2 \bigl(2 \bar{b}_3^2 \bar{F}_3+\bar{G}_2 \bigl(\bar{C}_2-\bar{c}_1\bigl) \nonumber \\
& \hskip 0.7cm-2 \bar{D}_4 \bar{F}_3\bigl)+16 \bar{b}_3^2 \bigl(2 \bar{C}_2 \bar{G}_1+\bar{d}_1 \bar{F}_1-\bigl(\bar{D}_2+\bar{D}_3-2 \bar{D}_4\bigl) \bar{F}_1\bigl)-24 \bar{b}_2 \bar{H}_1 \bigl(\bar{d}_1+\bar{D}_2-\bar{D}_3+2 \bar{D}_4\bigl)\nonumber \\
& \hskip 0.7cm+4 \bar{b}_2 \bar{F}_1 \bar{F}_2-8 \bar{c}_1 \bar{d}_1 \bar{G}_1+16 \bar{c}_1 \bar{d}_1 \bar{G}_2+8 \bar{c}_1 \bar{D}_3 \bar{G}_1-16 \bar{c}_1 \bar{D}_4 \bar{G}_1+16 \bar{c}_1 \bar{D}_2 \bar{G}_2-16 \bar{c}_1 \bar{D}_3 \bar{G}_2\nonumber \\
& \hskip 0.7cm+32 \bar{c}_1 \bar{D}_4 \bar{G}_2+16 \bar{E}_1 \bar{c}_1 \bar{F}_1-36 \bar{d}_1 \bar{D}_3 \bar{F}_1-16 \bar{d}_1 \bar{D}_4 \bar{F}_1+28 \bar{d}_1 \bar{D}_2 \bar{F}_3-32 \bar{d}_1 \bar{D}_3 \bar{F}_3+96 \bar{d}_1 \bar{D}_4 \bar{F}_3\nonumber \\
& \hskip 0.7cm+16 \bar{d}_1^2 \bar{F}_3+7 \bar{D}_2^2 \bar{F}_1+36 \bar{D}_3^2 \bar{F}_1-32 \bar{D}_4^2 \bar{F}_1-36 \bar{D}_2 \bar{D}_3 \bar{F}_1-56 \bar{D}_3 \bar{D}_4 \bar{F}_1+12 \bar{D}_2^2 \bar{F}_3+16 \bar{D}_3^2 \bar{F}_3\nonumber \\
& \hskip 0.7cm+128 \bar{D}_4^2 \bar{F}_3-28 \bar{D}_2 \bar{D}_3 \bar{F}_3+88 \bar{D}_2 \bar{D}_4 \bar{F}_3-96 \bar{D}_3 \bar{D}_4 \bar{F}_3+32 \bar{F}_2 \bar{H}_1\bigl)\, \nonumber \\
\{\bar{G}_2, \bar{H}_1\}&=-\frac{{\rm i}}{32} \bigl(2 \bar{b}_2 \bigl(\bar{b}_1 \bar{c}_1 \bigl(\bar{G}_1-3 \bar{G}_2\bigl)-\bar{b}_1 \bar{C}_2 \bar{G}_1+3 \bar{b}_1 \bar{C}_2 \bar{G}_2+4 \bar{b}_3 \bar{F}_1 \bigl(\bar{d}_1+2 \bar{D}_2-\bar{D}_3+4 \bar{D}_4\bigl)-10 \bar{b}_3 \bar{F}_3 \bigl(\bar{d}_1+\bar{D}_2\nonumber \\
& \hskip 0.7cm-\bar{D}_3+2 \bar{D}_4\bigl)+8 \bar{b}_1 \bar{D}_4 \bar{F}_1-9 \bar{b}_1 \bar{D}_4 \bar{F}_3+9 \bar{b}_1 \bar{b}_3^2 \bar{F}_3-4 \bar{b}_1 \bar{b}_3 \bar{H}_1-\bar{c}_1 \bar{C}_2 \bar{F}_1+\bar{c}_1^2 \bar{F}_1-8 \bar{H}_2 \bigl(\bar{d}_1+\bar{D}_2-\bar{D}_3\nonumber \\
& \hskip 0.7cm+2 \bar{D}_4\bigl)+4 \bar{d}_1 \bar{H}_1-4 \bar{D}_3 \bar{H}_1-12 \bar{F}_1 \bar{F}_2+10 \bar{F}_2 \bar{F}_3+8 \bar{E}_1 \bar{G}_1\bigl)+4 \bar{b}_3 \bigl(2 \bar{c}_1 \bigl(\bar{b}_1 \bar{G}_2+\bar{C}_2 \bigl(\bar{F}_3-\bar{F}_1\bigl)\bigl)+6 \bar{F}_1\nonumber \\
& \hskip 0.7cm \times \bigl(\bar{C}_2^2-\bar{b}_1 \bar{D}_3\bigl)-\bar{F}_3 \bigl(\bar{b}_1 \bigl(-2 \bar{D}_2+\bar{D}_3-4 \bar{D}_4\bigl)+3 \bar{C}_2^2+8 \bar{F}_2\bigl)+8 \bar{E}_1 \bar{G}_2\bigl)+16 \bar{b}_3^2 \bigl(-2 \bar{b}_1 \bar{H}_2+2 \bar{C}_2 \bar{G}_1\nonumber \\
& \hskip 0.7cm -\bar{C}_2 \bar{G}_2+\bar{D}_2 \bar{F}_1\bigl)+8 \bar{b}_1 \bar{D}_3 \bar{H}_1+8 \bar{b}_1 \bar{D}_3 \bar{H}_2+16 \bar{b}_1 \bar{b}_3^3 \bar{F}_3+8 \bar{b}_1 \bar{F}_1 \bar{F}_4-4 \bar{b}_1 \bar{F}_3 \bar{F}_4-8 \bar{E}_1 \bar{b}_1 \bar{G}_2-4 \bar{c}_1 \bar{D}_2 \bar{G}_1\nonumber \\
& \hskip 0.7cm -8 \bar{c}_1 \bar{D}_3 \bar{G}_1-8 \bar{E}_1 \bar{c}_1 \bar{F}_1+4 \bar{C}_2 \bar{D}_2 \bar{G}_1-8 \bar{C}_2^2 \bar{H}_1+8 \bar{d}_1 \bar{D}_3 \bar{F}_1-4 \bar{d}_1 \bar{D}_2 \bar{F}_3+12 \bar{d}_1 \bar{D}_3 \bar{F}_3+4 \bar{D}_2^2 \bar{F}_1-8 \bar{D}_3^2 \bar{F}_1\nonumber \\
& \hskip 0.7cm +4 \bar{D}_2 \bar{D}_3 \bar{F}_1-16 \bar{D}_2 \bar{D}_4 \bar{F}_1-9 \bar{D}_2^2 \bar{F}_3-12 \bar{D}_3^2 \bar{F}_3+16 \bar{D}_2 \bar{D}_3 \bar{F}_3+24 \bar{D}_3 \bar{D}_4 \bar{F}_3\bigl) \, , \nonumber \\
\{\bar{G}_1, \bar{H}_2\}&=-\frac{\rm i}{32} \bigl(4 \bar{b}_3 \bigl(2 \bar{b}_2 \bar{d}_1 \bar{F}_1-3 \bar{b}_2 \bar{d}_1 \bar{F}_3+3 \bar{b}_2 \bar{D}_2 \bar{F}_1-2 \bar{b}_2 \bar{D}_3 \bar{F}_1+4 \bar{b}_2 \bar{D}_4 \bar{F}_1-3 \bar{b}_2 \bar{D}_2 \bar{F}_3+3 \bar{b}_2 \bar{D}_3 \bar{F}_3-6 \bar{b}_2 \bar{D}_4 \bar{F}_3\nonumber \\
& \hskip 0.7cm+6 \bar{c}_1 \bar{C}_2 \bar{F}_1+\bar{C}_2^2 \bigl(3 \bar{F}_3-7 \bar{F}_1\bigl)+4 \bar{H}_2 \bigl(\bar{d}_1-\bar{D}_3+2 \bar{D}_4\bigl)+4 \bar{F}_2 \bar{F}_3+4 \bar{E}_1 \bar{G}_1\bigl)+\bar{b}_1 \bigl(2 \bar{b}_2 \bigl(\bar{b}_3^2 \bigl(8 \bar{F}_1-\bar{F}_3\bigl)\nonumber \\
& \hskip 0.7cm-3 \bar{c}_1 \bar{G}_1+3 \bar{C}_2 \bar{G}_1-8 \bar{D}_4 \bar{F}_1+5 \bar{D}_4 \bar{F}_3\bigl)-4 \bigl(\bar{b}_3 \bar{D}_3 \bigl(2 \bar{F}_1-3 \bar{F}_3\bigl)+6 \bar{D}_3 \bar{H}_1-2 \bar{F}_1 \bar{F}_4+5 \bar{F}_3 \bar{F}_4\bigl)+\bar{c}_1 \bar{F}_3\nonumber \\
& \hskip 0.7cm \times \bigl(\bar{C}_2-\bar{c}_1\bigl)+24 \bar{E}_1 \bar{G}_2\bigl)-16 \bar{b}_3^2 \bigl(\bar{C}_2 \bar{G}_1+\bar{F}_3 \bigl(\bar{d}_1-\bar{D}_3+2 \bar{D}_4\bigl)\bigl)-8 \bar{b}_2 \bar{H}_2 \bigl(\bar{d}_1+\bar{D}_2-\bar{D}_3+2 \bar{D}_4\bigl)\nonumber \\
& \hskip 0.7cm-8 \bar{b}_2 \bar{F}_1 \bar{F}_2+4 \bar{b}_2 \bar{F}_2 \bar{F}_3+8 \bar{E}_1 \bar{b}_2 \bar{G}_1-12 \bar{c}_1 \bar{C}_2 \bar{H}_1+4 \bar{c}_1 \bar{d}_1 \bar{G}_2-4 \bar{c}_1 \bar{D}_3 \bar{G}_2+8 \bar{c}_1 \bar{D}_4 \bar{G}_2-12 \bar{E}_1 \bar{c}_1 \bar{F}_1\nonumber \\
& \hskip 0.7cm-8 \bar{E}_1 \bar{c}_1 \bar{F}_3-8 \bar{C}_2 \bar{D}_2 \bar{G}_1-16 \bar{C}_2 \bar{D}_4 \bar{G}_1+16 \bar{E}_1\bar{C}_2 \bar{F}_1+8 \bar{C}_2^2 \bar{H}_1+8 \bar{d}_1 \bar{D}_2 \bar{F}_3+4 \bar{d}_1 \bar{D}_3 \bar{F}_3-8 \bar{D}_2^2 \bar{F}_1 \nonumber \\
& \hskip 0.7cm+4 \bar{D}_2 \bar{D}_3 \bar{F}_1+9 \bar{D}_2^2 \bar{F}_3-4 \bar{D}_3^2 \bar{F}_3-4 \bar{D}_2 \bar{D}_3 \bar{F}_3+8 \bar{D}_3 \bar{D}_4 \bar{F}_3\bigl)\, ,\nonumber
\end{align} and
\begin{align}
\{\bar{G}_2,\bar{H}_2\}&=-\frac{{\rm i}}{16}  \bigl(-2 \bigl(-2 \bigl(\bar{D}_3 \bigl(-3 \bar{b}_1 \bar{H}_2+2 \bar{c}_1 \bar{G}_1-\bar{c}_1 \bar{G}_2+\bar{d}_1 \bar{F}_3+4 \bar{D}_4 \bar{F}_1\bigl)+6 \bar{b}_1 \bar{F}_3 \bar{F}_4+2 \bar{E}_1 \bar{c}_1 \bar{F}_3-11 \bar{E}_1 \bar{C}_2 \bar{F}_3 \nonumber \\
& \hskip 0.7cm+\bar{D}_3^2 \bigl(2 \bar{F}_1-\bar{F}_3\bigl)+4 \bar{F}_4 \bar{H}_2\bigl)+4 \bar{b}_3^2 \bigl(-2 \bar{C}_2 \bar{G}_2+4 \bar{D}_3 \bar{F}_1-\bar{D}_3 \bar{F}_3\bigl)+2 \bar{b}_3 \bigl(\bar{b}_1 \bar{D}_3 \bar{F}_3-4 \bar{C}_2^2 \bar{F}_3-8 \bar{F}_1 \bar{F}_4 \nonumber \\
& \hskip 0.7cm+4 \bar{F}_3 \bar{F}_4+8 \bar{E}_1 \bar{G}_2\bigl)+\bar{D}_2 \bigl(-2 \bar{c}_1 \bar{G}_2+\bar{D}_3 \bar{F}_1-2 \bigl(\bar{D}_3+2 \bar{D}_4\bigl) \bar{F}_3\bigl)+\bar{D}_2^2 \bar{F}_3\bigl)+\bar{b}_2 \bigl(\bar{c}_1 \bigl(4 \bar{b}_1 \bar{G}_2-7 \bar{C}_2 \bar{F}_1 \nonumber \\
& \hskip 0.7cm+\bar{C}_2 \bar{F}_3\bigl)-2 \bar{b}_3 \bigl(2 \bar{c}_1 \bar{G}_1+4 \bar{D}_4 \bar{F}_1+\bar{D}_2 \bigl(3 \bar{F}_1+\bar{F}_3\bigl)\bigl)-4 \bar{b}_1 \bar{C}_2 \bar{G}_2+4 \bigl(-\bar{b}_1 \bar{D}_4 \bar{F}_3+\bigl(\bar{D}_2+\bar{D}_3\bigl) \bar{H}_1\nonumber \\
& \hskip 0.7cm+\bar{F}_2 \bar{F}_3+\bar{E}_1\bar{G}_1\bigl)+8 \bar{b}_3^3 \bar{F}_1+\bar{c}_1^2 \bigl(\bar{F}_1-\bar{F}_3\bigl)+6 \bar{C}_2^2 \bar{F}_1-8 \bar{F}_1 \bar{F}_4\bigl)+4 \bar{b}_3 \bar{b}_2^2 \bigl(\bar{H}_1-\bar{b}_3 \bar{F}_1\bigl)\bigl) \, .
\end{align} 

All these relations conclude the degree-fourteen Poisson brackets. We now turn our attention to the degree-fifteen Poisson brackets. 

\subsection{Expansions in the degree $15$ brackets}

The compact forms in degree fifteen brackets are given by $\{\bar{\textbf{G}},\bar{\textbf{I}}\}$ and $\{\bar{\textbf{H}},\bar{\textbf{H}}\}$. 
 Beginning with $\{\bar{\textbf{G}},\bar{\textbf{I}}\}$, a direct computation shows that the gradings are  
\begin{align}
      \mathcal{G} \left(\{\bar{G}_1,\bar{I}_1\}\right) =& \, (1,4,10) \tilde{+} (2,3,10) \tilde{+}(3,4,8) \tilde{+} (2,5,8)  ; \nonumber\\
    \mathcal{G} \left(\{\bar{G}_1,\bar{I}_2\}\right) =& \, (4,1,10) \tilde{+} (5,0,10) \tilde{+}(6,1,8) \tilde{+} (5,2,8)  ; \nonumber \\
    \mathcal{G} \left(\{\bar{G}_2,\bar{I}_1\}\right) =& \,  (0,5,10) \tilde{+} (1,4,10) \tilde{+}(2,5,8) \tilde{+} (1,6,8)   ; \nonumber \\
      \mathcal{G} \left(\{\bar{G}_2,\bar{I}_2\}\right) =& \, (3,2,10) \tilde{+} (4,1,10) \tilde{+}(5,2,8) \tilde{+} (4,3,8).
\end{align} 
 We consider the case of $\{\bar{G}_1,\bar{I}_2\}$, the remaining cases being similar. From Proposition \ref{grading}, all the allowed polynomials from the homogeneous grading in $\mathcal{G} \left(\{\bar{G}_1,\bar{I}_2\}\right)$ are listed below:
\begin{align*} 
     (4,1,10) = & \, \left\{\bar{b}_3 \bar{D}_2 \bar{I}_2,\bar{b}_1 \bar{b}_3 \bar{D}_4 \bar{G}_1,\bar{b}_1 \bar{b}_3^3 \bar{G}_1,\bar{b}_3 \bar{F}_2 \bar{G}_1,C^{(003)} \bar{C}_2 \bar{I}_2,C^{(003)} C^{(202)} \bar{H}_1,C^{(003)} \bar{F}_1 \bar{F}_2,\bar{b}_1 \bar{b}_3 C^{(003)} \bar{H}_1,\bar{b}_1 \left(C^{(003)}\right)^2 \bar{G}_1,\right. \\
     & \, \left. \bar{b}_1 C^{(003)} \bar{D}_4 \bar{F}_1,C^{(202)} \bar{D}_4 \bar{G}_1,\bar{b}_3^2 C^{(202)} \bar{G}_1,\bar{b}_3 C^{(003)} C^{(202)} \bar{F}_1,\bar{b}_1 \bar{b}_3^2 C^{(003)} \bar{F}_1 \right\}; \\
     (5,0,10) = & \, \left\{\bar{F}_2 \bar{I}_2,\bar{b}_1^2 \bar{b}_3^2 \bar{G}_1,\bar{b}_1 \bar{D}_4 \bar{I}_2,\bar{b}_1 \bar{b}_3^2 \bar{I}_2 ,\bar{b}_3 C^{(202)} \bar{I}_2\right\}; \\
     (6,1,8) = & \, \left\{ \bar{b}_1 \bar{D}_2 \bar{I}_2,\bar{b}_1 \bar{F}_2 \bar{G}_1,\bar{b}_1^2 \bar{D}_4 \bar{G}_1,\left(C^{(202)}\right)^2 \bar{G}_1,\bar{b}_1^2 C^{(003)} \bar{H}_1,\bar{b}_1 \bar{b}_3 C^{(202)} \bar{G}_1,\bar{b}_1 C^{(003)} C^{(202)} \bar{F}_1,\bar{b}_1^2 \bar{b}_3 C^{(003)} \bar{F}_1\right\}  ; \\
     (5,2,8) = & \, \left\{ \bar{b}_1 \bar{D}_3 \bar{I}_2,\bar{b}_1 \bar{E}_1 \bar{H}_1,\bar{b}_1 \bar{F}_2 \bar{G}_2,\bar{b}_1 \bar{b}_3 \bar{D}_2 \bar{G}_1,\bar{b}_1 \bar{b}_3 \bar{E}_1 \bar{F}_1,\bar{b}_1^2 \bar{D}_4 \bar{G}_2,\bar{b}_1 \bar{b}_2 \bar{b}_3 \bar{I}_2,\bar{b}_1 \bar{C}_2 \bar{D}_4 \bar{F}_1,\bar{b}_1 \bar{b}_3 \bar{C}_2 \bar{H}_1,\bar{C}_2 \bar{F}_1 \bar{F}_2,\bar{b}_1^2 \bar{b}_3^2 \bar{G}_2,\bar{b}_1 \bar{b}_3^2 \bar{C}_2 \bar{F}_1, \right. \\
     & \, \left.\bar{b}_2 C^{(202)} \bar{I}_2,\bar{C}_2^2 \bar{I}_2,\bar{C}_2 C^{(202)} \bar{H}_1,\left(C^{(202)}\right)^2 \bar{G}_2,C^{(202)} \bar{D}_2 \bar{G}_1,C^{(202)} \bar{E}_1 \bar{F}_1,\bar{b}_1^2 C^{(003)} \bar{H}_2,\bar{b}_1 \bar{b}_3 C^{(202)} \bar{G}_2,\bar{b}_1 C^{(003)} \bar{C}_2 \bar{G}_1,\right. \\
     & \, \left.\bar{b}_1 C^{(003)} C^{(202)} \bar{F}_3,\bar{b}_1 C^{(003)} \bar{D}_2 \bar{F}_1,\bar{b}_3 \bar{C}_2 C^{(202)} \bar{F}_1,\bar{b}_1^2 \bar{b}_3 C^{(003)} \bar{F}_3\right\}.
\end{align*} %Throughout the change of the generators, all the suitable generators are given by \begin{align*}
As usual, by substituting the generators $C^{(003)}$ and $C^{(202)}$ with $\bar{c}_1$ and $\bar{d}_1$, we find $117$ polynomials within the bracket $ \{\bar{G}_1,\bar{I}_2\}$, such that  
\begin{align*}
    \{\bar{G}_1, \bar{I}_2\} = \Gamma_{79}^1  \bar{b}_3 \bar{D}_2 \bar{I}_2 + \ldots + \Gamma_{79}^{116} \bar{b}_3 \bar{d}_1  \bar{C}_2  \bar{F}_1 + \Gamma_{79}^{117} \bar{b}_1^2 \bar{b}_3 \bar{c}_1 \bar{F}_3.
\end{align*} Here, the coefficients $\Gamma_{79}^1,\ldots,\Gamma_{79}^{117}$ are determined taking into account the  Poisson bracket relations of $\mathfrak{su}^*(4)$.  After finding the explicit values of the coefficients, the  expansion for $\{\bar{G}_1, \bar{I}_2\}$ reads  
\begin{align}
    \{\bar{G}_1, \bar{I}_2\} & =-\frac{{\rm i}}{8}  \bigl(\bar{b}_2 \bigl(\bar{c}_1 \bigl(2 \bar{b}_3 \bigl(\bar{b}_3 \bar{F}_3-\bar{H}_2\bigl)+\bar{D}_3 \bar{F}_1-\bigl(\bar{D}_3+2 \bar{D}_4\bigl) \bar{F}_3\bigl)+2 \bar{b}_3 \bar{D}_3 \bar{G}_2+\bar{C}_2 \bar{D}_3 \bigl(\bar{F}_3-\bar{F}_1\bigl)+4 \bar{F}_4 \bar{G}_1 \nonumber \\
& \hskip 0.7cm-4 \bar{F}_4 \bar{G}_2\bigl)-2 \bar{b}_3 \bigl(3 \bar{c}_1 \bar{D}_3 \bar{F}_3+4 \bar{F}_4 \bar{G}_2\bigl)+6 \bar{c}_1 \bigl(\bar{D}_3 \bar{H}_2+\bar{F}_3 \bar{F}_4\bigl)+2 \bar{D}_3 \bigl(4 \bar{D}_4 \bar{G}_2+\bar{D}_3 \bigl(\bar{G}_2-4 \bar{G}_1\bigl)\bigl)\bigl).
\end{align}
  Analogously, the remaining relations are computed: 
\begin{align}
\{\bar{G}_1, \bar{I}_1\}&=-\frac{{\rm i}}{32} \bigl(\bar{b}_2 \bigl(\bar{c}_1 \bigl(\bar{b}_1 \bigl(4 \bar{b}_3 \bar{F}_3-6 \bar{H}_2\bigl)+8 \bar{C}_2 \bigl(\bar{G}_1+2 \bar{G}_2\bigl)-\bigl(\bigl(\bar{D}_2+10 \bar{D}_3-8 \bar{D}_4\bigl) \bar{F}_1\bigl)\bigl)+\bar{C}_2 \bigl(4 \bar{b}_1 \bar{b}_3 \bar{F}_3 \nonumber \\
& \hskip 0.7cm-2 \bar{b}_1 \bar{H}_2-8 \bar{d}_1 \bar{F}_3+5 \bar{D}_2 \bar{F}_1-6 \bar{D}_3 \bar{F}_1+8 \bar{D}_3 \bar{F}_3\bigl)+4 \bigl(\bar{b}_3 \bigl(-2 \bar{b}_1 \bar{I}_1+\bar{D}_2 \bar{G}_1+2 \bar{D}_3 \bar{G}_1+\bar{D}_2 \bar{G}_2 \nonumber \\
& \hskip 0.7cm-10 \bar{E}_1 \bar{F}_1+8 \bar{E}_1 \bar{F}_3\bigl)-2 \bar{b}_1 \bar{D}_4 \bar{G}_2+2 \bar{b}_1 \bar{b}_3^2 \bar{G}_2+12 \bar{I}_1 \bigl(\bar{d}_1+\bar{D}_2-\bar{D}_3+2 \bar{D}_4\bigl)+20 \bar{F}_4 \bar{G}_1-6 \bar{F}_2 \bar{G}_2 \nonumber \\
& \hskip 0.7cm+12 \bar{E}_1 \bar{H}_1-32 \bar{E}_1 \bar{H}_2\bigl)\bigl)+2 \bar{D}_3 \bigl(\bar{b}_1 \bar{c}_1 \bar{F}_3-\bar{b}_1 \bar{C}_2 \bar{F}_3+20 \bar{b}_1 \bar{I}_1+6 \bar{D}_2 \bar{G}_1-16 \bar{D}_4 \bar{G}_1-10 \bar{D}_2 \bar{G}_2 \nonumber \\
& \hskip 0.7cm+8 \bar{E}_1 \bigl(2 \bar{F}_1+5 \bar{F}_3\bigl)\bigl)+4 \bar{b}_3 \bigl(-4 \bar{b}_1 \bar{D}_3 \bar{G}_2+\bar{c}_1 \bigl(2 \bar{C}_2 \bar{G}_2-2 \bar{D}_3 \bar{F}_1+\bar{D}_2 \bar{F}_3\bigl)+8 \bar{F}_4 \bar{G}_1\bigl)-2 \bar{b}_2^2 \nonumber \\
& \hskip 0.7cm \times \bigl(\bar{b}_3 \bigl(\bar{c}_1 \bar{F}_1-9 \bar{C}_2 \bar{F}_1+4 \bar{I}_2\bigl)+8 \bar{C}_2 \bar{H}_1\bigl)-4 \bar{G}_2 \bigl(-4 \bar{b}_1 \bar{F}_4+6 \bar{E}_1 \bar{c}_1+\bar{D}_2^2\bigl)+6 \bar{C}_2 \bar{F}_3 \bigl(-3 \bar{c}_1 \bar{C}_2 \nonumber \\
& \hskip 0.7cm+\bar{c}_1^2+2 \bar{C}_2^2-8 \bar{F}_4\bigl)-32 \bar{E}_1 \bar{D}_2 \bar{F}_3-40 \bar{D}_3^2 \bar{G}_1\bigl) \, \end{align} and \begin{align}
\{\bar{G}_2, \bar{I}_1\}&=-\frac{{\rm i}}{32}  \bigl(2 \bar{b}_1 \bigl(-\bar{b}_3 \bigl(-2 \bar{b}_2 \bar{I}_2+4 \bar{c}_1 \bar{H}_1+2 \bar{G}_2 \bigl(\bar{d}_1-\bar{D}_3+2 \bar{D}_4\bigl)+\bar{D}_2 \bigl(\bar{G}_1+2 \bar{G}_2\bigl)\bigl)+4 \bar{b}_3^2 \bar{c}_1 \bar{F}_1+\bar{c}_1 \bar{F}_1 \nonumber \\
& \hskip 0.7cm \times \bigl(-\bar{d}_1+\bar{D}_3-6 \bar{D}_4\bigl)+\bar{c}_1 \bar{F}_3 \bigl(\bar{d}_1+\bar{D}_2-\bar{D}_3+2 \bar{D}_4\bigl)+\bar{C}_2 \bar{F}_1 \bigl(\bar{d}_1-\bar{D}_3\bigl)-\bar{C}_2 \bar{F}_3\bigl(\bar{d}_1+\bar{D}_2 \nonumber \\
& \hskip 0.7cm -\bar{D}_3+2 \bar{D}_4\bigl)-2 \bigl(\bar{D}_2+2 \bar{D}_3\bigl) \bar{I}_2-4 \bar{F}_2 \bar{G}_1+8 \bar{F}_2 \bar{G}_2+4 \bar{E}_1 \bar{H}_1\bigl)+\bar{b}_1^2 \bigl(3 \bar{b}_3 \bar{F}_3 \bigl(\bar{c}_1-\bar{C}_2\bigl)+4 \bar{b}_3^2 \bar{G}_1 \nonumber \\
& \hskip 0.7cm -2 \bigl(\bar{H}_1+\bar{H}_2\bigl) \bigl(\bar{c}_1-\bar{C}_2\bigl)-4 \bar{D}_4 \bar{G}_1\bigl)+8 \bigl(-3 \bar{b}_3 \bar{c}_1 \bar{F}_1 \bigl(\bar{d}_1+\bar{D}_2-\bar{D}_3+2 \bar{D}_4\bigl)-4 \bar{b}_3 \bar{F}_2 \bar{G}_1+3 \bar{c}_1 \nonumber \\
& \hskip 0.7cm \times \bigl(\bar{H}_1 \bigl(\bar{d}_1+\bar{D}_2-\bar{D}_3+2 \bar{D}_4\bigl)+\bar{F}_1 \bar{F}_2\bigl)+\bigl(\bar{d}_1+\bar{D}_2-\bar{D}_3+2 \bar{D}_4\bigl) \bigl(\bar{G}_1 \bigl(\bar{d}_1+\bar{D}_2-\bar{D}_3+6 \bar{D}_4\bigl)\nonumber \\
& \hskip 0.7cm -4 \bar{G}_2 \bigl(\bar{d}_1+\bar{D}_2-\bar{D}_3+2 \bar{D}_4\bigl)\bigl)\bigl)\bigl) \, ,\nonumber \\
\{\bar{G}_2, \bar{I}_2\}&=\frac{{\rm i}}{16}  \bigl(-6 \bar{F}_1 \bar{C}_2^3-3 \bar{c}_1^2 \bar{F}_1 \bar{C}_2+15 \bar{b}_2 \bar{d}_1 \bar{F}_1 \bar{C}_2+15 \bar{b}_2 \bar{D}_2 \bar{F}_1 \bar{C}_2-15 \bar{b}_2 \bar{D}_3 \bar{F}_1 \bar{C}_2+30 \bar{b}_2 \bar{D}_4 \bar{F}_1 \bar{C}_2-56 \bar{b}_3 \bar{D}_4 \bar{F}_1 \bar{C}_2\nonumber \\
& \hskip 0.7cm+32 \bar{F}_1 \bar{F}_2 \bar{C}_2-48 \bar{d}_1 \bar{E}_1 \bar{F}_1-32 \bar{D}_2 \bar{E}_1 \bar{F}_1+48 \bar{D}_3 \bar{E}_1 \bar{F}_1+16 \bar{D}_4 \bar{E}_1 \bar{F}_1-16 \bar{d}_1 \bar{E}_1 \bar{F}_3-16 \bar{D}_2 \bar{E}_1 \bar{F}_3\nonumber \\
& \hskip 0.7cm+16 \bar{D}_3 \bar{E}_1 \bar{F}_3-32 \bar{D}_4 \bar{E}_1 \bar{F}_3+16 \bar{D}_2^2 \bar{G}_1+64 \bar{b}_2 \bar{b}_3 \bar{d}_1 \bar{G}_1+64 \bar{b}_2 \bar{b}_3 \bar{D}_2 \bar{G}_1+14 \bar{d}_1 \bar{D}_2 \bar{G}_1-64 \bar{b}_2 \bar{b}_3 \bar{D}_3 \bar{G}_1\nonumber \\
& \hskip 0.7cm-14 \bar{D}_2 \bar{D}_3 \bar{G}_1+128 \bar{b}_2 \bar{b}_3 \bar{D}_4 \bar{G}_1-64 \bar{b}_2 \bar{F}_2 \bar{G}_1+\bar{c}_1 \bigl(9 \bar{F}_1 \bar{C}_2^2-4 \bar{b}_3 \bar{G}_1 \bar{C}_2-2 \bar{b}_3 \bar{D}_2 \bar{F}_1-15 \bar{b}_2 \bigl(\bar{d}_1+\bar{D}_2\nonumber \\
& \hskip 0.7cm-\bar{D}_3+2 \bar{D}_4\bigl) \bar{F}_1+4 \bar{b}_3 \bigl(\bar{d}_1+\bar{D}_2-\bar{D}_3+2 \bar{D}_4\bigl) \bar{F}_3+12 \bar{E}_1 \bar{G}_1\bigl)+20 \bar{d}_1^2 \bar{G}_2+14 \bar{D}_2^2 \bar{G}_2+20 \bar{D}_3^2 \bar{G}_2\nonumber \\
& \hskip 0.7cm+34 \bar{d}_1 \bar{D}_2 \bar{G}_2-40 \bar{d}_1 \bar{D}_3 \bar{G}_2-34 \bar{D}_2 \bar{D}_3 \bar{G}_2+40 \bar{d}_1 \bar{D}_4 \bar{G}_2+28 \bar{D}_2 \bar{D}_4 \bar{G}_2-40 \bar{D}_3 \bar{D}_4 \bar{G}_2-16 \bar{b}_3 \bar{F}_2 \bar{G}_2\nonumber \\
& \hskip 0.7cm+\bar{b}_1 \bigl(-4 \bar{C}_2 \bar{F}_1 \bar{b}_3^2+4 \bigl(-4 \bar{E}_1 \bar{F}_3+5 \bar{D}_3 \bar{G}_1+2 \bar{D}_2 \bar{G}_2\bigl) \bar{b}_3+2 \bar{c}_1 \bar{D}_2 \bar{F}_1+2 \bar{C}_2 \bar{D}_2 \bar{F}_1-10 \bar{c}_1\bar{D}_3 \bar{F}_1\nonumber \\
& \hskip 0.7cm-6 \bar{C}_2 \bar{D}_3 \bar{F}_1-5 \bar{c}_1 \bar{d}_1 \bar{F}_3-3 \bar{C}_2 \bar{d}_1 \bar{F}_3+5 \bar{c}_1 \bar{D}_3 \bar{F}_3+3 \bar{C}_2 \bar{D}_3 \bar{F}_3+8 \bar{c}_1 \bar{C}_2 \bar{G}_1-24 \bar{F}_4 \bar{G}_1+16 \bar{c}_1 \bar{C}_2 \bar{G}_2\nonumber \\
& \hskip 0.7cm-32 \bar{F}_2 \bar{G}_2-2 \bar{b}_2 \bigl(16 \bar{G}_1 \bar{b}_3^2-2 \bigl(3 \bar{c}_1+\bar{C}_2\bigl) \bar{F}_1 \bar{b}_3-16 \bar{D}_4 \bar{G}_1+\bar{c}_1 \bar{H}_1+7 \bar{C}_2 \bar{H}_1\bigl)+8 \bar{E}_1 \bigl(4 \bar{H}_1+\bar{H}_2\bigl)\nonumber \\
& \hskip 0.7cm+4 \bigl(\bar{d}_1+\bar{D}_2-\bar{D}_3+2 \bar{D}_4\bigl) \bar{I}_1\bigl)+8 \bar{b}_1^2 \bigl(\bar{b}_3 \bar{c}_1 \bar{F}_3-\bar{c}_1 \bar{H}_2-2 \bar{b}_3 \bar{I}_1\bigl)+4 \bigl(2 \bar{C}_2^2+7 \bar{b}_2 \bigl(\bar{D}_3-\bar{d}_1\bigl)\bigl) \bar{I}_2\bigl) \, . %,\nonumber \\
\end{align}

Finally, for $\{\bar{\textbf{H}},\bar{\textbf{H}}\}$, which consists of the only Poisson bracket $\{\bar{H}_1,\bar{H}_2\}$, the grading reads
\begin{align}
    \mathcal{G} \left(\{\bar{H}_1,\bar{H}_2\}\right) = (2,3,10) \tilde{+}  (3,2,10) \tilde{+}  (4,3,8) \tilde{+}  (3,4,8) . \label{eq:grah1h2}
\end{align} 
 Proceeding in the usual way, the components in $\{\bar{H}_1,\bar{H}_2\}$ are given by
\begin{align}
\{\bar{H}_1, \bar{H}_2\}&=-\frac{{\rm i}}{64}  \bigl(-16 \bar{C}_2 \bigl(\bar{F}_1-2 \bar{F}_3\bigl) \bar{b}_3^3+16 \bigl(\bar{E}_1 \bigl(\bar{F}_1-2 \bar{F}_3\bigl)+\bar{D}_2 \bar{G}_1+\bigl(\bar{d}_1-\bar{D}_3+2 \bar{D}_4\bigl) \bar{G}_2-\bar{C}_2 \bar{H}_1\bigl) \bar{b}_3^2\nonumber \\
& \hskip 0.7cm+4 \bar{b}_1^2 \bar{I}_1 \bar{b}_3-8 \bigl(-2 \bar{I}_2 \bar{b}_2^2+14 \bar{E}_1 \bar{F}_1 \bar{b}_2+\bar{D}_2 \bar{G}_1 \bar{b}_2-4 \bar{C}_2 \bigl(\bar{D}_2+\bar{D}_4\bigl) \bar{F}_1+4 \bar{C}_2^2 \bar{G}_1-8 \bar{F}_4 \bar{G}_1\nonumber \\
& \hskip 0.7cm+\bar{c}_1 \bigl(\bar{D}_2 \bar{F}_1+2 \bar{D}_3 \bar{F}_1+3 \bar{d}_1 \bar{F}_3-3 \bar{D}_3 \bar{F}_3+6 \bar{D}_4 \bar{F}_3-2 \bar{C}_2 \bar{G}_1\bigl)+8 \bar{F}_2 \bar{G}_2\bigl) \bar{b}_3+6 \bar{c}_1^2 \bar{C}_2 \bar{F}_1\nonumber \\
& \hskip 0.7cm-\bar{b}_1 \bar{b}_2 \bigl(8 \bigl(\bar{G}_1+2 \bar{G}_2\bigl) \bar{b}_3^2+\bigl(-9 \bar{c}_1 \bar{F}_1+5 \bar{C}_2 \bar{F}_1-8 \bar{c}_1 \bar{F}_3+8 \bar{C}_2 \bar{F}_3\bigl) \bar{b}_3+8 \bar{D}_4 \bigl(\bar{G}_1-4 \bar{G}_2\bigl)+4 \bar{C}_2 \bar{H}_1\nonumber \\
& \hskip 0.7cm+8 \bar{c}_1 \bar{H}_2-8 \bar{C}_2 \bar{H}_2\bigl)-2 \bar{c}_1 \bigl(3 \bar{F}_1 \bar{C}_2^2+2 \bigl(\bar{b}_2 \bigl(\bar{d}_1-\bar{D}_3+4 \bar{D}_4\bigl) \bar{F}_1+\bar{b}_2 \bigl(\bar{d}_1+\bar{D}_2-\bar{D}_3+2 \bar{D}_4\bigl) \bar{F}_3 \nonumber \\
& \hskip 0.7cm-6 \bar{F}_2 \bar{F}_3-6 \bar{E}_1 \bar{G}_1+8 \bar{E}_1 \bar{G}_2-2 \bigl(\bar{d}_1-\bar{D}_3+2 \bar{D}_4\bigl) \bar{H}_2\bigl)\bigl)-2 \bar{b}_1 \bigl(2 \bar{G}_2 \bar{c}_1^2+\bigl(-10 \bar{F}_3 \bar{b}_3^2+8 \bar{H}_2 \bar{b}_3 \nonumber \\
& \hskip 0.7cm+4 \bar{D}_3 \bar{F}_1+\bar{D}_2 \bar{F}_3+4 \bar{D}_4 \bar{F}_3-4 \bar{C}_2 \bar{G}_2\bigl) \bar{c}_1+12 \bar{b}_3 \bar{E}_1 \bar{F}_3-2 \bar{C}_2 \bigl(\bar{D}_3 \bar{F}_1+\bar{D}_2 \bar{F}_3\bigl)-6 \bar{b}_3 \bar{D}_3 \bar{G}_1 \nonumber \\
& \hskip 0.7cm-16 \bar{F}_4 \bar{G}_1+\bar{b}_3 \bar{D}_2 \bar{G}_2-24 \bar{b}_3 \bar{D}_3 \bar{G}_2+32\bar{F}_4 \bar{G}_2-8 \bar{E}_1 \bar{H}_2-8 \bar{d}_1 \bar{I}_1+8 \bar{D}_3 \bar{I}_1\bigl)+4 \bigl(\bigl(\bar{G}_1+5 \bar{G}_2\bigl) \bar{D}_2^2 \nonumber \\
& \hskip 0.7cm-2 \bigl(\bar{E}_1 \bigl(7 \bar{F}_1+4 \bar{F}_3\bigl)-2 \bigl(\bar{D}_3+\bar{D}_4\bigl) \bar{G}_1+2 \bar{D}_3 \bar{G}_2\bigl) \bar{D}_2+4 \bigl(\bigl(\bar{G}_1+\bar{G}_2\bigl) \bar{D}_3^2+\bigl(\bar{E}_1 \bigl(6 \bar{F}_1-3 \bar{F}_3\bigl)\nonumber \\
& \hskip 0.7cm-2 \bar{D}_4 \bar{G}_2-\bar{d}_1 \bigl(\bar{G}_1+\bar{G}_2\bigl)\bigl) \bar{D}_3+\bar{E}_1 \bigl(2 \bar{D}_4 \bar{F}_1+3 \bar{d}_1 \bar{F}_3-5 \bar{C}_2 \bar{G}_1\bigl)\bigl)+\bar{b}_2 \bigl(\bar{C}_2 \bigl(\bar{d}_1-\bar{D}_3\bigl) \bar{F}_1 \nonumber \\
& \hskip 0.7cm+\bar{C}_2 \bigl(\bar{d}_1+\bar{D}_2-\bar{D}_3+2 \bar{D}_4\bigl) \bar{F}_3+4 \bar{F}_2 \bigl(\bar{G}_1-4 \bar{G}_2\bigl)+40 \bar{E}_1 \bar{H}_1-2 \bigl(\bar{D}_2+6 \bar{D}_3\bigl) \bar{I}_2\bigl)\bigl)\bigl) \, .
\end{align}

It is evident that all expansions in the Poisson brackets above have been shown to yield relations of degree fifteen. We now proceed to a detailed exploration of the expansions observed in the Poisson brackets of degrees sixteen and seventeen.

\subsection{Expansions in the degree $16$ and $17$}

 We first have a look at the degree sixteen expansion. In other words, we have to consider the compact form $\{\bar{\textbf{H}},\bar{\textbf{I}}\}$ including the Poisson brackets $\{\bar{H}_1,\bar{I}_1\},\{\bar{H}_1,\bar{I}_2\} ,\{\bar{H}_2,\bar{I}_1\}$ and $ \{\bar{H}_2,\bar{I}_2\}.$ Using Lemma \ref{4.1} and Proposition \ref{grading}, each of the gradings is given by
\begin{align*}
         \mathcal{G} \left(\{\bar{H}_1,\bar{I}_1\}\right) =& \, (1,4,11) \tilde{+} (2,3,11) \tilde{+}(3,4,9) \tilde{+} (2,5,9)  ; \\
    \mathcal{G} \left(\{\bar{H}_1,\bar{I}_2\}\right) =& \, (4,1,11) \tilde{+} (5,0,11) \tilde{+}(6,1,9) \tilde{+} (5,2,9)  ;
    \end{align*} and \begin{align*}
    \mathcal{G} \left(\{\bar{H}_2,\bar{I}_1\}\right) =& \,  (0,5,11) \tilde{+} (1,4,11) \tilde{+}(2,5,9) \tilde{+} (1,6,9)   ; \\
      \mathcal{G} \left(\{\bar{H}_2,\bar{I}_2\}\right) =& \, (3,2,11) \tilde{+} (4,1,11) \tilde{+}(5,2,9) \tilde{+} (4,3,9).
\end{align*} 

As previously detailed, we shall enumerate all admissible polynomials associated with one of the grading structures described earlier. For illustrative purposes, consider the polynomials that are permissible within the context of the homogeneous grading, specifically derived from the Poisson bracket denoted by $\{\bar{H}_1,\bar{I}_1\}$, which are as follows:
\begin{align*}
     (1,4,11) = &\, \left\{\bar{b}_3 \bar{E}_1 \bar{I}_1, \text{ } \bar{b}_3^3 \bar{D}_3 \bar{F}_3, \text{ } \bar{b}_2 \bar{b}_3^4 \bar{F}_3, \text{ } \bar{b}_3 \bar{F}_4 \bar{H}_2, \text{ } \bar{C}_2 \bar{D}_4 \bar{I}_1, \text{ } \bar{D}_3 \bar{D}_4 \bar{H}_2, \text{ } \bar{b}_2 \bar{b}_3^2 \bar{D}_4 \bar{F}_3, \text{ } \bar{D}_4 \bar{F}_3 \bar{F}_4, \text{ } \bar{b}_2 \bar{b}_3 \bar{D}_4 \bar{H}_2, \text{ }   \bar{b}_3 \bar{D}_3 \bar{D}_4 \bar{F}_3 , \text{ } \bar{b}_2 \bar{D}_4^2 \bar{F}_3,\right. \\
     &\,   \bar{b}_3^2 \bar{C}_2 \bar{I}_1,\text{ }  \bar{b}_3^2 \bar{F}_3 \bar{F}_4, \text{ } \bar{b}_3^2 \bar{D}_3 \bar{H}_2, \text{ } \bar{b}_2 \bar{b}_3^3 \bar{H}_2, \text{ } \bar{b}_2 \left(C^{(003)}\right)^2 \bar{H}_2, \text{ } \bar{b}_2 C^{(003)} \bar{D}_4 \bar{G}_2, \text{ }  \bar{b}_3 C^{(003)} \bar{D}_3 \bar{G}_2, \text{ }  \left(C^{(003)}\right)^2 \bar{D}_3 \bar{F}_3,   \\
     &\,  \left.C^{(003)} \bar{D}_2 \bar{I}_1, \text{ } C^{(003)} \bar{F}_4 \bar{G}_2, \text{ } \bar{b}_2 \bar{b}_3^2 C^{(003)} \bar{G}_2, \text{ } \bar{b}_2 \bar{b}_3 \left(C^{(003)}\right)^2 \bar{F}_3\right\}; \\
     (2,3,11) = &\, \left\{\bar{D}_2 \bar{D}_4 \bar{H}_2, \text{ } \bar{D}_3 \bar{D}_4 \bar{H}_1, \text{ } \bar{D}_4 \bar{E}_1 \bar{G}_2, \text{ } \bar{D}_4 \bar{F}_1 \bar{F}_4, \text{ } \bar{b}_2 \bar{b}_3 \bar{D}_4 \bar{H}_1, \text{ } \bar{b}_3^3 \bar{C}_2 \bar{G}_2, \text{ } \bar{b}_3^3 \bar{D}_2 \bar{F}_3, \text{ } \bar{b}_3^3 \bar{D}_3 \bar{F}_1, \text{ }  \bar{b}_3^2 \bar{C}_2^2 \bar{F}_3, \text{ } \bar{b}_1 \bar{b}_2 \bar{b}_3^3 \bar{F}_3,  \right. \\
     & \,     \bar{b}_3 \bar{C}_2 \bar{D}_4 \bar{G}_2, \text{ } \bar{b}_3 \bar{D}_2 \bar{D}_4 \bar{F}_3, \text{ } \bar{b}_3 \bar{D}_3 \bar{D}_4 \bar{F}_1, \text{ } \bar{b}_2 \bar{b}_3^2 \bar{D}_4 \bar{F}_1,  \text{ }  \bar{b}_2 \bar{b}_3^3 \bar{H}_1, \bar{b}_2 \bar{D}_4^2 \bar{F}_1, \text{ } \bar{b}_3^2 \bar{D}_2 \bar{H}_2, \text{ } \bar{b}_3^2 \bar{D}_3 \bar{H}_1, \text{ } \bar{b}_3^2 \bar{E}_1 \bar{G}_2, \text{ } \\ 
     &\, \bar{b}_3^2 \bar{F}_1 \bar{F}_4, \text{ } \bar{b}_3 \bar{F}_4 \bar{H}_1, \text{ } \bar{b}_2 \bar{b}_3^4 \bar{F}_1, \text{ } C^{(003)} C^{(202)} \bar{I}_1,\text{ } C^{(003)} \bar{E}_1 \bar{H}_2, \text{ } C^{(003)} \bar{F}_4 \bar{G}_1, \text{ } \bar{b}_2 \left(C^{(003)}\right)^2 \bar{H}_1, \text{ }  \\
     &\,\bar{b}_2 C^{(003)} \bar{D}_4, \text{ }  \bar{G}_1\bar{b}_3 C^{(003)} \bar{C}_2 \bar{H}_2, \text{ } \bar{b}_3 C^{(003)} \bar{E}_1 \bar{F}_3, \text{ } \left(C^{(003)}\right)^2 \bar{C}_2 \bar{G}_2, \text{ }\left(C^{(003)}\right)^2 \bar{D}_3 \bar{F}_1, \text{ } \\
     & \, C^{(003)} \bar{C}_2 \bar{D}_4 \bar{F}_3, \text{ } \bar{b}_3 C^{(003)} \bar{D}_2 \bar{G}_2, \text{ }  \bar{b}_3 C^{(003)} \bar{D}_3 \bar{G}_1, \text{ } \bar{b}_1 \bar{b}_3 C^{(003)} \bar{I}_1, \text{ } \bar{b}_2 \bar{b}_3^2 C^{(003)} \bar{G}_1, \\
     &\, \left.\left(C^{(003)}\right)^2 \bar{D}_2 \bar{F}_3, \text{ } \bar{b}_2 \bar{b}_3 \left(C^{(003)}\right)^2 \bar{F}_1, \text{ } \bar{b}_3^2 C^{(003)} \bar{C}_2 \bar{F}_3\right\} ; \\
     (3,  4,   9) = &\, \text{ }  \left\{\bar{b}_1 \bar{E}_1 \bar{I}_1, \text{ } \bar{b}_1 \bar{F}_4 \bar{H}_2, \text{ } \bar{b}_2 \bar{F}_2 \bar{H}_2, \text{ } \bar{C}_2 \bar{E}_1 \bar{H}_2, \text{ } \bar{C}_2 \bar{F}_4 \bar{G}_1, \text{ } \bar{D}_2^2 \bar{H}_2, \text{ } \bar{D}_2 \bar{D}_3 \bar{H}_1, \text{ } \bar{D}_2 \bar{E}_1 \bar{G}_2, \text{ } \bar{b}_2 \bar{b}_3 \bar{D}_2 \bar{H}_1, \text{ } \bar{b}_2 \bar{b}_3 \bar{E}_1 \bar{G}_1,  \right. \\
     & \,  \bar{b}_2 \bar{b}_3 \bar{F}_2 \bar{F}_3, \text{ } \bar{b}_1 \bar{b}_2 \bar{b}_3^2 \bar{H}_2, \text{ } \bar{D}_2 \bar{F}_1 \bar{F}_4, \text{ } \bar{b}_1 \bar{b}_3^2 \bar{D}_3 \bar{F}_3, \text{ }  \bar{D}_3 \bar{E}_1 \bar{G}_1, \text{ } \bar{D}_3 \bar{F}_2 \bar{F}_3, \text{ } \bar{E}_1^2 \bar{F}_3, \text{ } \bar{b}_1 \bar{b}_2 \bar{D}_4 \bar{H}_2,\text{ }  \bar{b}_1 \bar{b}_3 \bar{C}_2 \bar{I}_1,  \\ 
     & \, \bar{b}_1 \bar{b}_3 \bar{D}_3 \bar{H}_2, \text{ } \bar{b}_1 \bar{b}_3 \bar{F}_3 \bar{F}_4, \text{ } \bar{b}_2 \bar{D}_2 \bar{D}_4 \bar{F}_1, \text{ }  \bar{b}_3 \bar{D}_2^2 \bar{F}_3, \text{ }  \bar{b}_3 \bar{C}_2^2 \bar{H}_2, \text{ }  \bar{C}_2^2 \bar{D}_4 \bar{F}_3, \text{ } \bar{b}_3 \bar{C}_2 \bar{D}_2 \bar{G}_2, \text{ } \bar{b}_3 \bar{C}_2 \bar{D}_3 \bar{G}_1,  \\ 
     &\,    \bar{b}_3 \bar{C}_2 \bar{E}_1 \bar{F}_3, \text{ } \bar{b}_3 \bar{D}_2 \bar{D}_3 \bar{F}_1,\text{ } \bar{b}_2 \bar{C}_2 \bar{D}_4 \bar{G}_1, \text{ }\bar{b}_2 \bar{b}_3^2 \bar{C}_2 \bar{G}_1, \text{ } \bar{b}_2 \bar{b}_3^2 \bar{D}_2 \bar{F}_1, \text{ }    \bar{b}_1 \bar{D}_3 \bar{D}_4 \bar{F}_3, \text{ } \bar{b}_1 \bar{b}_2 \bar{b}_3 \bar{D}_4 \bar{F}_3, \\ 
     &\,    \bar{b}_2^2 C^{(003)} \bar{I}_2, \text{ } \bar{C}_2 C^{(202)} \bar{I}_1, \text{ }  C^{(003)} \bar{C}_2 \bar{D}_2 \bar{F}_3,\text{ } C^{(202)} \bar{D}_3 \bar{H}_2, \text{ } C^{(202)} \bar{F}_3 \bar{F}_4, \text{ }  \bar{b}_2 \bar{b}_3 C^{(202)} \bar{H}_2,   \\
     &\, \bar{b}_2 C^{(003)} \bar{C}_2 \bar{H}_1, \text{ } \bar{b}_2 C^{(003)} C^{(202)} \bar{G}_2,   \text{ } C^{(003)} \bar{C}_2 \bar{D}_3 \bar{F}_1, \text{ } \bar{b}_2 C^{(003)} \bar{D}_2 \bar{G}_1, \text{ } \bar{b}_2 C^{(003)} \bar{E}_1 \bar{F}_1,\\
     &\, \text{ }  \bar{b}_2 C^{(202)} \bar{D}_4 \bar{F}_3, \text{ } \bar{b}_3 C^{(202)} \bar{D}_3 \bar{F}_3, \text{ } C^{(003)} \bar{C}_2^2 \bar{G}_2,\text{ } \bar{b}_1 \bar{b}_2 \bar{b}_3 C^{(003)} \bar{G}_2, \text{ } \bar{b}_1 C^{(003)} \bar{D}_3 \bar{G}_2, \\
     &\,  \left.\bar{b}_1 \bar{b}_2 \left(C^{(003)}\right)^2 \bar{F}_3, \text{ }  \bar{b}_2 \bar{b}_3^2 C^{(202)} \bar{F}_3, \text{ } \bar{b}_2 \bar{b}_3 C^{(003)} \bar{C}_2 \bar{F}_1\right\} 
\end{align*} and
\begin{align*}
     (2,5,9) = &\, \left\{\bar{b}_2 \bar{F}_4 \bar{H}_1, \text{ } \bar{C}_2 \bar{D}_2 \bar{I}_1, \text{ } \bar{C}_2 \bar{F}_4 \bar{G}_2, \text{ } \bar{D}_2 \bar{D}_3 \bar{H}_2, \text{ } \bar{D}_2 \bar{F}_3 \bar{F}_4, \text{ } \bar{D}_3^2 \bar{H}_1, \text{ } \bar{D}_3 \bar{E}_1 \bar{G}_2, \text{ } \bar{D}_3 \bar{F}_1 \bar{F}_4, \text{ } \bar{b}_2^2 \bar{D}_4 \bar{H}_1, \text{ }   \bar{b}_2 \bar{b}_3 \bar{D}_2 \bar{H}_2, \text{ } \bar{b}_2^2 \bar{b}_3 \bar{D}_4 \bar{F}_1,  \right.\\
     &\, \bar{b}_2^2 \bar{b}_3^2 \bar{H}_1, \text{ } \bar{b}_3 \bar{D}_3^2 \bar{F}_1, \text{ } \bar{b}_2 \bar{b}_3 \bar{D}_3 \bar{H}_1, \text{ } \bar{b}_2 \bar{b}_3 \bar{E}_1 \bar{G}_2, \text{ } \bar{b}_2 \bar{C}_2 \bar{D}_4 \bar{G}_2, \text{ } \bar{b}_2 \bar{D}_2 \bar{D}_4 \bar{F}_3, \text{ } \bar{b}_2 \bar{b}_3^2 \bar{C}_2 \bar{G}_2, \text{ } \bar{b}_2 \bar{b}_3^2 \bar{D}_2 \bar{F}_3, \text{ } \bar{b}_2 \bar{D}_3 \bar{D}_4 \bar{F}_1, \text{ } \bar{b}_2 \bar{b}_3^2 \bar{D}_3 \bar{F}_1,  \\
     &\,\bar{b}_3 \bar{C}_2 \bar{D}_3 \bar{G}_2, \text{ } \bar{b}_2^2 \bar{b}_3^3 \bar{F}_1, \text{ } \bar{b}_3 \bar{D}_2 \bar{D}_3 \bar{F}_3, \text{ } \bar{b}_2 \bar{b}_3 \bar{F}_1 \bar{F}_4, \text{ } \bar{b}_2 C^{(003)} \bar{C}_2 \bar{H}_2, \text{ } \bar{b}_2 C^{(003)} \bar{D}_2 \bar{G}_2, \text{ } \bar{b}_2 C^{(003)} \bar{D}_3 \bar{G}_1, \text{ } \bar{b}_2 C^{(003)} \bar{E}_1 \bar{F}_3,\\
     &\, \left.   C^{(003)} \bar{C}_2 \bar{D}_3 \bar{F}_3, \text{ } \bar{b}_2^2 \bar{b}_3 C^{(003)} \bar{G}_1, \text{ }     \bar{b}_1 \bar{b}_2 C^{(003)} \bar{I}_1, \text{ } \bar{b}_2^2 \left(C^{(003)}\right)^2 \bar{F}_1,  \text{ } \bar{b}_2 \bar{b}_3 C^{(003)} \bar{C}_2 \bar{F}_3\right\} .
\end{align*}  
Following the usual transformation of generators,  these results in $191$ admissible  terms such that   
\begin{align*}
    \{\bar{H}_1,\bar{I}_1\} = \Gamma_{89}^1 \bar{b}_3\bar{E}_1\bar{I}_1 +\ldots + \Gamma_{89}^{190}\bar{b}_2^2\bar{c}_1^2 \bar{F}_1 +\Gamma_{89}^{191}\bar{b}_2\bar{b}_3\bar{c}_1\bar{C}_2\bar{F}_3 .
\end{align*} 
 Once the coefficients are determined, we find the explicit expression:
\begin{align}
\{\bar{H}_1, \bar{I}_1\}&=\frac{{\rm i}}{64}  \bigl(2 \bigl(24 \bar{F}_1 \bar{b}_3^3-24 \bar{H}_1 \bar{b}_3^2-\bigl(5 \bar{c}_1^2-6 \bar{C}_2 \bar{c}_1+\bar{C}_2^2\bigl) \bar{F}_1+8 \bar{C}_2 \bar{I}_2\bigl) \bar{b}_2^2+\bigl(-8 \bar{b}_1 \bar{F}_3 \bar{b}_3^3+8 \bigl(\bigl(\bar{D}_2+6 \bar{D}_3\nonumber \\
& \hskip 0.7cm-4 \bar{D}_4\bigl) \bar{F}_1+3 \bar{D}_2 \bar{F}_3+2 \bar{b}_1 \bar{H}_2\bigl) \bar{b}_3^2-4 \bigl(2 \bar{b}_1 \bar{D}_4 \bar{F}_3-4 \bar{F}_2 \bar{F}_3+4 \bar{c}_1 \bar{C}_2 \bigl(\bar{F}_1-2 \bar{F}_3\bigl)+20 \bar{F}_1 \bar{F}_4\nonumber \\
& \hskip 0.7cm+8 \bar{E}_1 \bar{G}_1+\bar{b}_1 \bar{c}_1 \bar{G}_2-\bar{b}_1 \bar{C}_2 \bar{G}_2-36 \bar{E}_1 \bar{G}_2+4 \bigl(\bar{D}_2+4 \bar{D}_3\bigl) \bar{H}_1\bigl) \bar{b}_3+32 \bar{D}_4^2 \bar{F}_1-32 \bar{D}_3 \bar{D}_4 \bar{F}_1\nonumber \\
& \hskip 0.7cm-4 \bar{c}_1 \bar{E}_1 \bar{F}_1+36 \bar{C}_2 \bar{E}_1 \bar{F}_1+5 \bar{b}_1 \bar{c}_1^2 \bar{F}_3+\bar{b}_1 \bar{C}_2^2 \bar{F}_3-6 \bar{b}_1 \bar{c}_1 \bar{C}_2 \bar{F}_3-72 \bar{c}_1 \bar{E}_1 \bar{F}_3+8 \bar{C}_2 \bar{E}_1 \bar{F}_3\nonumber \\
& \hskip 0.7cm+24 \bar{c}_1 \bar{D}_3 \bar{G}_1+8 \bar{C}_2 \bar{D}_3 \bar{G}_1+16 \bar{c}_1 \bar{D}_4 \bar{G}_1-16 \bar{C}_2 \bar{D}_4 \bar{G}_1+4 \bar{D}_2 \bigl(2 \bar{D}_4 \bar{F}_1+\bigl(\bar{c}_1-3 \bar{C}_2\bigl) \bar{G}_1\bigl)\nonumber \\
& \hskip 0.7cm-16 \bar{C}_2 \bar{d}_1 \bar{G}_2+16 \bar{C}_2 \bar{D}_3 \bar{G}_2-32 \bar{F}_4 \bar{H}_1-8 \bar{b}_1 \bigl(3 \bar{c}_1+\bar{C}_2\bigl) \bar{I}_1\bigl) \bar{b}_2+4 \bigl(-4 \bar{D}_2 \bar{F}_3 \bar{b}_3^3+2 \bigl(5 \bar{F}_3 \bar{C}_2^2\nonumber \\
& \hskip 0.7cm-6 \bar{b}_1 \bar{D}_3 \bar{F}_3+4 \bar{E}_1 \bar{G}_2\bigl) \bar{b}_3^2+2 \bigl(\bar{F}_3 \bar{D}_2^2+\bigl(-2 \bar{D}_3 \bar{F}_1+5 \bar{D}_3 \bar{F}_3+2 \bar{D}_4 \bar{F}_3\bigl) \bar{D}_2+10 \bar{b}_1 \bar{F}_3 \bar{F}_4\nonumber \\
& \hskip 0.7cm+2 \bar{D}_3 \bigl(-2 \bar{D}_3 \bar{F}_1+4 \bar{D}_4 \bar{F}_1+\bar{b}_1 \bar{H}_2\bigl)\bigl) \bar{b}_3+3 \bar{c}_1^2 \bar{D}_3 \bar{F}_1-32 \bar{E}_1^2 \bar{F}_3-2 \bar{C}_2^2 \bar{D}_4 \bar{F}_3+4 \bar{b}_1 \bar{D}_3 \bar{D}_4 \bar{F}_3\nonumber \\
& \hskip 0.7cm+12 \bar{D}_3 \bar{F}_2 \bar{F}_3+12 \bar{D}_3 \bar{F}_1 \bar{F}_4-16 \bar{D}_4 \bar{F}_1 \bar{F}_4-12 \bar{D}_2 \bar{F}_3 \bar{F}_4+28 \bar{D}_3 \bar{E}_1 \bar{G}_1-3 \bar{b}_1 \bar{C}_2 \bar{D}_3 \bar{G}_2\nonumber \\
& \hskip 0.7cm+14 \bar{D}_2 \bar{E}_1 \bar{G}_2-12 \bar{D}_3 \bar{E}_1 \bar{G}_2-8 \bar{D}_4 \bar{E}_1 \bar{G}_2+\bar{c}_1 \bigl(-3 \bar{C}_2 \bar{D}_3 \bar{F}_1+2 \bar{C}_2 \bigl(\bar{D}_4-3 \bar{b}_3^2\bigl) \bar{F}_3+4 \bar{b}_3 \bar{E}_1 \bar{F}_3\nonumber \\
& \hskip 0.7cm-12 \bar{F}_4 \bar{G}_1+2 \bar{b}_3 \bar{D}_2 \bar{G}_2+3 \bar{b}_1 \bar{D}_3 \bar{G}_2\bigl)+4 \bar{D}_3^2 \bar{H}_1-8 \bar{D}_2 \bar{D}_3 \bar{H}_1-16 \bar{b}_1 \bar{F}_4 \bar{H}_2\bigl)\bigl) \, .
\end{align} 

Repeating the procedure, the remaining expansions in $\{\bar{\textbf{H}},\bar{\textbf{I}}\}$ are given by 
\begin{align} 
\{\bar{H}_2, \bar{I}_1\}&=\frac{{\rm i}}{32} \bigl(b_2 \bigl(2 \bar{b}_3 \bigl(\bar{F}_3 \bigl(-\bar{c}_1 \bar{C}_2+\bar{c}_1^2+4 \bar{F}_4\bigl)-4 \bigl(4 \bar{D}_3 \bar{H}_1+2 \bar{D}_4 \bar{H}_2+\bar{F}_1 \bar{F}_4-\bar{E}_1 \bar{G}_2\bigl)\bigl)+10 \bar{b}_1 \bar{C}_2 \bar{I}_1-16 \bar{b}_3^4 \bar{F}_3\nonumber \\
& \hskip 0.7cm+\bar{b}_3^2 \bigl(28 \bar{D}_3 \bar{F}_1-6 \bar{D}_2 \bar{F}_3+8 \bigl(\bar{D}_3+4 \bar{D}_4\bigl) \bar{F}_3\bigl)+16 \bar{b}_3^3 \bar{H}_2+2 \bar{D}_3 \bigl(\bigl(6 \bar{G}_1+\bar{G}_2\bigl) \bigl(\bar{c}_1-\bar{C}_2\bigl)+6 \bar{D}_4 \bar{F}_1\nonumber \\
& \hskip 0.7cm-4 \bar{D}_4 \bar{F}_3\bigl)-10 \bar{c}_1 \bar{C}_2 \bar{H}_2+3 \bar{c}_1 \bar{D}_2 \bar{G}_2-4 \bar{E}_1 \bar{c}_1 \bar{F}_3-8 \bar{C}_2 \bar{D}_2 \bar{G}_2+14 \bar{E}_1 \bar{C}_2 \bar{F}_3-16 \bar{D}_4^2 \bar{F}_3+2 \bar{D}_2 \bar{D}_4 \bar{F}_3\nonumber \\
& \hskip 0.7cm-16 \bar{F}_4 \bar{H}_2+16 \bar{E}_1 \bar{I}_1\bigl)-2 \bar{b}_2^2 \bigl(-\bar{b}_3 \bigl(\bar{G}_1-\bar{G}_2\bigl) \bigl(\bar{c}_1-\bar{C}_2\bigl)-8 \bar{b}_3 \bar{D}_4 \bar{F}_1+6 \bar{H}_1 \bigl(\bar{D}_4-\bar{b}_3^2\bigl)+8 \bar{b}_3^3 \bar{F}_1\bigl)\nonumber \\
& \hskip 0.7cm-48 \bar{b}_3^2 \bigl(\bar{D}_3 \bar{H}_2+\bar{F}_3 \bar{F}_4\bigl)+8 \bar{b}_3 \bigl(\bar{D}_3 \bigl(5 \bar{D}_3 \bar{F}_1-2 \bigl(\bar{D}_3+3 \bar{D}_4\bigl) \bar{F}_3\bigl)+4 \bar{F}_4 \bar{H}_2\bigl)+48 \bar{b}_3^3 \bar{D}_3 \bar{F}_3\nonumber \\
& \hskip 0.7cm+6 \bar{c}_1 \bar{C}_2 \bar{D}_3 \bar{F}_3-6 \bar{c}_1^2 \bar{D}_3 \bar{F}_3+24 \bar{c}_1 \bar{F}_4 \bar{G}_2+32 \bar{C}_2 \bar{D}_4 \bar{I}_1-72 \bar{D}_3 \bar{F}_1 \bar{F}_4+24 \bar{D}_2 \bar{F}_3 \bar{F}_4+24 \bar{D}_3 \bar{F}_3 \bar{F}_4\nonumber \\
& \hskip 0.7cm+80 \bar{D}_4 \bar{F}_3 \bar{F}_4+24 \bar{E}_1 \bar{D}_3 \bar{G}_2-24 \bar{D}_3^2 \bar{H}_1-16 \bar{D}_3 \bar{D}_4 \bar{H}_2\bigl) \,   \\
\{\bar{H}_1, \bar{I}_2\}&=\frac{{\rm i}}{64}  \bigl(\bigl(\bigl(\bar{F}_1+14 \bar{F}_3\bigl) \bar{c}_1^2-2 \bigl(\bar{C}_2 \bigl(\bar{F}_1+14 \bar{F}_3\bigl)+2 \bar{b}_3 \bigl(2 \bar{G}_1+\bar{G}_2\bigl)\bigl) \bar{c}_1+\bar{C}_2^2 \bigl(\bar{F}_1+14 \bar{F}_3\bigl)+4 \bar{b}_3 \bar{C}_2 \bigl(2 \bar{G}_1+\bar{G}_2\bigl)\nonumber \\
& \hskip 0.7cm-8 \bigl(\bar{b}_3^2-\bar{D}_4\bigl) \bigl(3 \bar{b}_3 \bar{F}_3-4 \bar{H}_2\bigl)\bigl) \bar{b}_1^2+4 \bigl(-8 \bar{F}_1 \bar{b}_3^4+8 \bar{H}_1 \bar{b}_3^3+2 \bigl(\bigl(\bar{d}_1+\bar{D}_2-\bar{D}_3+10 \bar{D}_4\bigl) \bar{F}_1-\bigl(22 \bar{d}_1\nonumber \\
& \hskip 0.7cm+19 \bar{D}_2-22 \bar{D}_3+34 \bar{D}_4\bigl) \bar{F}_3\bigl) \bar{b}_3^2+\bigl(\bar{F}_1 \bar{c}_1^2-\bar{C}_2 \bar{F}_1 \bar{c}_1+4 \bigl(\bar{F}_2 \bigl(\bar{F}_1+12 \bar{F}_3\bigl)-11 \bar{E}_1 \bar{G}_1+\bigl(\bar{d}_1-\bar{D}_3\bigl) \bar{H}_1\nonumber \\
& \hskip 0.7cm+7 \bigl(\bar{d}_1+\bar{D}_2-\bar{D}_3+2 \bar{D}_4\bigl) \bar{H}_2-\bar{C}_2 \bar{I}_2\bigl)\bigl) \bar{b}_3-12 \bar{D}_4^2 \bar{F}_1+2 \bar{D}_3 \bar{D}_4 \bar{F}_1+20 \bar{c}_1 \bar{E}_1 \bar{F}_1-20 \bar{C}_2 \bar{E}_1 \bar{F}_1+4 \bar{D}_4^2 \bar{F}_3\nonumber \\
& \hskip 0.7cm+6 \bar{D}_2 \bar{D}_4 \bar{F}_3-12 \bar{D}_3 \bar{D}_4 \bar{F}_3-2 \bar{d}_1 \bar{D}_4 \bigl(\bar{F}_1-6 \bar{F}_3\bigl)-2 \bar{c}_1 \bar{D}_3 \bar{G}_1+2 \bar{C}_2 \bar{D}_3 \bar{G}_1-7 \bar{c}_1 \bar{D}_2 \bar{G}_2+7 \bar{C}_2 \bar{D}_2 \bar{G}_2\nonumber \\
& \hskip 0.7cm+4 \bar{c}_1 \bar{D}_3 \bar{G}_2-4 \bar{C}_2 \bar{D}_3 \bar{G}_2-18 \bar{c}_1 \bar{D}_4 \bar{G}_2+18 \bar{C}_2 \bar{D}_4 \bar{G}_2+2 \bigl(\bar{c}_1-\bar{C}_2\bigl) \bar{d}_1 \bigl(\bar{G}_1-2 \bar{G}_2\bigl)+4 \bar{F}_2 \bigl(\bar{H}_2-3 \bar{H}_1\bigl)\nonumber \\
& \hskip 0.7cm+8 \bar{b}_2 \bar{c}_1 \bar{I}_2-8 \bar{b}_2 \bar{C}_2 \bar{I}_2+12 \bar{E}_1 \bar{I}_2\bigl) \bar{b}_1+4 \bigl(24 \bigl(\bar{d}_1+\bar{D}_2-\bar{D}_3+2 \bar{D}_4\bigl) \bar{F}_1 \bar{b}_3^3-24 \bigl(\bar{F}_1 \bar{F}_2+\bigl(\bar{d}_1+\bar{D}_2-\bar{D}_3\nonumber \\
& \hskip 0.7cm+2 \bar{D}_4\bigl) \bar{H}_1\bigl) \bar{b}_3^2+2 \bigl(10 \bar{F}_3 \bigl(\bar{d}_1+\bar{D}_2-\bar{D}_3+2 \bar{D}_4\bigl)^2-\bigl(4 \bar{d}_1^2+11 \bar{D}_2 \bar{d}_1-8 \bar{D}_3 \bar{d}_1+4 \bar{D}_2^2+4 \bar{D}_3^2+40 \bar{D}_4^2\nonumber \\
& \hskip 0.7cm-11 \bar{D}_2 \bar{D}_3+4 \bigl(7 \bar{d}_1+6 \bar{D}_2-7 \bar{D}_3\bigl) \bar{D}_4\bigl) \bar{F}_1+8 \bar{F}_2 \bar{H}_1\bigl) \bar{b}_3-3 \bar{c}_1^2 \bigl(\bar{d}_1+\bar{D}_2-\bar{D}_3+2 \bar{D}_4\bigl) \bar{F}_1+3 \bar{c}_1\nonumber  \\
& \hskip 0.7cm \times \bigl(\bar{C}_2 \bigl(\bar{d}_1+\bar{D}_2-\bar{D}_3+2 \bar{D}_4\bigl)\bar{F}_1+4 \bar{F}_2 \bar{G}_1\bigl)+2 \bigl(8 \bar{H}_1 \bar{D}_4^2+24 \bar{F}_1 \bar{F}_2 \bar{D}_4-36 \bar{F}_2 \bar{F}_3 \bar{D}_4+20 \bar{E}_1 \bar{G}_1 \bar{D}_4\nonumber \\
& \hskip 0.7cm+4 \bar{d}_1 \bar{H}_1 \bar{D}_4-4 \bar{D}_3 \bar{H}_1 \bar{D}_4+6 \bar{d}_1 \bar{F}_1 \bar{F}_2+12 \bar{D}_2 \bar{F}_1 \bar{F}_2-6 \bar{D}_3 \bar{F}_1 \bar{F}_2-18 \bar{d}_1 \bar{F}_2 \bar{F}_3-18 \bar{D}_2 \bar{F}_2 \bar{F}_3+18 \bar{D}_3 \bar{F}_2 \bar{F}_3\nonumber \\
& \hskip 0.7cm+6 \bar{D}_2 \bar{E}_1 \bar{G}_1+3 \bar{d}_1 \bar{D}_2 \bar{H}_1-3 \bar{D}_2 \bar{D}_3 \bar{H}_1-2 \bigl(\bar{d}_1+\bar{D}_2-\bar{D}_3+2 \bar{D}_4\bigl) \bigl(3 \bar{D}_2+10 \bar{D}_4\bigl) \bar{H}_2\bigl)\bigl)\bigl)  
\end{align} and 
\begin{align}
\{\bar{H}_2, \bar{I}_2\}&=-\frac{\rm i}{256}  \bigl(\bigl(8 \bar{F}_3\bar{b}_3^3-16 \bar{H}_2 \bar{b}_3^2+8 \bar{D}_4 \bar{F}_3 \bar{b}_3+4 \bigl(\bar{c}_1-\bar{C}_2\bigl) \bar{G}_2 \bar{b}_3+35 \bigl(\bar{c}_1-\bar{C}_2\bigl) \bar{C}_2 \bar{F}_3-140 \bigl(\bar{c}_1+\bar{C}_2\bigl) \bar{I}_1\bigl) \bar{b}_1^2\nonumber \\
& \hskip 0.7cm+\bigl(\bar{2} \bigl(\bigl(24 \bar{d}_1 \bar{F}_3+70 \bar{D}_2 \bar{F}_3-8 \bar{D}_3 \bigl(20 \bar{F}_1+3 \bar{F}_3\bigl)+48 \bar{c}_1 \bar{G}_1\bigl) \bar{b}_3^2+2 \bigl(14 \bar{F}_1 \bar{c}_1^2-49 \bar{C}_2 \bar{F}_1 \bar{c}_1+70 \bar{c}_1\bar{C}_2 \bar{F}_3-4 \bar{F}_2 \bar{F}_3\nonumber \\
& \hskip 0.7cm +64 \bar{F}_1 \bar{F}_4-32 \bar{E}_1 \bar{G}_1+164 \bar{E}_1 \bar{G}_2+2 \bar{D}_2 \bar{H}_1+44 \bar{D}_3 \bar{H}_1\bigl) \bar{b}_3+16 \bar{c}_1 \bar{E}_1 \bar{F}_1+104 \bar{C}_2 \bar{E}_1 \bar{F}_1-24 \bar{d}_1 \bar{D}_4 \bar{F}_3\nonumber  \\
& \hskip 0.7cm -176 \bar{c}_1 \bar{E}_1 \bar{F}_3-56 \bar{C}_2 \bar{E}_1 \bar{F}_3-2 \bar{c}_1 \bar{D}_2 \bar{G}_1-23 \bar{C}_2 \bar{D}_2 \bar{G}_1-12 \bar{c}_1 \bar{D}_4 \bar{G}_1+4 \bar{C}_2 \bar{D}_4 \bar{G}_1-12 \bar{c}_1 \bar{d}_1 \bar{G}_2-58 \bar{C}_2 \bar{d}_1 \bar{G}_2\nonumber \\
& \hskip 0.7cm+2 \bar{D}_3 \bigl(-70 \bar{D}_4 \bar{F}_1+12 \bar{D}_4 \bar{F}_3+11 \bar{c}_1 \bar{G}_1+35 \bar{C}_2 \bar{G}_1+6 \bar{c}_1 \bar{G}_2+29 \bar{C}_2 \bar{G}_2\bigl)-4 \bigl(6 \bar{c}_1^2-11 \bar{C}_2 \bar{c}_1+88 \bar{F}_4\bigl) \bar{H}_1\bigl) \nonumber \\
& \hskip 0.7cm+\bar{b}_2 \bigl(456 \bar{F}_1 \bar{b}_3^3-416 \bar{H}_1 \bar{b}_3^2-16 \bigl(7 \bar{D}_4 \bar{F}_1+\bigl(\bar{c}_1-\bar{C}_2\bigl) \bar{G}_1\bigl) \bar{b}_3-55 \bar{c}_1^2 \bar{F}_1-9 \bar{C}_2^2 \bar{F}_1+64 \bar{c}_1 \bar{C}_2 \bar{F}_1+72 \bar{D}_4 \bar{H}_1\nonumber \\
& \hskip 0.7cm+100 \bar{C}_2 \bar{I}_2\bigl)\bigl) \bar{b}_1-4 \bigl(16 \bigl(7 \bar{d}_1+6 \bar{D}_2-7 \bar{D}_3+14 \bar{D}_4\bigl) \bar{F}_1 \bar{b}_3^3+2 \bigl(16 \bar{E}_1\bar{G}_1-\bar{F}_1 \bigl(12 \bar{c}_1 \bar{C}_2+29 \bar{b}_2 \bigl(\bar{d}_1+\bar{D}_2-\bar{D}_3\nonumber \\
& \hskip 0.7cm+2 \bar{D}_4\bigl)+112 \bar{F}_2\bigl)\bigl) \bar{b}_3^2+2 \bigl(12 \bar{F}_3 \bar{d}_1^2+\bigl(10 \bar{D}_2 \bar{F}_1+34 \bar{D}_3 \bar{F}_1-56 \bar{D}_4 \bar{F}_1-29 \bar{D}_2 \bar{F}_3-24 \bar{D}_3 \bar{F}_3+24 \bar{D}_4 \bar{F}_3-4 \bar{c}_1\bar{G}_1\nonumber \\
& \hskip 0.7cm+8 \bar{b}_2 \bar{H}_1\bigl) \bar{d}_1+\bar{D}_2^2 \bigl(29 \bar{F}_1-41 \bar{F}_3\bigl)+\bar{D}_2 \bigl(24 \bar{D}_3 \bar{F}_1+29 \bar{D}_3 \bar{F}_3-82 \bar{D}_4 \bar{F}_3+8 \bar{b}_2 \bar{H}_1\bigl)-2 \bigl(\bigl(17 \bar{F}_1-6 \bar{F}_3\bigl) \bar{D}_3^2\nonumber \\
& \hskip 0.7cm-2 \bigl(31 \bar{D}_4 \bar{F}_1-6 \bar{D}_4 \bar{F}_3+\bar{c}_1 \bar{G}_1-2 \bar{b}_2 \bar{H}_1\bigl) \bar{D}_3+56 \bar{D}_4^2 \bar{F}_1-4 \bar{c}_1 \bar{E}_1 \bar{F}_1-17 \bar{b}_2 \bar{F}_1 \bar{F}_2+4 \bar{c}_1 \bar{D}_4 \bar{G}_1-8 \bar{b}_2 \bar{D}_4 \bar{H}_1\nonumber\\
& \hskip 0.7cm -56 \bar{F}_2 \bar{H}_1+56 \bar{E}_1 \bar{I}_2\bigl)\bigl) \bar{b}_3+60 \bar{b}_2 \bar{D}_4^2 \bar{F}_1-40 \bar{E}_1^2 \bar{F}_1+30 \bar{b}_2 \bar{d}_1 \bar{D}_4 \bar{F}_1+30 \bar{b}_2 \bar{D}_2 \bar{D}_4 \bar{F}_1-30 \bar{b}_2 \bar{D}_3 \bar{D}_4 \bar{F}_1-12 \bar{D}_2 \bar{F}_1 \bar{F}_2\nonumber \\
& \hskip 0.7cm-72 \bar{D}_3 \bar{F}_1 \bar{F}_2-72 \bar{D}_4 \bar{F}_1 \bar{F}_2-12 \bar{c}_1^2 \bigl(\bar{d}_1+\bar{D}_2-\bar{D}_3+2 \bar{D}_4\bigl) \bigl(\bar{F}_1-\bar{F}_3\bigl)-24 \bar{d}_1 \bar{F}_2 \bar{F}_3+48 \bar{D}_2 \bar{F}_2 \bar{F}_3+24 \bar{D}_3 \bar{F}_2\bar{F}_3\nonumber \\
& \hskip 0.7cm-24 \bar{d}_1 \bar{F}_1 \bar{F}_4-24 \bar{D}_2 \bar{F}_1 \bar{F}_4+24 \bar{D}_3 \bar{F}_1 \bar{F}_4-48 \bar{D}_4 \bar{F}_1 \bar{F}_4-19 \bar{b}_2 \bar{C}_2 \bar{d}_1 \bar{G}_1-19 \bar{b}_2 \bar{C}_2 \bar{D}_2 \bar{G}_1+19 \bar{b}_2 \bar{C}_2 \bar{D}_3 \bar{G}_1\nonumber \\
& \hskip 0.7cm-38 \bar{b}_2 \bar{C}_2 \bar{D}_4 \bar{G}_1-32 \bar{d}_1 \bar{E}_1 \bar{G}_1+6 \bar{D}_2 \bar{E}_1 \bar{G}_1+32 \bar{D}_3 \bar{E}_1 \bar{G}_1-96 \bar{D}_4 \bar{E}_1 \bar{G}_1+72 \bar{C}_2 \bar{F}_2 \bar{G}_1+68 \bar{d}_1 \bar{E}_1 \bar{G}_2+68 \bar{D}_2 \bar{E}_1 \bar{G}_2\nonumber \\
& \hskip 0.7cm-68 \bar{D}_3 \bar{E}_1 \bar{G}_2+136 \bar{D}_4 \bar{E}_1 \bar{G}_2-37 \bar{D}_2^2 \bar{H}_1-36 \bar{D}_3^2 \bar{H}_1+144 \bar{D}_4^2 \bar{H}_1-8 \bar{d}_1 \bar{D}_2 \bar{H}_1+36 \bar{d}_1 \bar{D}_3 \bar{H}_1+44 \bar{D}_2 \bar{D}_3 \bar{H}_1\nonumber \\
& \hskip 0.7cm+72 \bar{d}_1 \bar{D}_4 \bar{H}_1-56 \bar{b}_2 \bar{F}_2 \bar{H}_1-72 \bar{C}_2 \bar{D}_4 \bar{I}_2+\bar{c}_1 \bigl(-24 \bar{G}_1 \bar{C}_2^2+12 \bigl(\bigl(\bar{d}_1-\bar{D}_3+5 \bar{D}_4\bigl) \bar{F}_1+\bigl(\bar{d}_1+\bar{D}_2-\bar{D}_3+2 \bar{D}_4\bigl) \bar{F}_3\bigl) \bar{C}_2\nonumber \\
& \hskip 0.7cm+19 \bar{b}_2 \bigl(\bar{d}_1+\bar{D}_2-\bar{D}_3+2 \bar{D}_4\bigl) \bar{G}_1-48 \bar{F}_2 \bigl(\bar{G}_1+\bar{G}_2\bigl)+24 \bar{D}_2 \bar{I}_2\bigl)\bigl)\bigl) \, \nonumber\\.
\end{align}

Finally, for degree $17$,  there is only one compact form $\{\bar{\textbf{I}},\bar{\textbf{I}}\}$, which contains the Poisson bracket $\{\bar{I}_1,\bar{I}_2\}$ with the grading 
\begin{align*}
    \mathcal{G} \left(\{\bar{I}_1,\bar{I}_2\}\right) = (2,3,12) \tilde{+} (3,2,12) \tilde{+} (4,3,10) \tilde{+} (3,4,10).
\end{align*} 
The generators are explicitly given by 
\begin{align*}
    (2,3,12) = & \, \left\{\bar{b}_3 \bar{F}_2 \bar{I}_1, \text{ }  \bar{D}_4 \bar{E}_1 \bar{H}_2, \text{ } \bar{D}_4 \bar{F}_4 \bar{G}_1, \text{ } \bar{b}_1 \bar{b}_3 \bar{D}_4 \bar{I}_1, \text{ } \bar{b}_2 \bar{D}_4^2 \bar{G}_1, \text{ } \bar{b}_3^2 \bar{E}_1 \bar{H}_2, \text{ } \bar{b}_3^2 \bar{F}_4 \bar{G}_1, \text{ } \bar{b}_3^2 \bar{C}_2 \bar{D}_4 \bar{F}_3, \text{ }  \bar{b}_3^3 \bar{C}_2 \bar{H}_2, \text{ } \bar{b}_3^3 \bar{d}_2 \bar{G}_2, \text{ } \bar{b}_3^3 \bar{d}_3 \bar{G}_1,   \right. \\
    & \, \bar{b}_3^3 \bar{E}_1 \bar{F}_3 , \text{ } \bar{b}_3^4 \bar{C}_2 \bar{F}_3, \text{ }   \bar{b}_2 \bar{b}_3^4 \bar{G}_1, \text{ } \bar{b}_1 \bar{b}_3^3 \bar{I}_1, \text{ } \bar{C}_2 \bar{D}_4^2 \bar{F}_3, \text{ } \bar{b}_3 \bar{C}_2 \bar{D}_4 \bar{H}_2, \text{ } \bar{b}_3 \bar{d}_2 \bar{D}_4 \bar{G}_2, \text{ } \bar{b}_3 \bar{d}_3 \bar{D}_4 \bar{G}_1, \text{ } \bar{b}_3 \bar{D}_4 \bar{E}_1 \bar{F}_3,  \text{ } \left(C^{(003)}\right)^2 \bar{d}_2 \bar{G}_2, \\
    & \, \text{ }  \left(C^{(003)}\right)^2 \bar{E}_1 \bar{F}_3, \text{ } C^{(003)} \bar{C}_2 \bar{D}_4 \bar{G}_2, \text{ } \bar{b}_3 C^{(003)} \bar{d}_2 \bar{H}_2, \text{ } \bar{b}_3 C^{(003)} \bar{d}_3 \bar{H}_1, \text{ } \bar{b}_3 C^{(003)} \bar{E}_1 \bar{G}_2, \text{ } \bar{b}_3 C^{(003)} \bar{F}_1 \bar{F}_4,  \text{ } \left(C^{(003)}\right)^2 \bar{d}_3 \bar{G}_1,   \\ 
   & \, \text{ }  C^{(003)} \bar{F}_4 \bar{H}_1, \text{ } C^{(202)} \bar{D}_4 \bar{I}_1, \text{ }  C^{(003)} \bar{d}_2 \bar{D}_4 \bar{F}_3, \text{ } C^{(003)} \bar{d}_3 \bar{D}_4 \bar{F}_1, \text{ } \bar{b}_2 \bar{b}_3^2 C^{(003)} \bar{H}_1, \text{ } \bar{b}_2 \bar{b}_3^2 \bar{D}_4 \bar{G}_1, \text{ } \bar{b}_2 \bar{b}_3 \left(C^{(003)}\right)^2 \bar{G}_1, \text{ } \\
   & \, \text{ } \bar{b}_2 \left(C^{(003)}\right)^3 \bar{F}_1, \text{ }  \bar{b}_3^2 C^{(202)} \bar{I}_1, \text{ } \bar{b}_3^2 C^{(003)} \bar{C}_2 \bar{G}_2, \text{ } \bar{b}_1 \left(C^{(003)}\right)^2 \bar{I}_1, \text{ } \bar{b}_2 C^{(003)} \bar{D}_4 \bar{H}_1, \text{ }  \bar{b}_2 \bar{b}_3 C^{(003)} \bar{D}_4 \bar{F}_1, \text{ }  \\ 
   & \, \text{ }  \left.\left(C^{(003)}\right)^2 \bar{C}_2 \bar{H}_2,\text{ } \bar{b}_3^2 C^{(003)} \bar{d}_2 \bar{F}_3, \text{ } \bar{b}_3^2 C^{(003)} \bar{d}_3 \bar{F}_1, \text{ } \bar{b}_3 \left(C^{(003)}\right)^2 \bar{C}_2 \bar{F}_3, \text{ } \bar{b}_2 \bar{b}_3^3 C^{(003)} \bar{F}_1  \right\} ; 
   \end{align*}
   \begin{align*} 
    (3,   2,   12)   = & \, \text{ }  \left\{\bar{D}_4 \bar{E}_1 \bar{H}_1, \text{ } \bar{D}_4 \bar{F}_2 \bar{G}_2, \text{ } \bar{b}_1 \bar{D}_4^2 \bar{G}_2, \text{ } \bar{b}_2 \bar{b}_3 \bar{D}_4 \bar{I}_2, \text{ } \bar{b}_3 \bar{C}_2 \bar{D}_4 \bar{H}_1, \text{ } \bar{b}_3 \bar{d}_2 \bar{D}_4 \bar{G}_1, \text{ } \bar{b}_3 \bar{D}_4 \bar{E}_1 \bar{F}_1, \text{ } \bar{b}_3^2 \bar{d}_3 \bar{I}_2, \text{ } \bar{b}_3^2 \bar{E}_1 \bar{H}_1, \text{ } \bar{b}_3^2 \bar{F}_2 \bar{G}_2, \text{ } \bar{b}_1 \bar{b}_3^4 \bar{G}_2\right. \\
    & \, \text{ } \bar{b}_3 \bar{F}_4 \bar{I}_2, \text{ } \bar{d}_3 \bar{D}_4 \bar{I}_2, \text{ } \bar{b}_3^3 \bar{d}_2 \bar{G}_1, \text{ } \bar{b}_3^3 \bar{E}_1 \bar{F}_1, \text{ }  \bar{b}_3^2 \bar{C}_2 \bar{D}_4 \bar{F}_1, \text{ } \bar{b}_3^4 \bar{C}_2 \bar{F}_1, \text{ } \bar{b}_2 \bar{b}_3^3 \bar{I}_2, \text{ } \bar{b}_1 \bar{b}_3^2 \bar{D}_4 \bar{G}_2, \text{ } \bar{b}_3^3 \bar{C}_2 \bar{H}_1, \text{ } \bar{b}_3 C^{(003)} \bar{F}_2 \bar{F}_3, \text{ } \bar{b}_3 C^{(202)} \bar{D}_4 \bar{G}_2, \text{ }   \\ 
    & \, \text{ }   \bar{b}_3 C^{(003)} \bar{E}_1, \text{ } \bar{b}_1 C^{(003)} \bar{D}_4 \bar{H}_2, \text{ } \bar{b}_3 C^{(003)} C^{(202)} \bar{H}_2, \text{ } \bar{b}_3 C^{(003)} \bar{d}_2 \bar{H}_1, \text{ } \left(C^{(003)}\right)^2 \bar{C}_2 \bar{H}_1, \text{ }  \left(C^{(003)}\right)^2 C^{(202)} \bar{G}_2, \text{ } \\
    & \, \text{ } \bar{b}_1 \bar{b}_3 C^{(003)} \bar{D}_4 \bar{F}_3, \text{ } \bar{b}_1 \left(C^{(003)}\right)^3 \bar{F}_3, \text{ } \left(C^{(003)}\right)^2 \bar{d}_2 \bar{G}_1, \text{ } \left(C^{(003)}\right)^2 \bar{E}_1 \bar{F}_1, \text{ } C^{(003)} \bar{C}_2 \bar{D}_4 \bar{G}_1, \text{ } C^{(003)} C^{(202)} \bar{D}_4 \bar{F}_3, \text{ } \\ 
    & \, \text{ } C^{(003)} \bar{d}_2 \bar{D}_4 \bar{F}_1, \text{ }   \bar{C}_2 \bar{D}_4^2 \bar{F}_1, \text{ } \bar{b}_1 \bar{b}_3^2 C^{(003)} \bar{H}_2, \text{ }  C^{(003)} \bar{F}_2 \bar{H}_2, \text{ } \bar{b}_3^3 C^{(202)} \bar{G}_2, \text{ } \bar{b}_3^2 C^{(003)} \bar{C}_2 \bar{G}_1, \text{ } \bar{b}_3^2 C^{(003)} C^{(202)} \bar{F}_3, \text{ } \\
    & \, \text{ }  \left. \bar{b}_1 \bar{b}_3 \left(C^{(003)}\right)^2 \bar{G}_2, \text{ } \bar{b}_3^2 C^{(003)} \bar{d}_2 \bar{F}_1, \text{ } \bar{b}_3 \left(C^{(003)}\right)^2 \bar{C}_2 \bar{F}_1, \text{ } \bar{b}_1 \bar{b}_3^3 C^{(003)} \bar{F}_3, \text{ } \bar{b}_2 \left(C^{(003)}\right)^2 \bar{I}_2 \right\} ; 
    \end{align*}
    \begin{align*}
    (4,  3,  10)  = & \, \text{ }  \left\{\bar{b}_1 \bar{F}_2 \bar{I}_1, \text{ } \bar{C}_2 \bar{F}_2 \bar{H}_2, \text{ } \bar{d}_2 \bar{d}_3 \bar{I}_2, \text{ } \bar{d}_2 \bar{E}_1 \bar{H}_1, \text{ } \bar{d}_2 \bar{F}_2 \bar{G}_2, \text{ } \bar{d}_3 \bar{F}_2 \bar{G}_1, \text{ } \bar{E}_1^2 \bar{G}_1, \text{ } \bar{E}_1 \bar{F}_2 \bar{F}_3, \text{ } \bar{b}_1^2 \bar{D}_4 \bar{I}_1, \text{ } \bar{b}_1 \bar{b}_3 \bar{E}_1 \bar{H}_2, \text{ } \bar{b}_1 \bar{b}_3 \bar{F}_4 \bar{G}_1, \text{ } \right.\\
    & \, \text{ } \bar{b}_1 \bar{b}_2 \bar{b}_3 \bar{D}_4 \bar{G}_1, \text{ } \bar{b}_3 \bar{C}_2 \bar{E}_1 \bar{G}_1, \text{ }  \bar{b}_1 \bar{b}_2 \bar{b}_3^3 \bar{G}_1, \text{ } \bar{b}_1 \bar{b}_3 \bar{C}_2 \bar{D}_4 \bar{F}_3, \text{ }  \bar{b}_3 \bar{C}_2 \bar{d}_2 \bar{H}_1, \text{ } \bar{b}_1 \bar{b}_3^2 \bar{C}_2 \bar{H}_2, \text{ } \bar{b}_1 \bar{b}_3^2 \bar{d}_2 \bar{G}_2, \text{ }  \text{ } \bar{b}_1 \bar{b}_3^3 \bar{C}_2 \bar{F}_3,  \\ 
    & \,  \bar{b}_3 \bar{d}_2^2 \bar{G}_1, \text{ } \bar{b}_3 \bar{d}_2 \bar{E}_1 \bar{F}_1, \text{ } \bar{C}_2 \bar{d}_2 \bar{D}_4 \bar{F}_1, \text{ } \bar{b}_1^2 \bar{b}_3^2 \bar{I}_1, \text{ } \bar{b}_1 \bar{b}_3^2 \bar{d}_3 \bar{G}_1, \text{ }  \bar{b}_1 \bar{b}_3^2 \bar{E}_1 \bar{F}_3, \text{ }  \bar{b}_1 \bar{d}_3 \bar{D}_4 \bar{G}_1, \text{ }   \bar{b}_3 \bar{C}_2 \bar{F}_2 \bar{F}_3, \text{ } \bar{C}_2^2 \bar{D}_4 \bar{G}_1,   \\ 
    & \, \text{ }  \bar{b}_1 \bar{C}_2 \bar{D}_4 \bar{H}_2, \text{ } \bar{b}_1 \bar{d}_2 \bar{D}_4 \bar{G}_2, \text{ } \bar{b}_1 \bar{D}_4 \bar{E}_1 \bar{F}_3, \text{ } \bar{b}_2 \bar{b}_3 \bar{d}_2 \bar{I}_2, \text{ } \bar{b}_2 \bar{b}_3 \bar{F}_2 \bar{G}_1, \text{ } \bar{b}_3^2 \bar{C}_2 \bar{d}_2 \bar{F}_1, \text{ } \bar{b}_3^2 \bar{C}_2^2 \bar{G}_1, \text{ } \bar{b}_1 C^{(003)} \bar{d}_2 \bar{H}_2, \text{ } \\ 
    & \, \text{ } \bar{b}_1 C^{(003)} \bar{d}_3 \bar{H}_1, \text{ } \bar{b}_1 C^{(003)} \bar{E}_1 \bar{G}_2, \text{ } \bar{b}_1 C^{(003)} \bar{F}_1 \bar{F}_4, \text{ } \bar{b}_2 C^{(003)} \bar{C}_2 \bar{I}_2, \text{ } \bar{b}_2 C^{(003)} C^{(202)} \bar{H}_1, \text{ } \bar{b}_2 C^{(003)} \bar{F}_1 \bar{F}_2,   \\
    & \, \text{ }  \left(C^{(202)}\right)^2 \bar{I}_1, \text{ } C^{(202)} \bar{E}_1 \bar{H}_2, \text{ } C^{(202)} \bar{F}_4 \bar{G}_1, \text{ } \bar{b}_3 \bar{C}_2 C^{(202)} \bar{H}_2, \text{ } C^{(003)} \bar{C}_2 \bar{d}_2 \bar{G}_1, \text{ } C^{(003)} \bar{C}_2 \bar{E}_1 \bar{F}_1, \text{ } \\ 
    &\, \text{ } C^{(003)} C^{(202)} \bar{d}_2 \bar{F}_3, \text{ } C^{(003)} C^{(202)} \bar{d}_3 \bar{F}_1, \text{ } C^{(003)} \bar{d}_2^2 \bar{F}_1, \text{ } C^{(003)} \bar{C}_2^2 \bar{H}_1, \bar{C}_2 C^{(202)} \bar{D}_4 \bar{F}_3, \text{ } \text{ }   \bar{b}_3 C^{(202)} \bar{d}_2 \bar{G}_2,\\ 
    & \,  \text{ } \bar{b}_3 C^{(202)} \bar{d}_3 \bar{G}_1, \text{ } \bar{b}_3 C^{(202)} \bar{E}_1 \bar{F}_3, \text{ }  \bar{b}_1 \bar{b}_2 \bar{b}_3 C^{(003)} \bar{H}_1, \text{ } \bar{b}_1 \bar{b}_2 \left(C^{(003)}\right)^2 \bar{G}_1, \text{ } \bar{b}_1 \bar{b}_2 C^{(003)} \bar{D}_4 \bar{F}_1, \text{ } \bar{b}_1 \bar{b}_3 C^{(202)} \bar{I}_1, \\
    & \, \text{ } \bar{b}_1 \bar{b}_3 C^{(003)} \bar{C}_2 \bar{G}_2, \text{ } \bar{b}_1 \bar{b}_3 C^{(003)} \bar{d}_2 \bar{F}_3, \text{ } \bar{b}_1 \bar{b}_3 C^{(003)} \bar{d}_3 \bar{F}_1, \text{ } \bar{b}_1 \left(C^{(003)}\right)^2 \bar{C}_2 \bar{F}_3, \text{ } \bar{b}_2 \bar{b}_3^2 C^{(202)} \bar{G}_1, \text{ }  \bar{b}_2 C^{(202)} \bar{D}_4 \bar{G}_1,   \\
    & \, \text{ }  \left.  C^{(003)} \bar{C}_2 C^{(202)} \bar{G}_2, \text{ }  \bar{b}_2 \bar{b}_3 C^{(003)} C^{(202)} \bar{F}_1, \text{ } \bar{b}_3^2 \bar{C}_2 C^{(202)} \bar{F}_3, \text{ } \bar{b}_3 C^{(003)} \bar{C}_2^2 \bar{F}_1, \text{ } \bar{b}_1 \bar{b}_2 \bar{b}_3^2 C^{(003)} \bar{F}_1\right\} \nonumber
    \end{align*} and 
    \begin{align*} 
    (3, 4, 10) = & \, \text{ }  \left\{\bar{b}_2 \bar{F}_4 \bar{I}_2, \text{ } \bar{C}_2 \bar{F}_4 \bar{H}_1, \text{ } \bar{d}_2 \bar{E}_1 \bar{H}_2, \text{ } \bar{d}_2 \bar{F}_4 \bar{G}_1, \text{ } \bar{d}_3^2 \bar{I}_2, \text{ } \bar{d}_3 \bar{E}_1 \bar{H}_1, \text{ } \bar{d}_3 \bar{F}_2 \bar{G}_2, \text{ } \bar{E}_1^2 \bar{G}_2, \text{ } \bar{E}_1 \bar{F}_1 \bar{F}_4, \text{ } \bar{b}_1 \bar{b}_3 \bar{d}_2 \bar{I}_1, \text{ } \bar{b}_1 \bar{b}_3 \bar{F}_4 \bar{G}_2, \text{ } \bar{b}_1 \bar{d}_3 \bar{D}_4 \bar{G}_2, \right.\\
    & \, \text{ } \bar{b}_2 \bar{b}_3^2 \bar{d}_2 \bar{G}_1, \text{ } \bar{b}_2 \bar{b}_3^2 \bar{E}_1 \bar{F}_1, \text{ } \bar{b}_3 \bar{d}_2 \bar{d}_3 \bar{G}_1, \text{ } \bar{b}_3 \bar{d}_2 \bar{E}_1 \bar{F}_3, \text{ } \bar{b}_3 \bar{d}_3 \bar{E}_1 \bar{F}_1, \text{ }   \bar{b}_2 \bar{D}_4 \bar{E}_1 \bar{F}_1, \text{ } \bar{b}_3 \bar{C}_2 \bar{d}_2 \bar{H}_2, \text{ } \bar{b}_3 \bar{C}_2 \bar{d}_3 \bar{H}_1, \text{ } \bar{b}_1 \bar{b}_3^2 \bar{d}_3 \bar{G}_2, \text{ } \bar{C}_2 \bar{d}_3 \bar{D}_4 \bar{F}_1, \text{ } \\ 
    & \, \text{ } \bar{b}_2 \bar{b}_3 \bar{C}_2 \bar{D}_4 \bar{F}_1, \text{ } \bar{b}_2^2 \bar{b}_3^2 \bar{I}_2, \text{ } \bar{b}_2 \bar{b}_3^2 \bar{C}_2 \bar{H}_1, \text{ } \bar{b}_3 \bar{C}_2 \bar{E}_1 \bar{G}_2, \text{ } \bar{b}_3 \bar{C}_2 \bar{F}_1 \bar{F}_4, \text{ }  \bar{b}_2 \bar{b}_3 \bar{d}_3 \bar{I}_2, \text{ } \bar{b}_2 \bar{b}_3 \bar{E}_1 \bar{H}_1, \text{ } \bar{b}_2 \bar{b}_3 \bar{F}_2 \bar{G}_2, \text{ }   \bar{b}_2^2 \bar{D}_4 \bar{I}_2, \text{ } \bar{b}_2 \bar{C}_2 \bar{D}_4 \bar{H}_1, \\ 
    & \,   \bar{b}_2 \bar{d}_2 \bar{D}_4 \bar{G}_1, \text{ } \bar{b}_3^2 \bar{C}_2^2 \bar{G}_2, \text{ } \bar{b}_1 \bar{b}_2 \bar{b}_3^3 \bar{G}_2, \text{ } \bar{b}_3^2 \bar{C}_2 \bar{d}_2 \bar{F}_3, \text{ } \bar{b}_3^2 \bar{C}_2 \bar{d}_3 \bar{F}_1, \text{ } \bar{b}_2 \bar{b}_3^3 \bar{C}_2 \bar{F}_1, \text{ } \bar{b}_3 \bar{d}_2^2 \bar{G}_2, \text{ } \bar{b}_1 \bar{b}_2 \bar{b}_3 \bar{D}_4 \bar{G}_2, \text{ }    \bar{C}_2^2 \bar{D}_4 \bar{G}_2,  \text{ } \bar{C}_2 \bar{d}_2 \bar{D}_4 \bar{F}_3,  \\
    & \,  \bar{b}_2 C^{(003)} C^{(202)} \bar{H}_2,\bar{b}_2 C^{(003)} \bar{d}_2 \bar{H}_1, \text{ } \bar{b}_2 C^{(003)} \bar{E}_1 \bar{G}_1, \text{ } \bar{b}_2 C^{(003)} \bar{F}_2 \bar{F}_3, \text{ } \bar{b}_2 C^{(202)} \bar{D}_4 \bar{G}_2,\text{ }  C^{(003)} \bar{C}_2^2 \bar{H}_2, \text{ }C^{(003)} \bar{C}_2 \bar{d}_2 \bar{G}_2, \\
    & \,    C^{(003)} \bar{C}_2 \bar{d}_3 \bar{G}_1, \text{ } C^{(003)} \bar{C}_2 \bar{E}_1 \bar{F}_3, \text{ } C^{(003)} C^{(202)} \bar{d}_3 \bar{F}_3, \text{ } \text{ }  C^{(003)} \bar{d}_2 \bar{d}_3 \bar{F}_1, \text{ }\bar{b}_3 C^{(202)} \bar{d}_3 \bar{G}_2, \text{ }    \bar{b}_1 C^{(003)} \bar{d}_3 \bar{H}_2, \text{ } \bar{b}_1 \bar{b}_2 \bar{b}_3 C^{(003)} \bar{H}_2, \\
    & \,     C^{(202)} \bar{d}_2 \bar{I}_1, \text{ } C^{(202)} \bar{F}_4 \bar{G}_2, \text{ } \bar{b}_1 C^{(003)} \bar{C}_2 \bar{I}_1, \text{ }   \bar{b}_1 \bar{b}_2 \left(C^{(003)}\right)^2 \bar{G}_2, \text{ } \bar{b}_1 \bar{b}_2 C^{(003)} \bar{D}_4 \bar{F}_3, \text{ } \bar{b}_1 \bar{b}_3 C^{(003)} \bar{d}_3 \bar{F}_3, \text{ } \bar{b}_2 \bar{b}_3^2 C^{(202)} \bar{G}_2,  \\
    & \,     \bar{b}_2 \bar{b}_3 C^{(003)} \bar{C}_2 \bar{G}_1, \text{ } \bar{b}_2 \bar{b}_3 C^{(003)} C^{(202)} \bar{F}_3, \text{ } \bar{b}_2 \bar{b}_3 C^{(003)} \bar{d}_2 \bar{F}_1, \text{ }   \bar{b}_1 C^{(003)} \bar{F}_3 \bar{F}_4, \text{ } C^{(003)} \bar{d}_2^2 \bar{F}_3,  \\
    & \,  \left.    \bar{b}_2 \left(C^{(003)}\right)^2 \bar{C}_2 \bar{F}_1, \text{ } \bar{b}_3 C^{(003)} \bar{C}_2^2 \bar{F}_3, \text{ } \bar{b}_1 \bar{b}_2 \bar{b}_3^2 C^{(003)} \bar{F}_3\right\}.
\end{align*} 
Subsequent to a modification of the generators, an extensive amount of analysis reveals that there exist $378$ permissible terms which constitute the components within the expansion of the expression $\{\bar{I}_1, \bar{I}_2\}$. That is, \begin{align*}
    \{\bar{I}_1, \bar{I}_2\} = \Gamma_{99}^1 \bar{b}_3 \bar{F}_2 \bar{I}_1 + \ldots +\Gamma_{99}^{377} \bar{b}_3 \bar{c}_1\bar{C}_2^2 \bar{F}_3 + \Gamma_{99}^{378} \bar{b}_1 \bar{b}_2 \bar{b}_3^2 \bar{c}_1 \bar{F}_3.
\end{align*} Following extensive calculations, the detailed and explicit form of the expansion for the brackets can be expressed as follows
\begin{align}
\{\bar{I}_1, \bar{I}_2\}&=-\frac{{\rm i}}{128} \bigl(20 \bar{b}_3 \bar{F}_1 \bar{C}_2^3-28 \bar{b}_3 \bar{c}_1 \bar{F}_1 \bar{C}_2^2+36 \bar{b}_3 \bar{c}_1 \bar{F}_3 \bar{C}_2^2-48 \bar{D}_4 \bar{G}_1 \bar{C}_2^2+8 \bar{D}_2^2 \bar{F}_1 \bar{C}_2+60 \bar{D}_3^2 \bar{F}_1 \bar{C}_2+32 \bar{b}_3^2 \bar{D}_2 \bar{F}_1 \bar{C}_2\nonumber \\
& \hskip 0.7cm-8 \bar{d}_1 \bar{D}_3 \bar{F}_1 \bar{C}_2-16 \bar{D}_2 \bar{D}_3 \bar{F}_1 \bar{C}_2+16 \bar{c}_1 \bar{E}_1 \bar{F}_1 \bar{C}_2-9 \bar{D}_2^2 \bar{F}_3 \bar{C}_2-52 \bar{D}_3^2 \bar{F}_3 \bar{C}_2+52 \bar{d}_1 \bar{D}_3 \bar{F}_3 \bar{C}_2+26 \bar{C}_2\bar{D}_2 \bar{D}_3 \bar{F}_3\nonumber  \\
& \hskip 0.7cm-72 \bar{c}_1 \bar{E}_1 \bar{F}_3 \bar{C}_2-16 \bar{c}_1 \bar{D}_3 \bar{G}_1 \bar{C}_2-6 \bar{c}_1 \bar{D}_2^2 \bar{F}_1-52 \bar{c}_1 \bar{D}_3^2 \bar{F}_1-88 \bar{b}_3 \bar{D}_2 \bar{E}_1 \bar{F}_1+320 \bar{b}_3 \bar{D}_3 \bar{E}_1 \bar{F}_1+26 \bar{c}_1 \bar{D}_2 \bar{D}_3 \bar{F}_3 \nonumber \\
& \hskip 0.7cm+72 \bar{b}_3 \bar{D}_2 \bar{E}_1 \bar{F}_3+208 \bar{b}_3 \bar{D}_3 \bar{E}_1 \bar{F}_3+32 \bar{E}_1 \bar{F}_2 \bar{F}_3-192 \bar{E}_1 \bar{F}_1 \bar{F}_4+28 \bar{b}_3 \bar{D}_2^2 \bar{G}_1-104 \bar{b}_3 \bar{D}_3^2 \bar{G}_1+96 \bar{E}_1^2 \bar{G}_1 \nonumber \\
& \hskip 0.7cm+40 \bar{b}_3 \bar{D}_2 \bar{D}_3 \bar{G}_1-80 \bar{D}_3 \bar{F}_2 \bar{G}_1-168 \bar{D}_2 \bar{F}_4 \bar{G}_1+208 \bar{D}_3 \bar{F}_4 \bar{G}_1+288 \bar{E}_1^2 \bar{G}_2-52 \bar{b}_3 \bar{D}_2 \bar{D}_3 \bar{G}_2+104 \bar{D}_2 \bar{F}_2 \bar{G}_2 \nonumber \\
& \hskip 0.7cm+160 \bar{D}_3 \bar{F}_2 \bar{G}_2+48 \bar{d}_1 \bar{F}_4 \bar{G}_2+48 \bar{D}_2 \bar{F}_4 \bar{G}_2-48 \bar{D}_3 \bar{F}_4 \bar{G}_2+96 \bar{D}_4 \bar{F}_4 \bar{G}_2-32 \bar{D}_2 \bar{E}_1 \bar{H}_1-640 \bar{D}_3 \bar{E}_1 \bar{H}_1 \nonumber \\
& \hskip 0.7cm-208 \bar{D}_3 \bar{E}_1 \bar{H}_2+\bar{b}_1 \bigl(\bar{b}_2 \bigl(2 \bigl(4 \bar{c}_1 \bar{F}_1-20 \bar{C}_2 \bar{F}_1-6 \bar{c}_1 \bar{F}_3+7 \bar{C}_2 \bar{F}_3\bigl) \bar{b}_3^2+28 \bar{C}_2 \bar{H}_1 \bar{b}_3+8  \bar{b}_3\bigl(\bar{c}_1+\bar{C}_2\bigl) \bar{H}_2 \nonumber \\
& \hskip 0.7cm-28 \bar{c}_1 \bar{D}_4 \bar{F}_1+32 \bar{C}_2 \bar{D}_4 \bar{F}_1+16 \bar{c}_1 \bar{D}_4 \bar{F}_3-34 \bar{C}_2 \bar{D}_4 \bar{F}_3-7 \bar{c}_1^2 \bar{G}_1-7 \bar{C}_2^2 \bar{G}_1+14 \bar{c}_1 \bar{C}_2 \bar{G}_1+9 \bigl(\bar{c}_1-\bar{C}_2\bigl)^2 \bar{G}_2\bigl)\nonumber \\
& \hskip 0.7cm+2 \bigl(\bigl(12 \bar{D}_3 \bar{G}_1-8 \bar{E}_1 \bar{F}_3\bigl) \bar{b}_3^2-\bar{c}_1 \bigl(4 \bar{D}_3 \bigl(\bar{F}_1+5 \bar{F}_3\bigl)+7 \bigl(\bar{D}_2 \bar{F}_3+2 \bar{C}_2 \bar{G}_2\bigl)\bigl) \bar{b}_3+2 \bigl(9 \bar{C}_2 \bar{D}_3 \bar{F}_1-3 \bar{C}_2 \bar{D}_3 \bar{F}_3\nonumber \\
& \hskip 0.7cm-20 \bar{F}_4 \bar{G}_1+8 \bar{E}_1 \bar{H}_2+26 \bar{D}_3 \bar{I}_1\bigl) \bar{b}_3+\bar{c}_1^2 \bar{C}_2 \bar{F}_3+\bar{D}_4 \bigl(28 \bar{D}_3 \bar{G}_1-8 \bar{E}_1 \bar{F}_3\bigl)-2 \bar{C}_2 \bigl(7 \bar{F}_1 \bar{F}_4-12 \bar{F}_3 \bar{F}_4+11 \bar{E}_1 \bar{G}_2\nonumber \\
& \hskip 0.7cm-2 \bar{D}_3 \bar{H}_1+6 \bar{D}_3 \bar{H}_2\bigl)+\bar{c}_1 \bigl(-\bar{C}_2^2 \bar{F}_3+14 \bar{F}_1 \bar{F}_4-6 \bar{F}_3 \bar{F}_4+18 \bar{E}_1 \bar{G}_2+20 \bar{D}_3 \bar{H}_2\bigl)\bigl)\bigl)+16 \bar{b}_2^2 \bar{b}_3^2 \bar{I}_2+96 \bar{D}_3^2 \bar{I}_2\nonumber \\
& \hskip 0.7cm-2 \bar{b}_2 \bigl(4 \bigl(-2 \bar{E}_1 \bar{F}_1+\bar{D}_2 \bar{G}_1+7 \bigl(\bar{d}_1+\bar{D}_2-\bar{D}_3+2 \bar{D}_4\bigl) \bar{G}_2\bigl) \bar{b}_3^2+\bigl(\bar{C}_2 \bigl(17 \bar{D}_3-4 \bar{d}_1\bigl) \bar{F}_1+24 \bar{C}_2 \bigl(\bar{d}_1+\bar{D}_2\nonumber \\
& \hskip 0.7cm-\bar{D}_3+2 \bar{D}_4\bigl) \bar{F}_3+16 \bigl(\bar{F}_2 \bar{G}_2-7 \bar{E}_1 \bar{H}_1\bigl)+28 \bigl(\bar{D}_2+2 \bar{D}_3\bigl) \bar{I}_2\bigl) \bar{b}_3+3 \bar{c}_1^2 \bar{C}_2 \bar{F}_1-\bar{c}_1 \bigl(3 \bar{F}_1 \bar{C}_2^2-8 \bar{b}_3 \bar{C}_2 \bar{G}_1\nonumber  \\
& \hskip 0.7cm+13 \bar{b}_3 \bigl(2 \bigl(\bar{d}_1+\bar{D}_2+2 \bar{D}_4\bigl) \bar{F}_3+\bar{D}_3 \bigl(\bar{F}_1-2 \bar{F}_3\bigl)\bigl)+2 \bar{F}_2 \bigl(5 \bar{F}_1-9 \bar{F}_3\bigl)-42 \bar{E}_1 \bar{G}_1\bigl)-8 \bar{D}_4^2 \bar{G}_2+4 \bar{D}_4 \bigl(18 \bar{E}_1 \bar{F}_1\nonumber \\
& \hskip 0.7cm-\bigl(\bar{d}_1+\bar{D}_2-\bar{D}_3\bigl) \bar{G}_2+6 \bar{C}_2 \bar{H}_2\bigl)+2 \bar{C}_2 \bigl(7 \bar{F}_1 \bar{F}_2-12 \bar{F}_3 \bar{F}_2-27 \bar{E}_1 \bar{G}_1+6 \bigl(\bar{d}_1+\bar{D}_2-\bar{D}_3\bigl) \bar{H}_2\bigl)\bigl)\bigl) \, .
\end{align}

\noindent  These relations show that the $20$ generators in $\mathcal{P}$ form a polynomial algebra of degree $4$, i.e., $\textbf{Alg} \left\langle \mathcal{P} \right \rangle: = \mathcal{Q}_{\mathfrak{su}(4)}(4)$ is closed under the Poisson bracket $\{\cdot,\cdot\}.$   It merely remains to determine the additional dependence relations satisfied by the generators.

\subsection{Extra polynomials relations} 
\label{4.12}

As mentioned, there appear algebraic relations of the type \eqref{algrel}. 
 In the following, we indicate these relations for increasing degree. The first nontrivial relation appears in degree eight, and it is given by
\begin{align}
		\bar{b}_3\bar{C}_2^2-\frac{1}{4}(\bar{D}_2-2 \bar{D}_3)^2+(\bar{d}_1-\bar{b}_1 \bar{b}_3+2 \bar{D}_4)\bar{D}_3+\frac{1}{2}\bar{b}_1 \bar{b}_2(\bar{b}_3^2-\bar{D}_4)&-\bar{b}_2 \bar{b}_3(\bar{d}_1+\bar{D}_2-\bar{D}_3+2 \bar{D}_4) \nonumber\\
		&-2 \bar{C}_2 \bar{E}_1+\bar{b}_2 \bar{F}_2+\bar{b}_1 \bar{F}_4=0 \, . \label{alge8}
\end{align}
 This relation was already obtained in \cite{nuclear}. At degree nine no algebraic relations appear, a result also remarked in the same reference. At degree ten we get the three algebraic relations:
	\begin{align}
		&(\bar{b}_2 \bar{b}_3-\bar{D}_3)\bar{F}_1+\frac{1}{2}\bar{D}_2 \bar{F}_3+\bar{C}_2 \bar{G}_2-\bar{b}_2 \bar{H}_1=0 \, ,\nonumber\\ 
		&\frac{1}{2}\bar{D}_2 \bar{F}_1+(\bar{b}_1 \bar{b}_3-\bar{d}_1-\bar{D}_2+\bar{D}_3-2 \bar{D}_4)\bar{F}_3+\bar{C}_2 \bar{G}_1-\bar{b}_1 \bar{H}_2=0 \, ,\\
		&\bar{E}_1^2-\frac{1}{16}\bigl(\bar{b}_1 \bar{b}_2(\bar{c}_1-\bar{C}_2)^2+4\bigl(2 \bar{b}_3^2 \bar{C}_2^2-4 \bar{b}_3 \bar{D}_3(\bar{d}_1+\bar{D}_2-\bar{D}_3+2\bar{D}_4)+\bar{C}_2\bigl(\bar{c}_1 \bar{D}_2-\bar{C}_2(\bar{D}_2+2\bar{D}_4)\bigl) \nonumber\\
		&\hskip 1.4cm+4 \bar{D}_3 \bar{F}_2\bigl)\bigl)+(\bar{d}_1+\bar{D}_2-\bar{D}_3+2\bar{D}_4)\bar{F}_4=0 \, .\nonumber
\end{align}
For degree-eleven, two algebraic relations appear
\begin{equation}
\bar{E}_1 \bar{F}_1-\frac{1}{2}\bar{D}_2 \bar{G}_1-\bar{C}_2 \bar{H}_1+\bar{b}_2 \bar{I}_2=0,\qquad 
\bar{E}_1 \bar{F}_3-\frac{1}{2}\bar{D}_2 \bar{G}_2-\bar{C}_2 \bar{H}_2+\bar{b}_1 \bar{I}_1=0,
\end{equation}
 while for degree-twelve only one appears:
\begin{align}
\bar{b}_1 \bigl(\bar{b}_2^2 \bigl(\bar{c}_1-\bar{C}_2\bigl)^2+16 \bar{D}_3 \bigl(\bar{b}_3 \bar{D}_3-\bar{F}_4\bigl)+8 \bar{b}_2 \bar{D}_3 \bigl(\bar{D}_4-\bar{b}_3^2\bigl)\bigl)\nonumber\\+4 \bigl(\bar{b}_2 \bar{C}_2 \bigl(2 \bar{b}_3^2 \bar{C}_2+\bar{c}_1 \bar{D}_2-\bar{C}_2 \bigl(\bar{D}_2+2 \bar{D}_4\bigl)\bigl)-4 \bar{b}_3 \bar{C}_2^2 \bar{D}_3+4 \bar{C}_2^2 \bar{F}_4+\bar{D}_2^2 \bar{D}_3+4 \bar{F}_3^2\bigl)=0 \, .
\end{align}
 A degree-thirteen relation is given by:
\begin{align}
\bar{c}_1 \bigl(\bar{b}_1 \bar{b}_2 \bar{F}_4+\bar{C}_2^2 \bar{D}_3-\bar{b}_2\bar{C}_2 \bar{E}_1 -\bar{b}_1 \bar{D}_3^2\bigl)+\bar{b}_1 \bar{C}_2 (\bar{D}_3^2-\bar{b}_2 \bar{F}_4)+ \bar{b}_2 \bar{C}_2^2\bar{E}_1+2 \bar{C}_2 \bar{D}_2 \bar{F}_4-(\bar{C}_2^3 +2 \bar{E}_1 \bar{D}_2) \bar{D}_3-4 \bar{F}_3 \bar{G}_2=0 \, .
\end{align}
An example of degree-fourteen algebraic relation is:
\begin{align}
&\bar{b}_2 \bigl(-\bar{b}_1 \bar{c}_1 \bar{G}_2+\bar{b}_1 \bar{C}_2 \bar{G}_2+2 \bar{b}_3 \bar{D}_2 \bar{F}_1+\bar{c}_1 \bar{C}_2 \bigl(\bar{F}_1-\bar{F}_3\bigl)+\bar{c}_1^2 \bar{F}_3-\bar{C}_2^2 \bar{F}_1-2 \bar{D}_2 \bar{H}_1+4 \bar{F}_1 \bar{F}_4-4 \bar{E}_1 \bar{G}_2\bigl)+\bar{D}_3 \bigl(4 \bar{b}_1 \bar{b}_3 \bar{F}_3+4 \bar{c}_1 \bar{G}_2 \nonumber \\
&-2 \bar{F}_3 \bigl(4 \bar{b}_3^2+2 \bar{d}_1+\bar{D}_2\bigl)\bigl) +4(4 \bar{b}_3- \bar{b}_1)  \bar{F}_3 \bar{F}_4+(8 \bar{E}_1 \bar{C}_2 -4 \bar{b}_3 \bar{C}_2^2+\bar{D}_2^2) \bar{F}_3+4 \bar{D}_3^2 \bigl(\bar{F}_3-\bar{F}_1\bigl)-16 \bar{F}_4 \bar{H}_2+16 \bar{E}_1 \bar{I}_1=0 \,\nonumber \, , \\ .
\end{align}
whereas an example of degree-fifteen one is:
\begin{align}
&\bar{b}_2 \bigl(2 \bar{c}_1 \bigl(\bar{b}_1 \bigl(\bar{b}_3 \bar{F}_3-2 \bar{H}_2\bigl)+2 \bar{C}_2 \bigl(\bar{G}_1+\bar{G}_2\bigl)+\bar{D}_2 \bar{F}_1-\bar{D}_3 \bar{F}_1\bigl)+\bar{G}_1 \bigl(4 \bar{b}_3 \bar{D}_3-2 \bar{c}_1^2\bigl)+\bar{F}_1 \bigl(8 \bar{E}_1 \bar{b}_3-\bar{C}_2 \bar{D}_2-2 \bar{C}_2 \bar{D}_3\bigl)+2 \bar{b}_1 \bar{C}_2 \bar{H}_2 \nonumber \\
&+\bar{G}_2 \bigl(2 \bar{b}_3 \bar{D}_2-4 \bar{b}_1 \bar{D}_4+4 \bar{b}_1 \bar{b}_3^2-8 \bar{F}_2\bigl)-4 \bar{b}_1 \bar{b}_3 \bar{I}_1+\bar{F}_3 \bigl(2 \bar{C}_2 \bar{D}_3-2 \bar{C}_2 \bar{d}_1\bigl)+8 \bar{I}_1 \bigl(\bar{d}_1+\bar{D}_2-\bar{D}_3+2 \bar{D}_4\bigl)-8 \bar{E}_1 \bar{H}_2\bigl) \nonumber\\
&-2 \bigl(2 \bar{b}_3 \bigl(2 \bar{b}_1 \bar{D}_3 \bar{G}_2+\bar{F}_3 \bigl(\bar{c}_1 \bar{D}_2-4 \bar{C}_2 \bar{D}_4\bigl)+2 \bar{C}_2\bar{D}_3 \bar{F}_1+8 \bar{F}_4 \bar{G}_1\bigl)+\bar{D}_3 \bigl(\bar{b}_1 \bar{F}_3 \bigl(\bar{C}_2-\bar{c}_1\bigl)+4 \bigl(\bar{D}_2+2 \bar{D}_4\bigl) \bar{G}_1+4 \bar{E}_1 \bar{F}_1\bigl) \nonumber\\
&+\bar{G}_2 \bigl(\bar{D}_2^2-4 \bigl(\bar{b}_1 \bar{F}_4+\bar{E}_1 \bar{c}_1\bigl)\bigl)+8 \bar{b}_3^3 \bar{C}_2 \bar{F}_3-8 \bar{b}_3^2 \bar{D}_3 \bar{G}_1-\bar{C}_2 \bar{F}_3 \bigl(\bar{C}_2-2 \bar{c}_1\bigl) \bigl(\bar{C}_2-\bar{c}_1\bigl)-4 \bar{c}_1 \bar{F}_1 \bar{F}_4-8 \bar{E}_1 \bigl(\bar{D}_2+2 \bar{D}_4\bigl) \bar{F}_3 \nonumber \\
&-16 \bar{F}_2 \bar{I}_1\bigl)+\bar{b}_2^2 \bigl(\bar{b}_3 \bar{F}_1 \bigl(\bar{c}_1+3 \bar{C}_2\bigl)-2 \bar{H}_1 \bigl(\bar{c}_1+\bar{C}_2\bigl)\bigl)=0 \, . \nonumber \\
\end{align}
In conclusion, we provide an illustration of algebraic relations of degrees sixteen and seventeen, respectively. To begin, the relation of degree sixteen is presented as follows:
\begin{align}
&-\bar{b}_2 \bigl(-4 \bar{b}_3 \bigl(\bar{b}_1 \bar{G}_2 \bigl(\bar{c}_1-\bar{C}_2\bigl)-\bar{c}_1 \bar{C}_2 \bar{F}_1+2 \bar{c}_1 \bar{C}_2 \bar{F}_3+2 \bigl(\bar{D}_2+2 \bar{D}_3\bigl) \bar{H}_1+4 \bar{F}_4 \bar{F}_1\bigl)+5 \bar{b}_1 \bar{c}_1 \bar{C}_2 \bar{F}_3+4 \bar{b}_1 \bar{I}_1 \bigl(\bar{c}_1+\bar{C}_2\bigl)-3 \bar{b}_1 \bar{c}_1^2 \bar{F}_3 \nonumber \\
&-2 \bar{b}_1 \bar{C}_2^2 \bar{F}_3+4 \bar{b}_3^2 \bigl(2 \bar{D}_2 \bar{F}_1+4 \bar{D}_3 \bar{F}_1-\bar{D}_2 \bar{F}_3\bigl)-4 \bar{D}_3 \bigl(\bar{G}_1 \bigl(\bar{c}_1+\bar{C}_2\bigl)+\bar{C}_2 \bar{G}_2-2 \bar{D}_4 \bar{F}_1\bigl)+8 \bar{E}_1 \bar{c}_1 \bar{F}_1+12 \bar{E}_1 \bar{c}_1 \bar{F}_3+4 \bar{C}_2 \bar{d}_1 \bar{G}_2 \nonumber \\
&+2 \bar{C}_2 \bar{D}_2 \bar{G}_1-16 \bar{E}_1 \bar{C}_2 \bar{F}_1+4 \bar{E}_1 \bar{C}_2 \bar{F}_3+16 \bar{F}_4 \bar{H}_1\bigl)+4 \bigl(2 \bar{b}_1 \bar{c}_1 \bar{D}_3 \bar{G}_2+\bar{F}_3 \bigl(\bar{b}_3^2 \bigl(6 \bar{C}_2^2-8 \bar{b}_1 \bar{D}_3\bigl)+4 \bar{b}_1 \bar{D}_3 \bar{D}_4-2 \bar{C}_2^2 \bar{D}_4+4 \bar{D}_3 \bar{F}_2 \nonumber \\
&-20 \bar{E}_1^2\bigl)-2 \bar{b}_1 \bar{C}_2 \bar{D}_3 \bar{G}_2+\bar{F}_4 \bar{F}_3 \bigl(12 \bar{b}_1 \bar{b}_3+4 \bar{d}_1+4 \bar{D}_2-4 \bar{D}_3+8 \bar{D}_4\bigl)+\bar{b}_3 \bigl(4 \bar{D}_2 \bar{D}_3 \bar{F}_1-\bar{D}_2^2 \bar{F}_3\bigl)-8 \bar{b}_1 \bar{F}_4 \bar{H}_2-4 \bar{D}_2 \bar{F}_1 \bar{F}_4\nonumber \\
&+8 \bar{E}_1 \bar{D}_3 \bar{G}_1+8 \bar{E}_1 \bar{D}_2 \bar{G}_2-4 \bar{D}_2 \bar{D}_3 \bar{H}_1\bigl)+2 \bar{b}_2^2 \bigl(\bar{F}_1 \bigl(4 \bar{b}_3^3+\bar{c}_1 \bar{C}_2-\bar{c}_1^2\bigl)-4 \bar{b}_3^2 \bar{H}_1+2 \bar{C}_2 \bar{I}_2\bigl)=0\nonumber \, , \\ 
\end{align}
whereas a degree-seventeen algebraic relation reads:
\begin{align}
&8 \bar{b}_2^2  \bar{b}_3^2  \bar{I}_2 -2 \bigl(16 \bar{b}_3^4 \bar{C}_2  \bar{F}_3  +4 \bigl( \bar{b}_1  \bar{c}_1  \bar{F}_3-4  \bar{D}_3  \bar{G}_1\bigl)  \bar{b}_3^3+2 \bigl(-2  \bar{C}_2  \bar{D}_2  \bar{F}_1+16  \bar{F}_4  \bar{G}_1-4  \bar{b}_1 \bigl( \bar{E}_1  \bar{F}_3+ \bar{D}_3 \bigl( \bar{G}_2- \bar{G}_1\bigl)\bigl)\nonumber \\
&+ \bar{c}_1 \bigl( \bar{D}_2  \bar{F}_1-4  \bar{D}_3  \bar{F}_1-4  \bar{d}_1  \bar{F}_3+4  \bar{D}_3  \bar{F}_3-8  \bar{D}_4  \bar{F}_3+2  \bar{C}_2  \bar{G}_1-2  \bar{b}_1  \bar{H}_2\bigl)\bigl)  \bar{b}_3^2+\bigl( \bar{C}_2 \bigl( \bar{F}_1-2  \bar{F}_3\bigl)  \bar{c}_1^2+\bigl(\bigl(6  \bar{F}_3-5  \bar{F}_1\bigl)  \bar{C}_2^2 +2  \bar{b}_1  \bar{G}_2  \bar{C}_2\nonumber \\
&+ \bar{b}_1  \bar{D}_2  \bar{F}_3+4  \bar{F}_2  \bar{F}_3+2  \bar{b}_1  \bar{D}_3 \bigl( \bar{F}_3-3  \bar{F}_1\bigl)-4  \bar{E}_1 \bigl( \bar{G}_1+2  \bar{G}_2\bigl)+4 \bigl( \bar{d}_1- \bar{D}_3+2  \bar{D}_4\bigl)  \bar{H}_2\bigl)  \bar{c}_1-2 \bigl(-2  \bar{F}_1  \bar{C}_2^3+ \bar{b}_1  \bar{D}_3 \bigl(2  \bar{F}_1+3  \bar{F}_3\bigl)\bar{C}_2 \nonumber\\
& -16  \bar{d}_1  \bar{E}_1  \bar{F}_3-10  \bar{D}_2  \bar{E}_1  \bar{F}_3+ \bar{D}_2^2  \bar{G}_1-4  \bar{D}_3^2  \bar{G}_1-4  \bar{b}_1  \bar{F}_4  \bar{G}_1+ \bar{D}_3 \bigl(-20  \bar{E}_1  \bar{F}_1+8  \bar{E}_1  \bar{F}_3+8  \bar{d}_1  \bar{G}_1+6  \bar{D}_2  \bar{G}_1\bigl)+ \bar{D}_2^2  \bar{G}_2+2  \bar{d}_1  \bar{D}_2  \bar{G}_2\nonumber\\
&-4  \bar{b}_1  \bar{F}_4  \bar{G}_2+8  \bar{b}_1  \bar{E}_1  \bar{H}_2-4  \bar{b}_1  \bar{d}_1  \bar{I}_1+16  \bar{F}_2  \bar{I}_1\bigl)\bigl)  \bar{b}_3+4  \bar{C}_2  \bar{D}_3^2  \bar{F}_1- \bar{C}_2  \bar{D}_2^2  \bar{F}_3-4  \bar{C}_2  \bar{D}_3^2  \bar{F}_3+2  \bar{b}_1  \bar{c}_1^2  \bar{C}_2  \bar{F}_3+4  \bar{C}_2  \bar{d}_1  \bar{D}_3  \bar{F}_3+2  \bar{C}_2  \bar{D}_2  \bar{D}_3  \bar{F}_3\nonumber \\
&+16  \bar{b}_1  \bar{D}_4  \bar{E}_1  \bar{F}_3-32  \bar{E}_1  \bar{F}_2  \bar{F}_3+4  \bar{b}_1  \bar{C}_2  \bar{F}_1  \bar{F}_4-16  \bar{E}_1  \bar{F}_1  \bar{F}_4+2  \bar{b}_1  \bar{C}_2  \bar{F}_3  \bar{F}_4-8  \bar{b}_1  \bar{D}_3  \bar{D}_4  \bar{G}_1+16  \bar{D}_3  \bar{F}_2  \bar{G}_1-4  \bar{D}_2  \bar{F}_4  \bar{G}_1+16  \bar{D}_3  \bar{F}_4  \bar{G}_1\nonumber \\
&+32  \bar{E}_1^2  \bar{G}_2-12  \bar{b}_1  \bar{C}_2  \bar{E}_1  \bar{G}_2+8  \bar{D}_2  \bar{F}_2  \bar{G}_2-8  \bar{d}_1  \bar{F}_4  \bar{G}_2-8  \bar{D}_2  \bar{F}_4  \bar{G}_2+8  \bar{D}_3  \bar{F}_4  \bar{G}_2-16  \bar{D}_4  \bar{F}_4  \bar{G}_2+2  \bar{c}_1 \bigl(-2  \bar{F}_1  \bar{D}_3^2+ \bar{D}_2  \bar{D}_3\bar{F}_3\nonumber \\
&  -4  \bar{C}_2  \bar{E}_1  \bar{F}_3+ \bar{b}_1 \bigl(- \bar{C}_2^2  \bar{F}_3+2  \bar{F}_1  \bar{F}_4+ \bar{F}_3  \bar{F}_4+2  \bar{E}_1  \bar{G}_2\bigl)\bigl)+8  \bar{b}_1  \bar{C}_2  \bar{D}_3  \bar{H}_1-32  \bar{D}_3  \bar{E}_1  \bar{H}_1-16  \bar{D}_3  \bar{E}_1  \bar{H}_2\bigl)- \bar{b}_2 \bigl(8 \bigl( \bar{c}_1  \bar{F}_1+ \bar{b}_1  \bar{G}_2\bigl)  \bar{b}_3^3 \nonumber \\
&+2 \bigl(8  \bar{E}_1  \bar{F}_1+ \bar{b}_1 \bigl(2  \bar{C}_2  \bar{F}_1-5  \bar{c}_1  \bar{F}_3+7  \bar{C}_2  \bar{F}_3\bigl)+2  \bar{D}_2  \bar{G}_1\bigl)  \bar{b}_3^2+4 \bigl( \bar{c}_1 \bigl( \bar{D}_2  \bar{F}_3+2  \bar{D}_4  \bar{F}_3- \bar{D}_3 \bigl( \bar{F}_1+ \bar{F}_3\bigl)+ \bar{d}_1 \bigl(2  \bar{F}_1+ \bar{F}_3\bigl)+ \bar{b}_1  \bar{H}_2\bigl)\nonumber \\
&- \bar{C}_2 \bigl(\bigl( \bar{D}_2+2  \bar{D}_4\bigl)  \bar{F}_3- \bar{D}_3 \bigl( \bar{F}_1+ \bar{F}_3\bigl)+ \bar{d}_1 \bigl(2  \bar{F}_1+ \bar{F}_3\bigl)+ \bar{b}_1 \bigl( \bar{H}_1+ \bar{H}_2\bigl)\bigl)+2 \bigl(- \bar{b}_1  \bar{D}_4  \bar{G}_2+4  \bar{F}_2  \bar{G}_2-6  \bar{E}_1  \bar{H}_1+ \bar{D}_2  \bar{I}_2+6  \bar{D}_3  \bar{I}_2\bigl)\bigl)  \bar{b}_3\nonumber \\
&+2  \bar{c}_1  \bar{C}_2^2  \bar{F}_1-2  \bar{c}_1^2  \bar{C}_2  \bar{F}_1-16  \bar{D}_4  \bar{E}_1  \bar{F}_1-8  \bar{c}_1  \bar{F}_1  \bar{F}_2+8  \bar{C}_2  \bar{F}_1  \bar{F}_2-4  \bar{b}_1  \bar{C}_2  \bar{D}_4  \bar{F}_3-4  \bar{c}_1  \bar{F}_2  \bar{F}_3+4  \bar{C}_2  \bar{F}_2  \bar{F}_3-2  \bar{b}_1  \bar{c}_1^2  \bar{G}_1-2  \bar{b}_1  \bar{C}_2^2  \bar{G}_1\nonumber \\
&+4  \bar{b}_1  \bar{c}_1  \bar{C}_2  \bar{G}_1-4  \bar{c}_1  \bar{E}_1  \bar{G}_1+12  \bar{C}_2  \bar{E}_1  \bar{G}_1+ \bar{b}_1  \bar{c}_1^2  \bar{G}_2+ \bar{b}_1  \bar{C}_2^2  \bar{G}_2+16  \bar{D}_4^2  \bar{G}_2-2  \bar{b}_1  \bar{c}_1  \bar{C}_2  \bar{G}_2+8  \bar{d}_1 \bar{D}_4  \bar{G}_2+8  \bar{D}_2  \bar{D}_4  \bar{G}_2-8  \bar{D}_3  \bar{D}_4  \bar{G}_2 \nonumber \\
&-8  \bar{C}_2 \bigl( \bar{d}_1+ \bar{D}_2- \bar{D}_3+2  \bar{D}_4\bigl)  \bar{H}_2\bigl)=0\nonumber \, .\\ 
\end{align}

%With all the identities above, it is enough to verify that the Poisson bracket $\{\cdot,\cdot\}$ satisfies the Jacobi identity with all the generators in $\mathcal{P}.$ From this we can finally claim that $\mathcal{Q}_{\mathfrak{su}(4)}(5)$ is a quintic Poisson algebra.

\section{Conclusions}
\label{sec5}

In this paper, we have examined the polynomial algebra arising from the reduction chain $\mathfrak{su}(4) \supset \mathfrak{su}(2) \times \mathfrak{su}(2)$ in detail. In our analysis of the commutant associated with the $\mathfrak{su}(2) \times \mathfrak{su}(2)$ subalgebra, we have been able to identify $20$ linearly independent and indecomposable polynomials, up to degree nine in the underlying Lie algebra generators, which ultimately close in a polynomial algebra. This process required us to deal with polynomial expansions up to degree $17$ in the Lie algebra generators. This polynomial structure is characterized by the presence of five central elements of various degrees, i.e. three of degree two, one of degree three, and another of degree four. Appropriate combinations of these central elements have been shown to give rise to the three (unsymmetrized) Casimir elements of $\mathfrak{su}(4)$. Although the existence of these $20$ linearly independent and indecomposable polynomials is already documented in the literature (see \cite{nuclear}), no prior information has been available on the specific polynomial structure defined by these elements. With this paper, we have filled this gap.
 We emphasize that the computations required to obtain such a closed polynomial structure proved to be extremely demanding from a computational perspective. Initially, when trying to solve the problem by considering all possible terms in expansions based solely on their fixed degree, we were unable to proceed beyond degree $11$. For this reason, in the paper we have highlighted several key computational aspects encountered in the study of polynomial expansions of such high degrees.
Additionally, we discussed a potential efficient technique that could be employed, leading us to define new polynomials using Poisson brackets to ensure algebra closure. As mentioned previously, the results we obtained in this way match others available in the literature, which were derived using alternative techniques.  A novel aspect of our approach is the implementation of a grading process for the generators, following the method introduced in \cite{campoamor2025On}. This approach streamlines computations and identifies permissible terms within specific degree brackets, demonstrating how gradings can effectively clarify complex algebraic structures. Thanks to this technique, we have been able to reveal that the hidden symmetry algebra emerging from the embedding chain $\mathfrak{su}(4) \supset \mathfrak{su}(2) \times \mathfrak{su}(2)$ gives rise to a quartic polynomial Poisson algebra. As a byproduct, this paper has also demonstrated that the grading method not only simplifies complex expansions in the Poisson bracket relations of a given degree but is, in fact, indispensable for deriving these relations when the degree becomes so high that standard ``brute-force" methods are no longer applicable.

\medskip

 This ansatz can be generalized and applied to any other reduction chain $\mathfrak{g}\supset \mathfrak{g}^{\prime}$ involving reductive Lie algebras where the subalgebra embedding is not regular, and for which the usual argumentation using the root system is no longer applicable \cite{MR4660510}. However, as shown in \cite{Hav3}, there are several alternatives to define non-canonical gradings in Lie algebras, that can be conveniently adapted to various non-regular reduction chains (see \cite{Hav1,Hav2} and references therein). Following this argumentation, the reduction chains of interest in physical applications (see e.g. \cite[Chapter 12]{gtp},\cite{MR2270799}) examined in the context of the missing label problem can be reevaluated from the perspective of polynomial algebras, allowing us to describe the complete algebraic structure, and potentially allowing us to make more effective of labeling operators to separate degenerate states. For applications in superintegrable systems \cite{MR2105429,MR2143019}, the precise knowledge of the associated polynomial algebra can be useful for the systematic construction of Hamiltonians admitting constants of the motion of degrees higher than two, providing new hierarchies of systems for which the classical criteria for the separation of variables no longer necessarily apply \cite{MR1939624,MR2385271}. It is worthy to be mentioned in this context that the observed relations between polynomial algebras and quasi-exactly solvable problems (see e.g. \cite{tur1988}) could be systematized by means of a detailed analysis of the (hidden) symmetry algebras and their associated polynomial structures. Work in these directions is currently in progress.

\section*{Acknowledgements}
This work was partially supported by the Future Fellowship FT180100099 and the Discovery Project DP190101529 from the Australian Research Council.  RCS  acknowledges financial support by the Agencia Estatal de Investigaci\'on (Spain) under  the grant PID2023-148373NB-I00 funded by MCIN/AEI/10.13039/501100011033/FEDER, UE.  The research of DL is partially funded by MUR - Dipartimento di Eccellenza 2023-2027, codice CUP \textsf{G43C22004580005} - codice progetto \textsf{DECC23\_012\_DIP} and partially supported by INFN-CSN4 (Commissione Scientifica Nazionale 4 - Fisica Teorica), MMNLP project. DL is a member of GNFM, INdAM.

\newpage

\appendix

\section{Explicit expressions of polynomials of degree six}
\label{appendixA}

The four linearly independent and indecomposable polynomials of degree six, satisfying the commutant constraint \eqref{pdes}, explicitly read

{\small \begin{equation}
\begin{split}
F_1=&-q_{12} q_{22} q_{31} s_{1}^2 t_{1} - q_{13} q_{23} q_{31} s_{1}^2 t_{1} + q_{12} q_{21} q_{32} s_{1}^2 t_{1} + 
q_{13} q_{21} q_{33} s_{1}^2 t_{1} + q_{12}^2 q_{31} s_{1} s_{2} t_{1} \\
&+ q_{13}^2 q_{31} s_{1} s_{2} t_{1} - 
q_{22}^2 q_{31} s_{1} s_{2} t_{1} - q_{23}^2 q_{31} s_{1} s_{2} t_{1} - q_{11} q_{12} q_{32} s_{1} s_{2} t_{1} + 
q_{21} q_{22} q_{32} s_{1} s_{2} t_{1}\\
& - q_{11} q_{13} q_{33} s_{1} s_{2} t_{1} + q_{21} q_{23} q_{33} s_{1} s_{2} t_{1} + 
q_{12} q_{22} q_{31} s_{2}^2 t_{1} + q_{13} q_{23} q_{31} s_{2}^2 t_{1} - q_{11} q_{22} q_{32} s_{2}^2 t_{1} \\
&- 
q_{11} q_{23} q_{33} s_{2}^2 t_{1} - q_{12}^2 q_{21} s_{1} s_{3} t_{1} - q_{13}^2 q_{21} s_{1} s_{3} t_{1} + 
q_{11} q_{12} q_{22} s_{1} s_{3} t_{1} + q_{11} q_{13} q_{23} s_{1} s_{3} t_{1} \\
&- q_{22} q_{31} q_{32} s_{1} s_{3} t_{1} + 
q_{21} q_{32}^2 s_{1} s_{3} t_{1} - q_{23} q_{31} q_{33} s_{1} s_{3} t_{1} + q_{21} q_{33}^2 s_{1} s_{3} t_{1} - 
q_{12} q_{21} q_{22} s_{2} s_{3} t_{1} \\
&+ q_{11} q_{22}^2 s_{2} s_{3} t_{1} - q_{13} q_{21} q_{23} s_{2} s_{3} t_{1} + 
q_{11} q_{23}^2 s_{2} s_{3} t_{1} + q_{12} q_{31} q_{32} s_{2} s_{3} t_{1} - q_{11} q_{32}^2 s_{2} s_{3} t_{1} \\
&+ 
q_{13} q_{31} q_{33} s_{2} s_{3} t_{1} - q_{11} q_{33}^2 s_{2} s_{3} t_{1} - q_{12} q_{21} q_{32} s_{3}^2 t_{1} + 
q_{11} q_{22} q_{32} s_{3}^2 t_{1} - q_{13} q_{21} q_{33} s_{3}^2 t_{1}\\
& + q_{11} q_{23} q_{33} s_{3}^2 t_{1} + 
q_{11} q_{22} q_{31} s_{1}^2 t_{2} - q_{11} q_{21} q_{32} s_{1}^2 t_{2} - q_{13} q_{23} q_{32} s_{1}^2 t_{2} + 
q_{13} q_{22} q_{33} s_{1}^2 t_{2}\\
& - q_{11} q_{12} q_{31} s_{1} s_{2} t_{2} + q_{21} q_{22} q_{31} s_{1} s_{2} t_{2} + 
q_{11}^2 q_{32} s_{1} s_{2} t_{2} + q_{13}^2 q_{32} s_{1} s_{2} t_{2} - q_{21}^2 q_{32} s_{1} s_{2} t_{2} \\
&- 
q_{23}^2 q_{32} s_{1} s_{2} t_{2} - q_{12} q_{13} q_{33} s_{1} s_{2} t_{2} + q_{22} q_{23} q_{33} s_{1} s_{2} t_{2} - 
q_{12} q_{21} q_{31} s_{2}^2 t_{2} + q_{11} q_{21} q_{32} s_{2}^2 t_{2} \\
&+ q_{13} q_{23} q_{32} s_{2}^2 t_{2} - 
q_{12} q_{23} q_{33} s_{2}^2 t_{2} + q_{11} q_{12} q_{21} s_{1} s_{3} t_{2} - q_{11}^2 q_{22} s_{1} s_{3} t_{2} - 
q_{13}^2 q_{22} s_{1} s_{3} t_{2}\\
& + q_{12} q_{13} q_{23} s_{1} s_{3} t_{2} + q_{22} q_{31}^2 s_{1} s_{3} t_{2} - 
q_{21} q_{31} q_{32} s_{1} s_{3} t_{2} - q_{23} q_{32} q_{33} s_{1} s_{3} t_{2} + q_{22} q_{33}^2 s_{1} s_{3} t_{2}\\
& + 
q_{12} q_{21}^2 s_{2} s_{3} t_{2} - q_{11} q_{21} q_{22} s_{2} s_{3} t_{2} - q_{13} q_{22} q_{23} s_{2} s_{3} t_{2} + 
q_{12} q_{23}^2 s_{2} s_{3} t_{2} - q_{12} q_{31}^2 s_{2} s_{3} t_{2} \\
&+ q_{11} q_{31} q_{32} s_{2} s_{3} t_{2} + 
q_{13} q_{32} q_{33} s_{2} s_{3} t_{2} - q_{12} q_{33}^2 s_{2} s_{3} t_{2} + q_{12} q_{21} q_{31} s_{3}^2 t_{2} - 
q_{11} q_{22} q_{31} s_{3}^2 t_{2} \\
&- q_{13} q_{22} q_{33} s_{3}^2 t_{2} + q_{12} q_{23} q_{33} s_{3}^2 t_{2} + 
q_{11} q_{23} q_{31} s_{1}^2 t_{3} + q_{12} q_{23} q_{32} s_{1}^2 t_{3} - q_{11} q_{21} q_{33} s_{1}^2 t_{3} \\
&- 
q_{12} q_{22} q_{33} s_{1}^2 t_{3} - q_{11} q_{13} q_{31} s_{1} s_{2} t_{3} + q_{21} q_{23} q_{31} s_{1} s_{2} t_{3} - 
q_{12} q_{13} q_{32} s_{1} s_{2} t_{3} + q_{22} q_{23} q_{32} s_{1} s_{2} t_{3} \\
&+ q_{11}^2 q_{33} s_{1} s_{2} t_{3} + 
q_{12}^2 q_{33} s_{1} s_{2} t_{3} - q_{21}^2 q_{33} s_{1} s_{2} t_{3} - q_{22}^2 q_{33} s_{1} s_{2} t_{3} - 
q_{13} q_{21} q_{31} s_{2}^2 t_{3} \\
&- q_{13} q_{22} q_{32} s_{2}^2 t_{3} + q_{11} q_{21} q_{33} s_{2}^2 t_{3} + 
q_{12} q_{22} q_{33} s_{2}^2 t_{3} + q_{11} q_{13} q_{21} s_{1} s_{3} t_{3} + q_{12} q_{13} q_{22} s_{1} s_{3} t_{3} \\
&- 
q_{11}^2 q_{23} s_{1} s_{3} t_{3} - q_{12}^2 q_{23} s_{1} s_{3} t_{3} + q_{23} q_{31}^2 s_{1} s_{3} t_{3} + 
q_{23} q_{32}^2 s_{1} s_{3} t_{3} - q_{21} q_{31} q_{33} s_{1} s_{3} t_{3} \\
&- q_{22} q_{32} q_{33} s_{1} s_{3} t_{3} + 
q_{13} q_{21}^2 s_{2} s_{3} t_{3} + q_{13} q_{22}^2 s_{2} s_{3} t_{3} - q_{11} q_{21} q_{23} s_{2} s_{3} t_{3} - 
q_{12} q_{22} q_{23} s_{2} s_{3} t_{3} \\
&- q_{13} q_{31}^2 s_{2} s_{3} t_{3} - q_{13} q_{32}^2 s_{2} s_{3} t_{3} + 
q_{11} q_{31} q_{33} s_{2} s_{3} t_{3} + q_{12} q_{32} q_{33} s_{2} s_{3} t_{3} + q_{13} q_{21} q_{31} s_{3}^2 t_{3} \\
&- 
q_{11} q_{23} q_{31} s_{3}^2 t_{3} + q_{13} q_{22} q_{32} s_{3}^2 t_{3} - q_{12} q_{23} q_{32} s_{3}^2 t_{3} \, ,
\end{split}
\end{equation}}

{\small \begin{equation}
\begin{split}
F_2&=q_{11}^4 s_{1}^2 + 2 q_{11}^2 q_{12}^2 s_{1}^2 + q_{12}^4 s_{1}^2 + 2 q_{11}^2 q_{13}^2 s_{1}^2 + 
2 q_{12}^2 q_{13}^2 s_{1}^2 + q_{13}^4 s_{1}^2 - q_{21}^4 s_{1}^2 - 2 q_{21}^2 q_{22}^2 s_{1}^2 
- q_{22}^4 s_{1}^2 \\
&- 2 q_{21}^2 q_{23}^2 s_{1}^2 - 2 q_{22}^2 q_{23}^2 s_{1}^2 - q_{23}^4 s_{1}^2 - 
2 q_{21}^2 q_{31}^2 s_{1}^2 - q_{31}^4 s_{1}^2 - 4 q_{21} q_{22} q_{31} q_{32} s_{1}^2 - 
2 q_{22}^2 q_{32}^2 s_{1}^2 - 2 q_{31}^2 q_{32}^2 s_{1}^2 \\
&- q_{32}^4 s_{1}^2 - 4 q_{21} q_{23} q_{31} q_{33} s_{1}^2 - 4 q_{22} q_{23} q_{32} q_{33} s_{1}^2 - 
2 q_{23}^2 q_{33}^2 s_{1}^2 - 2 q_{31}^2 q_{33}^2 s_{1}^2 - 2 q_{32}^2 q_{33}^2 s_{1}^2 - 
q_{33}^4 s_{1}^2 \\
&+ 4 q_{11}^3 q_{21} s_{1} s_{2} + 4 q_{11} q_{12}^2 q_{21} s_{1} s_{2} + 
4 q_{11} q_{13}^2 q_{21} s_{1} s_{2} + 4 q_{11} q_{21}^3 s_{1} s_{2} + 4 q_{11}^2 q_{12} q_{22} s_{1} s_{2} + 
4 q_{12}^3 q_{22} s_{1} s_{2} \\
&+ 4 q_{12} q_{13}^2 q_{22} s_{1} s_{2} + 4 q_{12} q_{21}^2 q_{22} s_{1} s_{2} + 
4 q_{11} q_{21} q_{22}^2 s_{1} s_{2} + 4 q_{12} q_{22}^3 s_{1} s_{2} + 4 q_{11}^2 q_{13} q_{23} s_{1} s_{2} + 
4 q_{12}^2 q_{13} q_{23} s_{1} s_{2} \\
&+ 4 q_{13}^3 q_{23} s_{1} s_{2} + 4 q_{13} q_{21}^2 q_{23} s_{1} s_{2} + 
4 q_{13} q_{22}^2 q_{23} s_{1} s_{2} + 4 q_{11} q_{21} q_{23}^2 s_{1} s_{2} + 
4 q_{12} q_{22} q_{23}^2 s_{1} s_{2} + 4 q_{13} q_{23}^3 s_{1} s_{2} \\
&+ 4 q_{11} q_{21} q_{31}^2 s_{1} s_{2} + 
4 q_{12} q_{21} q_{31} q_{32} s_{1} s_{2} + 4 q_{11} q_{22} q_{31} q_{32} s_{1} s_{2} + 
4 q_{12} q_{22} q_{32}^2 s_{1} s_{2} + 4 q_{13} q_{21} q_{31} q_{33} s_{1} s_{2} \\
&+ 
4 q_{11} q_{23} q_{31} q_{33} s_{1} s_{2} + 4 q_{13} q_{22} q_{32} q_{33} s_{1} s_{2} + 
4 q_{12} q_{23} q_{32} q_{33} s_{1} s_{2} + 4 q_{13} q_{23} q_{33}^2 s_{1} s_{2} - q_{11}^4 s_{2}^2 - 
2 q_{11}^2 q_{12}^2 s_{2}^2 - q_{12}^4 s_{2}^2\\ 
&- 2 q_{11}^2 q_{13}^2 s_{2}^2
- 2 q_{12}^2 q_{13}^2 s_{2}^2 - q_{13}^4 s_{2}^2 + q_{21}^4 s_{2}^2 + 2 q_{21}^2 q_{22}^2 s_{2}^2 + 
q_{22}^4 s_{2}^2 + 2 q_{21}^2 q_{23}^2 s_{2}^2 + 2 q_{22}^2 q_{23}^2 s_{2}^2 + q_{23}^4 s_{2}^2 - 
2 q_{11}^2 q_{31}^2 s_{2}^2\\
& - q_{31}^4 s_{2}^2 - 4 q_{11} q_{12} q_{31} q_{32} s_{2}^2 - 
2 q_{12}^2 q_{32}^2 s_{2}^2 - 2 q_{31}^2 q_{32}^2 s_{2}^2 - q_{32}^4 s_{2}^2 - 
4 q_{11} q_{13} q_{31} q_{33} s_{2}^2 - 4 q_{12} q_{13} q_{32} q_{33} s_{2}^2 - 
2 q_{13}^2 q_{33}^2 s_{2}^2 \\
&- 2 q_{31}^2 q_{33}^2 s_{2}^2 - 2 q_{32}^2 q_{33}^2 s_{2}^2 - 
q_{33}^4 s_{2}^2 + 4 q_{11}^3 q_{31} s_{1} s_{3} + 4 q_{11} q_{12}^2 q_{31} s_{1} s_{3} + 
4 q_{11} q_{13}^2 q_{31} s_{1} s_{3} + 4 q_{11} q_{21}^2 q_{31} s_{1} s_{3} \\
&+ 
4 q_{12} q_{21} q_{22} q_{31} s_{1} s_{3} + 4 q_{13} q_{21} q_{23} q_{31} s_{1} s_{3} + 
4 q_{11} q_{31}^3 s_{1} s_{3} + 4 q_{11}^2 q_{12} q_{32} s_{1} s_{3} + 4 q_{12}^3 q_{32} s_{1} s_{3} + 
4 q_{12} q_{13}^2 q_{32} s_{1} s_{3} \\
&+ 4 q_{11} q_{21} q_{22} q_{32} s_{1} s_{3} + 
4 q_{12} q_{22}^2 q_{32} s_{1} s_{3} + 4 q_{13} q_{22} q_{23} q_{32} s_{1} s_{3} + 
4 q_{12} q_{31}^2 q_{32} s_{1} s_{3} + 4 q_{11} q_{31} q_{32}^2 s_{1} s_{3} + 4 q_{12} q_{32}^3 s_{1} s_{3} \\
&+ 
4 q_{11}^2 q_{13} q_{33} s_{1} s_{3} + 4 q_{12}^2 q_{13} q_{33} s_{1} s_{3} + 4 q_{13}^3 q_{33} s_{1} s_{3} + 
4 q_{11} q_{21} q_{23} q_{33} s_{1} s_{3} + 4 q_{12} q_{22} q_{23} q_{33} s_{1} s_{3} + 
4 q_{13} q_{23}^2 q_{33} s_{1} s_{3} \\
&+ 4 q_{13} q_{31}^2 q_{33} s_{1} s_{3} + 
4 q_{13} q_{32}^2 q_{33} s_{1} s_{3} + 4 q_{11} q_{31} q_{33}^2 s_{1} s_{3} + 
4 q_{12} q_{32} q_{33}^2 s_{1} s_{3} + 4 q_{13} q_{33}^3 s_{1} s_{3} + 4 q_{11}^2 q_{21} q_{31} s_{2} s_{3} \\
&+ 
4 q_{21}^3 q_{31} s_{2} s_{3} + 4 q_{11} q_{12} q_{22} q_{31} s_{2} s_{3} + 
4 q_{21} q_{22}^2 q_{31} s_{2} s_{3} + 4 q_{11} q_{13} q_{23} q_{31} s_{2} s_{3} + 
4 q_{21} q_{23}^2 q_{31} s_{2} s_{3} + 4 q_{21} q_{31}^3 s_{2} s_{3} \\
&+ 
4 q_{11} q_{12} q_{21} q_{32} s_{2} s_{3} + 4 q_{12}^2 q_{22} q_{32} s_{2} s_{3} + 
4 q_{21}^2 q_{22} q_{32} s_{2} s_{3} + 4 q_{22}^3 q_{32} s_{2} s_{3} + 
4 q_{12} q_{13} q_{23} q_{32} s_{2} s_{3} + 4 q_{22} q_{23}^2 q_{32} s_{2} s_{3} \\
&+ 
4 q_{22} q_{31}^2 q_{32} s_{2} s_{3} + 4 q_{21} q_{31} q_{32}^2 s_{2} s_{3} + 4 q_{22} q_{32}^3 s_{2} s_{3} + 
4 q_{11} q_{13} q_{21} q_{33} s_{2} s_{3} + 4 q_{12} q_{13} q_{22} q_{33} s_{2} s_{3} + 
4 q_{13}^2 q_{23} q_{33} s_{2} s_{3} \\
&+ 4 q_{21}^2 q_{23} q_{33} s_{2} s_{3} + 
4 q_{22}^2 q_{23} q_{33} s_{2} s_{3} + 4 q_{23}^3 q_{33} s_{2} s_{3} + 4 q_{23} q_{31}^2 q_{33} s_{2} s_{3} + 
4 q_{23} q_{32}^2 q_{33} s_{2} s_{3} + 4 q_{21} q_{31} q_{33}^2 s_{2} s_{3} \\
&+ 4 q_{22} q_{32} q_{33}^2 s_{2} s_{3} + 4 q_{23} q_{33}^3 s_{2} s_{3} - q_{11}^4 s_{3}^2 - 
2 q_{11}^2 q_{12}^2 s_{3}^2 - q_{12}^4 s_{3}^2 - 2 q_{11}^2 q_{13}^2 s_{3}^2 - 
2 q_{12}^2 q_{13}^2 s_{3}^2 - q_{13}^4 s_{3}^2 - 2 q_{11}^2 q_{21}^2 s_{3}^2\\
& - q_{21}^4 s_{3}^2 - 
4 q_{11} q_{12} q_{21} q_{22} s_{3}^2 - 2 q_{12}^2 q_{22}^2 s_{3}^2 - 2 q_{21}^2 q_{22}^2 s_{3}^2 - 
q_{22}^4 s_{3}^2 - 4 q_{11} q_{13} q_{21} q_{23} s_{3}^2 - 4 q_{12} q_{13} q_{22} q_{23} s_{3}^2 - 
2 q_{13}^2 q_{23}^2 s_{3}^2\\
& - 2 q_{21}^2 q_{23}^2 s_{3}^2 - 2 q_{22}^2 q_{23}^2 s_{3}^2 - 
q_{23}^4 s_{3}^2 + q_{31}^4 s_{3}^2 + 2 q_{31}^2 q_{32}^2 s_{3}^2 + q_{32}^4 s_{3}^2 + 
2 q_{31}^2 q_{33}^2 s_{3}^2 + 2 q_{32}^2 q_{33}^2 s_{3}^2 + q_{33}^4 s_{3}^2 \, ,
\end{split}
\end{equation}}

{\small \begin{equation}
\begin{split}
F_3&=-q_{13} q_{21} q_{22} s_{1} t_{1}^2 + q_{12} q_{21} q_{23} s_{1} t_{1}^2 - q_{13} q_{31} q_{32} s_{1} t_{1}^2 + 
q_{12} q_{31} q_{33} s_{1} t_{1}^2 + q_{11} q_{13} q_{22} s_{2} t_{1}^2 - q_{11} q_{12} q_{23} s_{2} t_{1}^2\\
& - q_{23} q_{31} q_{32} s_{2} t_{1}^2 + q_{22} q_{31} q_{33} s_{2} t_{1}^2 + q_{11} q_{13} q_{32} s_{3} t_{1}^2 + 
q_{21} q_{23} q_{32} s_{3} t_{1}^2 - q_{11} q_{12} q_{33} s_{3} t_{1}^2 - q_{21} q_{22} q_{33} s_{3} t_{1}^2 \\
&+ 
q_{13} q_{21}^2 s_{1} t_{1} t_{2} - q_{13} q_{22}^2 s_{1} t_{1} t_{2} - q_{11} q_{21} q_{23} s_{1} t_{1} t_{2} + 
q_{12} q_{22} q_{23} s_{1} t_{1} t_{2} + q_{13} q_{31}^2 s_{1} t_{1} t_{2} - q_{13} q_{32}^2 s_{1} t_{1} t_{2} \\
&- 
q_{11} q_{31} q_{33} s_{1} t_{1} t_{2} + q_{12} q_{32} q_{33} s_{1} t_{1} t_{2} - q_{11} q_{13} q_{21} s_{2} t_{1} t_{2} + 
q_{12} q_{13} q_{22} s_{2} t_{1} t_{2} + q_{11}^2 q_{23} s_{2} t_{1} t_{2} - q_{12}^2 q_{23} s_{2} t_{1} t_{2} \\
&+ 
q_{23} q_{31}^2 s_{2} t_{1} t_{2} - q_{23} q_{32}^2 s_{2} t_{1} t_{2} - q_{21} q_{31} q_{33} s_{2} t_{1} t_{2} + 
q_{22} q_{32} q_{33} s_{2} t_{1} t_{2} - q_{11} q_{13} q_{31} s_{3} t_{1} t_{2} - q_{21} q_{23} q_{31} s_{3} t_{1} t_{2} \\
&+ 
q_{12} q_{13} q_{32} s_{3} t_{1} t_{2} + q_{22} q_{23} q_{32} s_{3} t_{1} t_{2} + q_{11}^2 q_{33} s_{3} t_{1} t_{2}- 
q_{12}^2 q_{33} s_{3} t_{1} t_{2} + q_{21}^2 q_{33} s_{3} t_{1} t_{2} - q_{22}^2 q_{33} s_{3} t_{1} t_{2} \\
&+ 
q_{13} q_{21} q_{22} s_{1} t_{2}^2 - q_{11} q_{22} q_{23} s_{1} t_{2}^2 + q_{13} q_{31} q_{32} s_{1} t_{2}^2 - 
q_{11} q_{32} q_{33} s_{1} t_{2}^2 - q_{12} q_{13} q_{21} s_{2} t_{2}^2 + q_{11} q_{12} q_{23} s_{2} t_{2}^2 \\
&+ 
q_{23} q_{31} q_{32} s_{2} t_{2}^2 - q_{21} q_{32} q_{33} s_{2} t_{2}^2 - q_{12} q_{13} q_{31} s_{3} t_{2}^2 - 
q_{22} q_{23} q_{31} s_{3} t_{2}^2 + q_{11} q_{12} q_{33} s_{3} t_{2}^2 + q_{21} q_{22} q_{33} s_{3} t_{2}^2 \\
&- 
q_{12} q_{21}^2 s_{1} t_{1} t_{3} + q_{11} q_{21} q_{22} s_{1} t_{1} t_{3} - q_{13} q_{22} q_{23} s_{1} t_{1} t_{3} + 
q_{12} q_{23}^2 s_{1} t_{1} t_{3} - q_{12} q_{31}^2 s_{1} t_{1} t_{3} + q_{11} q_{31} q_{32} s_{1} t_{1} t_{3} \\
&- 
q_{13} q_{32} q_{33} s_{1} t_{1} t_{3} + q_{12} q_{33}^2 s_{1} t_{1} t_{3} + q_{11} q_{12} q_{21} s_{2} t_{1} t_{3} - 
q_{11}^2 q_{22} s_{2} t_{1} t_{3} + q_{13}^2 q_{22} s_{2} t_{1} t_{3} - q_{12} q_{13} q_{23} s_{2} t_{1} t_{3} \\
&- 
q_{22} q_{31}^2 s_{2} t_{1} t_{3} + q_{21} q_{31} q_{32} s_{2} t_{1} t_{3} - q_{23} q_{32} q_{33} s_{2} t_{1} t_{3} + 
q_{22} q_{33}^2 s_{2} t_{1} t_{3} + q_{11} q_{12} q_{31} s_{3} t_{1} t_{3} + q_{21} q_{22} q_{31} s_{3} t_{1} t_{3}\\
& - q_{11}^2 q_{32} s_{3} t_{1} t_{3} + q_{13}^2 q_{32} s_{3} t_{1} t_{3} - q_{21}^2 q_{32} s_{3} t_{1} t_{3} + 
q_{23}^2 q_{32} s_{3} t_{1} t_{3} - q_{12} q_{13} q_{33} s_{3} t_{1} t_{3} - q_{22} q_{23} q_{33} s_{3} t_{1} t_{3} \\
&- 
q_{12} q_{21} q_{22} s_{1} t_{2} t_{3} + q_{11} q_{22}^2 s_{1} t_{2} t_{3} + q_{13} q_{21} q_{23} s_{1} t_{2} t_{3} - 
q_{11} q_{23}^2 s_{1} t_{2} t_{3} - q_{12} q_{31} q_{32} s_{1} t_{2} t_{3} + q_{11} q_{32}^2 s_{1} t_{2} t_{3} \\
&+ 
q_{13} q_{31} q_{33} s_{1} t_{2} t_{3} - q_{11} q_{33}^2 s_{1} t_{2} t_{3} + q_{12}^2 q_{21} s_{2} t_{2} t_{3} - 
q_{13}^2 q_{21} s_{2} t_{2} t_{3} - q_{11} q_{12} q_{22} s_{2} t_{2} t_{3} + q_{11} q_{13} q_{23} s_{2} t_{2} t_{3}\\
& - 
q_{22} q_{31} q_{32} s_{2} t_{2} t_{3} + q_{21} q_{32}^2 s_{2} t_{2} t_{3} + q_{23} q_{31} q_{33} s_{2} t_{2} t_{3} - 
q_{21} q_{33}^2 s_{2} t_{2} t_{3} + q_{12}^2 q_{31} s_{3} t_{2} t_{3} - q_{13}^2 q_{31} s_{3} t_{2} t_{3} \\
&+ 
q_{22}^2 q_{31} s_{3} t_{2} t_{3} - q_{23}^2 q_{31} s_{3} t_{2} t_{3} - q_{11} q_{12} q_{32} s_{3} t_{2} t_{3} - 
q_{21} q_{22} q_{32} s_{3} t_{2} t_{3} + q_{11} q_{13} q_{33} s_{3} t_{2} t_{3} + q_{21} q_{23} q_{33} s_{3} t_{2} t_{3}\\
& - 
q_{12} q_{21} q_{23} s_{1} t_{3}^2 + q_{11} q_{22} q_{23} s_{1} t_{3}^2 - q_{12} q_{31} q_{33} s_{1} t_{3}^2 + 
q_{11} q_{32} q_{33} s_{1} t_{3}^2 + q_{12} q_{13} q_{21} s_{2} t_{3}^2 - q_{11} q_{13} q_{22} s_{2} t_{3}^2 \\
&- 
q_{22} q_{31} q_{33} s_{2} t_{3}^2 + q_{21} q_{32} q_{33} s_{2} t_{3}^2 + q_{12} q_{13} q_{31} s_{3} t_{3}^2 + 
q_{22} q_{23} q_{31} s_{3} t_{3}^2 - q_{11} q_{13} q_{32} s_{3} t_{3}^2 - q_{21} q_{23} q_{32} s_{3} t_{3}^2 \, ,
\end{split}
\end{equation}}

{\small \begin{equation}
\begin{split}
F_4&=q_{11}^4 t_{1}^2 - q_{12}^4 t_{1}^2 - 2 q_{12}^2 q_{13}^2 t_{1}^2 - q_{13}^4 t_{1}^2 + 
2 q_{11}^2 q_{21}^2 t_{1}^2 + q_{21}^4 t_{1}^2 - 2 q_{12}^2 q_{22}^2 t_{1}^2 - q_{22}^4 t_{1}^2 - 
4 q_{12} q_{13} q_{22} q_{23} t_{1}^2 - 2 q_{13}^2 q_{23}^2 t_{1}^2 \\
&- 2 q_{22}^2 q_{23}^2 t_{1}^2 - 
q_{23}^4 t_{1}^2 + 2 q_{11}^2 q_{31}^2 t_{1}^2 + 2 q_{21}^2 q_{31}^2 t_{1}^2 + q_{31}^4 t_{1}^2 - 
2 q_{12}^2 q_{32}^2 t_{1}^2 - 2 q_{22}^2 q_{32}^2 t_{1}^2 - q_{32}^4 t_{1}^2 - 
4 q_{12} q_{13} q_{32} q_{33} t_{1}^2 \\
&- 4 q_{22} q_{23} q_{32} q_{33} t_{1}^2 - 
2 q_{13}^2 q_{33}^2 t_{1}^2 - 2 q_{23}^2 q_{33}^2 t_{1}^2 - 2 q_{32}^2 q_{33}^2 t_{1}^2 - 
q_{33}^4 t_{1}^2 + 4 q_{11}^3 q_{12} t_{1} t_{2} + 4 q_{11} q_{12}^3 t_{1} t_{2} + 
4 q_{11} q_{12} q_{13}^2 t_{1} t_{2} \\
&+ 4 q_{11} q_{12} q_{21}^2 t_{1} t_{2} + 
4 q_{11}^2 q_{21} q_{22} t_{1} t_{2} + 4 q_{12}^2 q_{21} q_{22} t_{1} t_{2} + 4 q_{21}^3 q_{22} t_{1} t_{2} + 
4 q_{11} q_{12} q_{22}^2 t_{1} t_{2} + 4 q_{21} q_{22}^3 t_{1} t_{2} \\
&+ 
4 q_{12} q_{13} q_{21} q_{23} t_{1} t_{2} + 4 q_{11} q_{13} q_{22} q_{23} t_{1} t_{2} + 
4 q_{21} q_{22} q_{23}^2 t_{1} t_{2} + 4 q_{11} q_{12} q_{31}^2 t_{1} t_{2} + 
4 q_{21} q_{22} q_{31}^2 t_{1} t_{2} + 4 q_{11}^2 q_{31} q_{32} t_{1} t_{2} \\
&+ 
4 q_{12}^2 q_{31} q_{32} t_{1} t_{2} + 4 q_{21}^2 q_{31} q_{32} t_{1} t_{2} + 
4 q_{22}^2 q_{31} q_{32} t_{1} t_{2} + 4 q_{31}^3 q_{32} t_{1} t_{2} + 4 q_{11} q_{12} q_{32}^2 t_{1} t_{2} + 
4 q_{21} q_{22} q_{32}^2 t_{1} t_{2} \\
&+ 4 q_{31} q_{32}^3 t_{1} t_{2} + 4 q_{12} q_{13} q_{31} q_{33} t_{1} t_{2} + 4 q_{22} q_{23} q_{31} q_{33} t_{1} t_{2} + 
4 q_{11} q_{13} q_{32} q_{33} t_{1} t_{2} + 4 q_{21} q_{23} q_{32} q_{33} t_{1} t_{2} + 
4 q_{31} q_{32} q_{33}^2 t_{1} t_{2} \\
&- q_{11}^4 t_{2}^2 + q_{12}^4 t_{2}^2 - 
2 q_{11}^2 q_{13}^2 t_{2}^2 - q_{13}^4 t_{2}^2 - 2 q_{11}^2 q_{21}^2 t_{2}^2 - q_{21}^4 t_{2}^2 + 
2 q_{12}^2 q_{22}^2 t_{2}^2 + q_{22}^4 t_{2}^2 - 4 q_{11} q_{13} q_{21} q_{23} t_{2}^2 - 
2 q_{13}^2 q_{23}^2 t_{2}^2 \\
&- 2 q_{21}^2 q_{23}^2 t_{2}^2 - q_{23}^4 t_{2}^2 - 
2 q_{11}^2 q_{31}^2 t_{2}^2 - 2 q_{21}^2 q_{31}^2 t_{2}^2 - q_{31}^4 t_{2}^2 + 
2 q_{12}^2 q_{32}^2 t_{2}^2 + 2 q_{22}^2 q_{32}^2 t_{2}^2 + q_{32}^4 t_{2}^2 - 
4 q_{11} q_{13} q_{31} q_{33} t_{2}^2\\
& - 4 q_{21} q_{23} q_{31} q_{33} t_{2}^2 - 
2 q_{13}^2 q_{33}^2 t_{2}^2 - 2 q_{23}^2 q_{33}^2 t_{2}^2 - 2 q_{31}^2 q_{33}^2 t_{2}^2 - 
q_{33}^4 t_{2}^2 + 4 q_{11}^3 q_{13} t_{1} t_{3} + 4 q_{11} q_{12}^2 q_{13} t_{1} t_{3} + 
4 q_{11} q_{13}^3 t_{1} t_{3}\\
& + 4 q_{11} q_{13} q_{21}^2 t_{1} t_{3} + 
4 q_{12} q_{13} q_{21} q_{22} t_{1} t_{3} + 4 q_{11}^2 q_{21} q_{23} t_{1} t_{3} + 
4 q_{13}^2 q_{21} q_{23} t_{1} t_{3} + 4 q_{21}^3 q_{23} t_{1} t_{3} + 
4 q_{11} q_{12} q_{22} q_{23} t_{1} t_{3}\\
& + 4 q_{21} q_{22}^2 q_{23} t_{1} t_{3} + 
4 q_{11} q_{13} q_{23}^2 t_{1} t_{3} + 4 q_{21} q_{23}^3 t_{1} t_{3} + 4 q_{11} q_{13} q_{31}^2 t_{1} t_{3} + 
4 q_{21} q_{23} q_{31}^2 t_{1} t_{3} + 4 q_{12} q_{13} q_{31} q_{32} t_{1} t_{3} \\
&+ 
4 q_{22} q_{23} q_{31} q_{32} t_{1} t_{3} + 4 q_{11}^2 q_{31} q_{33} t_{1} t_{3} + 
4 q_{13}^2 q_{31} q_{33} t_{1} t_{3} + 4 q_{21}^2 q_{31} q_{33} t_{1} t_{3} + 
4 q_{23}^2 q_{31} q_{33} t_{1} t_{3} + 4 q_{31}^3 q_{33} t_{1} t_{3} \\
&+ 
4 q_{11} q_{12} q_{32} q_{33} t_{1} t_{3} + 4 q_{21} q_{22} q_{32} q_{33} t_{1} t_{3} + 
4 q_{31} q_{32}^2 q_{33} t_{1} t_{3} + 4 q_{11} q_{13} q_{33}^2 t_{1} t_{3} + 
4 q_{21} q_{23} q_{33}^2 t_{1} t_{3} + 4 q_{31} q_{33}^3 t_{1} t_{3} \\
&+ 4 q_{11}^2 q_{12} q_{13} t_{2} t_{3} + 
4 q_{12}^3 q_{13} t_{2} t_{3} + 4 q_{12} q_{13}^3 t_{2} t_{3} + 4 q_{11} q_{13} q_{21} q_{22} t_{2} t_{3} + 
4 q_{12} q_{13} q_{22}^2 t_{2} t_{3} + 4 q_{11} q_{12} q_{21} q_{23} t_{2} t_{3} \\
&+ 
4 q_{12}^2 q_{22} q_{23} t_{2} t_{3} + 4 q_{13}^2 q_{22} q_{23} t_{2} t_{3} + 
4 q_{21}^2 q_{22} q_{23} t_{2} t_{3} + 4 q_{22}^3 q_{23} t_{2} t_{3} + 4 q_{12} q_{13} q_{23}^2 t_{2} t_{3} + 
4 q_{22} q_{23}^3 t_{2} t_{3} \\
&+ 4 q_{11} q_{13} q_{31} q_{32} t_{2} t_{3} + 
4 q_{21} q_{23} q_{31} q_{32} t_{2} t_{3} + 4 q_{12} q_{13} q_{32}^2 t_{2} t_{3} + 
4 q_{22} q_{23} q_{32}^2 t_{2} t_{3} + 4 q_{11} q_{12} q_{31} q_{33} t_{2} t_{3} \\
&+ 
4 q_{21} q_{22} q_{31} q_{33} t_{2} t_{3} + 4 q_{12}^2 q_{32} q_{33} t_{2} t_{3} + 
4 q_{13}^2 q_{32} q_{33} t_{2} t_{3} + 4 q_{22}^2 q_{32} q_{33} t_{2} t_{3} + 
4 q_{23}^2 q_{32} q_{33} t_{2} t_{3} + 4 q_{31}^2 q_{32} q_{33} t_{2} t_{3} \\
&+ 4 q_{32}^3 q_{33} t_{2} t_{3} + 
4 q_{12} q_{13} q_{33}^2 t_{2} t_{3} + 4 q_{22} q_{23} q_{33}^2 t_{2} t_{3} + 4 q_{32} q_{33}^3 t_{2} t_{3} - 
q_{11}^4 t_{3}^2 - 2 q_{11}^2 q_{12}^2 t_{3}^2 - q_{12}^4 t_{3}^2 + q_{13}^4 t_{3}^2 - 
2 q_{11}^2 q_{21}^2 t_{3}^2\\
& - q_{21}^4 t_{3}^2 - 4 q_{11} q_{12} q_{21} q_{22} t_{3}^2 - 
2 q_{12}^2 q_{22}^2 t_{3}^2 - 2 q_{21}^2 q_{22}^2 t_{3}^2 - q_{22}^4 t_{3}^2 + 
2 q_{13}^2 q_{23}^2 t_{3}^2 + q_{23}^4 t_{3}^2 - 2 q_{11}^2 q_{31}^2 t_{3}^2 - 
2 q_{21}^2 q_{31}^2 t_{3}^2 - q_{31}^4 t_{3}^2 \\
&- 4 q_{11} q_{12} q_{31} q_{32} t_{3}^2 - 
4 q_{21} q_{22} q_{31} q_{32} t_{3}^2 - 2 q_{12}^2 q_{32}^2 t_{3}^2 - 2 q_{22}^2 q_{32}^2 t_{3}^2 - 
2 q_{31}^2 q_{32}^2 t_{3}^2 - q_{32}^4 t_{3}^2 + 2 q_{13}^2 q_{33}^2 t_{3}^2 + 
2 q_{23}^2 q_{33}^2 t_{3}^2 + q_{33}^4 t_{3}^2
\end{split}
\end{equation}}

\bibliographystyle{unsrt}
\bibliography{bibliography.bib}

\end{document}